\documentclass[%
superscriptaddress,
twocolumn,
nofootinbib,
]{revtex4-2}
\usepackage{relsize}

\usepackage{tabularray} %%%
\usepackage{multirow}
\usepackage{rotating}
\usepackage{multirow}

\usepackage{graphicx}% 

\usepackage{dcolumn}% Align table columns on decimal point
\usepackage{bm,bbm}% bold math
\usepackage{physics}
\usepackage{amsmath}
\usepackage{amsthm}
\usepackage{amsfonts}
\usepackage[colorlinks=true,citecolor=blue,linkcolor=magenta]{hyperref}
\usepackage[dvipsnames]{xcolor}
\graphicspath{{Images/}}

\usepackage{algorithm}
\usepackage{algorithmicx}
\usepackage{algpseudocode} 
\newtheorem{theorem}{Theorem}
\newtheorem{corollary}{Corollary} %{corollary}{Corollary}[theorem]
\newtheorem{lemma}{Lemma} %{lemma}[theorem]{Lemma}

\newtheorem{proposition}{Proposition}

\newcommand{\poly}{\operatorname{poly}}

\newcommand{\Ebb}{\mathbb{E}}

\newcommand{\AC}{\mathcal{A}}

\newcommand{\EC}{\mathcal{E}}

\newcommand{\OC}{\mathcal{O}}
\newcommand{\PC}{\mathcal{P}}
\newcommand{\LC}{\mathcal{L}}

\newcommand{\RC}{\mathcal{R}}
\newcommand{\SC}{\mathcal{S}}

\newcommand{\UC}{\mathcal{U}}
\newcommand{\VC}{\mathcal{V}}

\newcommand{\Var}{{\rm Var}}

\renewcommand{\geq}{\geqslant}
\renewcommand{\leq}{\leqslant}

\DeclareMathOperator*{\argmax}{arg\,max}
\DeclareMathOperator*{\argmin}{arg\,min}
\renewcommand{\vec}[1]{\boldsymbol{#1}}  % Bold vectors instead of arrow vectors

\renewcommand{\th}{\theta } %Latex \th = thor n

\newcommand{\thv}{\vec{\theta}}
 
\newcommand{\channel}{\mathcal{E}}

\graphicspath{{Images/}}

\newcommand{\be}{\begin{equation}}
\newcommand{\ee}{\end{equation}}

\newcommand{\1}{\mathbbm{1}}

\newcommand{\vtheta}{\Vec{\theta}}

\renewcommand{\vec}[1]{\boldsymbol{#1}}  % Bold vectors instead of arrow vectors

\usepackage{amssymb}

\newcommand{\uni}{\boldsymbol{\mathcal{D}}} 
\newcommand{\vol}{\boldsymbol{\mathcal{V}}}

\newcommand{\vphi}{\vec{\phi}}

\newcommand{\nparams}{m}%this omega needs to be small m, but i want to be able to check
\newcommand{\mindex}{\mu}
\newcommand{\nHam}{M}
\newcommand{\layerindex}{l}
\newcommand{\paramindex}{p}
\renewcommand{\norm}[1]{\left\| #1 \right\|_{\infty}}

\newcommand{\gap}{\Delta_{\rm gap}}

\usepackage[normalem]{ulem}

\usepackage{cancel}

\usepackage{tikz}

\usepackage{minitoc}
\usepackage{tocloft}
\renewcommand{\arraystretch}{1.5}

\makeatletter
\renewcommand\onecolumngrid{
\do@columngrid{one}{\@ne}
\def\set@footnotewidth{\onecolumngrid}
\def\footnoterule{\kern-6pt\hrule width 1.5in\kern6pt}
}
\renewcommand\twocolumngrid{
        \def\footnoterule{
        \dimen@\skip\footins\divide\dimen@\thr@@
        \kern-\dimen@\hrule width.5in\kern\dimen@}
        \do@columngrid{mlt}{\tw@}
}
\makeatother
\begin{document}
% NECESSARY TO MAKE TOC IN APPENDIX
\doparttoc % Tell to minitoc to generate a toc for the parts
\faketableofcontents % Run a fake tableofcontents command for the partocs
\part{}

\title{A unifying account of warm start guarantees for patches of quantum landscapes}

\author{Hela Mhiri}
\thanks{The first two authors contributed equally to this work.}
\affiliation{Institute of Physics, Ecole Polytechnique F\'{e}d\'{e}rale de Lausanne (EPFL), CH-1015 Lausanne, Switzerland}
\affiliation{Laboratoire d’Informatique de Paris 6, CNRS, Sorbonne Universite, 4 Place Jussieu, 75005 Paris, France}

\author{Ricard Puig}
\thanks{The first two authors contributed equally to this work.}
\affiliation{Institute of Physics, Ecole Polytechnique F\'{e}d\'{e}rale de Lausanne (EPFL), CH-1015 Lausanne, Switzerland}

\author{Sacha Lerch}
\affiliation{Institute of Physics, Ecole Polytechnique F\'{e}d\'{e}rale de Lausanne (EPFL), CH-1015 Lausanne, Switzerland}

\author{Manuel S. Rudolph}
\affiliation{Institute of Physics, Ecole Polytechnique F\'{e}d\'{e}rale de Lausanne (EPFL), CH-1015 Lausanne, Switzerland}

\author{Thiparat Chotibut}
\affiliation{Chula Intelligent and Complex Systems, Department of Physics, Faculty of Science, Chulalongkorn University, Bangkok, Thailand, 10330}

\author{Supanut Thanasilp}
\affiliation{Institute of Physics, Ecole Polytechnique F\'{e}d\'{e}rale de Lausanne (EPFL), CH-1015 Lausanne, Switzerland}
\affiliation{Chula Intelligent and Complex Systems, Department of Physics, Faculty of Science, Chulalongkorn University, Bangkok, Thailand, 10330}

\author{Zo\"{e} Holmes}
\affiliation{Institute of Physics, Ecole Polytechnique F\'{e}d\'{e}rale de Lausanne (EPFL), CH-1015 Lausanne, Switzerland}

\begin{abstract}
Barren plateaus are fundamentally a statement about quantum loss landscapes on average but there can, and generally will, exist patches of barren plateau landscapes with substantial gradients. Previous work has studied certain classes of parameterized quantum circuits and found example regions where gradients vanish at worst polynomially in system size. Here we present a general bound that unifies all these previous cases and that can tackle physically-motivated ans\"atze that could not be analyzed previously. Concretely, we analytically prove a lower-bound on the variance of the loss that can be used to show that in a non-exponentially narrow region around a point with curvature the loss variance cannot decay exponentially fast. This result is complemented by numerics and an upper-bound that suggest that any loss function with a barren plateau will have exponentially vanishing gradients in any constant radius subregion. Our work thus suggests that while there are hopes to be able to warm-start variational quantum algorithms, any initialization strategy that cannot get increasingly close to the region of attraction with increasing problem size is likely inadequate. 

\end{abstract}

\maketitle

\section{Introduction}
There is widespread interest in the potential of using warm starts~\cite{dborin2022matrix, truger2024warm,rudolph2022synergy, goh2023lie,sauvage2021flip, verdon2019learning, okada2023classically, ravi2022cafqa,gibbs2024exploiting, mitarai2022quadratic, tate2021classically, niu2023warm, egger2021warm, wurtz2021fixed, mari2020transfer, wilson2019optimizing, liu2023mitigating,zhou2020quantum, akshay2021parameter, grimsley2022adapt, mele2022avoiding},  to sidestep the barren plateau phenomenon~\cite{mcclean2018barren, larocca2024review} in variational quantum computing~\cite{cerezo2020variationalreview}. The hope lies in the fact that the barren plateaus phenomenon, whereby loss gradients vanish exponentially with problem size, is an average case statement. Thus a landscape can have a barren plateau and yet also have special subregions with substantial gradients. Initializing in these regions offers a chance of potentially training across an otherwise barren landscape. 

A growing body of literature aims to provide gradient guarantees for small subregions (``patches'') of quantum loss landscapes. Examples include small-angle regions around zero for tailored circuits~\cite{wang2023trainability, park2023hamiltonian,park2024hardware, zhang2022escaping, chang2024latent,shi2024avoiding,cao2024exploiting}, also known as ``identity initializations’'~\cite{grant2019initialization}, and perturbative approaches which ensure that one initializes close to a solution~\cite{haug2021optimal, puig2024variational}. Concretely, these guarantees are formulated as lower bounds on loss variances which seek to ensure that the loss differences are large enough to have some hope of training. However, these prior works are limited to specialized ans\"atze and specialized regions that are generally far from those considered for warm starting elsewhere~\cite{dborin2022matrix, truger2024warm,rudolph2022synergy, goh2023lie,sauvage2021flip, verdon2019learning, okada2023classically, ravi2022cafqa,gibbs2024exploiting, mitarai2022quadratic, tate2021classically, niu2023warm, egger2021warm, wurtz2021fixed, mari2020transfer, wilson2019optimizing, liu2023mitigating,zhou2020quantum, akshay2021parameter, grimsley2022adapt, mele2022avoiding}. 

In this work, we provide a unifying framework that both captures and extends these prior analyses. Our core observation is rather simple: prior variance lower bounds for patches of loss landscapes boil down to showing that if there are non-vanishing gradients at a single point then there will also be substantial gradients in a small region around that point. Prior work on small-angle initializations can be viewed through this lens~\cite{wang2023trainability, park2023hamiltonian,park2024hardware, zhang2022escaping, shi2024avoiding, grant2019initialization, haug2021optimal, chang2024latent, puig2024variational,cao2024exploiting}. Namely, they rely implicitly on assuming that some point in the landscape has non-vanishing gradients, to then prove that there will be a region with gradients nearby. 

We extend this account with the equally simple observation that the region of attraction around a minimum corresponds to a patch with guaranteed gradients. This can be attributed to the fact that a well-defined minimum necessarily has a non-exponentially vanishing curvature. Prior work on guarantees for iterative update strategies~\cite{puig2024variational}, as well as those using pre-training strategies to find an approximate solution classically~\cite{rudolph2022synergy, goh2023lie, gibbs2024exploiting, sauvage2021flip, truger2024warm, verdon2019learning,okada2023classically,ravi2022cafqa, mitarai2022quadratic} fall broadly under this umbrella. These approaches would fail if exponential precision in each of the initial parameters were required to initialize within a region of attraction. 

Concretely, we prove that the loss variance in patches with a radius that shrinks as $\order{1/\left(\sqrt{m} \poly(n)\right)}$, where $m$ is the number of independent variational parameters and $n$ is the number of qubits, vanishes at worst polynomially with the system size. This guarantee holds as long as there is some (non-exponentially vanishing) curvature to the landscape.

We argue that this is true in all patch guarantees previously considered in the literature and thus explains these prior results~\cite{wang2023trainability, park2023hamiltonian,park2024hardware, zhang2022escaping, shi2024avoiding, grant2019initialization, haug2021optimal, chang2024latent, puig2024variational, cao2024exploiting} and also generalizes to more complex ans\"{a}tze such as the Unitary Coupled Cluster ansatz, which so far have eluded analysis~\cite{zhou2021quantum, chai2022shortcuts, chandarana2022digitized, vizzuso2024convergence}. In particular, we use our bounds to discuss how correlating parameters can affect the magnitude of loss variances and the width of the gorge.

Finally, our numerical results (which are supported by an analytic argument for certain cases) suggest that any landscape that exhibits a barren plateau, will also exhibit a barren plateau on constant width subregions. Thus while special initializations offer hope of finding substantial gradients on quantum loss landscapes, scaling up the problem size looks likely to demand increasingly precise initialization. We therefore conclude that while warm starting strategies offer some hope for sidestepping the barren plateau phenomenon, for this hope to be realized we will likely need increasingly clever initialization strategies.

\section{Preliminaries}

\paragraph*{Variational Quantum Algorithms (VQAs).}
In this work we will consider loss functions of the form
\begin{equation}\label{eq:loss}
    \LC(\thv)=\Tr[U(\thv)\rho U^\dagger(\thv)O]\,,
\end{equation}
where $\rho$ is an $n$-qubit input state, $O$ is a Hermitian measurement operator, and $U(\thv)$ is a quantum circuit parametrized by an $\nparams$ component vector of trainable parameters $\thv$. By varying the choice in $\rho$, $O$ and $U(\thv)$ this expression can be used to capture the losses used in a broad range of variational quantum algorithms~\cite{cerezo2020variationalreview}. One could also consider losses composed of a sum of terms of the form of Eq.~\eqref{eq:loss} with different initial states and observables, or a loss formulated from such terms and classical post processing. This is particularly common in the context of quantum machine learning~\cite{cerezo2020variationalreview, biamonte2017quantum, perez2020data, nguyen2022atheory, caro2022outofdistribution, gibbs2022dynamical}. While we do not explicitly cover such losses here, much of our conclusions carry over to those settings.

\begin{figure}
    \centering
    \includegraphics[width=0.99\linewidth]{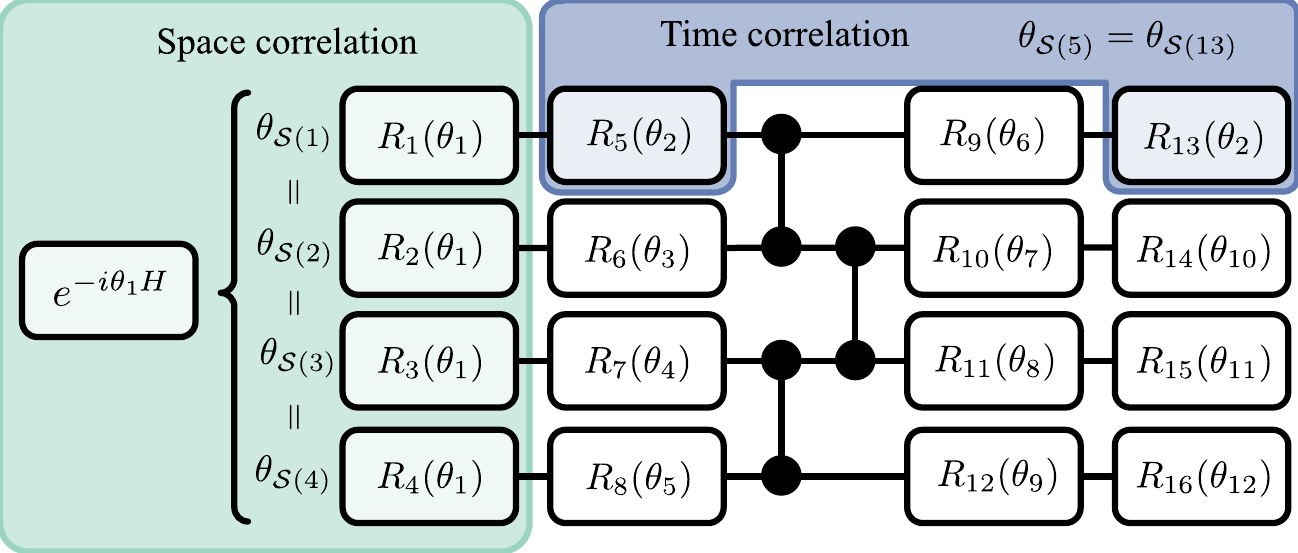}
    \caption{\textbf{Schematic of different types of correlations.} Spatial correlations, where we correlate gates in different layers, are shown in green. These can result from a gate with a generator that acts on multiple qubits with one variational parameter. Time correlations between parameters in different layers are shown in blue. The function $\SC$ maps a gate label to a corresponding parameter label.}
    \label{fig:corr_schematic}
\end{figure}

We consider parameterized quantum circuits of the form
\begin{equation}
\label{eq:circuit}
U(\thv) = \prod_{l=1}^M V_l \, U_l\bigl(\theta_{\SC(l)}\bigr)\,,
\end{equation}
where $\{V_l\}_{l=1}^M$ are fixed unitaries, each $U_l\bigl(\theta_{\SC(l)}\bigr) = e^{-i\theta_{\SC(l)}H_l}$ is a parametrized rotation, $\{H_l\}_{l=1}^M$ are Hermitian generators and $\thv = \{\theta_k\}_{k=1}^m$ is a vector of the \emph{independent} parameters with $m \le M$. Here, $\SC: \{1,\dots,M\} \to \{1,\dots,m\}$ is a map that assigns each generator index $l$ to a parameter index $\SC(l)$, allowing different generators to share the same parameter, i.e. allows for \emph{correlated parameters} as sketched in Fig.~\ref{fig:corr_schematic}. For instance, if all parameters are uncorrelated, we have $m = M$ and $\SC(l)=l$, so each $H_l$ has its own unique angle. If instead all parameters are fully correlated, then $m=1$ and $\SC(l)=1$ for every $l$, so all $H_l$ share the same angle. Throughout this work, ``$m$ independent parameters'' refers to the count after considering all unique indices output by the map $\SC$.

\medskip

\paragraph*{Gradient magnitudes and barren plateaus.}
To successfully train a variational quantum algorithm the loss landscape must exhibit sufficiently large loss gradients or, more generally, loss differences. Chebyshev’s inequality bounds the probability that the
cost value deviates from its average as
\begin{equation}\label{eq:Cheb}
    {\rm Pr}_{\thv\sim\PC}( | \LC(\thv) -  \Ebb_{\thv\sim\PC}[\LC(\thv)] | \geq \delta ) \leq \frac{\text{Var}_{\thv\sim\PC}[\mathcal{L}(\thv)]}{\delta^2} \, ,
\end{equation}
for some $\delta>0$. The variance of the loss is defined as
\begin{equation}\label{eq:vardef}
    \Var_{\vec{\theta}\sim\PC}[\mathcal{L}(\thv)] = \Ebb_{\thv\sim\PC}\left[\LC^2(\thv)\right] -  \left(\Ebb_{\thv\sim\PC}\left[\LC(\thv)\right]\right)^2  \,, 
\end{equation}
where the expectation value is taken over the circuit parameters $\vec{\theta}$ sampled from some distribution $\mathcal{P}$. 
If the loss variance vanishes exponentially, i.e. $  \Var[\mathcal{L}] \in \mathcal{O}(b^{-n})$ with $b> 1$, the probability of observing non-negligible loss differences is exponentially small and the landscape is said to exhibit a \textit{barren plateau}~\cite{mcclean2018barren, larocca2024review,marrero2020entanglement,sharma2020trainability,patti2020entanglement,wang2020noise,arrasmith2021equivalence,larocca2021diagnosing,holmes2021connecting, cerezo2020cost,khatri2019quantum,rudolph2023trainability,kieferova2021quantum,thanaslip2021subtleties,tangpanitanon2020expressibility,holmes2021barren,martin2022barren,fontana2023theadjoint,ragone2023unified, thanasilp2022exponential, letcher2023tight,anschuetz2024unified, chang2024latent,xiong2023fundamental, crognaletti2024estimates,  mao2023barren,deshpande2024dynamic}. On such landscapes, exponentially precise loss evaluations are typically required to navigate towards the global minimum, and hence the resources (shots) required for training are expected to scale exponentially. 

\medskip

\paragraph*{Small \textit{patches} of quantum landscape.}
The variance of the loss in Eq.~\eqref{eq:vardef} depends crucially on the parameter region considered. Much of the analysis on quantum loss landscapes has focused on analyzing the loss over the entire loss landscape. This can be viewed either as providing an ``average case'' trainability analysis for the entire landscape or as quantifying the ability to start training having randomly initialized. This prompts the question of how the properties of the landscape differ in sub-regions corresponding to non-random initializations. 

Motivated by these thoughts, one line of research has focused on small-angle initialization strategies. More concretely, let us define
\begin{equation}\label{eq:hypercube}
	\vol(\vec{\phi}, r) := \{ \vec{\theta} \} \; \; \text{such that} \; \; \theta_i \in [\phi_i -r , \phi_i + r ]  
\end{equation}
for $i = 1, ... \, , m$, 
as the hypercube of parameter space centered around the point $\vec{\phi}$, and define $\uni(\vec{\phi}, r) = \text{Unif}[\vol(\vec{\phi}, r)]$ as the uniform distribution over the hypercube $\vol(\vec{\phi}, r)$. We note that while we focus on uniform distributions over the patch in this work, the extension to other symmetric distributions, such as a Gaussian distribution, is straightforward and similar conclusions are expected to hold.

It was shown in Ref.~\cite{wang2023trainability}, that if the parameters are uniformly sampled in a small hypercube with $r \in \mathcal{O}\left(\frac{1}{\sqrt{L}}\right)$ around $\vec{\phi} = \vec{0}$ for a particular hardware efficient architecture with $L$ being the number of layers, then the variance decays only polynomially with the depth of the circuit,
\begin{equation}
    \Var_{\thv\sim\uni(\vec{0}, r)}[\mathcal{L}(\thv)] \in \Omega\left(\frac{1}{\text{poly}(L)}\right) \, .
\end{equation}
Similar conclusions were reached for the Hamiltonian Variational Ansatz in Refs.~\cite{park2023hamiltonian,cao2024exploiting} and for Gaussian initializations in Refs.~\cite{zhang2022escaping,chang2024latent,shi2024avoiding}. In these cases the small-angle initialization corresponds to initializing close to either identity or a Clifford circuit.

\begin{figure*}
    \centering
    \includegraphics[width=0.95\linewidth]{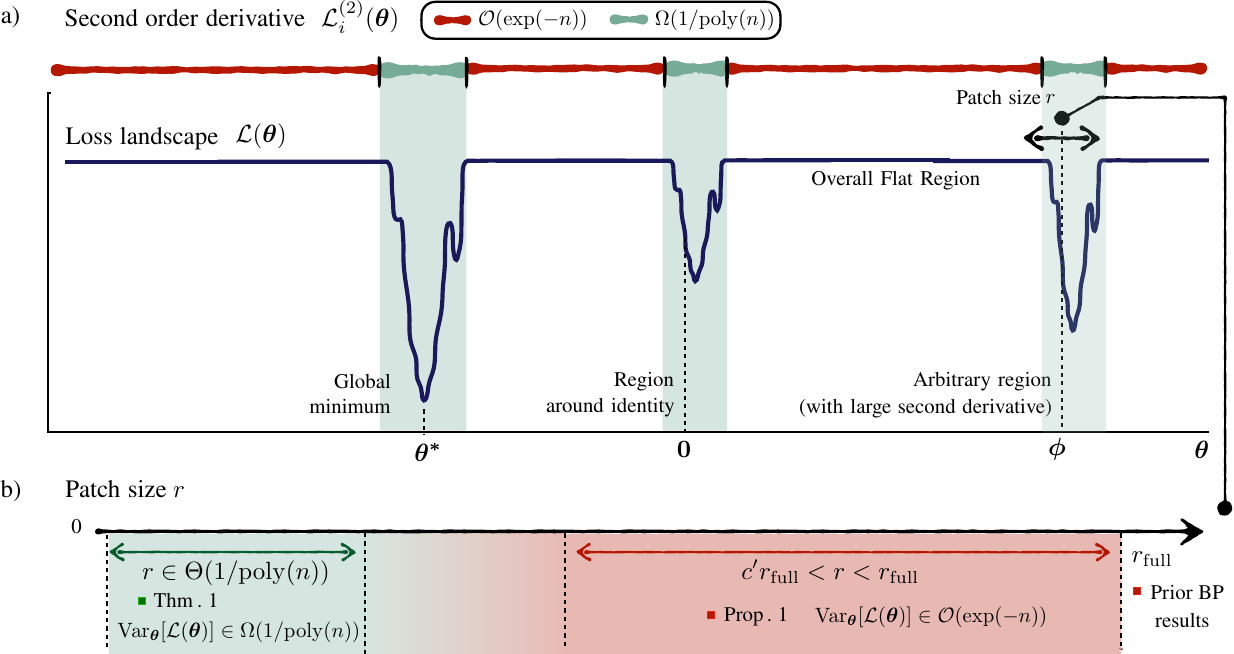}
    \caption{\textbf{Schematic summary of our main results.}     
    In panel (a), the solid blue curve sketches a generic loss landscape $\LC(\thv)$ and the horizontal line above its second order derivative $\LC^{(2)}_i(\thv)$.  The red sections indicate vanishing second derivatives and the green sections represent curvature-rich regions with non-negligible second derivatives which can, for example, occur near a  global minimum $\thv^*$, around identity $\mathbf{0}$ or simply around arbitrary local minima $\vphi$. Panel (b) shows the variance of $\LC(\thv)$ for uniformly sampled parameters $\thv$ in a hypercube of width $2r$ centered around any point with a non-negligible curvature. On the full landscape $(r = r_{\mathrm{full}})$, prior works (e.g., Ref.~\cite{cerezo2020cost}) show that certain families of circuits exhibit a barren plateau.
    Our Proposition~\ref{prop:upperbound}, which strictly only applies to a restricted family of circuits, formalizes that \emph{any} subregion with the patch's size $r \in (c' r_{\rm full}, r_{\rm full})$ with some constant $c' < 1$ still inherits an exponentially vanishing loss variance over that region.
    By contrast, the main Theorem~\ref{th:var} states that for polynomial depth circuits, even if $r$ shrinks \emph{no faster} than $1/\mathrm{poly}(n)$, the corresponding patch (green) still supports non-exponentially vanishing variance.}
    \label{fig:schematic} 
\end{figure*}

However, these small-angle regions considered above are either low-entanglement regions~\cite{park2023hamiltonian,cao2024exploiting}, and therefore can possibly be simulated using tensor network methods~\cite{orus2014practical,orus2019tensor,pan2022simulation}, or low magic regions~\cite{wang2023trainability, park2024hardware,zhang2022escaping,shi2024avoiding,chang2024latent}, and so can be simulated using Clifford perturbation methods~\cite{lerch2024efficient, nemkov2023fourier, beguvsic2024fast, beguvsic2023simulating}. Moreover, these usually require that the circuit has non-vanishing gradient at $\thv=\vec{0}$ (which is not always the case) and a good solution may be far from these rather arbitrarily chosen regions. Thus these methods can (in full generality) only work on a vanishing fraction of problem instances~\cite{nietner2023unifying}.

Recently, Ref.~\cite{puig2024variational} explored warm starts for a family of variational quantum algorithms for quantum evolution in real and imaginary time~\cite{otten2019noise, benedetti2020hardware, barison2021efficient, lin2021real, berthusen2022quantum, haug2021optimal, gentinetta2023overhead}. In line with previous work on small-angle initializations, Ref.~\cite{puig2024variational} proves that the variance of the loss vanishes at worst polynomially in a hypercube with $r_{\rm patch} \in \Theta\left(\frac{1}{\sqrt{\nparams}} \right)$. However, in contrast to the small-angle guarantees discussed above, the patches considered here are not necessarily in a low-entanglement or low-magic region of the landscape. Rather, the iterative approach of these algorithms implies that one initializes close to a good minimum with guaranteed gradients but this region can be arbitrary and therefore could correspond to high-entanglement or high-magic circuits. Nonetheless, these guarantees are only valid for a particular VQA protocol and cannot be applied directly to most conventional VQAs.

\section{Main results}

\subsection{Overview of the analysis}
We provide a general framework for computing variance lower bounds for sub-regions of quantum loss landscapes corresponding to different initialization strategies.
The core intuition underlying our results, outlined in Section~\ref{sec:approxbound} and sketched in Fig.~\ref{fig:schematic}, is that to identify a patch with guaranteed gradients it suffices to find a single point with non-vanishing curvature. It then follows from the smoothness of typical quantum loss landscapes, i.e., the fact that their first and second order derivatives are bounded, that there will be non-vanishing gradients in the region around that point.  While not framed in these terms, we argue that all prior lower bound variance guarantees for small-angle initializations implicitly relied on this simple observation~\cite{wang2023trainability, park2023hamiltonian,park2024hardware, zhang2022escaping, shi2024avoiding, grant2019initialization, haug2021optimal, chang2024latent,puig2024variational, cao2024exploiting}. 

Our first main result is Theorem~\ref{th:var} in Section~\ref{sec:MainBound}, which captures all of the prior bounds outlined in the previous section (Refs.~\cite{wang2023trainability, park2023hamiltonian,park2024hardware, zhang2022escaping, shi2024avoiding, grant2019initialization, haug2021optimal, chang2024latent,puig2024variational, cao2024exploiting}) and applies to problem-inspired circuits that have no existing analytical results. In particular, our bound applies to all circuits of the form of Eq.~\eqref{eq:circuit} with or without correlated parameters. Moreover, Theorem~\ref{th:var} can be used to describe the variance of a patch centered around any point of the loss landscape. This generality is achieved by abstractifying away from the underlying circuits and treating the loss as an arbitrary function that can be Taylor expanded.

In Sec.~\ref{sec:minimum} we further discuss how Theorem~\ref{th:var} can be used to characterize i. small-angle initializations and ii. regions of attraction. For the former, we show how we can capture the prior small-angle guarantees proposed in Refs.~\cite{wang2023trainability, park2023hamiltonian,park2024hardware, zhang2022escaping, shi2024avoiding, grant2019initialization, haug2021optimal, chang2024latent,puig2024variational, cao2024exploiting, cao2024exploiting}. We further discuss limitations of these approaches, stressing that they are only successful if gradients can be guaranteed at identity - which for many families of states, circuits and observables will not be the case. For the latter, Corollary~\ref{cor:var_minimum} shows that, under mild assumptions, the width of a region of attraction shrinks at worst polynomially. Thus, crucially, our work lends hope to the potential of ``warm-starting'' variational quantum algorithms as exponential precision in each parameter is not required to initialize in a region of attraction.  

In Sec.~\ref{sec:fourierfreq} we provide an intuitive description of the properties of the loss landscape and the Fourier frequencies of the loss function. In particular, we introduce the notion of \textit{maximal frequencies} and show that  the width of patches with guaranteed gradients is inversely proportional to these frequencies. We then further prove that this bound can be tightened when the parameters are not correlated in time.

In Sec.~\ref{sec:architectures_main} we show how Theorem~\ref{th:var} can be used to re-derive bounds for each of the previous case-by-case analyses~\cite{wang2023trainability, park2023hamiltonian,park2024hardware, zhang2022escaping, shi2024avoiding, grant2019initialization, haug2021optimal, chang2024latent, puig2024variational, cao2024exploiting} and to push beyond these cases to problem-inspired ans\"atze, such as the unitary coupled cluster ansatz~\cite{mao2023barren, zhou2021quantum, arrazola2022universal, chai2022shortcuts, chandarana2022digitized, vizzuso2024convergence}. We further emphasize how our bounds allow us to move away from small-angle initializations that are close to identity or a Clifford circuit to other relevant points in the landscape. In particular, we use our analysis to understand how the region of attraction around  $\thv = \vec{0}$ depends on properties of the ansatz class.  

Finally, our positive findings are counterbalanced by the observation that regions of attraction do shrink with increasing problem size. Namely, in all cases numerically studied, the width of the region with guaranteed gradients decreases inverse polynomially with the number of trainable parameters $m$. This finding is further supported by Proposition~\ref{prop:upperbound} that shows that any landscape that has a barren plateau for a circuit where the number of parameters scales linearly with system size, $m \in \Theta(n)$,  will also have a barren plateau in a wide range of constant-width landscape patches. Thus, initializing close to a region of attraction necessarily becomes harder with increasing problem sizes. Our main results are summarized in Fig.~\ref{fig:schematic}. 

\subsection{Warm up: approximate variance}\label{sec:approxbound}

We start with an approximate expression for the variance to outline the core intuition underlying our analysis. The core idea is that to identify a patch with guaranteed gradients it suffices to find a single point in the landscape with a non-vanishing first or second derivative. 

To make this idea a little (but not much) more formal, let us consider a differentiable function $\LC(\thv)$ that depends on a vector $\thv$ of size $\nparams$. Next we write the Taylor expansion of the loss around a point $\vec{\phi}$ as 
\begin{align}
    \LC(\thv)  =&\; \LC(\vec{\phi}) + \sum_{i=1}^\nparams \LC^{(1)}_i (\theta_i - \phi_i) \nonumber \\
     &+ \sum_{i,j=1}^\nparams \frac{1}{2} \LC^{(2)}_{ij} (\theta_i-\phi_i) (\theta_j - \phi_j) + {\Phi} \;,
\end{align}
where we defined partial derivatives evaluated at $\vec{\phi}$ as
\begin{equation}
    \LC^{(k)}_{i_1i_2...i_k} = \left( \frac{\partial^k  \LC(\thv)}{\partial \theta_{i_1} \partial \theta_{i_2} ... \partial \theta_{i_k}} \right)\bigg|_{\vec{\theta} = \vec{\phi}}\;,
\end{equation}
and $\Phi$ collects higher order terms.
We can use this expression to compute the variance over $\uni(\vec{\phi}, r)$, the uniform distribution over a hypercube of width $2r$ around a point $\vec{\phi}$, up to the fourth order in $r$ (see Appendix~\ref{app:approximate-variance-bound}). On doing so, we find that the variance can be approximated as
\begin{align}\label{eq:approxvar}
    \Var_{\thv\sim\uni(\vec{\phi}, r)}[\LC(\thv)]&=  \; \frac{r^2}{3}\sum_{i=1}^m \LC^{(1) \, 2}_i \nonumber \\ \nonumber
    &+ 
    \frac{r^4}{9}\sum_{i,j=1}^m\LC_i^{(1)}\LC_{ijj}^{(3)}\left( 1-\frac{2\delta_{ij}}{5} \right)\\ %\nonumber
    &+\frac{r^4}{18}\sum_{i,j=1}^m\LC_{ij}^{(2) \,2} \left( 1-\frac{3\delta_{ij}}{5} \right) + \mathcal{R}\;,
\end{align}
where $\mathcal{R}\in\OC(r^6)$. Thus we see that if there exists some curvature in a region such that either $\LC_i^{(1)}$ or $\LC_{ij}^{(2)}$ is non-vanishing then we will generally have substantial gradients in that region. 

However, as the \textit{approximate} expression in Eq.~\eqref{eq:approxvar} only holds for small deviations from $\vec{\phi}$, it only gives good approximation for small $r$. Thus it cannot be used to rigorously quantify the scaling of regions of attraction. To address this, in the following section we develop an exact general lower-bound for the variance of a general class of quantum circuits. As in practice any non-vanishing $\LC_i^{(1)}$ will be accompanied by a non-vanishing $\LC_{ij}^{(2)}$ at most points on the landscape, in our formal bounds we will work with the assumption that at least some $\LC_{ij}^{(2)}$ are not 
exponentially vanishing. 

\subsection{General variance lower bounds}\label{sec:MainBound}
In this section we present our most general result: a lower bound on the variance of a loss $\LC(\thv)$ for a uniformly sampled hypercube of width $2r$ around an arbitrary point $\vec{\phi}$ i.e., $\vol(\vec{\phi}, r)$. Crucially, gradients are ensured via the requirement that the second derivatives around $\vec{\phi}$ are not exponentially small. This condition will allow us to derive further guarantees around special points in the landscape that we know satisfy this condition. In particular, it will enable us to analyze small angle regions usually studied in the previous works~\cite{wang2023trainability, park2023hamiltonian,park2024hardware, zhang2022escaping, shi2024avoiding, cao2024exploiting} and further study the properties of regions of attraction around minima.  

\begin{theorem}[Lower bound on the loss variance, Informal]\label{th:var}
Consider a generic loss $\LC(\thv)$ of the form in Eq.~\eqref{eq:loss} and a parametrized quantum circuit $U(\thv)$ of the form in Eq.~\eqref{eq:circuit}. We consider uniformly sampling parameters in a hypercube of width $2r$ around any point of the landscape $\vec{\phi}$ as in Eq.~\eqref{eq:hypercube}. Then as long as there exists at least one parameter $\th_p$ such that the second derivative at $\vec{\phi}$ is at worst polynomially vanishing with the number of qubits, 
\begin{align}\label{eq:second_deriv_cond}
\left| \left.\left(\frac{\partial^2\mathcal{L}(\thv)}{\partial\th_{p}^2}\right)\right|_{\thv=\vec{\phi}} \right| \in\Omega\left( \frac{1}{{\rm poly}(n)} \right) \, ,
\end{align}
we can find a region with $r_{\rm patch}$ where,
\begin{align}\label{eq:region_cond}  r_{\rm patch}\in\Theta\left(\frac{1}{\sqrt{\nparams}\cdot{\rm poly}(n)}\right) \, ,
\end{align}
such that $\forall\, r \leq r_{\rm patch}$
\begin{align}
    \Var_{\vtheta \sim \uni(\vec{\phi}, r)} \left[ \LC (\vtheta)\right] \in \Omega\left(r^4\right) \; .
\end{align}
\end{theorem}
\noindent The formal version of Theorem~\ref{th:var} is presented in Theorem~\ref{th:var_formal} in Appendix~\ref{appendix:warm_up_proof_th1} for the uncorrelated case and Theorem~\ref{th:var_formal_cor} in Appendix~\ref{appendix:lower_bound_correlated} for the correlated case.

Theorem~\ref{th:var} establishes that as long as there exist points on a quantum loss landscape where the second derivatives decrease at worst polynomially in $n$, then there exist regions around these points where gradients vanish at worst polynomially in the number of qubits, $n$, and the number of parameters, $m$.  Consequently, for $m \in \OC(\poly(n))$ the loss variance vanishes at worst polynomially in $n$ in these regions. We denote the width of regions with guaranteed polynomially vanishing variance by $r_{\rm patch}$.

In this crude form, our argument may appear somewhat circular and borderline trivial. In fact, the formal proof, as shown in Appendix~\ref{appendix:lower_bound_correlated}, was rather cumbersome and relied on a careful analysis of the average Taylor remainder together with some non-trivial operator bounds. More importantly, Theorem~\ref{th:var} acts as a skeleton from which we can derive more concrete results for specific families of circuits, observables, and patches by showing that the second derivative condition in Eq.~\eqref{eq:second_deriv_cond} does hold in those cases. 
However, before we get to specific problem classes in Section~\ref{sec:architectures_main}, let us highlight some general characteristics of quantum loss landscapes that are elucidated by our analysis.

\medskip

\paragraph*{Characterizing small-angle initializations.}\label{sec:IdentityInitialization}  
Theorem~\ref{th:var} is very general and can be applied to any point in a quantum loss landscape with guaranteed curvature but some points are more natural to study than others. The first family of points we will consider are so-called small-angle initializations whereby a circuit is initialized in a small region around the all zero parameter vector. In the context of Theorem~\ref{th:var}, the success of these approaches can be attributed to the fact that there are a variety of circuits that exhibit gradients/curvature around zero~\cite{wang2023trainability, park2023hamiltonian,park2024hardware, zhang2022escaping, shi2024avoiding, grant2019initialization, cao2024exploiting}. However, not all possible circuits and loss function combinations will necessarily have gradients around zero. We provide an example of this in Appendix~\ref{appx:iden-initialization}. This means that the characterization of gradients around zero will have to be circuit dependent and, in contrast to suggestions in the literature~\cite{grant2019initialization}, small-angle initializations do not in general guarantee gradients. We will explore small-angle initializations more concretely in Section~\ref{sec:architectures_main} where we apply Theorem~\ref{th:var} to different families of circuits. 
% Concretely, we will show how Theorem~\ref{th:var} can be used to re-derive prior identity initialization guarantees. 

\medskip

\paragraph*{Characterizing the region of attraction.}\label{sec:minimum}
The second set of points we will consider are those around the global minimum. There has been much interest in the potential of ``warm starting'' quantum algorithms such that they are sufficiently close to the solution that it is possible to train. Most of this work so far has been heuristically driven by numerical simulations~\cite{rudolph2022synergy, goh2023lie, gibbs2024exploiting, sauvage2021flip, truger2024warm, verdon2019learning,okada2023classically,ravi2022cafqa, mitarai2022quadratic}. To investigate the potential of such strategies at larger system sizes we apply Theorem~\ref{th:var} to bound the region of attraction surrounding a global minimum for a given circuit $U(\thv)$. This is captured via Corollary~\ref{cor:var_minimum}.

\begin{corollary}[Scaling of regions of attraction, Informal]\label{cor:var_minimum}
Consider a generic loss $\LC(\thv)$ of the form in Eq.~\eqref{eq:loss} and a parametrized quantum circuit $U(\thv)$ of the form in Eq~\eqref{eq:circuit}. Furthermore, denote $\thv^*$ as the parameter corresponding to the global minimum $\thv^* = \argmin_{\thv} \LC(\thv)$. Assume that the  fidelity between the parametrized state and the ground state of $O$ is at  least $1-|\epsilon|^2$, with $|\epsilon|\in\order{1/{\rm poly}(n)}$. Further assume that $O$ has a non-degenerate ground state with a gap $\gap$ to the first excited state that scales as $\gap \in\Omega(1/{\rm poly}(n))$. Then there exists a region around the minimum with
\begin{align}
    r_{\rm patch}\in\Theta\left( \frac{1}{\sqrt{\nparams} \, {\rm poly}(n)} \right)
\end{align}
such that the variance is not exponentially vanishing, as long as $\Tr[\rho H_{\rm 1}^2] - \Tr[\rho H_{\rm 1}]^2\in\Omega\left( 1/{\rm poly}(n)\right)$, where $H_{\rm 1}$ is the generator closest to the state, and its corresponding parameter $\theta_1$ is not correlated with other parameters in the circuit. 
\end{corollary}

Corollary~\ref{cor:var_minimum} allows us to characterize the region around the global minimum. In particular, it establishes that there exists a hypercube with a non-exponentially vanishing width around the solution in which the loss landscape has non-exponentially vanishing gradients under mild assumptions. A formal version of Corollary~\ref{cor:var_minimum} can be found in Appendix ~\ref{subsec:corollary-1-formal-app} followed by its proof in Appendix~\ref{subsec:var_minimum_app}.

The assumption on the observable $O$, i.e. $\gap\in\Omega\left(1/{\rm poly}(n)\right)$, comes from the need to be able to distinguish the target space from others. Indeed, trying to discriminate between two exponentially close subspaces is a hard task. 
However, we expect that in practice, having a degenerate ground-state space will not make it harder to reach the minimum because the ``target'' subspace is larger when the ground state is degenerate.

It is worth recalling that not every choice of circuit leads to a polynomial vanishing variance in a region around a global solution. Indeed, an extreme example of such an adversarial circuit is one that completely commutes with the observable or the state. In that case, the cost function is fixed for all $\thv$ values and its variance trivially  becomes zero. Such adversarial circuits are ruled out by the requirement that $H_1$, the gate closest to the observable, has a non-trivial effect on the initial state. In particular, we need the variance of the generator $H_1$ with respect to the initial state to be non exponentially vanishing, i.e. $\Tr[\rho H_1^2] - \Tr[\rho H_1]^2\in\Omega\left(1/{\rm poly}(n)\right)$. Furthermore, note  that this is a condition on the initial state $\rho$ and the first gate on the circuit. Therefore, in most instances this can be trivially checked in practice. 

Corollary~\ref{cor:var_minimum}, put simply, thus tells us that the width of regions of attraction does not vanish exponentially and thus there is (potentially!) some hope of warm starting variational quantum algorithms using pre-training and smart initialization strategies. Importantly, our results here hold for practically all parameterized circuits currently under consideration by the community, from unstructured circuits on arbitrary topologies~\cite{kandala2017hardware} to problem-inspired circuits where the parameters are highly correlated~\cite{zhou2021quantum, chai2022shortcuts, chandarana2022digitized, vizzuso2024convergence}. 

\subsection{The role of Fourier frequencies}\label{sec:fourierfreq}

\begin{figure}
    \centering
    \includegraphics[width=0.99\linewidth]{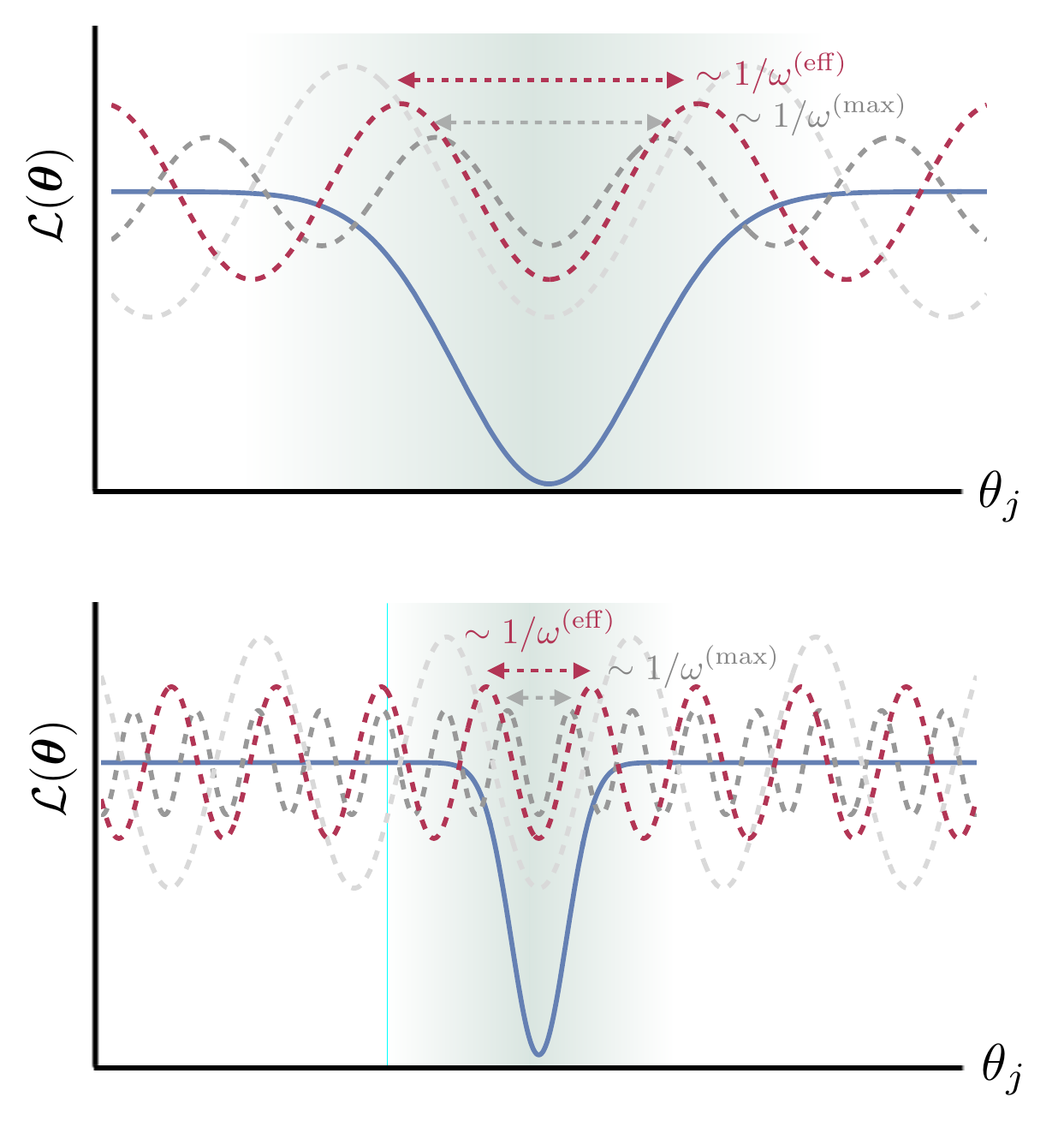}
    \caption{\textbf{Role of Fourier frequencies.} Here we sketch the Fourier decomposition of the loss $\LC(\thv)$ with respect to a parameter $\theta_j$. Intuitively, the width of patches with gradients depend inversely on the magnitude of the frequencies in the Fourier decomposition. 
    In cases where many high frequencies are present both the maximum and effective (dominant) frequencies are high and the minima tend to be narrower compared to when those frequencies are low.
    Note that this figure is merely to be understood as an illustration of the frequencies and the role they play in the loss function. 
    }
    \label{fig:fourierdecomp}
\end{figure}

The scaling of $r_{\rm patch}$, the width of any patch of a quantum landscape with guaranteed polynomial loss variances, depends strongly on the frequencies that appear in the Fourier decomposition of the loss. Concretely, as highlighted originally in Ref.~\cite{schuld2021effect} and shown in full generality in Appendix~\ref{App:Fourier_expansion}, any loss function of the form Eq.~\eqref{eq:loss} with a parameterized quantum circuit of the form Eq.~\eqref{eq:circuit} can be written as
\begin{equation}\label{eq:fourierloss}
    \LC(\thv) = \sum_{\boldsymbol{\omega} \in \Omega_1 \times \dots \times \Omega_\nHam} e^{-i \thv^T \boldsymbol{\omega}} a_{\boldsymbol{\omega}},
\end{equation}
where the components $\omega_l$ of the Fourier frequency vector $\boldsymbol{\omega}$  are given by the differences of the eigenvalues of $H_l$. Precisely, the Fourier frequencies are defined as $\omega_l \in \Omega_l:= \{ \lambda_i^{(l)}-\lambda_j^{(l)} ,\; \forall i,j \in [1,2^n]\}$ where $\{\lambda^{(l)}_i\}_i$ is the set of eigenvalues of $H_l$, and $\Omega_l$ represents the spectrum of distinct frequencies associated to a generator $H_l$. We refer the reader to Appendix~\ref{App:Fourier_expansion} for a formal definition of the Fourier coefficient $a_{\boldsymbol{\omega}}$. Intuitively, as sketched in Fig.~\ref{fig:fourierdecomp}, functions with high frequency terms can have narrower minima whereas functions with only lower frequency terms have broader minima.

Our analysis indicates that this is true also for patches of quantum losses. In particular, we find that $r_{\rm patch}$ is roughly inversely proportional to the sum of the maximal frequencies $\omega^{(\max)}_{j}$ associated with each trainable parameter $\theta_j$. Concretely, for a set of (potentially correlated) gate parameters $\{e^{-i\theta_{j}H_{l}}\}_{l\in \mathcal{S}^{-1}(j)}$, where $\mathcal{S}^{-1}(j)$ is an inverse map of $\SC$ and represents all the generators with the same parameter $\theta_j$, the \textit{maximal frequency} associated with $\theta_j$ is defined as
\begin{equation}\label{eq:max_freqs}
\omega^{(\max)}_{j} := \sum_{l\in\mathcal{S}^{-1}(j)} \omega^{(\max)}(H_l) \, , 
\end{equation}
where $\omega^{(\max)}(H_l) = \max(\Omega_l)$
denotes the maximal Fourier frequency component corresponding to the generator $H_l$. As proven in Appendix~\ref{appendix:lower_bound_correlated}, we then find that 
\begin{equation}\label{eq:r_patch_maxfreq}
    r^2_{\rm patch} \propto \frac{1}{\sum_{j=1}^m  (\omega^{(\max)}_{j})^2} \, ,
\end{equation}
where the sum is taken over the indices of each of the $m$ independent trainable parameters. 

We thus see that there are two key elements that affect the maximal frequencies and thereby the scaling of $r_{\rm patch}$. The first factor is the frequencies, $\omega^{(\max)}(H_l)$, associated with each of the generators $H_l$. The second factor is whether different parameters are correlated.
In particular, as the frequencies of the generators will be greater than $1$, it follows from a simple application of the triangle inequality to the squared sum in Eq.~\eqref{eq:r_patch_maxfreq} that, for a fixed number of generators $M$, correlating parameters will in general increase the sum of the maximal frequencies and so reduce $r_{\rm patch}$.

A tighter estimate on the scaling of $r_{\rm patch}$ can be obtained when we allow spatial correlations, i.e., correlations between parameters of qubits within a layer, but not time correlations, i.e., correlations between layers, as sketched in Fig.~\ref{fig:corr_schematic}.  This is done by introducing the notion of \textit{effective frequencies}. As sketched in Fig.~\ref{fig:fourierdecomp}, these capture the idea that the size of any minimum will depend more closely on the magnitude of the dominant frequencies rather than the largest possible frequency. Concretely, the effective frequencies are defined as 
\begin{equation}\label{eq:ef_freq}
    \omega^{\rm (eff)}_{j}(\vec{\phi}) = \sqrt{\norm{  \left.\frac{\partial^2}{\partial\th_j^2}\left[ U(\thv)OU^\dagger(\thv) \right]\right|_{\thv=\vec{\phi}} }} \, .
\end{equation}
While it may not be obvious from this expression, in Appendix~\ref{App:Fourier_expansion} we show that $\omega^{\rm (eff)}_{j}$ can be thought of as effectively a weighted sum of the dominant frequencies in the Fourier decomposition of the loss given in Eq.~\eqref{eq:fourierloss}. The effective frequencies are upper bounded by the maximal frequencies with $\omega^{\rm (eff)}_{j} \leq 2 \sqrt{\norm{O}}\omega^{(\rm max)}_{j}$.

As with the maximal frequencies, the effective frequencies determine the width of the gorge. In broad terms, as shown explicitly in Appendix~\ref{appendix:warm_up_proof_th1}, we have 
\begin{equation}
    r^2_{\rm patch} \propto \frac{1}{\sum_{j=1}^m (\omega^{\rm (eff)}_{j}(\vphi))^2 }\, ,
\end{equation}
such that the larger the effective frequencies are, the narrower any patch with guaranteed polynomial gradients becomes. 

Crucially, in contrast to the maximal frequencies, it is clear from the form of Eq.~\eqref{eq:ef_freq} that effective frequencies are  dependent both on properties of the observable and the light-cone of the circuit. Thus the notion of effective frequencies can be used, for example, to explain the effect of local versus global observables on the scaling of regions of attraction. In particular, the effective frequency $ \omega^{\rm (eff)}_{j}$, associated with a parameter $\theta_j$, vanishes if the back-propagated observable (i.e., the observable $O$ which is Heisenberg evolved up to the location of the $\theta_j$ gate) commutes with the generator $H_j$. It follows that the effective frequencies of local observables are substantially reduced compared to those of global observables and so $r_{\rm patch}$ is larger. We explore this effect in more detail in the next section. 

\begin{figure}
        \centering
    \includegraphics[width=0.85\linewidth]{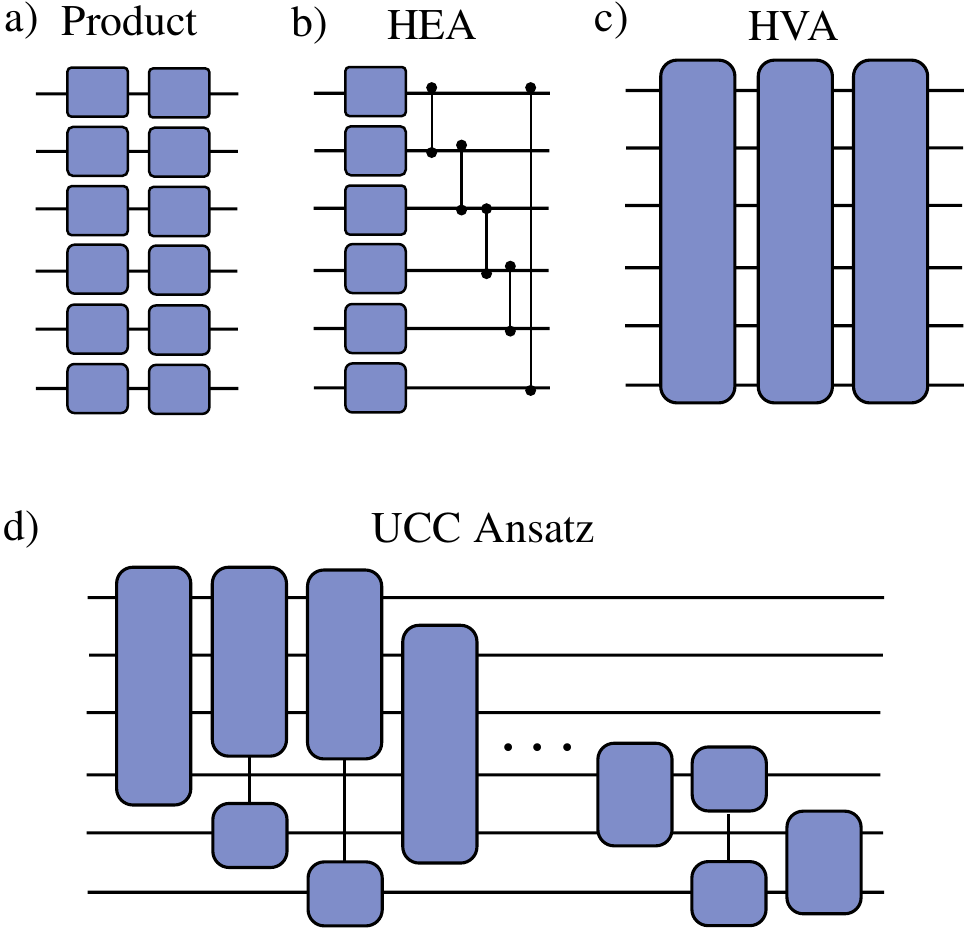} % Adjust the filename and size
    \caption{\textbf{Example circuit architectures schematic.} Here we sketch the structure of the four families of circuits we analyze: a) a tensor product ansatz, b) a hardware efficient ansatz (HEA), c) the Hamiltonian variational ansatz (HVA) and d) the unitary coupled cluster (UCC) ansatz.}
    \label{fig:architectures}
\end{figure}

\subsection{Applications to different architectures}\label{sec:architectures_main}
In this section, we refine Theorem~\ref{th:var} to different circuits to highlight the generality of our theoretical bound and to provide more tailored results for the ans\"atze of interest. As sketched in Fig.~\ref{fig:architectures}, we consider four different circuit families of ans\"atze comprised of i.~a tensor product example, ii.~a Hardware Efficient Ansatz (HEA)~\cite{kandala2017hardware}, iii.~a Hamiltonian Variational Ansatz (HVA)~\cite{wecker2015progress,mele2022avoiding} and the Quantum Approximate Optimization Algorithm (QAOA)~\cite{farhi2014quantum,zhou2020quantum, akshay2021parameter}, and lastly iv.~Unitary Coupled Cluster (UCC) ansatz~\cite{mao2023barren, zhou2021quantum, arrazola2022universal, chai2022shortcuts, chandarana2022digitized, vizzuso2024convergence}. In particular, we apply the formal version of Theorem~\ref{th:var} to analyze the scalings of the patch's width and the loss around $\vphi = \vec{0}$. Our bounds are summarized in Table~\ref{table:scalings}. More details can be found in Appendix~\ref{appendix:architectures}, including a general recipe to compute the scalings for an arbitrary circuit of interest.

To accompany our theoretical analysis, we numerically study the maximal variance,
\begin{equation}\label{eq:varmax}
    \Var_{\rm max} :=   \max_{r} \Var_{\vtheta \sim \uni(\vec{\phi}, r)} \left[ \LC (\vtheta)\right] \, ,
\end{equation}
and the corresponding $r$ that maximizes the variance
\begin{equation}\label{eq:rmax}
    r_{\rm max } = \argmax_{r} \Var_{\vtheta \sim \uni(\vec{\phi}, r)} \left[ \LC (\vtheta)\right] \, 
\end{equation}
as indicated in Fig.~\ref{fig:corr_vs_uncorr}.
We emphasize that despite being closely related, $r_{\rm max}$ and $r_{\rm patch}$ are two different ways of characterizing the size of a patch with substantial gradients. Hence it comes with no surprise below that while we numerically find that $r_{\rm max}$ scales in a similar manner to our analytic estimates of $r_{\rm patch}$, the scalings are not exactly the same.

\begin{figure}
    \centering
    \includegraphics[width=0.99\linewidth]{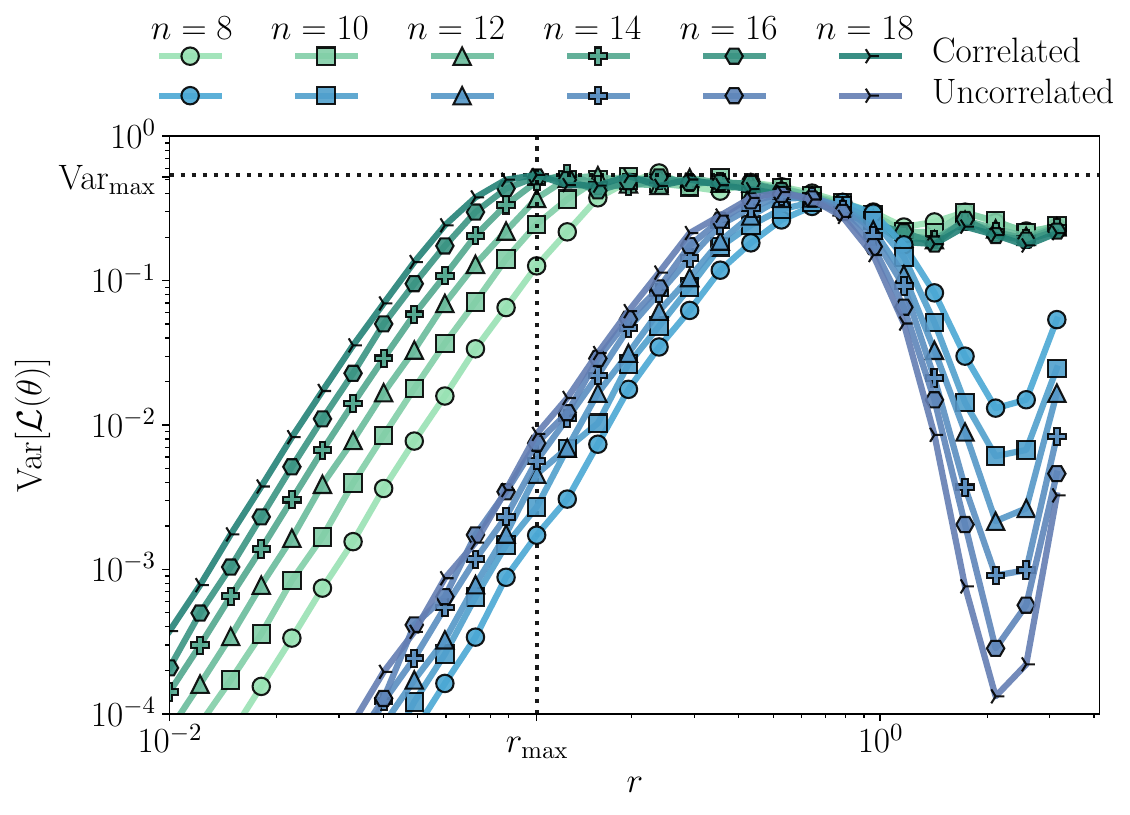} % Adjust the filename and size
    \caption{\textbf{Patch variance for correlated and uncorrelated product ansatz.} Here we study the landscape of a loss function of the form in Eq.~\eqref{eq:loss} with $\rho = |\psi \rangle \langle \psi |$ where $\ket{\psi} = \frac{1}{\sqrt{2}}(\ket{+}^{\otimes n}+ \ket{-}^{\otimes n})$ and $ O = \bigotimes_{i=1}^{n} \sigma_{z}^{(i)}$. We consider a tensor product ansatz composed of $R_X(\th),R_Z(\th), R_X(\th)$ rotations applied on each qubit. We plot the relationship between the variance of $\LC(\thv)$ in a hypercube around $\vec{0}$ as a function of $r$ when the parameters are correlated (green) and uncorrelated (blue). The max variance, $\text{Var}_{\rm max}$, and its location, $r_{\rm max}$, are indicated for $n = 18$ in the correlated case.}
    \label{fig:corr_vs_uncorr}
\end{figure}

\medskip
\paragraph*{Architectures considered.} Here we provide a high-level description of the circuit families we consider, as sketched in Fig.~\ref{fig:architectures}. 

\begin{description}
    \item[Tensor Product] An ansatz composed of layers of single qubit $R_X(\th),R_Z(\th),R_X(\th)$ rotations. For our numerics, we take the observable to be $\sigma_z^{\otimes n}$ and the initial state $\ket{\psi} ~=~\frac{1}{\sqrt{2}}(\ket{+,...,+} + \ket{-,...,-})$. The aim of this example is to study the effect of correlating parameters. We compare the case where the parameters are completely uncorrelated and the case where they are partially correlated such that all the parameters within the layer of $R_X(\th),R_Z(\th),R_X(\th)$ rotations are the same. 

    \item[Hardware Efficient Ansatz (HEA)] Parameterized quantum circuits composed of alternating independent single qubit Pauli rotations and fixed unitary gates~\cite{kandala2017hardware}. That is, ans\"atze of the form of Eq.~\eqref{eq:circuit} where the generators $H_i$ are single qubit Pauli operators with $H_i^2 = \1$ and with uncorrelated parameters such that $ \nparams = \nHam$. In our numerical simulations, we analyze an ansatz constructed via layers of single qubit rotations $R_Y(\th),R_Z(\th)$ and ladders of CZ gates in a 1D topology. We compare the maximum value of the variance for a local observable $\sigma_z^{(1)} \otimes \sigma_z^{(2)}$ applied on the first two qubits, versus a global observable, $\bigotimes_{i=1}^n\sigma_z^{(i)}$. The initial state is $\ket{0}^{\otimes n}$.

    \item[Hamiltonian Variational Ansatz (HVA)] A physically motivated ans\"atze for finding the ground state of a Hamiltonian~\cite{wecker2015progress,mele2022avoiding} and solving combinatorial optimization~\cite{farhi2014quantum}. Suppose that the target Hamiltonian $H$ can be expressed as a sum of non-commuting bounded terms $H = \sum_{k=1}^{K} H_k$ where $K\in\Theta(1)$. The HVA then takes the form
    \begin{equation}
    U(\thv) = \prod_{l=1}^L\left(\prod_{k=1}^K e^{-i\theta_{k,l} H_k}\right) \, ,
    \end{equation}
    where $k$ loops over all the terms in $H$ and then this composite layer is repeated $L$ times. The Quantum Approximate Optimization Algorithm (QAOA)~\cite{farhi2014quantum,zhou2020quantum, akshay2021parameter} is a special case of the HVA.

    Usually one considers a fully \textit{relaxed} version of the HVA and QAOA where the rotation angles $\theta_{k,l}$ are assumed to be uncorrelated such that $\nparams = KL$. To explore the effect of time correlations between different layers of the circuit we will also consider the case where each of the $l$ layers uses the same parameters, i.e., $\theta_{k,l} = \theta_{k, l'}$. This captures the Trotterized evolution under a time-independent parameterized Hamiltonian. 

    In our numerical analysis, we consider a Heisenberg Hamiltonian $O= \sum_{i=1}^n \sigma_z^{(i)}\otimes\sigma_z^{(i+1)} + \sigma_x^{(i)}\otimes\sigma_x^{(i+1)}+\sigma_y^{(i)}\otimes\sigma_y^{(i+1)}$ and the N\'eel state $\ket{\psi} ~= ~\frac{|01\rangle^{\otimes n/2} + |10\rangle^{\otimes n/2}}{\sqrt{2}}$ as our initial state.

    \item[Unitary Coupled Cluster Ansatz (UCC)] A chemically inspired ansatz used for finding the ground state of a Fermionic molecular Hamiltonian $H$. The UCC ansatz is formulated in terms of parameterized excitations applied to some reference state (usually the Hartree-Fock state).

    In this work we consider the widely used standard Unitary Coupled-Cluster
    with Singles and Doubles (UCCSD) ansatz~\cite{mao2023barren}. This is formulated by truncating the excitation operator $T(\thv)$ such that only first and second order excitations remain which are then approximated using a $L$-step Trotter approximation~\cite{mao2023barren} leading to the form
    \begin{align}\label{eq:ucc_first_main}
    U(\thv) = \prod_{l=1}^L\prod_{k=1}^K e^{{\theta_{k,l}}(\tau_k-\tau^\dagger_k)} \;,
    \end{align}
    where $\tau_k - \tau_k^\dagger$ can be further decomposed into Pauli operators for practical implementations~\cite{mao2023barren}. The ansatz consists of all possible gates which excite a single and double particles to different orbitals (see a sketch in Fig.~\ref{fig:architectures}). 
    
    We again consider two versions of this ansatz: a \textit{relaxed} version where all parameters are uncorrelated and a \textit{Trotter} version with time correlations such that different Trotter layers share the same parameters, i.e., $\theta_{k,l} = \theta_{k,l'}$. In our numerics, we consider an observable of the form $O= \sum_{i=1}^{n}\sigma_z^{(i)}\otimes\sigma_z^{(i+1)}$
    with periodic boundary conditions and an initial state $\ket{\psi} = \ket{1}^{\otimes \frac{n}{2}} \otimes \ket{0}^{\otimes \frac{n}{2}}$.
\end{description}

\begin{table*}[]
{
\setlength{\tabcolsep}{6pt}
\renewcommand{\arraystretch}{1.8}
\begin{tabular}{c|cc|cc|cc|cc|}
\cline{2-9}
                                                                                                                               & \multicolumn{2}{c|}{Tensor product }                                                                                                                                       & \multicolumn{2}{c|}{HEA}                                                                               & \multicolumn{2}{c|}{HVA}                                               & \multicolumn{2}{c|}{UCC}                                                                                                                \\ \cline{2-9} 
                                                                                                                               & \multicolumn{1}{c|}{Correlated}                             & Uncorrelated                                                                                             & \multicolumn{1}{c|}{Local}                                   & Global                                  & \multicolumn{1}{c|}{Trotter} & Relaxed                                 & \multicolumn{1}{c|}{Trotter}                                                 & Relaxed                                                  \\ \hline
\multicolumn{1}{|c|}{$r_{\rm patch}$}                                                                                          & \multicolumn{1}{c|}{$\Theta\left(\frac{1}{n}\right)$} & $\Theta\left(\frac{1}{\sqrt{n}}\right)$                                                                  & \multicolumn{1}{c|}{$\Theta\left(\frac{1}{(nL)^{1/4}}\right)$}  & $\Theta\left(\frac{1}{\sqrt{n L}}\right)$ & \multicolumn{1}{c|}{$\Theta\left(\frac{1}{L n^3}\right)$}        & $\Theta\left(\frac{1}{\sqrt{L}}\right)$ & \multicolumn{1}{c|}{$\Theta\left(\frac{1}{L\sqrt{n(K + n)}}\right)$} & $\Theta\left(\frac{1}{\sqrt{K L + n^2}} \right)$ \\ \hline
\multicolumn{1}{|c|}{${\rm Var}_{\rm patch}$} & \multicolumn{1}{c|}{$\Omega\left(1\right)$}       & $\Omega\left( \frac{1}{n}\right)$ & \multicolumn{1}{c|}{$\Omega\left(\frac{1}{\sqrt{nL}}\right)$} & $\Omega\left(\frac{1}{nL}\right)$        & \multicolumn{1}{c|}{$\Omega\left(\frac{1}{n^{10}}\right)$}        & $\Omega\left(\frac{n^2}{L^2}\right)$      & \multicolumn{1}{c|}{$\Omega\left(\frac{1}{n^2(K + n)^2}\right)$} & $\Omega\left(\frac{L n^2}{(K L+n^2)^2}\right)$                      \\ \hline
\end{tabular}
}
    \caption{\textbf{Analytic scalings of the  patch variance for example architectures.} Here we provide the analytic scalings of $r_{\rm patch}$ and $\Var_{\rm patch} := \Var_{\vtheta \sim \uni(\vec{\phi}, r_{\rm patch})}$ with $n$ being the number of qubits, $L$ the number of circuit layers and $K$ the number of generators in one layer for the UCC ansatz (as described in Eq.~\eqref{eq:ucc_first_main}) for different example ans\"atze.}
    \label{table:scalings}
\end{table*}

\begin{figure*}
    \centering
    \begin{tikzpicture}
    \pgftext{\includegraphics[width=.99\textwidth]{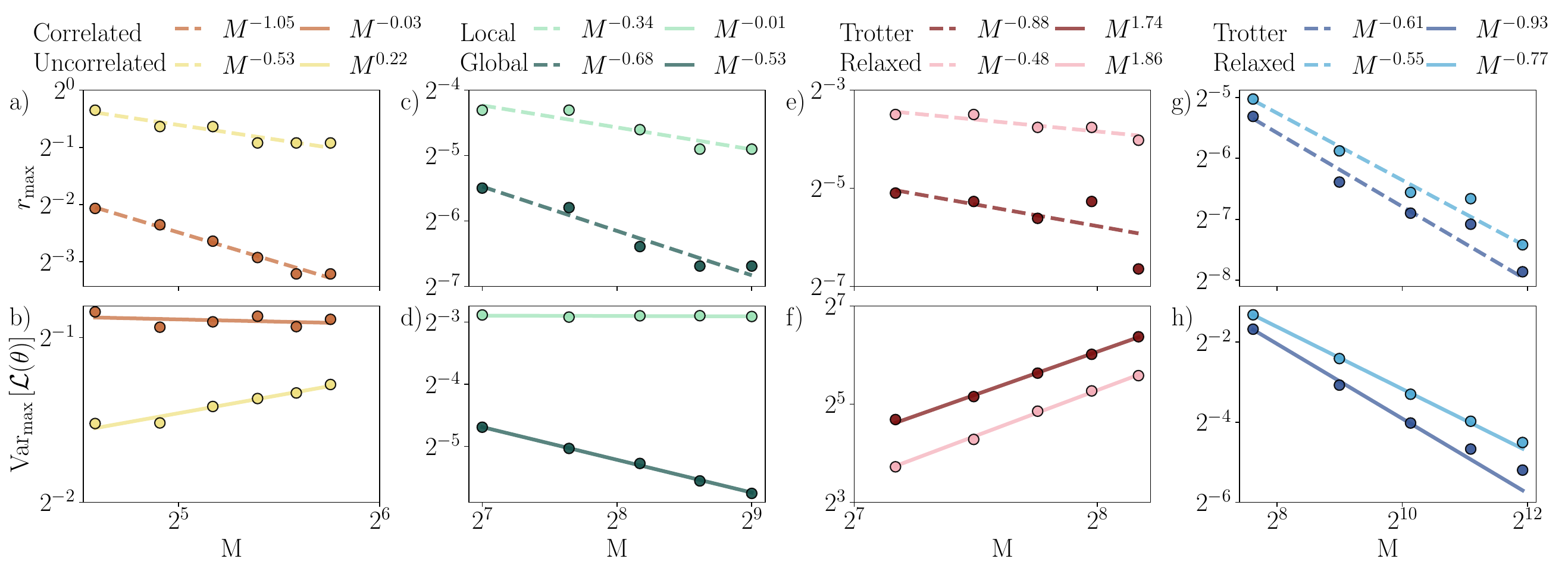}} 

    \node at (-6.3,3.5) {Tensor product};
    
    \node at (-1.7,3.5) {HEA};
    \node at (2.6,3.5) {HVA};
    \node at (6.7,3.5) {UCC};
    \end{tikzpicture}
    \caption{\textbf{Numerical scalings of the patch variance for example architectures.} Here we study the variance of the loss landscape $\LC(\thv)$ for different ans\"{a}tze with parameters drawn from a uniformly sampled hypercube centered around zero, $\uni(\vec{0}, r)$. Particularly, in the upper panels we show the value of $r$ that maximizes the variance, $r_{\max} = \argmax_{r}\Var_{\thv\sim\uni(\vec{0},r)}[\LC(\thv)]$, and in the lower panels we show the maximum value of the variance, $\Var_{\max} = \max_r \Var_{\thv\sim\uni(\vec{0},r)}[\LC(\thv)]$. We plot these values with respect to $M$, the number of generators. We calculate the variance by sampling $1000$ times inside hypercubes with different $r$. All the plots except for the tensor product ansatz are done for the following number of qubits $n\in\{8,10,12,14,16\}$. \textbf{a)} and \textbf{b)} here we consider a single layer of the tensor product ansatz, we plot it for $n\in\{8,10,12,14,16, 18\}$ qubits. \textbf{c)} and \textbf{d)} we study the Hardware Efficient Ansatz, with the number of layers scaling linearly with the number of qubits, i.e, $L = n$. \textbf{e)} and \textbf{f)} we analyze the Hamiltonian Variational Ansatz with the number of layers growing linearly with the number of qubits $L = 8n$. \textbf{g)} and \textbf{h)} we consider the Unitary Coupled Cluster ansatz with a constant number of layers $L = 2$.
    }
    \label{fig:scalings}
\end{figure*}

\paragraph*{The role of correlations.} The tensor product example nicely demonstrates the role of \textit{spatial} correlations, i.e., correlations between different qubits in the same layer as sketched in Fig.~\ref{fig:corr_schematic}, whereas the role of time correlations, i.e., correlations between parameters in different layers of the ansatz, is explored by our comparison between the \textit{relaxed} and \textit{Trotter} versions of the HVA and UCC. 
By computing the maximal frequencies for the correlated and uncorrelated cases, we find that, in line with our informal argument in the previous section, correlations do decrease the size of the region with gradients. This behavior is theoretically captured by our analytical predictions of the width of patches with guaranteed gradients, where we see that $r_{\rm patch}$ is smaller in the correlated case compared to the uncorrelated one, as summarized in the Table~\ref{table:scalings}.
This is further confirmed numerically in Fig.~\ref{fig:scalings}. We see that  $r_{\max}$ is indeed smaller in the case of the correlated ansatz compared to the uncorrelated ansatz as predicted.

However, the relative scaling of the variance within the patches with guarantees between the correlated and uncorrelated settings is not fully clear. Although prior works ~\cite{holmes2021connecting, volkoff2021large, fontana2023theadjoint, ragone2022representation} indicate that correlating parameters tends to increase the variance when looking at the full parameter range, this may not always be the case when restricting the analysis to small regions of the loss landscape. Or in other words, even if the variance across the whole landscape is bigger in the correlated case, $\Var_{\max }$ can still be larger in the  uncorrelated parameter setting. For example, in the tensor product ansatz, both $\Var_{\rm max}$ and $\Var_{\rm patch}$, as shown in Table \ref{table:scalings}, are larger in the correlated setting compared to the uncorrelated one. Conversely, the Relaxed UCC circuit exhibits greater loss variance at $r_{\rm max}$ than its Trotter counterpart both numerically and analytically. 

\medskip

\paragraph*{Locality versus globality.}
The role of observable locality is illustrated in our study of the HEA. In line with our informal discussion on the role of the measured observable locality   through the notion of effective frequencies, here we find that a circuit with a local observable exhibits a larger $r_{\rm patch}$ around zero compared to the one with a global observable.  In particular, both our analytic bounds and numerics indicate that  $r_{\rm patch}^{(G)} \sim (r_{\rm patch}^{(L)})^2$ and $r_{\max}^{(G)} \sim (r_{\max}^{(L)})^2$ respectively. Moreover, in line with prior works, we find that the variance of the local loss is large than that of the global loss~\cite{cerezo2020cost} for this unstructured ans\"atze. We expect that the story would prove more complex for ans\"atze with symmetries~\cite{ragone2023unified, fontana2023theadjoint}; however, we leave this investigation for future work.

\subsection{Fundamental limitations on warm start strategies}\label{sec:upperbound}
All prior work on loss variances in patches of quantum landscapes, including our results presented here so far, provide only \textit{lower} bounds on the variance. That is, it has been shown that in certain small-angle regions, that shrink \textit{at worst} polynomially with the number of trainable parameters, the variance is guaranteed to decrease \textit{at worse} polynomially. However, the variance could in theory be significant in larger regions, and indeed one might hope that is the case. Our numerics above strongly indicate that this optimism is misplaced. Indeed, in Fig.~\ref{fig:scalings} we find that the width of patches with gradients shrinks polynomially with the number of trainable parameters for all cases we have studied in line with our bounds. We stress that our analysis covers a wide range of ansatz classes including physically motivated circuits with correlated parameters. Thus the polynomial shrinking of patches with guaranteed gradients would seem to be a common phenomenon. 

In this section, we proceed to discuss analytic progress deriving upper bounds to pin down the scaling of patches with guaranteed gradients.  
In general, computing a non-trivial upper-bound on the loss variance for patches of loss landscapes is a hard task. However, we conjecture that any circuit that suffers from the barren plateau phenomenon will also do so in any patch with a constant width, 
i.e., in any uniformly sampled hypercube with $r\in\Theta(1)$. 
To support this claim we present Proposition~\ref{prop:upperbound}, with the proof in Appendix~\ref{app:upperbound}.  

\begin{proposition}
[Upper bound on the variance]\label{prop:upperbound}
Consider a generic loss $\LC(\thv)$ of the form in Eq.~\eqref{eq:loss}. Assume that the average of this loss over the parameter hyper-space $\vol(\vec{\phi},r_{\rm full})$ is zero, that is $\Ebb_{\thv \sim \uni(\vec{\phi}, r_{\rm full})}[\LC(\thv)] = 0$. Assume that the number of independent parameters is proportional to the number of qubits $m = c n$ ($c>1$ a constant), and that the variance over the full landscape exponentially vanishes
\begin{align}
    \Var_{\vtheta \sim \uni(\vphi, r_{\rm full})} \left[ \LC (\vtheta)\right] \in \OC\left( \frac{1}{b^n}\right) \;,
\end{align}
with $b > 1$. Then the variance on any uniformly sampled hypercube $\vol(\vec{\phi}, r)$ with $r>\frac{r_{\rm full}}{b^{1/c}}$, will also be exponentially vanishing
\begin{align}
    \Var_{\vtheta \sim \uni(\vec{\phi}, r)} \left[ \LC (\vtheta)\right]\in\order{\frac{1}{\beta^n}} \;,
\end{align}
with $\beta>1$.
\end{proposition}

Proposition~\ref{prop:upperbound} shows that any loss function and ansatz with a linearly scaling number of trainable parameters that suffers from a barren plateau over the entire parameter space, i.e., a variance that vanishes as $1/b^n$ in a hypercube\footnote{We note that in our framing here we have implicitly assumed that $r_{\rm full}$ is the same for all the variational parameters. This does not need to be the case, however, this theorem can be trivially extended to any hyper-rectangle.} with $r_{\rm full}$, will do so in any patch with $r>\frac{r_{\rm full}}{b^{1/c}}$.
Thus, as sketched in the inset to Fig.~\ref{fig:schematic}, this result allows us to extend any prior result which has proven the presence of a barren plateau over the full landscape (for an ansatz with a linearly scaling number of trainable parameters) to a wide range of patches of constant width.

The conditions of Proposition~\ref{prop:upperbound} are firstly met in the context of globality-induced barren plateaus~\cite{cerezo2020cost} which arise for constant-depth circuits. To take a simple example, if the initial state is pure and the final measurement is a projector then an ansatz made up of a single layer of single qubit Pauli rotation gates has $r_{\rm full} = \pi$ and $b = 8/3$ so we can extend the barren plateau result of Ref.~\cite{cerezo2020cost} from precisely $r = \pi$ to any $r>\frac{3\pi}{8}$. Perhaps more interestingly, the conditions of Proposition~\ref{prop:upperbound} can also be met for the HVA or QAOA at linear depths.

However, Proposition~\ref{prop:upperbound} does not generally apply to expressivity induced barren plateaus for unstructured circuits~\cite{mcclean2018barren} because there the number of parameters required for a barren plateau will scale faster than linear~\cite{harrow2009random}. Furthermore, even in cases where Proposition~\ref{prop:upperbound} does apply, there remains a gap between our upper and lower bounds (sketched in Fig.~\ref{fig:schematic}) where it is as of yet unknown whether or not loss landscapes exhibit exponentially vanishing variances. 
We leave the question of whether these gaps can be closed to future work.

\medskip

\section{Discussion}

Variance lower bounds have previously been presented on a case-by-case basis for different parameter regions of certain families of quantum circuits. Here we provided a unifying and generalizing framework to understand these prior results. Core to our results is the observation that around any point with a substantial second derivative (be this close to a minimum or another special point of the landscape) there must be a patch with non-vanishing loss variances. This observation is borderline trivial - and yet we would argue that this unifying perspective was lacking in prior case-by-case analyses. Moreover, the unifying perspective in turn allowed us to derive general variance lower bounds for patches of landscapes for physically motivated ans\"atze that had previously eluded analysis. 

The region of attraction with gradients around a solution has previously been dubbed a \textit{narrow gorge}~\cite{cerezo2020cost, arrasmith2021equivalence}. This terminology was motivated by the idea that any patch of a barren plateau landscape with substantial gradients must necessarily have an exponentially small volume relative to the total volume of the parameter space. However, the name is potentially slightly misleading as this does not mean that regions of attraction are narrow in the sense of having exponentially vanishing widths. In fact, our results show the opposite. The width of regions of attraction vanish at worst polynomially in the number of trainable parameters. Thus exponentially increasing precision in each trainable parameter is not required in order to initialize within a region of attraction.

Our general lower bounds are complemented by a numerical analysis and an upper bound for certain ans\"atze that demonstrates that the radius of regions of attraction necessarily decreases with system size. It follows that the quality of the initialization will seemingly need to increase with problem size. This challenge compounds with the fact that we of course need this  precision in \textit{all} parameters simultaneously to initialize within the region with guarantees, and this becomes exponentially more challenging in $m$ as the number of parameters $m$ increases. 

Variational quantum algorithms were originally proposed, in part, as a means of finding the approximate ground states needed as inputs for more established quantum algorithms for energetic structure calculations (e.g., quantum phase estimation). However, our findings here indicate that good approximate states are themselves needed to initialize variational quantum algorithms and the quality of approximate states will need to increase with increasing problem sizes. This thus pushes the question of how to find approximate initial states further down the line.

It has recently been observed that there is a strong correlation between being able to prove that a (subregion of a) quantum landscape does not have a barren plateau (i.e,  proving that the variance of expectation values vanishes, at worst, polynomially) and the ability to classically simulate or surrogate that landscape~\cite{cerezo2023does, angrisani2024classically, bermejo2024quantum, lerch2024efficient, mele2024noise,martinez2025efficient, angrisani2025simulating}. A classical surrogate is generated using an initial data collection phase on a quantum computer. 
In our companion paper Ref.~\cite{lerch2024efficient} we demonstrate that in all the reduced parameter regimes where it is possible to provide polynomial (in $n$ or $\nparams$) lower bounds on the variance of the expectation value, it is possible to construct a classical surrogate of the expectation value landscape. In the case of uncorrelated small-angle initializations that are close to identity (or a Clifford circuit) the measurements required to generate a surrogate are usually very simple. However, for general warm starts at arbitrary points in the landscape sophisticated circuits will generally be needed to generate this surrogate and thus it remains open whether it is preferable to surrogate or just run the variational quantum algorithm as normal. 

Finally, it is important to emphasize that our analysis, both here and in Ref.~\cite{lerch2024efficient}, focuses on uniform hypercube patches of quantum landscapes. In practice, for a successful optimization we do not need gradients in a uniform hypercube around an initialization point, nor in a uniform hypercube around a solution, but rather along a trajectory from an initialization point to a solution. Put another way, ``all'' we need to train on a barren plateau is the existence of a \textit{fertile valley} with gradients from an initialization to a solution and this fertile valley is not typically going to resemble or sit within a hypercube. Thus, while we have pushed beyond  average case analyses of the full quantum loss landscapes much more remains to be done. However, for the truly beyond-average-case analyses that are required to study fertile valleys we will need to find new theoretical tools.

\section{Acknowledgments}
HM, MSR and ZH acknowledge support of the NCCR MARVEL, a National Centre of Competence in Research, funded by the Swiss National Science Foundation (grant number 205602). 
RP acknowledges the support of the SNF Quantum Flagship Replacement Scheme (grant No. 215933). MSR acknowledges funding from the 2024 Google PhD Fellowship and the Swiss National Science Foundation [grant number 200021-219329].
TC and ST acknowledge funding support from the NSRF via the Program Management Unit for Human Resources \& Institutional Development, Research and Innovation [grant number B39G680007]. ST and ZH acknowledge support from the Sandoz Family Foundation-Monique de Meuron program for Academic Promotion.

\bibliography{quantum.bib,extra.bib}

\onecolumngrid

\newpage

\appendix

\part{Appendix}
\parttoc 

%tofix the footnotes
% \count\footins = 1000

\section{Notation}

    \noindent\begin{tabular}{ |p{1.2cm}|p{14.3cm}|  }
        \hline
        \multicolumn{2}{|c|}{\textbf{General notation table}} \\
        \hline
        Symbol & Definition  \\
        \hline
        $n$ & Number of qubits. \\
        $m$ & Number of independent variational parameters. \\
        $M$ & Number of generators in the circuit. \\
         $\SC$ & Mapping a generator $H_l$ index  $l \in \{1,\dots,M\}$ to its corresponding parameter index $\SC(l) \in \{1,\dots,m\}$.
        \\
        $\thv$ & Variational parameters. \\
        $\vphi$ & Fixed parameters.\\
        $\vol(\vec{\phi},r)$ & Hypercube in parameters space centered at $\vec{\phi}$ of width $2r$. \\
        $\uni(\vec{\phi},r)$ & Uniform distribution over $\vol(\vec{\phi},r)$. \\
        $U(\thv)$ & Parametrized quantum circuit. \\
        $H_k$ & $k$-th hermitian generator in the circuit. \\
        $\Omega_k$ & Fourier spectrum associated to the generator $H_k$ defined as the pairwise difference between eigenvalues $\lambda_i^{(k)}$ of $H_k$, i.e. $\Omega_k:= \{ \lambda_i^{(k)}-\lambda_j^{(k)} ,\; \forall i,j \in [1,2^n]\}$.\\
        $V_k$ & $k$-th non-parametrized unitary in the circuit.\\
    $\UC_\th$ & Hamiltonian evolution superoperator in the Heisenberg picture for time $\th$.\\
    $\UC_{\th,k}$ & Hamiltonian evolution superoperator in the Heisenberg picture for time $\th$ and Hamiltonian $H_k$.\\
    $\EC_{\vec{\theta}}$ &Parametrized unitary channel of the full circuit in Heisenberg picture.  \\
    $\UC_{\phi,k}^{(p)}$, $\EC^{(p)}_{\phi}$& $p$-th order derivative of $\UC_{\theta,k}$ and $\EC_{\theta}$ with respect to $\th$ evaluated at $\th=\phi$.\\
    $r_{\rm patch}$ & Characterize the width of the region in which the variance scales at least inversely polynomial in $n$ and $m$.\\
    $c_p(\vec{\phi})$ & Partial second derivative of the loss function with respect to the parameter $\th_p$ evaluated at $\vec{\phi}$.\\
    $\omega^{(\rm max)}(H) $& Maximal frequency of the generator $H$, i.e. $|\lambda_{\rm max}(H) - \lambda_{\rm min}(H)|$ where $\lambda_{\rm max}(H)$ and $\lambda_{\rm min}(H)$ are the max/min eigenvalues of $H$ respectively.\\
    $\omega^{(\rm max)}_{p}$ & Maximal frequency corresponding to the parameter $\th_p$ defined in Eq.~\eqref{eq:omega-max-corr-def}. \\
    $\omega^{\rm (eff)}_{p}$ & 
    Effective frequency corresponding to the parameter $\th_p$, defined in Eq.~\eqref{eq:omega-eff-1-def}.\\
    $\widetilde{\omega}_{p,q}^{\rm (eff)}$& Effective frequency corresponding to the parameters $\th_p$ and $\th_q$, defined in Eq.~\eqref{eq:omega-eff-2-def}.\\
    $\gap$ & Energy difference between the ground state and the first excited state (spectral gap). \\
    
        \hline
    \end{tabular}

\section{Preliminaries}
In this section we present preliminary results and theorems that will be used throughout the different calculations. We divide these into two sections. General preliminaries (mostly includes well known theorems and non-quantum related results) and Preliminaries for the main proofs (mostly containing results that will be used to prove Theorem~\ref{th:var}).

\subsection{General preliminaries}

In this section, we present some theorems and elementary analytical tools and results that will be used in the next sections.

\subsubsection{Multinomial expansion and the general Leibniz rule}
Here, we recall two fundamental statements, the multinomial expansion and the general Leibniz rule, that will be used repeatedly in subsequent derivations.

\paragraph{Multinomial expansion.} 
For any nonnegative integer $p$ and real (or complex) variables $x_1,\dots,x_K$, 
\begin{equation}\label{eq:multinomial-theorem}
(x_1 + \cdots + x_K)^p \;=\; \sum_{k_1+\cdots +k_K = p} \binom{p}{k_1, \dots, k_K}\, \prod_{i=1}^K x_i^{\,k_i},
\end{equation}
where the sum is over all nonnegative integer $k_i$ such that $k_1 + \cdots + k_K = p$, and the multinomial coefficient is given by
\begin{equation}\label{eq:multinomial}
    \binom{p}{\mathbf{k}}=\binom{p}{k_1,\dots,k_K} \;=\; \frac{p!}{k_1! \cdots k_K!}\, .
\end{equation}

\paragraph{General Leibniz rule.} 
Consider $K$ functions $f_1,\dots,f_K$ each at least $p$-times differentiable with respect to some variable.  Let $f_i^{(k)}$ denote the $k$-th derivative of $f_i$.  Then, the $p$-th derivative of the product of these $K$ functions is given by
\begin{equation}\label{eq:general-Leibniz-rule}
\Bigl(\prod_{i=1}^K f_i\Bigr)^{(p)} 
\;=\; 
\sum_{k_1 + \cdots + k_K = p} 
\binom{p}{k_1,\dots,k_K}
\;\prod_{i=1}^K f_i^{(k_i)}.
\end{equation}
This rule can be seen as a direct extension of the product rule to higher-order derivatives, taking the form similar to the multinomial expansion but with derivatives instead of exponents.

\subsubsection{Taylor remainder theorem}

We recall the Taylor remainder theorem which expresses  a single variable differentiable function as a series expansion.
\begin{theorem}[Taylor remainder theorem for a single variable real function]\label{th:taylor_remainder}
    Consider a single variable  function  $f(x)$ such that $f:\mathbb{R} \rightarrow \mathbb{R}$ is differentiable up to order $N+1$ for some positive integer $N$. The function $f(x)$  can be expanded around some fixed point $a \in \mathbb{R}$ as
    \begin{equation}
        f(x)= \sum_{k=0}^N \frac{1}{k!} \left.\frac{d^kf(x)}{dx^k}\right|_{x=a} (x-a)^k + R_{N,a}(x)
    \end{equation}
    where the remainder is of the form
    \begin{equation}
        R_{N,a}(x) = \frac{1}{(N+1)!} \left.\frac{d^{N+1}f(x)}{dx^{N+1}}\right|_{x=\nu} (x-a)^{N+1}
    \end{equation}
    with $\nu = c x + (1-c) a $ for some $c \in (0,1)$.
\end{theorem}

Next we use the Taylor expansion to find a compact expression for a series. 

\begin{lemma}\label{lemma:sacha-magic1}
A series $S(x)$ which is of the form
\begin{align}
    S(x) =  \sum_{\substack{k,k'\geq 1 \\ k+k'\geq 3}} \frac{x^{2(k+k')}}{(2k)!(2k')!}\left(\frac{1}{2(k+k')+1} - \frac{1}{(2k+1)(2k'+1)}\right) \;,
\end{align}
can be expressed in the closed form as
\begin{align}
    S(x) = \frac{1}{2}\left(1 + \frac{\sinh(2x)}{2x} - 2\left(\frac{\sinh(x)}{x}\right)^2-\frac{2x^4}{45}\right) \;.
\end{align}
\end{lemma}
\begin{proof}
Let us consider the series 
\begin{equation}
    \sum_{\substack{k,k'\geq 1 \\ k+k'\geq 3}} \frac{x^{2(k+k')}}{(2k)!(2k')!}\left(\frac{1}{2(k+k')+1} - \frac{1}{(2k+1)(2k'+1)}\right)\;.
\end{equation}
    The condition $k+k'\geq 3$ only discard the case $k=k'=1$, so we will add this term to the sum in order to compute it and remove it afterward 
    \begin{equation}
   \sum_{\substack{k,k'\geq 1 \\ k+k'\geq 3}} \frac{x^{2(k+k')}}{(2k)!(2k')!}\left(\frac{1}{2(k+k')+1} - \frac{1}{(2k+1)(2k'+1)}\right) = -\frac{x^4}{45}+\sum_{k,k'\geq 1} \frac{x^{2(k+k')}}{(2k)!(2k')!}\left(\frac{1}{2(k+k')+1} - \frac{1}{(2k+1)(2k'+1)}\right)\;.   
    \end{equation}
    Now, we will compute the sum on the RHS of this expression. First, we compute the sum over the positive terms as follows using $\frac{x^{2l}}{2l+1}=\frac{1}{x}\int_{0}^x y^{2l} \dd y$ (for any real $x>0$ and integer $l$). 
    \begin{align}
        \sum_{k,k'\geq 1} \frac{1}{(2k)!(2k')!}\frac{x^{2(k+k')}}{2(k+k')+1} & = \sum_{k,k'\geq 1} \frac{1}{(2k)!(2k')!}\frac{1}{x}\int_{0}^x  y^{2(k+k')}\dd y\\
        &=\frac{1}{x}\int_{0}^x\left(\sum_{k\geq 1} \frac{y^{2k}}{(2k)!}\right)^2 \dd y\\
        &= \frac{1}{x}\int_{0}^x (\cosh(y)-1)^2 \dd y\\
        &=\frac{1}{x}\int_{0}^x \left(\frac{3}{2}-2\cosh(y)+\frac{\cosh(2y)}{2}\right) \dd y\;,
    \end{align}
    where we first recognise the hyperbolic cosine taylor series i.e. $\cosh(y)=\sum_{k\geq 0} \frac{y^{2k}}{(2k)!}$, and then we used the property $\cosh^2(y)=\frac{1+\cosh(2y)}{2}$. Now, we can compute the integral in previous expression using $\frac{1}{x}\int_0^{x} \cosh(ay)\dd y=\frac{\sinh(ax)}{ax}$ (for any real $a$) to get
    \begin{equation}
    \label{eq:bound-remainder-single-layer-positive-terms-proof}
        \sum_{k,k'\geq 1} \frac{1}{(2k)!(2k')!}\frac{x^{2(k+k')}}{2(k+k')+1}=\frac{1}{2}\left(3-4\frac{\sinh(x)}{x} +\frac{\sinh(2x)}{2x}\right)\;.
    \end{equation}
    Secondly, we compute the sum over the negative terms as follows
    \begin{align}
        \sum_{k,k'\geq 1} \frac{x^{2(k+k')}}{(2k+1)!(2k'+1)!} &= \left(\sum_{k\geq 1} \frac{x^{2k}}{(2k+1)!}\right)^2 \\
        &=\left(-1+\sum_{k\geq 0} \frac{x^{2k}}{(2k+1)!}\right)^2\\
        &=\left(-1+\frac{1}{x}\sum_{k\geq 0} \frac{x^{2k+1}}{(2k+1)!}\right)^2 \\
        &=\left(-1+\frac{\sinh(x)}{x}\right)^2\\
        &= 1-2\frac{\sinh(x)}{x}+\left(\frac{\sinh(x)}{x}\right)^2\;.
    \end{align}
    Finally, combining it with Eq.~\eqref{eq:bound-remainder-single-layer-positive-terms-proof} leads to the desired result, i.e.
    \begin{equation}
     \sum_{\substack{k,k'\geq 1 \\ k+k'\geq 3}} \frac{x^{2(k+k')}}{(2k)!(2k')!}\left(\frac{1}{2(k+k')+1} - \frac{1}{(2k+1)(2k'+1)}\right)=\frac{1}{2}\left(1 + \frac{\sinh(2x)}{2x} - 2\left(\frac{\sinh(x)}{x}\right)^2-\frac{2x^4}{45}\right)\;.
    \end{equation}
\end{proof}

And we also use the Taylor series to find an upper-bound of the following function for a constrained regime of $x$. 

\begin{lemma}\label{lemma:sacha-magic2}
    Consider a real function $S(x)$ of the form
    \begin{align}
        S(x) = \frac{1}{2}\left(1 + \frac{\sinh(2x)}{2x} - 2\left(\frac{\sinh(x)}{x}\right)^2 - \frac{2x^4}{45}\right) \;.
    \end{align}
    Provided that $x\leq \frac{3}{2}$, we have the following lower bound
    \begin{align}
        S(x) \leq  \frac{x^6}{270} \;.
    \end{align}
\end{lemma}
\begin{proof}
%\sacha{Upper bounding $S(x)$ at $\OC(x^6)$: }
Let us write $S(x)$ as a series using the series expansion of $\sinh(x)$ i.e. $\sinh(x)=\sum_{k=0}^{\infty} \frac{x^{2k+1}}{(2k+1)!}$. First, the expansion of the positive term in $S(x)$ is given by
\begin{align}
    \frac{1}{2}\left(1+\frac{\sinh(2x)}{2x}\right) &= \frac{1}{2}\left(1+\sum_{k=0}^{\infty}\frac{(2x)^{2k}}{(2k+1)!}\right) \\
    &= 1 +\frac{x^2}{3}+ \frac{1}{2}\sum_{k=2}^{\infty}\frac{(2x)^{2k}}{(2k+1)!}\;,
    \label{eq:S-function-remainder-expansion-positive-part}
\end{align}
where we explicitly separated the terms $k=0$ and $k=1$ from the sum in the last equality. 
 Now, for the negative term, we have 
 \begin{align}
     \left(\frac{\sinh(x)}{x}\right)^2 &= \sum_{k,k'=0}^\infty \frac{x^{2(k+k')}}{(2k+1)!(2k'+1)!} \\
     &= \sum_{l=0}^{\infty}\sum_{k=0}^l \frac{x^{2l}}{(2k+1)!(2(l-k)+1)!} \\
     &= \sum_{l=0}^{\infty}\frac{x^{2l}}{(2l+2)!}\sum_{k=0}^l \binom{2l+2}{2k+1} \\
     &= \sum_{l=0}^{\infty}\frac{x^{2l}}{(2l+2)!} 2^{2l+1} \\ 
     &= 1+\frac{x^2}{3} + \sum_{l=2}^\infty \frac{(2x)^{2l}}{(2l+1)!}\frac{1}{l+1}\;,
     \label{eq:S-function-remainder-expansion-negative-part}
 \end{align}
where the second equality is obtained by introducing the variable $l=k+k'$ and rewrite the sum as a sum over $l$. The third equality is obtained by introducing the binomial coefficient. The binomial sum is computed in the fourth equality as follows 
\begin{align}
  \sum_{k=0}^l \binom{2l+2}{2k+1} &= \sum_{k=0}^{2l+2} \binom{2l+2}{k} \frac{1-(-1)^k}{2} \\
  &= \frac{2^{2l+2}-(1-1)^{2l+2}}{2} \\
  &=2^{2l+1}\;,
\end{align}
which is the sum over odd binomial coefficients (equal to the sum over even ones). 
Finally, the last equality is obtained by rearranging the term in the sum and separating the terms $k=0$ and $k=1$ as done in Eq.~\eqref{eq:S-function-remainder-expansion-positive-part}. Therefore, combining Eq.~\eqref{eq:S-function-remainder-expansion-positive-part} with Eq.~\eqref{eq:S-function-remainder-expansion-negative-part} leads to
\begin{equation}
\label{eq:S-function-remainder-expansion-full-series}
    S(x) = \sum_{k=2}^\infty \frac{(2x)^{2k}}{(2k+1)!}\left(\frac{1}{2}-\frac{1}{k+1}\right) - \frac{x^4}{45}\;.
\end{equation}
This is a polynomial of even powers of  $x$ (starting at $x^4$) with positive coefficients. Now, let us bound the series at $\OC(x^6)$ by rewriting it as follows:
\begin{align}
    S(x) &= \frac{x^4}{45} +\sum_{k=3}^\infty \frac{(2x)^{2k}}{(2k+1)!}\left(\frac{1}{2}-\frac{1}{k+1}\right)- \frac{x^4}{45} \\
    &= \sum_{k=0}^\infty \frac{(2x)^{2k+6}}{(2k+7)!}\left(\frac{1}{2}-\frac{1}{k+4}\right)\\
    &= x^6\sum_{k=0}^\infty x^{2k}\frac{2^{2k+6}}{(2k+7)!}\left(\frac{1}{2}-\frac{1}{k+4}\right)\;,
\end{align}
where we first separated the term $k=2$ from the sum, then we changed variable $k\to k+3$ and finally factorize $x^6$ outside the sum. The series multiplying $x^6$ is a polynomial of even order of $x$ with positive coefficients and therefore it increases with $x\geq 0$. So, assuming $0< x\leq x_0$ allows to bound this series by evaluating it at $x=x_0$ (i.e. its largest possible value). Notice that this series corresponds to the function
\begin{equation}
    S(x)x^{-6}\;.
\end{equation}
If we assume $x\leq 1.5$, we can bound $S(x)$ by evaluating previous expression at $x=1.5$ which numerically gives $\approx 0.00739377 < 1/270$. Therefore, $S(x)\leq x^6/270$ for $x\leq 1.5$. This completes the proof. One final remark is that one could potentially keep higher order terms by following the same approach.
\end{proof}

\subsubsection{Lower-bound on the variance of a single variable function}\label{sec:lower_bound_variance_single_variable}

In this section we use the Taylor Reminder Theorem (Theorem~\ref{th:taylor_remainder}) to find a lower-bound on the variance of a single variable function with bounded derivatives. To prove the bound, first we need the following proposition.

\begin{proposition}\label{prop:expectation_taylor_product}
    Consider the parameterized superoperator $\Lambda_{\th}(A)$ for some real parameter $\th \in \mathbb{R}$ with bounded second derivative with respect to $\th$,i.e. $\norm{\Lambda^{(2)}_{\th}(A)} \leq \gamma$.
    We also consider $\uni(0, r)$ as the uniform distribution over the interval $[-r,r]$. Then,  we have
    \begin{equation}
        \norm{\Ebb_{\theta  \sim \uni(0, r)}[\Lambda_{\th}(A)]-\Lambda_{0}(A)} \leq \frac{1}{6} \gamma r^2
    \end{equation}
\end{proposition}

\begin{proof}

The infinite norm of an operator can be rewritten as its maximum expectation value magnitude over any quantum state as follows
\begin{equation}
    \label{eq:infinite-norm-definition-maximum}
\norm{O}=\max_{\ket{\psi}}|\bra{\psi} O \ket{\psi}|\;.
\end{equation}
Let us consider the function $f_{\ket{\psi}}(\th)=\bra{\psi}\Lambda_\th(A)\ket{\psi}$ and let $f_{\ket{\psi}}^{(p)}(\phi)=\left.\frac{d^pf_{\ket{\psi}}(\th)}{d\theta^p}\right|_{\theta=\phi}$. The Taylor Reminder Theorem (Theorem~\ref{th:taylor_remainder}) states that there exists a parameter $\nu_{\th,\ket{\psi}} = c_{\ket{\psi}}\th$ for some $c_{\ket{\psi}} \in (0,1)$ such that
\begin{equation}
  f_{\ket{\psi}}(\th)= f_{\ket{\psi}}(0)+\th f_{\ket{\psi}}^{(1)}(0)+\frac{\th^2}{2}f_{\ket{\psi}}^{(2)}(\nu_{\th,\ket{\psi}})\;.
\end{equation}
Therefore, we have 
\begin{align}
   \norm{\Ebb_{\th \sim \uni(0, r)}[\Lambda_\th(A)] -\Lambda_0(A)} &= \max_{\ket{\psi}}\left|\Ebb_{\th \sim \uni(0, r)}[f_{\ket{\psi}}(\th)] -f_{\ket{\psi}}(0)\right| \\ 
   &=\max_{\ket{\psi}}\left|\Ebb_{\th \sim \uni(0, r)}\left[\frac{\th^2}{2}f_{\ket{\psi}}^{(2)}(\nu_{\th,\ket{\psi}})\right]\right|\;,
\end{align}
where we used the fact that the first moment vanishes in the second equality.
Now, we can further upperbound the previous equation as follows

\begin{align}
\max_{\ket{\psi}}\left|\Ebb_{\th \sim \uni(0, r)}\left[\frac{\th^2}{2}f_{\ket{\psi}}^{(2)}(\nu_{\th,\ket{\psi}})\right]\right|
& \leq \Ebb_{\th \sim \uni(0, r)}\left[\max_{\ket{\psi}}\left|\frac{\th^2}{2}f_{\ket{\psi}}^{(2)}(\nu_{\th,\ket{\psi}})\right|\right]\\
& \leq \Ebb_{\th \sim \uni(0, r)}\left[\frac{\th^2}{2}\max_\nu\max_{\ket{\psi}}\left|f_{\ket{\psi}}^{(2)}(\nu_{\th,\ket{\psi}})\right|\right]\\
&= \Ebb_{\th \sim \uni(0, r)}\left[\frac{\th^2}{2}\right] \max_{\nu}\max_{\ket{\psi}}\left|f_{\ket{\psi}}^{(2)}(\nu)\right|\\ \label{eq:distance-expectation-vs-center-channel-proof}
&= \frac{r^2}{6}\max_{\nu} \norm{\Lambda_{\nu}^{(2)} (A)} \\
&\leq \frac{r^2}{6} \gamma\;,
\end{align}
where we first used Jensen's inequality (for the convex function $\max_{\ket{\psi}}|\cdot|$) and then removed the dependence of $\max_{\ket{\psi}}\left|f_{\ket{\psi}}^{(2)}(\nu)\right|$ on $\th$ by maximizing over $\nu$. Moreover,  the last equality is obtained by explicitly evaluating the second moment and the final inequality is obtained using the bounded derivative assumption. 

Therefore, we indeed get 
\begin{align}
    \norm{\Ebb_{\th \sim \uni(0, r)}[\Lambda_{\th}(A)] -\Lambda_{0}(A)}&\leq\frac{\gamma}{6}  r^2\;.
\end{align}

\end{proof}

Once the necessary preliminary result has been proven, we can use Theorem~\ref{th:taylor_remainder} (Taylor Reminder Theorem) to find the lower-bound on the variance for a single variable function mentioned above.

\begin{proposition}[Variance lower bound of a function with bounded derivatives]\label{prop:variance_LB_bounded} 
    Consider a single variable differentiable function $f(\th)$ such that $f:\mathbb{R} \rightarrow \mathbb{R}$. We assume that the function $f$ has bounded even derivatives at zero in the sense that
    \begin{align} \label{eq:assumption-bounded-grad}
        \left|f^{(2p)}(0) \right| := \left| \left.\frac{d^{2p}}{d\th^{2p}} f(\th)\right|_{\th=0}\right|  \leq \alpha  \gamma^{2p} \;\;, \;\;\forall p \geq 1 \;,
    \end{align}
    for some constants $\alpha$ and $\gamma$. We also consider the parameter $\th$ to be sampled from the uniform distribution over the interval of length $2r$ centered around zero $\uni(0, r)$. Then, given that the perturbation $r$ satisfies \begin{equation}\label{eq:prop_r_remainder_cond}
        r \leq \frac{3}{2 \gamma} \;,
    \end{equation}
    the variance of the function $f$ can be lower bounded as 
    \begin{equation}\label{eq:var_f_LB}
        \Var_{\th \sim \uni(0,r)}[f(\theta)] \geq \frac{1}{45}\left[f^{(2)}(0)\right]^2  r^4 - \frac{ \alpha^2 \gamma^6 }{270} r^6 \;.
    \end{equation}
\end{proposition}

\begin{proof}
    First, we can lower bound the variance of the function by using the Taylor expansion around zero as follows
    
    \begin{align}
        \Var_{\th \sim \uni(0,r)}[f(\th)] &= \Var_{\th \sim \uni(0,r)}\left[f(0) + \sum_{k=1}^\infty \frac{\th^k}{k!} f^{(k)}(0)\right]\\
        &= \Var_{\th \sim \uni(0,r)}\left[ \sum_{k=1}^\infty \frac{\th^k}{k!} f^{(k)}(0)\right]\\
        &= \Var_{\th \sim \uni(0,r)}\left[ \sum_{k=0}^\infty \frac{\th^{2k+1}}{(2k+1)!} f^{(2k+1)}(0)\right] + \Var_{\th \sim \uni(0,r)}\left[\sum_{k=1}^\infty \frac{\th^{2k}}{(2k)!} f^{(2k)}(0)\right]\\
        & \geq \Var_{\th \sim \uni(0,r)}\left[\sum_{k=1}^\infty \frac{\th^{2k}}{(2k)!} f^{(2k)}(0)\right]
    \end{align}
where in the second equality, we drop the constant term and in the third equality, we split the terms corresponding to even and odd powers in the parameter $\theta$ and use the fact that their covariance is vanishing for the symmetric uniform distribution $\uni(0,r)$. Indeed, the product of an odd and even function is an odd function, and thus the integral over an even space is zero. In the final inequality, we only keep the even order terms.

Now let us compute the variance explicitly by using the fact that $\Ebb_{\th \sim \uni(0,r)}[\th^{2k}]=\frac{r^{2k}}{2k+1}$, which leads to

\begin{align}    \Var_{\th \sim \uni(0,r)}\left[f(\th)\right]  & \geq \Var_{\th \sim \uni(0,r)}\left[ \sum_{k=1}^\infty \frac{\theta^{2k}}{(2k)!} f^{(2k)}(0)\right]\\
&= \Var_{\th \sim \uni(0,r)}\left[\sum_{k=1}^\infty  \frac{\theta^{2k}}{(2k)!} f^{(2k)}(0)\right]\\
&= \Ebb\left[\left(\sum_{k=1}^\infty  \frac{\theta^{2k}}{(2k)!} f^{(2k)}(0)\right)^2\right] -  \Ebb\left[\sum_{k=1}^\infty  \frac{\theta^{2k}}{(2k)!} f^{(2k)}(0)\right]^2\\
&= \sum_{k,k'=1}^\infty  \left(\frac{\Ebb[\theta^{2k + 2k'}]}{(2k)! (2k')!} - \frac{\Ebb[\theta^{2k}] \Ebb[\theta^{2k'}]}{(2k)!(2k')!} \right) f^{(2k)}(0) f^{(2k')}(0)\\
&=\sum_{k,k'\geq 1}^\infty \frac{r^{2(k+k')}}{(2k)!(2k')!}\left(\frac{1}{2(k+k')+1} - \frac{1}{(2k+1)(2k'+1)}\right) f^{(2k)}(0) f^{(2k')}(0) \\
    &= \frac{r^4}{45}[ f^{(2)}(0)]^2 + \RC \\
    &\geq \frac{r^4}{45}[ f^{(2)}(0)]^2 -|\RC| \;, \label{eq:varf_LB}
\end{align}
where we introduce the remainder $\RC$ which contains all higher order terms in $r$ i.e., $\RC\sim \OC(r^6)$ and is of the form
\begin{equation}
    \RC=\sum_{\substack{k,k'\geq 1 \\ k+k'\geq 3}}^\infty \frac{r^{2(k+k')}}{(2k)!(2k')!}\left(\frac{1}{2(k+k')+1} - \frac{1}{(2k+1)(2k'+1)}\right)f^{(2k)}(0) f^{(2k')}(0) \;. \label{eq:remainder_f}
\end{equation}

To further lower bound in Eq.~\eqref{eq:varf_LB}, we provide the following upper bound on the remainder $\abs{\mathcal{R}}$:
\begin{align}
    \abs{\mathcal{R}} &\leq  \sum_{\substack{k,k'\geq 1 \\ k+k'\geq 3}} \frac{r^{2(k+k')}}{(2k)!(2k')!}\left(\frac{1}{2(k+k')+1} - \frac{1}{(2k+1)(2k'+1)}\right) |f^{(2k)}(0)| |f^{(2k')}(0)| \\
    &\leq \sum_{\substack{k,k'\geq 1 \\ k+k'\geq 3}} \frac{r^{2(k+k')}}{(2k)!(2k')!}\left(\frac{1}{2(k+k')+1} - \frac{1}{(2k+1)(2k'+1)}\right) \alpha^2 \gamma^{2(k+k')} \label{eq:remainder_UB_S0}\\ 
    &= \alpha^2  S( \gamma r)\label{eq:remainder_UB_S}
\end{align}
where the first inequality is from the triangle inequality and the second inequality is by the assumption on the derivative bound in Eq.~\eqref{eq:assumption-bounded-grad}. To reach the final equality, we invoke Lemma~\ref{lemma:sacha-magic1} where $S(x) =\frac{1}{2}\left(1 + \frac{\sinh(2x)}{2x} - 2\left(\frac{\sinh(x)}{x}\right)^2 - \frac{2x^4}{45}\right)$. Next, by further invoking Lemma~\ref{lemma:sacha-magic2}, we can bound $S(x)$ for some region. Particularly, one can see that for
\begin{equation}\label{eq:remainder_r_cond}
    \gamma r\leq \frac{3}{2} \Rightarrow r \leq \frac{3}{2 \gamma}
\end{equation}
we can upper-bound $S(\gamma r)$ as follows 
\begin{equation}
    S(r\gamma)\leq \frac{\gamma^6}{270} r^6
\end{equation}
and thus we can recover Eq.~\eqref{eq:remainder_UB_S} to upper-bound the remainder as
\begin{equation}\label{eq:final_remainder_UB}
    |\mathcal{R}| \leq \alpha^2 \frac{\gamma^6}{270} r^6  \;.
\end{equation}

By plugging the remainder upper bound in Eq.~\eqref{eq:final_remainder_UB} and for the perturbation $r$ satisfying the condition in Eq.~\eqref{eq:remainder_r_cond}, we retrieve the promised variance lower bound in Eq.~\eqref{eq:var_f_LB}.
\end{proof}

\subsubsection{Variance decomposition of a multivariable function}
Here we show how the variance of a multivariate function can be decomposed into a sum of expected values and variances for different variables.

\begin{proposition}\label{prop:var_decomp}
Consider a multivariable function $f(\vec{\th})=f(\th_1,\dots,\th_\nparams)$ depending on $\nparams$ parameters such that $f: \mathbb{R}^\nparams \rightarrow \mathbb{R}$. We assume that each parameter is sampled independently from some distribution $\PC$ i.e., $\thv \sim \PC^{\otimes m}$. Then for any permutation $\pi$, 
the variance of the function $f$ can be expressed as 
\begin{align}
    \Var_{\vec{\th} \sim \mathcal{P}^{\otimes \nparams}}\left[f(\vec{\th})\right] = \sum_{k=1}^\nparams \Ebb_{\pi(\nparams),\dots,\pi(k+1)}[\Var_{\pi(k)}[\Ebb_{\pi(k-1),\dots,\pi(1)}[f(\vec{\th})]]]
\end{align}
\end{proposition}
\begin{proof}

     The proof can be obtained by recursion over the number of independent parameters $\nparams$. 

    Precisely, let us introduce the function $f_k$ defined recursively as follows
    \begin{align}
        f_0&= f(\vec{\th})\\
        f_{k+1} &= \Ebb_{k+1}[f_k]
    \end{align}
    Here, one can see that the function $f_k$ only depends on the parameters $\vec{\th}_{\overline{k}}:=(\th_{k+1},\dots,\th_\nparams)$. 
    Moreover, we can show that $\forall \; 0\leq  k \leq m$
    \begin{align}
        \Var_{\overline{k}}[f_k] &= \Ebb_{\overline{k}}[f_k^2] - \Ebb_{\overline{k}}[f_k]^2 \label{eq:recursive_f_start}\\
        &=  \Ebb_{\overline{k+1}}[\Ebb_{k+1}[f_k^2]] - \Ebb_{\overline{k+1}}[\Ebb_{k+1}[f_k]]^2\\
        &= \Ebb_{\overline{k+1}}[\Ebb_{k+1}[f_k^2]]- \Ebb_{\overline{k+1}}[\Ebb_{k+1}[f_k]^2] + \Ebb_{\overline{k+1}}[\Ebb_{k+1}[f_k]^2] - \Ebb_{\overline{k+1}}[\Ebb_{k+1}[f_k]]^2\\
        &= \Ebb_{\overline{k+1}}[\Var_{k+1}[f_k]] + \Var_{\overline{k+1}}[\Ebb_{k+1}[f_k]]\\
        &= \Ebb_{\overline{k+1}}[\Var_{k+1}[f_k]] + \Var_{\overline{k+1}}[f_{k+1}]
    \end{align}

    Hence, we have that
    \begin{align}
    \Var[f(\vec{\th})] &= 
    \sum_{k=0}^{m-1}  \left(\Var_{\overline{k}}[f_k] - \Var_{\overline{k+1}}[f_{k+1}] \right) \\
    &= \sum_{k=0}^{m-1} \Ebb_{\overline{k+1}}[\Var_{k+1}[f_k]]\\
    &= \sum_{k=1}^{m} \Ebb_{\overline{k}}[\Var_{k}[\Ebb_{1,\dots,k-1}[f(\vec{\th})]]] \label{eq:recursive_f_end}
    \end{align}

    Here, we note that the indexing of the parameters is arbitrary. Consequently, for any permutation $\pi:\{1,\dots,\nparams\} \rightarrow \{1,\dots,\nparams\}$, we obtain the same variance decomposition. 
    
    Let us choose any reordering $\pi$ of the parameter labels such that 
    $(\th_{\pi(1)},\dots,\th_{\pi(\nparams)})= (\th_{p_1},\dots,\th_{p_\nparams})$. Then, we can rewrite the recursive definition of $f_k$ as follows
    \begin{align}
        f_0 &= f(\vtheta)\\
        f_{k+1} &=  \Ebb_{p_{k+1}}[f_k]
    \end{align}
    where $f_k$ now depends only on the parameters $\vtheta_{\overline{p_k}} := (\th_{p_{k+1}},\dots,\th_{p_\nparams})$.
    Then using the same recursive proof technique as in Eq.~\eqref{eq:recursive_f_start}-Eq.~\eqref{eq:recursive_f_end}, we obtain 
    \begin{align}
        \Var[f(\vtheta)]= \sum_{k=1}^\nparams \Ebb_{\overline{p_k}}[\Var_{p_k}[\Ebb_{p_1,\dots,p_{k-1}}[f(\vtheta)]]]
    \end{align} 
    which concludes the proof.
\end{proof}

\subsection{Preliminaries for the main proofs}
In this section we derive different results that will be used throughout the main proofs of this paper. Particularly, to prove Theorem~\ref{th:var}. 

\subsubsection{Bounded derivatives of parameterized unitary channels}\label{sec:bounded_derivatives_param_unitary_channels}
We now show that, under mild assumptions on the generators, the derivatives of certain parameterized unitaries (and their induced channels) remain norm-bounded.  These lemmas lead to many of our subsequent arguments about vanishing (or non-vanishing) gradients.

\begin{lemma}[Bound on nested commutators]\label{lemma:bound-nested-commutator}
Consider two bounded operators $H$ and $A$, as well as denote  the $p^{\rm th}$-order nested commutator as
    \begin{align}
        [(H)^p, A]:= \underbrace{[H, [H, \dots, [H, A]]]}_{p \text{ times}} \;.
    \end{align} 
The infinity norm of the nested commutator can be bounded as
\begin{align}
    \left\| [(H)^p,A] \right\|_\infty & 
     \leq (2 \omega^{(\rm max)}(H))^p \left\| A \right\|_\infty \;.
\end{align}
where 
\begin{equation}\label{eq:freq_lemma_commutator}
\omega^{(\max)}(H) := |\lambda_{\rm max}(H) - \lambda_{\rm min}(H)|
\end{equation}
is the maximum frequency of $H$ (i.e., the largest spectral gap among the eigenvalues of $H$.)

\end{lemma}

\medskip

\begin{proof} 

We first notice that for a single commutator ($p=1$) we have
\begin{align}
    \left\| [H,A] \right\|_\infty & = 
     \norm{[H - \lambda_{\rm min}(H)\1,A]}\\
    &=\left\| (H - \lambda_{\rm min}(H)\1)A-A(H - \lambda_{\rm min}(H)\1) \right\|_\infty \\
    &\leq \left\| (H - \lambda_{\rm min}(H)\1)A \right\|_\infty+\left\| A(H - \lambda_{\rm min}(H)\1) \right\|_\infty \\
    &\leq 2\left\| H - \lambda_{\rm min}(H)\1 \right\|_\infty\left\| A \right\|_\infty\\\
    &= 2 \omega^{(\rm max)}(H) \norm{A}\;,
\end{align}
where we used the commutator invariance to a term proportional to the identity operator in the first equality with $\lambda_{\rm min}(H)$ being the lowest eigenvalue of $H$. We then applied triangle inequality in the first inequality and sub-multiplicativity of the norm $ \| H A \|_{\infty} \leq \| H \|_\infty \| A\|_\infty$ in the second inequality. Finally, we used the shorthand $\omega^{(\rm max)}(H)$ (introduced in Eq.~\eqref{eq:freq_lemma_commutator}) which denotes the maximal difference in absolute value of the hermitian operator $H$ eigenvalues. From this result, the following recursive relation can be obtained
\begin{align}
    \left\| [(H)^p,A] \right\|_\infty &= \left\| [H,[(H)^{p-1},A]] \right\|_\infty \\
    &\leq 2\omega^{(\rm max)}(H)\left\| [(H)^{p-1},A] \right\|_\infty \\
    & \leq (2\omega^{(\rm max)}(H))^p\left\| A \right\|_\infty\;.
\end{align}
Hence, the claim follows by induction on $p$.

\end{proof}
\noindent \textit{Remark.} Notice that Lemma~\ref{lemma:bound-nested-commutator} also holds if we replace $\omega^{(\rm max)}(H)$ by $\norm{H}$. If $H$ has negative eigenvalues, the spectral norm is tighter than the maximal frequency, i.e., $\norm{H} \leq \omega^{(\rm max)}(H)$. Moreover, we can further tighten the nested commutator bound in  Lemma~\ref{lemma:bound-nested-commutator} under some locality assumptions, which we detail in the following Lemma.

\begin{lemma}[Bound on nested commutators under locality assumptions]\label{lemma:bound_nested_com_local}
    Consider two bounded operators $H$ and $P$, acting each on a constant number of qubits. Precisely, $H$ can be decomposed into a sum of commuting  Pauli strings $H= \sum_{j=1}^N h_j$ , where each $h_j$  acts on at most a constant number of neighboring qubits.  We also assume that  the operator $P$ is a local Pauli string, acting on a constant number of neighboring qubits. 
    Their $p^{\rm th}$- order nested commutator can be written as
    \begin{equation}\label{eq:comm_local_lemma}
        [(H)^p,P] = [(\widetilde{H}(P))^p,P] 
    \end{equation}
    where $\widetilde{H}(P)$ is the part of $H$ which anti commute with $P$.
    Moreover, the infinity norm of the nested commutator can be upper bounded by a constant, i.e. 
    \begin{align}\label{eq:lemma_bound_com_local}
        \norm{[(H)^p,P]} \leq (2  s(P))^2
    \end{align}
    where $s(P)$ is the number of Pauli terms in $\widetilde{H}(P)$ that is constant in the system size.
    
\end{lemma}

\begin{proof}

We begin by restating the key assumptions of the Lemma and introducing  some new definitions.
First, the Hamiltonian can be written as a sum of commuting Pauli terms, i.e. $[h_i, h_j]=0\, \forall \, i,\, j$, acting each on a constant number of neighboring qubits or what we simply call geometrically local, i.e.
\begin{align}\label{eq:H_decomp}
    H = \sum_{j=1}^N h_j\;.
\end{align}
We also consider the operator $P$ to be a geometrically local Pauli string.

    Let us now introduce the set $\mathcal{C}(P)$ denoting the ensemble of Pauli terms $h_j$ in the decomposition of $H$ in Eq.~\eqref{eq:H_decomp} which anti commute with $P$. Formally, we have
\begin{equation}
    \mathcal{C}(P) = \{ j \in \{1,\dots,N_k\} \;, \{P,h_j\}=0 \}\;.
\end{equation}
Here we can see that the size of this set $s(P):=|\mathcal{C}(P)|$ is a constant in the system size due to the locality of $h_j$ and $P$. 

Hence, in what follows we will prove that we can show recursively that the $p$-th order nested commutator of the Hamiltonian term $H$ can be expressed as 
\begin{align}\label{eq:comm_reduce_lemma}
    [(H)^p,P] &= \left[\left(\sum_{j \in \mathcal{C}(P)} h_j\right)^p,P\right] \;, \forall p\geq1\\
    &= \left[\left(\widetilde{H}(P)\right)^p,P\right] \;, \forall p\geq1\label{eq:only_com_bit_idont_know_howtolabel_this_anymore}
\end{align}

We start by seeing that the property in Eq.~\eqref{eq:comm_reduce_lemma} can be trivially verified for $p=1$. Besides, we show that the following recursive relation holds
\begin{align}
    [(H)^{p+1},P] &= [H,[(H)^{p},P]]\\
    &= \left[H,\left[\left(\sum_{j \in \mathcal{C}(P)} h_j\right)^p,P\right]\right]\\
    &= \left[\sum_{j \in \mathcal{C}(P)} h_j + \sum_{j \notin \mathcal{C}(P)} h_j,\left[\left(\sum_{j \in \mathcal{C}(P)} h_j\right)^p,P\right]\right]\\
    &= \left[\sum_{j \in \mathcal{C}(P)} h_j,\left[\left(\sum_{j \in \mathcal{C}(P)} h_j\right)^p,P\right]\right]\\
    &= \left[\left(\sum_{j \in \mathcal{C}(P)} h_j\right)^{p+1},P\right]
\end{align}
where we used in the fourth equality the fact that all the $h_j$ mutually commute and that the terms $h_j$ for $j\notin \mathcal{C}(P)$ commute with $P$.
Thus, this proves the claim in Eq.~\eqref{eq:comm_local_lemma} by induction on $p$.

In addition, we can upper bound the infinity norm of the nested commutator as
\begin{align}
     \norm{[(H)^p,P]} &\leq \norm{\left(2\sum_{j \in \mathcal{C}(P)} h_j\right)^p} \norm{P}\\
     &\leq \norm{2\sum_{j \in \mathcal{C}(P)} h_j}^p\\
     & \leq  \left(2\sum_{j \in \mathcal{C}(P)} \norm{h_j}\right)^p\\
     & = (2s(P))^p \;.   \label{eq:comm_reduce_norm_lemma}  
\end{align}
where in the first inequality, we used  the result from Lemma \ref{lemma:bound-nested-commutator} and the results we just derived presented in Eq.~\eqref{eq:only_com_bit_idont_know_howtolabel_this_anymore}. In the second inequality, we used the triangular inequality and that the infinity norm of a Pauli string is equal to one.

\end{proof}

This last lemma presented is quite specific, and will just be used a handful of times in certain instances. Thus we go back to Lemma \ref{lemma:bound-nested-commutator} in order to find an upper-bound on the derivatives of a unitary channel.

\begin{lemma}[Single-parameter derivative bounds]\label{lemma:bounded_unitary}

Consider the parametrized unitary channel 
\begin{equation}
\mathcal{U}_\theta(A) \;=\; e^{i\,\theta\,H}\,A\,e^{-i\,\theta\,H},
\end{equation}
where $H$ is a Hermitian operator, $A$ is a bounded operator, and $\theta \in \mathbb{R}$. 

1. The $p$-th derivative of $\mathcal{U}_\theta(A)$ with respect to $\theta$, evaluated at $\theta=\phi$, is 
\begin{equation}\label{eq:p-order-grad-A}
\mathcal{U}_\phi^{(p)}(A) 
\;:=\;\frac{d^p}{d\theta^p}\,\mathcal{U}_\theta(A)\Big|_{\theta=\phi} 
\;=\; i^p\,\mathcal{U}_{\phi}\Bigl(\bigl[(H)^p,\,A\bigr]\Bigr).
\end{equation}

2. For $p\ge 2$, 
\begin{equation}\label{eq:bound-p-order-grad-A}
\bigl\|\mathcal{U}_\phi^{(p)}(A)\bigr\|_\infty 
\;\le\;\bigl(2\,\omega^{\rm (max)}(H)\bigr)^p 
\;\frac{\bigl\|\bigl[H,\,[H,A]\bigr]\bigr\|_\infty}{4\,\omega^{\rm (max)}(H)^2}.
\end{equation}

\end{lemma}

\medskip

\begin{proof} We prove the two parts of the lemma separately.

1. We will prove the $p^{\rm th}$-order derivative expression in Eq.~\eqref{eq:p-order-grad-A} by induction. First let us consider the single derivative of $\UC_\theta(A)=e^{i\theta H}Ae^{-i\theta H}$ with respect to $\theta$
\begin{align}
    \frac{d}{d\theta}\UC_\theta(A) %&= iHe^{i\theta H}Ae^{-i\theta H}-ie^{i\theta H}Ae^{-i\theta H}H \\
    &= ie^{i\theta H}HAe^{-i\theta H}-ie^{i\theta H}AHe^{-i\theta H} \\
    &=ie^{i\theta H}(HA-AH)e^{-i\theta H} \\
    &=i\UC_\theta([H,A])\;. \label{eq:single-derivative-application-proof} \;,
\end{align}
which shows that the form in Eq.~\eqref{eq:p-order-grad-A} is true for $p=1$.

Now, we assume that the $p$-th order derivative is given by 
\begin{equation}
     \frac{d^p }{d\th^p} \UC_\th(A)  = i^p \UC_{\th}([(H)^p,A])\;.
\end{equation}
Then, we can show that the $p+1$-th order derivative is given by
\begin{align}
    \frac{d^{p+1} }{d\th^{p+1}} \UC_\th(A) &= \frac{d }{d\th} i^p \UC_{\th}([(H)^p,A]) \\
    &=  i^{p+1} \UC_{\th}([H,[(H)^p,A]]) \\
    &=i^{p+1} \UC_{\th}([(H)^{p+1},A])\;,
\end{align}
where the second equality is given by Eq.~\eqref{eq:single-derivative-application-proof} for a single derivative (replacing $A$ by $[(H)^p,A]$), and the last equality uses the above definition of the nested commutator. This completes the proof of the first statement of the lemma. 

\medskip

2. Now we will prove the bound in Eq.~\eqref{eq:bound-p-order-grad-A} as follows

\begin{align}
    \left\| \UC_\phi^{(p)}(A) \right\|_\infty  &= \left\| i^p \UC_{\phi}([(H)^p,A]) \right\|_\infty \\
    &=\left\| [(H)^p,A] \right\|_\infty\; \\
    & \leq( 2\omega^{(\rm max)}(H) )^p \frac{\norm{[H,[H,A]]}}{4 ( \omega^{(\rm max)}(H) )^2} \, ,
\end{align}
where we used unitary invariance of the norm in the second equality and to reach the last inequality we invoke Lemma~\ref{lemma:bound-nested-commutator} with a slight modification. Precisely, in Lemma~\ref{lemma:bound-nested-commutator}, we showed  that $\norm{[(H)^p,A]} \leq (2\omega^{(\rm max)}(H))^p \norm{A}$ based on the recursive relation $\norm{[(H)^p,A]} = \norm{[H,[(H)^{p-1},A]]}\leq (2\omega^{(\rm max)}(H))\norm{[(H)^{p-1},A]}$. Thus, if $p\geq 2$ we can also have $\norm{[(H)^p,A]} \leq (2\omega^{(\rm max)}(H))^{p-2} \norm{[H,[H,A]]}$. This completes the proof of the lemma.

\end{proof}

We can generalize this bound for a composition of multiple unitary channels that depend on the same parameter.

\begin{lemma}[Multi-generator single-parameter channels]\label{lemma:bouded_unitary_product}
    Consider the parameterized unitary channel of the form $\mathcal{E}_\th(A) = U(\th)^\dagger A U(\th)$ with some real parameter $\th \in \mathbb{R}$ where $U(\theta) = \prod_{k=1}^K V_k e^{-i\th H_k}$, $\{H_k\}_{k=1}^K$ is a set of $K$ Hamiltonians and $\{V_k\}_{k=1}^K$ are some unitaries that do not depend on the parameter $\theta$. Then 
    we can use the multinomial coefficient introduced in Eq.~\eqref{eq:multinomial}, to write the $p^{\rm th}$-order derivative of $\mathcal{E}_\th(A)$ with respect to the parameter $\th$ evaluated at a point $\phi \in \mathbb{R}$ can be expressed as
   
    \begin{align}\label{eq:pth_deriv_corr}
        \channel_\phi^{(p)}(A) := \frac{d^p }{d\th^p} \mathcal{E}_\th(A) \Big|_{\th = \phi} = \sum_{\substack{\vec{a}=(a_1,\dots,a_K)\\
        a_1 + \dots + a_K = p}} \binom{p}{\vec{a}}\;\widetilde{\UC}_{\phi,K}^{(a_K)}\circ \widetilde{\UC}_{\phi,K-1}^{(a_{K-1})} \circ \dots \circ \widetilde{\UC}_{\phi,1}^{(a_1)}(A)
    \end{align}
    where we denote the unitary channel $\widetilde{\UC}_{\th,k}(A):= e^{i\th H_k}V_k^\dagger A V_k e^{-i\th H_k}$, and 
    $\widetilde{\UC}_{\phi,k}^{(a_k)}(A)$ as the $a_k^{th}$-order derivative of the channel $\widetilde{\UC}_{\th,k}(A)$ %$:= e^{i\th H_k}V_k^\dagger A V_k e^{-i\th H_k}$ 
    with respect to $\th$  evaluated at $\phi$. 
    
    Moreover, the infinity norm of the $p^{\rm th}$-order derivative $\channel_\phi^{(p)}(A)$ can be bounded as 
    \begin{align}
        \left\| \channel_\phi^{(p)}(A) \right\|_\infty \leq \left(\sum_{k=1}^K 2 \omega^{(\rm max)}(H_k)\right)^p \norm{A} \;.
    \end{align}
\end{lemma}

\begin{proof}

We begin by rewriting the parametrized channel $\EC_{\th}(A)$ as a composition of unitary channels. 
\begin{align}\label{eq:paramchanneldecomp}
  \mathcal{E}_{\th}(A)
  = \widetilde{\UC}_{\th,K} \circ  \widetilde{\UC}_{\th,K-1} \circ \dots \circ \widetilde{\UC}_{\th,1} (A)\;.
\end{align}
Next, we note that the product rule for derivatives also applies for composition of unitary channels as 
\begin{align}
    \frac{d }{d \th} \left(\widetilde{\UC}_{\th,2}\circ \widetilde{\UC}_{\th,1} (A)\right) &= \frac{d }{d \th} \left(e^{iH_2\th}V_2^\dagger e^{iH_1\th}V_1^\dagger A V_1 e^{-iH_1\th}V_2 e^{-iH_2\th}\right) \\
    &= iH_2e^{iH_2\th}V_2^\dagger e^{iH_1\th}V_1^\dagger A V_1 e^{-iH_1\th}V_2 e^{-iH_2\th} \\
    &+e^{iH_2\th}V_2^\dagger iH_1e^{iH_1\th}V_1^\dagger A V_1 e^{-iH_1\th}V_2 e^{-iH_2\th} \\
    &-e^{iH_2\th}V_2^\dagger e^{iH_1\th}V_1^\dagger A V_1 e^{-iH_1\th}iH_1V_2 e^{-iH_2\th} \\
    &-e^{iH_2\th}V_2^\dagger e^{iH_1\th}V_1^\dagger A V_1 e^{-iH_1\th}V_2 e^{-iH_2\th}iH_2 \\
    &=\widetilde{\UC}_{\th,2}^{(1)}\circ \widetilde{\UC}_{\th,1}(A)+\widetilde{\UC}_{\th,2}\circ \widetilde{\UC}_{\th,1}^{(1)}(A)\;,
\end{align}
where we applied the product rule in the second equality and the last equality is obtained by recognising both 
\begin{align}
  \widetilde{\UC}_{\th,2}^{(1)}\circ \widetilde{\UC}_{\th,1}(A)&=iH_2e^{iH_2\th}V_2^\dagger e^{iH_1\th}V_1^\dagger A V_1 e^{-iH_1\th}V_2 e^{-iH_2\th}-e^{iH_2\th}V_2^\dagger e^{iH_1\th}V_1^\dagger A V_1 e^{-iH_1\th}V_2 e^{-iH_2\th}iH_2 \\
  \widetilde{\UC}_{\th,2}\circ \widetilde{\UC}_{\th,1}^{(1)}(A)&=e^{iH_2\th}V_2^\dagger iH_1e^{iH_1\th}V_1^\dagger A V_1 e^{-iH_1\th}V_2 e^{-iH_2\th}-e^{iH_2\th}V_2^\dagger e^{iH_1\th}V_1^\dagger A V_1 e^{-iH_1\th}iH_1V_2 e^{-iH_2\th}\;.
\end{align}
Therefore, this generalises to the $p^{\rm th}$-order derivative with $K$ unitary channels given in Eq.~\eqref{eq:proof-lemma-correlated-1} as it is the case with the conventional general Leibnitz rule in Eq.~\eqref{eq:general-Leibniz-rule}.
Thus the $p^{\rm th}$-order derivative of $\mathcal{E}_\th(A)$ with respect to the parameter $\th$ evaluated at a point $\phi$ can be expressed as: 
    \begin{align} \label{eq:proof-lemma-correlated-1}
        \channel_\phi^{(p)}(A) =  \sum_{\substack{\vec{a}=(a_1,\dots,a_K)\\
        a_1 + \dots + a_K = p}} \binom{p}{\vec{a}}\;\widetilde{\UC}_{\phi,K}^{(a_K)}\circ \widetilde{\UC}_{\phi,K-1}^{(a_{K-1})} \circ \dots \circ \widetilde{\UC}_{\phi,1}^{(a_1)}(A)\;.
    \end{align}

We can further upper bound the infinity norm of the operator $\channel_\phi^{(p)}(A)$ (by iteratively using Lemmas~\ref{lemma:bound-nested-commutator} and~\ref{lemma:bounded_unitary}) as follows:
\begin{align}
    \norm{ \channel_\phi^{(p)}(A)} &\leq \sum_{\substack{\vec{a}=(a_1,\dots,a_K)\\a_1+\dots+a_K = p}} \binom{p}{\vec{a}} \norm{\widetilde{\UC}_{\phi,K}^{(a_K)}\circ\widetilde{\UC}_{\phi,K-1}^{(a_{K-1})}\circ\dots \circ \widetilde{\UC}_{\phi,1}^{(a_1)}(A)}\\
    & \leq \sum_{\substack{\vec{a}=(a_1,\dots,a_K)\\a_1+\dots+a_K = p}}  \binom{p}{\vec{a}} (2 \omega^{(\rm max)}(H_K))^{a_K} 
    \norm{\widetilde{\UC}_{\phi,K-1}^{(a_{K-1})}\circ\dots \circ \widetilde{\UC}_{\phi,1}^{(a_1)}(A)}\\
    %&\vdots\\
     &  \leq \sum_{\substack{\vec{a}=(a_1,\dots,a_K)\\a_1+\dots+a_K = p}}  \binom{p}{\vec{a}} \left(\prod_{k=1}^{K} \left(2\omega^{(\rm max)}(H_k)\right)^{a_k}\right) 
     \norm{A}\\
    & =  \left(\sum_{k=1}^K 2\omega^{(\rm max)}(H_k)\right)^p \norm{A}\;,
\end{align} 
where in the first inequality, we use the triangular inequality and in the next ones we apply iteratively Lemmas~\ref{lemma:bound-nested-commutator} and~\ref{lemma:bounded_unitary}. The final equality is obtained using the multinomial Newton formula \ref{eq:multinomial-theorem}.
\end{proof}

Finally, we can further extend this result for any second derivative of a composition of unitary channel with respect to two arbitrary parameters. 

\begin{lemma}[Two-parameter derivative bound]\label{lemma:double_deriv_bound}
     Consider the parameterized unitary channel of the form $\mathcal{E}_{\th_1,\th_2}(A) = U(\th_1,\th_2)^\dagger A U(\th_1,\th_2)$ with two real parameters $\th_1,\th_2 \in \mathbb{R}$ where $U(\theta_1,\th_2) = \prod_{k=1}^K V_k e^{-i\th_{\mathcal{S}(k)} H_k}$, $\{H_k\}_{k=1}^K$ is a set of $K$ Hamiltonians, $\{V_k\}_{k=1}^K$ are some unitaries that do not depend on the parameters $\th_1,\th_2$ and $\mathcal{S}: \{1,\dots,K\} \rightarrow \{1,2\}$ maps a Hamiltonian index $k$ to the associated parameter $\th_{\mathcal{S}(k)}$ . Then, the infinity norm of the  $p_1^{\rm th}$, $p_2^{\rm th}$-order partial derivative of the operator $\channel_{\th_1,\th_2}(A)$ evaluated at a point $\phi_1,\phi_2 \in \mathbb{R}$ can be upper bounded as
   
    \begin{align}
        \norm{\frac{\partial^{p_1+p_2} }{\partial\th_1^{p_1} \partial\th_2^{p_2}} \mathcal{E}_{\th_1,\th_2}(A) \Big|_{\th_1 = \phi_1,\th_2=\phi_2}} \leq \left(\sum_{l \in \mathcal{S}^{-1}(1)}2 \omega^{(\rm max)}(H_l)\right)^{p_1} \left(\sum_{l \in \mathcal{S}^{-1}(2)}2 \omega^{(\rm max)}(H_l)\right)^{p_2} \norm{A} 
    \end{align}
    where we denote the unitary channel $\widetilde{\UC}_{\th,k}(A):= e^{i\th H_k}V_k^\dagger A V_k e^{-i\th H_k}$, and 
    $\widetilde{\UC}_{\phi,k}^{(a_k)}$(A) as the $a_k^{th}$-order derivative of the channel $\widetilde{\UC}_{\th,k}(A)$ %$:= e^{i\th H_k}V_k^\dagger A V_k e^{-i\th H_k}$ 
    with respect to $\th$  evaluated at $\phi$.

\end{lemma}

\begin{proof} We separate this proof in the different crucial steps.

\textit{Step 1: Partial derivatives of the unitary.} We start by applying the general Leibniz rule (see Eq.~\eqref{eq:general-Leibniz-rule}) to compute the $\alpha_1^{\rm th}$, $\alpha_2^{\rm th}$-order partial derivative of the unitary $U(\th_1,\th_2)$ as done in Eq.~\eqref{eq:proof-lemma-correlated-1}, in the proof of the previous lemma. We also emphasize that we use the notation for the multinomial coefficient introduced in Eq.~\eqref{eq:multinomial}. Compared to the results in Lemma \ref{lemma:bouded_unitary_product}, we have to deal here with the differentiation with respect to not a single but two distinct parameters, requiring a trickier tracking of layer indices associated to each parameter defined by the map $\mathcal{S}$. To do so, we consider the unitary $\alpha_1^{\rm th}$-order partial derivative with respect to $\th_1$  while fixing $\th_2$. We use the short-hand notation of $X^{(a_k)}$ to define the $\partial_{a_k}X$.
    \begin{align}
        \frac{\partial^{\alpha_1}}{\partial\th_1^{\alpha_1}} U(\th_1,\th_2)  &= \sum_{\vec{a}\in\mathcal{A}(\alpha_1)} \binom{\alpha_1}{\vec{a}} \prod_{k=1}^K \left(V_k e^{-i\th_{\mathcal{S}(k)} H_k} \right)^{(a_k)}\;,
    \end{align}
    where we introduced a short hand notation to represent the condition of $\vec{a} = (a_1,... a_K)$ such that $a_1 + ... + a_K=\alpha_1$ used in the previous proof, particularly in Eq.~\eqref{eq:proof-lemma-correlated-1}. Indeed, each element of $\mathcal{A}(\alpha_1)$ is a vector of $K$ positive integers such that the components with indices in $\mathcal{S}^{-1}(1)$ are summing variables and the components with indices in $\mathcal{S}^{-1}(2)$ are set to zero. The additional condition $\sum_{k \in \mathcal{S}^{-1}(1)} a_k = \alpha_1$ leads to multinomial series of $|\mathcal{S}^{-1}(1)|$ integer variables. Formally, the notation is defined as $\mathcal{A}(\alpha_1) := \{\vec{a}\in\mathbb{N}^K | a_k=0\;, \forall k\in \mathcal{S}^{-1}(2)\; {\rm and}\; \sum_{k\in \mathcal{S}^{-1}(1)} a_k = \alpha_1 \}$.

    Next, we add on top the $\alpha_2^{\rm th}$-order partial derivative with respect to $\th_2$ for each term in the sum above and apply the Leibniz rule again. To do so, we define a similar notation as $\mathcal{A}(\alpha_1)$. We define the set $\mathcal{B}(\alpha_2) := \{\vec{b}\in\mathbb{N}^K | b_k=0\;, \forall k\in \mathcal{S}^{-1}(1)\; {\rm and}\; \sum_{k\in \mathcal{S}^{-1}(2)} b_k = \alpha_2 \}$. 
    \begin{align}
        \frac{\partial^{\alpha_1+\alpha_2}}{\partial\th_1^{\alpha_1} \partial\th_2^{\alpha_2}} U(\th_1,\th_2) &= \sum_{\vec{a}\in\mathcal{A}(\alpha_1)} \binom{\alpha_1}{\vec{a}} \frac{\partial^{\alpha_2}}{\partial \th_2^{\alpha_2}}\left(\prod_{k=1}^K \left(V_k e^{-i\th_{\mathcal{S}(k)} H_k} \right)^{(a_k)}\right)\\
        &= \sum_{\vec{a}\in\mathcal{A}(\alpha_1)} \binom{\alpha_1}{\vec{a}} \sum_{\vec{b}\in\mathcal{B}(\alpha_2)} \binom{\alpha_2}{\vec{b}} \prod_{k=1}^K \left(V_k e^{-i\th_{\mathcal{S}(k)} H_k} \right)^{(a_k+b_k)}\\
        &= \sum_{\vec{a} \in \mathcal{A}(\alpha_1)}  \sum_{\vec{b}\in \mathcal{B}(\alpha_2)} \binom{\alpha_1}{\vec{a}} \binom{\alpha_2}{\vec{b}} \prod_{k=1}^K \left(V_k e^{-i\th_{\mathcal{S}(k)} H_k} \right)^{(a_k+b_k)} \label{eq:partial_deriv_unitary}\;.
    \end{align}

\textit{Step 2: Bounding the norm of these derivatives.} Now we consider the norm of the unitary partial derivative, and we upper-bound it as follows
\begin{align}
    \norm{\left.\frac{\partial^{\alpha_1+\alpha_2}}{\partial\th_1^{\alpha_1} \partial\th_2^{\alpha_2}} [U(\th_1,\th_2)]\right|_{\th_1=\phi_1,\th_2=\phi_2}}
        & \leq \sum_{\vec{a} \in \mathcal{A}(\alpha_1)}  \sum_{\vec{b}\in \mathcal{B}(\alpha_2)} \binom{\alpha_1}{\vec{a}} \binom{\alpha_2}{\vec{b}} \norm{\prod_{k=1}^K \left.\left(V_k e^{-i\th_{\mathcal{S}(k)} H_k} \right)^{(a_k+b_k)}\right|_{\th_{\mathcal{S}(k)}=\phi_{\mathcal{S}(k)}} }\\
        & \leq \sum_{\vec{a} \in \mathcal{A}(\alpha_1)}  \sum_{\vec{b}\in \mathcal{B}(\alpha_2)} \binom{\alpha_1}{\vec{a}} \binom{\alpha_2}{\vec{b}} \norm{\prod_{k=1}^K \left(V_k e^{i\phi_{\mathcal{S}(k)} H_k} H_k^{a_k+b_k} \right) }\\
        & \leq \sum_{\vec{a} \in \mathcal{A}(\alpha_1)}  \sum_{\vec{b}\in \mathcal{B}(\alpha_2)} \binom{\alpha_1}{\vec{a}} \binom{\alpha_2}{\vec{b}} \prod_{k=1}^K \norm{\left(V_k e^{i\phi_{\mathcal{S}(k)} H_k} H_k^{a_k+b_k} \right) }\\
        &   \leq \sum_{\vec{a} \in \mathcal{A}(\alpha_1)}  \sum_{\vec{b}\in \mathcal{B}(\alpha_2)} \binom{\alpha_1}{\vec{a}} \binom{\alpha_2}{\vec{b}} \prod_{k=1}^K \left(\norm{ H_k}^{a_k} \norm{ H_k }^{b_k}\right)\\
        &= \left(\sum_{\vec{a} \in \mathcal{A}(\alpha_1)} \binom{\alpha_1}{\vec{a}} \prod_{k \in \mathcal{S}^{-1}(1)} \norm{ H_k}^{a_k} \right) \left(\sum_{\vec{b} \in \mathcal{B}(\alpha_2)} \binom{\alpha_2}{\vec{b}} \prod_{ k \in \mathcal{S}^{-1}(2)} \norm{ H_k}^{b_k} \right)\\
        &= \left(\sum_{k \in \mathcal{S}^{-1}(1)} \norm{H_k}\right)^{\alpha_1} \left(\sum_{k \in \mathcal{S}^{-1}(2)} \norm{H_k}\right)^{\alpha_2} ,
        \label{eq:derivative-double-param-U-norm-proof}
\end{align}
where in the first inequality we used the triangle inequality, in the second one, we explicitly computed the $a_k+b_k^{\rm th}$-order derivative of matrix exponentiation with respect to $\th_{\mathcal{S}(k)}$ evaluated at $\phi_{\mathcal{S}(k)}$. In the last two inequalities, we used the norm sub-multiplicativity and unitary invariance. Finally, we apply the multinomial theorem (see Eq.~\eqref{eq:multinomial-theorem}) to retrieve the last equality.

\textit{Step 3: Bounding the norm of the derivatives of the channel.} Now, we express the partial derivative of the operator $\channel_{\th_1,\th_2}(A)$ as a function of the unitary $U(\th_1,\th_2)$ partial derivatives using again the general Leibniz rule as follows: 
\begin{align}
    \frac{\partial^{p_1}}{\partial\th_1^{p_1}} [\channel_{\th_1,\th_2}(A)] &= \frac{\partial^{p_1}}{\partial\th_1^{p_1}} [U^\dagger(\th_1,\th_2) A U(\th_1,\th_2)] \\
    &= \sum_{\substack{c_1,d_1\\c_1 + d_1 = p_1}} \binom{p_1}{c_1,d_1} \frac{\partial^{c_1}}{\partial\th_1^{c_1}}\left[U^\dagger(\th_1,\th_2)\right] A \frac{\partial^{d_1}}{\partial\th_1^{d_1}}\left[U(\th_1,\th_2)\right]\;, \\
    \frac{\partial^{p_2}}{\partial\th_2^{p_2}}\left[\frac{\partial^{p_1}}{\partial\th_1^{p_1}} [\channel_{\th_1,\th_2}(A)]\right] &= \sum_{\substack{c_1,d_1\\c_1 + d_1 = p_1}} \binom{p_1}{c_1,d_1}\frac{\partial^{p_2}}{\partial\th_2^{p_2}}\left[\frac{\partial^{c_1}}{\partial\th_1^{c_1}}\left[U^\dagger(\th_1,\th_2)\right] A \frac{\partial^{d_1}}{\partial\th_1^{d_1}}\left[U(\th_1,\th_2)\right]\right]\\
    &= \sum_{\substack{c_1,d_1\\c_1 + d_1 = p_1}} \binom{p_1}{c_1,d_1} \sum_{\substack{c_2,d_2\\c_2+d_2 = p_2}} \binom{p_2}{c_2,d_2}  \frac{\partial^{c_1+c_2}}{\partial\th_1^{c_1}\partial\th_2^{c_2}} [U^\dagger(\th_1,\th_2)] A \frac{\partial^{d_1+d_2}}{\partial\th_1^{d_1}\partial\th_2^{d_2}} [U(\th_1,\th_2)]\;.
    %&= \sum_{\substack{c_1,d_1\\c_1 + d_1 = p_1}}  \sum_{\substack{c_2,d_2\\c_2+d_2 = p_2}} \binom{p_1}{c_1,d_1}\binom{p_2}{c_2,d_2} \left(\sum_{\vec{a} \in \mathcal{A}(c_1)}  \sum_{\vec{b}\in \mathcal{B}(c_2)} \binom{c_1}{\vec{a}} \binom{c_2}{\vec{b}} \prod_{k=1}^K \left(V_k e^{i\th_{\mathcal{S}(k)} H_k} \right)^{(a_k+b_k)}\right)
\end{align}

Hence, the infinity norm of the operator $\channel_{\th_1,\th_2}(A)$ partial derivative can be upper bounded as 
\begin{align}
    \norm{\frac{\partial^{p_2}}{\partial\th_2^{p_2}}\left[\frac{\partial^{p_1}}{\partial\th_1^{p_1}} [\channel_{\th_1,\th_2}(A)]\right]} &= \norm{\sum_{\substack{c_1,d_1\\c_1 + d_1 = p_1}}
    \sum_{\substack{c_2,d_2\\c_2+d_2 = p_2}}\binom{p_1}{c_1,d_1} \binom{p_2}{c_2,d_2}   \frac{\partial^{c_1+c_2}}{\partial\th_1^{c_1}\partial\th_2^{c_2}} [U^\dagger(\th_1,\th_2)] A \frac{\partial^{d_1+d_2}}{\partial\th_1^{d_1}\partial\th_2^{d_2}} [U(\th_1,\th_2)]}\\
    & \leq \sum_{\substack{c_1,d_1\\c_1 + d_1 = p_1}}
    \sum_{\substack{c_2,d_2\\c_2+d_2 = p_2}}\binom{p_1}{c_1,d_1} \binom{p_2}{c_2,d_2}   \norm{\frac{\partial^{c_1+c_2}}{\partial\th_1^{c_1}\partial\th_2^{c_2}} [U^\dagger(\th_1,\th_2)] A \frac{\partial^{d_1+d_2}}{\partial\th_1^{d_1}\partial\th_2^{d_2}} [U(\th_1,\th_2)]}\\
    & \leq \sum_{\substack{c_1,d_1\\c_1 + d_1 = p_1}}
    \sum_{\substack{c_2,d_2\\c_2+d_2 = p_2}}\binom{p_1}{c_1,d_1} \binom{p_2}{c_2,d_2}   \norm{\frac{\partial^{c_1+c_2}}{\partial\th_1^{c_1}\partial\th_2^{c_2}} [U^\dagger(\th_1,\th_2)]} \norm{ \frac{\partial^{d_1+d_2}}{\partial\th_1^{d_1}\partial\th_2^{d_2}} [U(\th_1,\th_2)]} \norm{A}\\
    & \leq \sum_{\substack{c_1,d_1\\c_1 + d_1 = p_1}}
    \sum_{\substack{c_2,d_2\\c_2+d_2 = p_2}}\binom{p_1}{c_1,d_1} \binom{p_2}{c_2,d_2}  \left(\sum_{k \in \mathcal{S}^{-1}(1)} \norm{H_k}\right)^{c_1+d_1} \left(\sum_{k \in \mathcal{S}^{-1}(2)} \norm{H_k}\right)^{c_2+d_2}\\
    &= \left(\sum_{k \in \mathcal{S}^{-1}(1)} \norm{H_k}\right)^{p_1} \left(\sum_{k \in \mathcal{S}^{-1}(2)} \norm{H_k}\right)^{p_2}\sum_{\substack{c_1,d_1\\c_1 + d_1 = p_1}}
    \sum_{\substack{c_2,d_2\\c_2+d_2 = p_2}}\binom{p_1}{c_1,d_1} \binom{p_2}{c_2,d_2}  \\
    &= \left(\sum_{k \in \mathcal{S}^{-1}(1)} 2\norm{H_k}\right)^{p_1} \left(\sum_{k \in \mathcal{S}^{-1}(2)} 2\norm{H_k}\right)^{p_2}\label{eq:double_deriv_norm}\;,
\end{align}
where we first used triangle inequality, the second inequality is due to sub-multiplicativity of the norm ($\|A B\|_{\infty} = \|A\|_{\infty} \|B\|_{\infty}$), the last inequality is obtained by applying Eq.~\eqref{eq:derivative-double-param-U-norm-proof}, and we retrieve the last equality using multinomial theorem introduced in Eq.~\eqref{eq:multinomial-theorem}.
Finally, $\norm{H_k}$ can be simply replaced in Eq.~\eqref{eq:double_deriv_norm} by $\omega^{(\rm max)}(H_k) = |\lambda_{\max}(H_k)-\lambda_{\min}(H_k)|$ by noticing that the operator $\channel_{\th_1,\th_2}(A)$ can be written as 
\begin{align}
    \channel_{\th_1,\th_2}(A) &= \left(\prod_{k=1}^K V_k e ^{-i \th_{\mathcal{S}(k)} H_k}\right)^\dagger A \left(\prod_{k=1}^K V_k e ^{-i \th_{\mathcal{S}(k)} H_k}\right)\\
    &= \left(\prod_{k=1}^K V_k e ^{-i \th_{\mathcal{S}(k)} (H_k - \lambda_{\rm min}(H_k)\1)}\right)^\dagger A \left(\prod_{k=1}^K V_k e ^{-i \th_{\mathcal{S}(k)} (H_k - \lambda_{\rm min}(H_k)\1)}\right)\;,
\end{align}
where we use the simple observation that $(- \lambda_{\rm min}(H_k)\1)$ commutes with all operators, and that the maximal frequency can be written as $\omega^{(\max)}(H_k) = \| H_k -\lambda_{\min}(H_k) \1 \|_{\infty} $.

\end{proof}

% Now, we apply the Taylor-Lagrange theorem to prove the following statement.  

\subsubsection{Lower-bounds on the variance for one variable unitary channels}
In this section we apply the results of Section~\ref{sec:lower_bound_variance_single_variable} to the variance of functions with unitary channels. We recall that in the previous Section~\ref{sec:bounded_derivatives_param_unitary_channels} we proved that unitary channels have bounded derivatives, and thus we can use this result and directly apply it to the results of Section~\ref{sec:lower_bound_variance_single_variable} to derive the following lower-bounds. We start lower-bounding a unitary channel with a single generator.

\begin{corollary}[Variance lower bound for a unitary channel with a single generator]\label{cor:var_LB_1param_1layer}
  Consider a loss function of the form $\mathcal{L}(\theta) = \Tr[\rho \; \UC_{\theta} (V^{\dagger}OV)]$  where $\UC_{\theta}(\cdot):= e^{i\th H}(\cdot)e^{-i\th H} $ denotes the unitary superoperator describing the evolution of Hamiltonian $H$ for time $\theta$. Further  consider uniformly sampling the parameter $\theta$ from the interval of length $2r$ centered around $\phi$ i.e., $\th \sim \uni(\phi, r)$. Provided that $r \leq \frac{3}{4 \omega^{(\rm max)}(H)}$,
     the variance of  $\mathcal{L}(\theta)$ can be lower bounded as 

    \begin{equation}
        \Var_{\th \sim \uni(\phi, r)}[\mathcal{L}(\theta)] \geq \frac{r^4}{45} \Tr[\rho \UC_{\phi}^{(2)}(V^{\dagger}O V) ]^2 -  \frac{2 (\omega^{(\rm max)})^2(H) \norm{[H,[H,e^{i\phi H} V^{\dagger}O V e^{-i\phi H}]]}^2}{135} r^6
    \end{equation}
    where $\UC_{\phi}^{(2)}$ denotes the second derivative of the superoperator $\UC_{\theta}$ with respect to $\theta$ evaluated at $\theta=\phi$.
\end{corollary}

\begin{proof}

Our strategy is to invoke Proposition~\ref{prop:variance_LB_bounded} by showing that the loss function $\LC(\th) = \Tr[\rho \; \UC_{\theta} (V^{\dagger}OV)]$ satisfies the assumption on the bounded derivatives in Eq.~\eqref{eq:assumption-bounded-grad} as well as identifying the associated $\alpha$ and $\gamma$. In particular, we prove below that in this case $\alpha =\frac{\norm{[H,[H,e^{i\phi H} V^{\dagger}O V e^{-i\phi H}]]}}{4(\omega^{(\rm max)}(H))^2}$ and $\gamma = 2 \omega^{(\rm max)}(H)$.
    
First, we note that our proof can be carried out around $\phi=0$ without any loss of generality. Specifically, one can absorb $e^{-i\phi H}$ within the non parametrized unitary $V$
    \begin{align}
        \mathcal{L}(\theta + \phi) &:= \Tr[\rho \UC_{\theta+ \phi}(V^{\dagger}OV)]\\
        &=  \Tr[\rho \UC_{\theta}(e^{i\phi H}V^{\dagger}OVe^{-i\phi H})]\\
        &= \Tr[\rho \UC_{\theta}(\tilde{V}^{\dagger}(\phi)O \tilde{V}(\phi))]
    \end{align}
where we introduce the unitary $\tilde{V}(\phi) = Ve^{-i\phi H}$ in the final equality. 
    %In what follows, we denote $\tilde{V}$ simply by $V$ for ease of notation.

Now, we can show that the loss function $\mathcal{L}(\th+\phi)$ has bounded derivatives of order $k\geq 2$ at $\th=0$.
\begin{align}
    \left|\mathcal{L}^{(k)}(\phi)
\right| &= \left|\Tr\left[ \rho \; \UC_0^{(k)}(\tilde{V}(\phi)^{\dagger}O \tilde{V}(\phi))\right]\right|\\
    &\leq \|\rho\|_1 \norm{\UC_0^{(k)}(\tilde{V}^{\dagger}(\phi)O \tilde{V}(\phi))}\\
    &= \norm{\UC_0^{(k)}(\tilde{V}^{\dagger}(\phi)O \tilde{V}(\phi))} \\
    & \leq (2\omega^{(\rm max)}(H))^k \frac{\norm{[H,[H,\tilde{V}^{\dagger}(\phi)O \tilde{V}(\phi)]]}}{4(\omega^{(\rm max)}(H))^2}
\end{align}
where in the first inequality, we use Holder inequality. The second equality is due to $\| \rho \|_1 = 1$ for any state $\rho$. To reach the last inequality, we apply Lemma~\ref{lemma:bounded_unitary}. 
Hence, the bounded derivative assumption in Eq.~\eqref{eq:assumption-bounded-grad} is satisfied with $\alpha =\frac{\norm{[H,[H,\tilde{V}^{\dagger}(\phi)O \tilde{V}(\phi)]]}}{4(\omega^{(\rm max)}(H))^2}$ and $\gamma = 2 \omega^{(\rm max)}(H)$, allowing us to invoke Proposition~\ref{prop:variance_LB_bounded}. Consequently, we obtain the promised bound 

\begin{equation}
     \Var[\mathcal{L}(\theta)] \geq \frac{r^4}{45} \Tr[\rho \; \UC_{\phi}^{(2)}(\tilde{V}^{\dagger}(\phi)O \tilde{V}(\phi))]^2 - \frac{2 (\omega^{(\rm max)})^2(H) \norm{[H,[H,\tilde{V}^{\dagger}(\phi)O \tilde{V}(\phi)]]}^2}{135} r^6
\end{equation} 
provided that the perturbation $r$ obeys the condition in Eq.~\eqref{eq:prop_r_remainder_cond}.

\end{proof}

If we further assume some locality assumptions on the generator and the observable as well as considering that the variance is computed over the interval centered around zero, we show a tighter  variance lower bound obtained in the following Corollary.

\begin{corollary}[Variance lower bound for a unitary channel with a single generator under locality assumptions]\label{cor:var_LB_1param_1layer_local}
  Consider a loss function of the form $\mathcal{L}(\theta) = \Tr[\rho \; \UC_{\theta} (O)]$  where $\UC_{\theta}(\cdot):= e^{i\th H}(\cdot)e^{-i\th H} $ denotes the unitary superoperator describing the evolution of a geometrically local Hamiltonian $H$ for time $\theta$. Precisely, the Hamiltonian $H$ can be written as $H = \sum_{j=1}^N h_j$
  where $h_j$ are local Pauli strings acting on a constant number of neighboring qubits. Let us also assume that the observable $O$ is a sum of $N_O$ geometrically local Pauli terms $P_i$, i.e. $O = \sum_{i=1}^{N_O} P_i$. Further  consider uniformly sampling the parameter $\theta$ from the interval of length $2r$ centered around zero i.e., $\th \sim \uni(0, r)$. Provided that $r \leq \frac{3}{4 s}$ where $s$ is a constant that depends on the generator $H$ defined as
  \begin{equation}
        s(O) = \max_{1 \leq i \leq N_O} |\{1\leq j \leq N \;, \{h_j,P_i\}=0  \}|\;,
    \end{equation}
     the variance of  $\mathcal{L}(\theta)$ can be lower bounded as 

    \begin{equation}
        \Var_{\th \sim \uni(0, r)}[\mathcal{L}(\theta)] \geq \frac{r^4}{45} \Tr[\rho \UC_{0}^{(2)}(O) ]^2 -  \frac{32 N_O^2 s^6(O) }{135} r^6
    \end{equation}
    where $\UC_{0}^{(2)}$ denotes the second derivative of the superoperator $\UC_{\theta}$ with respect to $\theta$ evaluated at $\theta=0$.
\end{corollary}

\begin{proof}
    Similarly to the proof of the previous Corollary \ref{cor:var_LB_1param_1layer}, our strategy consists in invoking Proposition \ref{prop:variance_LB_bounded} by showing that the loss function $\mathcal{L}(\thv)$ satisfies the bounded derivatives assumption in Eq.~\eqref{eq:assumption-bounded-grad}. The key difference compared to the previous Corollary \ref{cor:var_LB_1param_1layer}  is that  the bound on the derivatives is significantly smaller under  locality assumptions on the generator and the observable. In particular, we show that $\alpha= N_O $ and $\gamma = 2s(O)$.
    
    To do so, we first upper bounding the loss derivatives with a simple nested commutator and then invoke the upper bound from Lemma \ref{lemma:bound_nested_com_local}. Precisely, we get
    \begin{align}
    \left|\mathcal{L}^{(k)}(0)
    \right| &= \left|\Tr\left[ \rho \; \UC_0^{(k)}(O)\right]\right|\\
    &\leq \|\rho\|_1 \norm{\UC_0^{(k)}(O)}\label{eq:random_eq_imtiredofnames1}\\
    &= \norm{\UC_0^{(k)}(O)}\label{eq:random_eq_imtiredofnames2} \\
    &= \norm{[(H)^k,\sum_{i=1}^{N_O} P_i]}\\
    & \leq \sum_{i=1}^{N_O} \norm{[(H)^k,P_i]}\\
    &\leq N_O \max_{1 \leq i \leq N_O} \norm{[(H)^k,P_i]}\label{eq:random_eq_imtiredofnames3}\\
    & \leq N_O (2s(O))^k \;.
    \end{align}

where we used Holder's inequality in Eq.~\eqref{eq:random_eq_imtiredofnames1}. In Eq.~\eqref{eq:random_eq_imtiredofnames2} we use that the 1-norm of a density matrix is 1. In the next last equality we use Eq.~\eqref{eq:p-order-grad-A} from Lemma \ref{lemma:bounded_unitary}. In the first inequality we use the triangular inequality to say that the norm of a sum of terms is upperbonded by the sum of norms. In Eq.~\eqref{eq:random_eq_imtiredofnames3} we use that $\sum_i^n x_i \leq n \max x_i$ and in the last inequality we use Eq.~\eqref{eq:lemma_bound_com_local} from Lemma \ref{lemma:bound_nested_com_local}.

Hence, the bounded derivative assumption in Eq.~\eqref{eq:assumption-bounded-grad} is satisfied with $\alpha =N_O$ and $\gamma = 2s$, allowing us to invoke Proposition~\ref{prop:variance_LB_bounded}. Consequently, we obtain the promised bound 
%we can now apply  the result from  Proposition~\ref{prop:variance_LB_bounded} and obtain the final result.

      \begin{equation}
        \Var_{\th \sim \uni(0, r)}[\mathcal{L}(\theta)] \geq \frac{r^4}{45} \Tr[\rho \UC_{0}^{(2)}(O) ]^2 -  \frac{32 N_O^2 s^6(O) }{135} r^6
    \end{equation}
 
provided that the perturbation $r$ obeys the condition in Eq.~\eqref{eq:prop_r_remainder_cond}.

\end{proof}

Now, we extend the general variance lower bound for a single generator in Corollary \ref{cor:var_LB_1param_1layer} to multiple generators, still dependent on a single parameter.
% Once we have derived the result for a single generator, we can extend the lower-bound on the variance to multiple generators, still dependent

\begin{corollary}[Variance lower bound for a unitary channel with many generators]\label{cor:var_LB_1param_layers}
 Consider a loss function of the form $\LC(\th) = \Tr\left[\rho \; \EC_{\th}(O) \right]$ with a state $\rho$, an observable $O$ and the parameterized channel $\mathcal{E}_{\th}(\cdot)=\left(\prod_{k=1}^K V_k e^{-i \th H_k} \right)^\dagger (\cdot) \left(\prod_{k=1}^K V_k e^{-i \th H_k} \right)$  where the Hamiltonians $\{H_k\}_{k=1}^K$ do not necessarily commute. We consider uniformly sampling the parameter $\theta$ from the interval of length $2r$ centered around $\phi$ i.e., $\th \sim \uni(\phi, r)$.

Provided that $r \leq \frac{3}{4 \sum_{i=1}^K \omega^{(\rm max)}(H_i)}$, the variance of  $\mathcal{L}(\theta)$ can be lower bounded as 
    \begin{equation}
        \Var_{\th \sim [-r,r]}[\mathcal{L}(\theta)] \geq \frac{r^4}{45} \Tr[\rho \; \mathcal{E}_{\phi}^{(2)}(O) ]^2 - \frac{32 \left(\sum_{i=1}^K \omega^{(\rm max)}(H_i)\right)^6 \norm{O}^2}{135} r^6
    \end{equation}
    where $\mathcal{E}_{\phi}^{(2)}$ denotes the second derivative of the superoperator $\mathcal{E}_{\theta}$ with respect to $\theta$ evaluated at $\th=\phi$.
\end{corollary}

\begin{proof}
Our proof strategy is the same as the proof in Corollary~\ref{cor:var_LB_1param_1layer}, which is to show that with the loss function $\LC(\th) = \Tr\left[ \EC_{\th}(O) \rho \right]$ satisfying the assumption in Eq.~\eqref{eq:assumption-bounded-grad} and consequently invoke Proposition~\ref{prop:variance_LB_bounded}. 

First, note that, identically to the proof step in Corollary~\ref{cor:var_LB_1param_1layer}, without loss of generality we can consider the centered point around $\phi = 0$. This is since we can absorb $e^{-i \phi H_k}$ to a non-parametrized $V_k$ and redefine a non-parametrized unitary. 

Next, we show that the loss function has bounded derivatives as follows
\begin{align}
    \left|\mathcal{L}^{(p)}(\phi)\right|&= \left|\Tr[\rho \; \channel_\phi^{(p)}(O)] \right|\\
    & \leq \|\rho\|_1 \norm{\channel_\phi^{(p)}(O)}\\
    & \leq \left(\sum_{k=1}^K 2\omega^{(\rm max)}(H_k)\right)^p \norm{O} \;, 
\end{align} 
where the first inequality is by Holder inequality and the second inequality is due to $\| \rho \|_1 \leq 1$ for any state $\rho$, and finally we use Lemma~\ref{lemma:bouded_unitary_product} to reach the last inequality. One can see that the assumption in Eq.~\eqref{eq:assumption-bounded-grad} for employing Proposition~\ref{prop:variance_LB_bounded} is satisfied with $\alpha = \norm{O}$ and $\gamma = \sum_{k=1}^K 2\omega^{(\rm max)}(H_k)$. Hence, provided that the perturbation $r$ satisfies
\begin{equation}
    r \leq \frac{3}{4 \sum_{k=1}^K \omega^{(\rm max)}(H_k)} \;,
\end{equation} 
the variance lower bound of the loss can be expressed as
\begin{align}
    \Var_{\th \sim \uni(\phi,r)}[\mathcal{L}(\theta)] 
    %&\geq  \frac{r^4}{45} [ \mathcal{L}^{(2)}(0)]^2 - \frac{ \left(\sum_{i=1}^K 2\norm{H_i}\right)^6 \norm{A}^2}{270} r^6  \\
    &\geq \frac{r^4}{45} \Tr[\rho \; \mathcal{E}_{0}^{(2)}(A) ]^2 - \frac{2 \left(\sum_{i=1}^K \omega^{(\rm max)}(H_i)\right)^2 \norm{O}^2}{135} r^6 \;,
\end{align}
which completes the proof.

\end{proof}

\section{Approximate variance bound}\label{app:approximate-variance-bound}

 In this appendix we derive an approximate expression for the variance of a generic multivariate function $\LC(\thv)$ when its parameters $\thv\in\mathbb{R}^{\nparams}$ are sampled within a small hypercube. Concretely, we expand $\LC$ in a Taylor series around $\vec{0}$ and work up to third order in $\thv$. Our goal is to illustrate how derivatives multiplied by the leading orders in $x$ contribute to the overall variance in a small neighborhood around $\vec{0}$. 

We begin by writing the Taylor series of $\LC(\thv)$ around $\thv=\vec{0}$:
\begin{align}
    \LC(\thv) & = \sum_{k = 0}^{\infty} \sum_{i_1, i_2, ...,i_k}^\nparams \frac{1}{k!} \left( \frac{\partial^k  \LC(\thv)}{\partial \theta_{i_1} \partial \theta_{i_2} ... \partial \theta_{i_k}} \right)\bigg|_{\vec{\theta} = \vec{0}} (\theta_{i_1} \theta_{i_2} ... \theta_{i_k}) \\
    & = \LC(\vec{0}) + \sum_{i=1}^\nparams \LC^{(1)}_i \theta_i + \sum_{i,j=1}^\nparams \frac{1}{2} \LC^{(2)}_{ij} \theta_i \theta_j + \sum_{i,j,k=1}^\nparams \frac{1}{6} \LC^{(3)}_{ijk} \theta_i \theta_j \theta_k+ \OC(\theta_i\theta_j\theta_k\theta_l) \;,
\end{align}
where we defined partial derivatives evaluated at $\vec{0}$ as
\begin{equation}
    \LC^{(k)}_{i_1i_2...i_k} = \left( \frac{\partial^k  \LC(\thv)}{\partial \theta_{i_1} \partial \theta_{i_2} ... \partial \theta_{i_k}} \right)\bigg|_{\vec{\theta} = \vec{0}}\;.
\end{equation}
Note that each $\LC^{(k)}_{i_1i_2...i_k}$ is independent of $\thv$. Moreover, it is invariant under any permutations of its sub-indices i.e. for any permutation $\pi$ of $k$ elements we have $\LC^{(k)}_{i_1i_2...i_k}=\LC^{(k)}_{i_{\pi(1)}i_{\pi(2)}...i_{\pi(k)}}$ (equality of mixed partial derivatives).  
% We only express the terms explicitly up to the third order here and $\thv$ will be considered a small perturbation around the fixed point $\vec{0}$.
Note that considering the expansion around $\vec{0}$ is simply a way to ease notations and does not lose any generality since if we are interested in a perturbation around a fixed point $\vec{\phi}$ we can always redefine $\thv - \vec{\phi} \rightarrow \thv$.

We are interested in the variance with respect to a small parameter region around the fixed point. In particular, let us denote the small region of interest $\vol(\vec{0},r) = \{ \thv \;|\; \theta_i \in [-r , r] \}$ with $r$ being the perturbation strength and  $\uni(\vec{0}, r)$ denote the uniform distribution where each parameter is sampled from this region. 
Before jumping into the variance calculation of the function itself, we note the following which will be useful:
\begin{align}
\label{eq:moment-single-param-proof}
    \Ebb_{\theta_i \sim \uni(\vec{0}, r)} \left[ \theta_i ^\mindex \right] = \left\{ 
    \begin{array}{ll}
         \frac{r^\mindex}{\mindex+1} \;\;&{\rm if\;}\mindex{\rm\;even} \;, \\
         0 \;\;&{\rm if\;}\mindex{\rm\;odd}\;.
    \end{array}
    \right.
\end{align}

Now, the variance of the function can be \textit{approximately} computed as 
\begin{align}
    \Var_{\thv \sim \uni(\vec{0}, r)} \left[ \LC(\thv) \right] & = \Var_{\thv \sim \uni(\vec{0}, r)} \left[  \LC(\vec{0}) + \sum_{i=1}^\nparams \LC^{(1)}_i \theta_i + \sum_{i,j=1}^\nparams \frac{1}{2} \LC^{(2)}_{ij} \theta_i \theta_j + \sum_{i,j,k=1}^\nparams \frac{1}{6} \LC^{(3)}_{ijk} \theta_i \theta_j \theta_k+ \OC(\theta_i\theta_j\theta_k\theta_l) \right] \\
    &=\Var_{\thv \sim \uni(\vec{0}, r)} \left[   \sum_{i=1}^\nparams \LC^{(1)}_i \theta_i + \sum_{i,j=1}^\nparams \frac{1}{2} \LC^{(2)}_{ij} \theta_i \theta_j + \sum_{i,j,k=1}^\nparams \frac{1}{6} \LC^{(3)}_{ijk} \theta_i \theta_j \theta_k+ \OC(\theta_i\theta_j\theta_k\theta_l) \right]\\
    & = \Var_{\thv}\left[ \sum_i \LC_i^{(1)} \theta_i \right] + \Var_{\thv}\left[ \sum_{i,j} \frac{1}{2} \LC^{(2)}_{ij} \theta_i \theta_j \right] +2\text{Cov}_{\thv}\left[\sum_i \LC_i^{(1)} \theta_i,\, \sum_{i,j,k} \frac{1}{6} \LC^{(3)}_{ijk} \theta_i \theta_j \theta_k \right]+ \OC(r^6) \\
    & = \sum_i \LC_i^{(1) \, 2}\,  \Var_{\thv}\left[ \theta_i \right]+ \frac{1}{4}\sum_{i\neq j}\LC_{ij}^{(2)}(\LC_{ij}^{(2)}+\LC_{ji}^{(2)})\, \Var_{\thv}\left[ \theta_i \theta_j \right]+ \frac{1}{4}\sum_{i}\LC_{ii}^{(2)\, 2}\, \Var_{\thv}\left[ \theta_i^2  \right] \\ 
    &\; + \frac{1}{3}\sum_{i}\LC_i^{(1)}\LC_{iii}^{(3)}\text{Cov}_{\thv}[\theta_i,\theta_i^3]+\frac{1}{3}\sum_{i\neq j} \LC_i^{(1)}\left(\LC_{ijj}^{(3)}+\LC_{jij}^{(3)}+\LC_{jji}^{(3)}\right)\text{Cov}_{\thv}[\theta_i,\theta_i\theta_j^2]+ \OC(r^6)\;,
    \end{align}
where the second equality is obtained by removing the constant term which does not affect the variance. In the third equality, we used the fact that odd order of perturbation terms vanishes and we explicitly kept the second and fourth order terms to get a remainder of order $\OC(r^6)$. In the fourth equality, we removed the terms that cancel due independence between parameters and distribution symmetry e.g. $\text{Cov}_{\thv}[\theta_i,\theta_j]=\delta_{ij}\Var_{\thv}\left[ \theta_i \right]$ or $\text{Cov}_{\thv}[\theta_i\theta_j,\theta_k\theta_l]$ is non-vanishing if either $k=i\neq j=l$ or $k=j\neq i=l$ or $i=j=k=l$.
Now, let us compute expectations and variances involved using Eq.~\eqref{eq:moment-single-param-proof} and simplify the expression using equality of mixed partials mentioned above (e.g. $\LC_{ij}^{(2)}=\LC_{ji}^{(2)}$). Therefore, previous expression becomes
    \begin{align}
   \Var_{\thv \sim \uni(\vec{0}, r)} \left[ \LC(\thv) \right] & = \sum_i \LC^{(1) \, 2}_i \left(\frac{r^2}{3}\right) + \frac{1}{2}\sum_{i\neq j}\LC_{ij}^{(2) \,2} \left( \frac{r^4}{9} \right)+ \frac{1}{4}\sum_{i}\LC_{ii}^{(2) \,2} \left( \frac{4r^4}{45} \right) \\
    &\; +\frac{1}{3}\sum_i \LC_i^{(1)}\LC_{iii}^{(3)}\frac{r^4}{5}+\sum_{i\neq j} \LC_i^{(1)}\LC_{ijj}^{(3)}\frac{r^4}{9}+ \OC(r^6) \\
    &=\left(\frac{r^2}{3}\right)\sum_i \LC^{(1) \, 2}_i  + \left(\frac{r^4}{9}\right)\sum_{i,j}\left(\frac{1}{2}\LC_{ij}^{(2) \,2} \left( 1-\frac{3\delta_{ij}}{5} \right)+  \LC_i^{(1)}\LC_{ijj}^{(3)}\left( 1-\frac{2\delta_{ij}}{5} \right)\right)+ \OC(r^6)\;,
    \label{eq:approximate-variance-order-4}
\end{align}
  where, in the last equality, we have grouped the terms according to their order of perturbation and introduced the Kronecker delta to write it in a more compact form.

This \textit{approximate} variance expression is rather general and quite useful to give some insights about trainability for small perturbation $r$ as long as we have some information about the first and second derivatives around $\vec{0}$. The flexibility of this expression comes in as we can play around what the function $\LC(\thv)$ is (i.e., it does not have to always be a loss) and what the fixed point $\vec{0}$ is (i.e., it does not have to always be setting all parameters in the ansatz to be zeros).

\section{General variance lower bound and theoretical guarantees on the region of attraction}
In this section, we present the formal version of Theorem~\ref{th:var} as well as the formal version of Corollary~\ref{cor:var_minimum} in the main text, together with further discussion and their detailed derivations. 

\medskip

The section is structured as follows:
\begin{itemize}
    \item In Appendix~\ref{appx:main-thm-summay}, we provide the formal statements for Theorem~\ref{th:var} and Corollary~\ref{cor:var_minimum}, as well as further discuss their technical subtleties.
    \item In Appendix~\ref{appx:main-thm-overview},  we go over the proof strategy of the main theorems which hopefully helps preparing the readers for the battle to come. 
    \item In Appendix~\ref{appendix:proof-analytical-results}, the full detailed proofs of analytical results are presented.
\end{itemize}

\subsection{Summary of the key technical results}\label{appx:main-thm-summay}
\subsubsection{General variance lower bound around any point}

The formal version of Theorem~\ref{th:var} is separated into two parts: Theorem~\ref{th:var_formal} which describes circuits with uncorrelated or spatial correlated parameters, and Theorem~\ref{th:var_formal_cor} which applies more generally to circuits with correlated parameters including temporal correlated ones (see Fig.~\ref{fig:corr_schematic}). 

\medskip

We now present them in details and after have some further discussions, including the main technical difference between them.

\medskip

\noindent For the uncorrelated/spatial correlated parameter case, we have
\begin{theorem}\label{th:var_formal} [Lower bound on the loss variance for uncorrelated/spatial correlated parameters, Formal]
Consider a loss function of the form $\LC(\thv) = \Tr\left[U(\thv) \rho U^\dagger(\thv) O \right]$ as defined in Eq.~\eqref{eq:loss} with a state $\rho$, an observable $O$ and a parametrized circuit of the form in Eq.~\eqref{eq:circuit} with $\nparams$ parameters and $\nHam$ generators such that $\nparams=\nHam$. That is, all parameters are uncorrelated. We further consider uniformly sampling parameters in a hypercube of width $2r$ around any arbitrary point on the loss landscape $\vec{\phi}$ i.e., $\thv \sim \uni(\vec{\phi},r)$.

\medskip

\noindent If $r$ satisfies 
\begin{equation}\label{eq:th1_cond_uncor}
    r^2 \leq r_{{\rm patch}}^{2} := \min_{l \in \Lambda} \frac{9c_l(\vec{\phi})^2}{8 c_l(\vec{\phi})\left(4 \left[\omega_{l}^{\rm (max)}\right]^2 \sum_{\mindex=1}^{l-1} \left[\omega_{\mu}^{\rm (eff) }(\vec{\phi})\right]^2  + \sum_{\mindex=l+1}^{\nparams}\left[\widetilde{\omega}_{l, \mu }^{(\rm eff)}(\vec{\phi})\right]^2  \right) + 24\beta_l} \;,
\end{equation}
for  any subset of layers $\Lambda \subset \{1,\dots,\nparams\}$, then the variance of $\LC(\thv)$ over the hypercube $\vol(\vec{\phi}, r)$  can be bounded as 
\begin{equation}\label{eq:var-lower-bound-general}
     \Var_{\thv \sim \uni(\vec{\phi}, r)}[\LC(\thv)]  \geq \sum_{l \in \Lambda} \frac{r^4}{45} \left(\left(c_l(\vec{\phi})- \frac{1}{6} \left(4 \left[\omega_{l}^{\rm (max) }\right]^2 \sum_{\mindex=1}^{l-1} (\omega_{\mu}^{\rm (eff) }(\vec{\phi}))^2  + \sum_{\mindex=l+1}^{\nparams}(\widetilde{\omega}_{l, \mu }^{(\rm eff)}(\vec{\phi}))^2  \right)r^2\right)^2 - \beta_l r^2\right)
\end{equation}
Moreover, over the hypercube $\vol(\vec{\phi}, r_{\rm patch})$, the variance is lower bounded as follows   
\begin{equation}\label{eq:th1_LB_var_patch}
    \Var_{\thv \sim \uni(\vec{\phi}, r)}[\LC(\thv)]  \geq  \frac{1}{72}    \left( \sum_{l \in \Lambda} c_l(\vec{\phi})^2 \right)   r_{\rm patch}^4 
\end{equation}
    with 
    \begin{align}
        c_l(\vec{\phi}) &:= \left|\left.\left(\frac{\partial^2\LC(\thv)}{\partial\th_l^2}\right)\right|_{\thv=\vec{\phi}}\right| \;, \\ \label{eq:omega-max-corr-def}
       \omega_{\mu}^{(\rm max)}&:= \omega^{(\max)}(H_\mu) =\lambda_{\max}(H_\mu) - \lambda_{\min}(H_\mu) \;, \\ \label{eq:omega-eff-1-def}
         \omega_{\mu}^{\rm (eff) }(\vec{\phi})&= \sqrt{\norm{\left.\frac{\partial^2 [U(\thv)^\dagger O U(\thv)]}{\partial \th_\mu^2}\right|_{\thv= \vec{\phi}}}} \;,\\ \label{eq:omega-eff-2-def}
    \widetilde{\omega}_{l, \mu }^{(\rm eff)}(\vec{\phi}) &= \sqrt{\norm{ \left.\frac{\partial^4 [U(\thv)^\dagger O U(\thv)]}{\partial \th_\mu^2 \partial \th_l^2}\right|_{\thv= \vec{\phi}}}} \;,\\
    \beta_l &=\begin{cases}
       \frac{32( \omega_{l}^{\rm (max) })^{6}  \norm{O}^2}{3}\; {\rm if }\; l>1 \;,\\
       \frac{2 \left[\omega^{(\rm max)}_{1}\right]^2  (\omega^{\rm (eff)}_{1}(\phi_1))^4}{3} \; {\rm if }\; l=1\;.
    \end{cases}  \label{eq:beta_l}
\end{align}

\medskip

\noindent Lastly, for 
$\norm{O}, \norm{H_k} \in \mathcal{O}(\poly(n))$ and $c_q(\vec{\phi}) \in \Omega\left(\frac{1}{\poly(n)}\right)$ for some $q \in \Lambda$, if we choose $r$ such that 
\begin{equation}\label{eq:proof-thm1-region-attraction0}
    r \in \Theta\left(\frac{1}{ \poly(n) \cdot \sqrt{m }}\right)
\end{equation}
then the variance of the loss function is lower bounded as 
\begin{align} \label{eq:proof-thm1-region-attraction}
     \Var_{\thv \sim \uni(\vec{\phi}, r)}[\LC(\thv)] \in \Omega\left(\frac{1}{\poly(n) \cdot m^2}\right) \;.
\end{align}    
\end{theorem}

\bigskip

\bigskip

\noindent For arbitrary correlated parameter case, we have
\begin{theorem}[Lower bound on the loss variance for correlated parameters, Formal]\label{th:var_formal_cor}
Consider a generic loss $\LC(\thv)$ of the form in Eq.~\eqref{eq:loss} and a parametrized quantum circuit $U(\thv)$ of the form of Eq.~\eqref{eq:circuit} with $\nparams$ parameters. We consider uniformly sampling parameters in a hypercube of width $2r$ around any point of the landscape $\vec{\phi} \in \mathbb{R}^\nparams$.
If $r$ satisfies 
\begin{equation}\label{eq:th1_cond_cor}
 r^2 \leq r^2_{\rm patch} := \frac{9c^2_{p}(\vec{\phi})}{8\left(16 c_{p}(\vec{\phi}) \norm{O} \left[\omega^{(\rm max)}_{p}\right]^2\sum_{q\neq p} \left[\omega^{(\rm max)}_{q}\right]^2 + 32 \norm{O}^2 \left[\omega^{(\rm max)}_{p}\right]^6 \right) }
\end{equation}
for any $\paramindex \in \{1,\dots,\nparams\}$, then the variance of $\LC(\thv)$ over the hypercube $\vol(\vec{\phi}, r)$  can be bounded as 
\begin{equation}\label{eq:thm-correlated-variance}
     \Var_{\theta \sim \uni(\vec{\phi}, r)}[\LC(\thv)]  \geq  \frac{r^4}{45} \left(\left(c_p(\vec{\phi})-\frac{8}{3} \norm{O}  \left[ \omega^{\rm (max)}_{p} \right]^2 \left(\sum_{q \neq p}  \left[ \omega^{\rm (max)}_{q} \right]^2 \right)r^2\right)^2 - \frac{32}{3} \norm{O}^2 \left[ \omega^{\rm (max)}_{p} \right]^6 r^2\right)
\end{equation}

 Moreover, within the hypercube $\vol(\vec{\phi}, r_{\rm patch})$, the variance is lower bounded as follows   
\begin{equation}
    \Var_{\sim \uni(\vec{\phi}, r_{\rm patch})}[\LC(\thv)]  \geq \frac{1}{72}    c_p(\vec{\phi})^2   r_{\rm patch}^4
\end{equation}

    with 
    \begin{align}
        c_p(\vec{\phi}) &:= \left|\left.\left(\frac{\partial^2\LC(\thv)}{\partial\th_p^2}\right)\right|_{\thv=\vec{\phi}}\right| \\
        \omega^{\rm (max)}_{q}  &:=\sum_{l \in \mathcal{S}^{-1}(q)} \omega^{(\rm max)}(H_l)
\end{align}

\medskip

\noindent Lastly, for 
$\norm{O}, \norm{H_k} \in \mathcal{O}(\poly(n))$ and $c_p(\vec{\phi}) \in \Omega\left(\frac{1}{\poly(n)}\right)$. If we choose $r$ such that 
\begin{equation}\label{eq:proof-thm12-region-attraction0}
    r \in \Theta\left(\frac{1}{ \poly(n) \cdot \sqrt{m }}\right)
\end{equation}
then the variance of the loss function is lower bounded as 
\begin{align}\label{eq:proof-thm12-region-attraction}
     \Var_{\theta \sim \uni(\vec{\phi}, r)}[\LC(\thv)] \in \Omega\left(\frac{1}{\poly(n) m^2}\right) \;.
\end{align}

\end{theorem}

\bigskip

\bigskip

\noindent\paragraph*{\underline{Key technical notes.}} After presenting them in details, we discuss a few important remarks. 
\begin{enumerate}
    \item It is worth to emphasize that for each theorem it contains analytical formulae of the variance lower bounds together with the region of patches for \textit{any point of interest} on the quantum landscape. 
    These analytical expressions 
    in turn depend on various different quantities including second order derivatives at that point, maximal frequencies and some forms of effective frequencies in the case of Theorem~\ref{th:var_formal}.
    \item The theoretical guarantee on the region with substantial gradients follows with a further assumption that the point of interest has at least one non-vanishing second derivative together with polynomially bounded infinity norms of an observable and a corresponding generator.
    \item Indeed, Theorem~\ref{th:var_formal_cor} works more generally and also includes the setting in Theorem~\ref{th:var_formal}. However, due to being more specific with the circuit structure in the proof, Theorem~\ref{th:var_formal} is expressed in terms of effective frequencies (as compared to only maximal frequencies in Theorem~\ref{th:var_formal_cor}), leading to a larger size of a patch with guaranteed polynomial gradients. In particular, albeit the same asymptotic polynomial scalings, the prefactor obtained with Theorem~\ref{th:var_formal} is larger in general and could result in polynomially large patches when coming to analyze specific circuit architectures.
\end{enumerate}

\subsubsection{Theoretical guarantees of the region of attraction with the polynomial width around a minimum}
\label{subsec:corollary-1-formal-app}
By applying Theorem~\ref{th:var_formal} and Theorem~\ref{th:var_formal_cor} to a minimum of a quantum landscape together with some mild assumptions, we obtain a theoretical guarantee on a region of attraction with a polynomial width. This is presented formally in the following corollary.

\setcounter{corollary}{0}
\begin{corollary}[Scaling of regions of attraction, Formal]\label{cor:var_minimum_app}
Consider a generic loss $\LC(\thv)$ of the form in Eq.~\eqref{eq:loss} with a parametrized quantum circuit $U(\thv)$ of the form in Eq.~\eqref{eq:circuit} and an observable $O$. Further consider a minimum point of the landscape $\thv^*$ with sufficiently large overlap between an evolved state and the ground state $|\lambda_1\rangle$ of $O$. In particular, we assume that the overlap $\Tr[U(\thv^*)\rho U^\dagger(\thv^*)\ketbra{\lambda_1}] = 1-|\epsilon|^2$ is large enough such that
\begin{align}
    |\epsilon| < \frac{\gap \Var_\rho(H_M)}{(6\|O\|_{\infty} + 11\gap)\|H_M\|_{\infty}^2}\;.
\end{align}

Under the following assumptions:
\begin{itemize}
    \item Assumption~1: The gap $\gap$ between the ground state and first excited state energies of $O$ vanishes at worst polynomially with the number of qubits. The ground state is also assumed to be non-degenerate. In particular, with the eigen decomposition of $O = \sum_{i=1}^{2^n} \lambda_i| \lambda_i \rangle\langle\lambda_i|$ where $\lambda_i$ is an eigen energy with an associated eigenstate $|\lambda_i\rangle$ labeled in increasing order i.e.,  $\lambda_1 < \lambda_2 \leq \dots \leq \lambda_{2^n}$, the gap satisfies
    \begin{align}
        \gap = \lambda_2 - \lambda_1 \in \Omega\left( \frac{1}{\poly(n)}\right) \;.
    \end{align}
    \item Assumption~2: The interaction between the gate closest to the initial state and the initial state itself has to be non-trivial. More specifically, given the initial state $\rho$ and the gate generator (closest to the initial state) $H_M$, we require
    \begin{align}
        \Var_\rho(H_M) = \Tr[\rho H_M^2] - \Tr[\rho H_M]^2\in\Omega(1/{\rm poly}(n)) \;.
    \end{align}
    Additionally, the parameter associated with the gate generator $H_M$ is required to be non-temporally correlated with other parameters. Remark that all other parameters can be arbitrarily correlated among themselves. 
\end{itemize}

Then, for the patch's width $2r^*_{\rm patch}$ around $\thv^*$ such that 
\begin{equation}
        (r_{\rm patch}^*)^2 =  \frac{9 [ \gap \Var_{\rho}(H) - |\epsilon|\|H\|_{\infty}^2 (11\gap + 6\|O\|_{\infty}) ]^2 }{128 \norm{O}^2 \left[ \omega^{\rm (max)}_{m} \right]^4 \sum_{j=1}^{\nparams-1} \left[ \omega^{\rm (max)}_{j} \right]^2 }
    \end{equation}
the loss variance with the uniformly sampled parameters inside the patch scales as
\begin{align}
    \Var_{\thv \sim \uni(\thv^*,r_{\rm patch}^*)}[\mathcal{L}(\thv)] \geq \frac{[ \gap \Var_{\rho}(H) - |\epsilon|\|H\|_{\infty}^2 (11\gap + 6\|O\|_{\infty}) ]^2}{18}  (r_{\rm patch}^*)^4\;.
\end{align}
Crucially, for at most polynomial number of parameters $m \in \OC(\poly(n))$ and a bounded infinity norm of the observable $\| O\|_{\infty} \in \OC(\poly(n))$, this variance around this region is guaranteed to scale at least polynomially with the qubits
\begin{align}
    \Var_{\thv \sim \uni(\thv^*,r_{\rm patch}^*)}[\mathcal{L}(\thv)] \in \Omega\left( \frac{1}{\sqrt{m}\poly( n)}\right) \;.
\end{align}
\end{corollary}

\subsection{Proof overview for the main theorems}\label{appx:main-thm-overview}

The proofs of both main theorems consist of a series of long mathematical steps as well as a bunch of hideous (in a cute way) notations, which could be found cumbersome at some point. Nevertheless, they both share some common proof structure. In particular, the proof of Theorem~\ref{th:var_formal_cor} is a generalization of the proof of Theorem~\ref{th:var_formal} in the sense that long-range temporal correlations between parameters are taken into account. This amounts to a more complicated initial setting but once that has been done the rest of the proof largely follows the same structure. In other words, the proof of Theorem~\ref{th:var_formal} can be seen as a \textit{warm-start} for the proof of Theorem~\ref{th:var_formal_cor}.

To aid the keen readers who wish to embark on a long journey in either proof, we use the proof structure for Theorem~\ref{th:var_formal} as an example and summarize the key proof steps below.
\begin{enumerate}
    \item \textbf{Expressing the loss variance as a sum of single parameter variance contributions} with an average over the remaining parameters. In particular, by using 
    Proposition~\ref{prop:var_decomp}, the loss variance can be expressed as
    \begin{equation}%\label{eq:VarM}
    \Var_{\thv \sim \uni(\vec{0},r)}\left[\mathcal{L}(\boldsymbol{\theta})\right]  =\sum_{l=1}^{\nparams} \Ebb_{\overline{l}}[\Var_l[\Ebb_{1,\dots,l-1}[\mathcal{L}(\thv)]]] \;,
    \end{equation}
    with $\Var_{l} [\cdot] := \Var_{\theta_l \sim \uni(0,r)} [\cdot]$ and $ \Ebb_{l}[\cdot] = \Ebb_{\theta_l \sim \uni(0,r)}[\cdot]$, and $\overline{l}:= l+1,\dots,\nparams$.
    \item \textbf{Lower bounding each individual single variance term} (without the average on later parameters) $\Var_l[\Ebb_{1,\dots,l-1}[\mathcal{L}(\thv)]]$ by using Proposition~\ref{prop:variance_LB_bounded}.
    \item \textbf{Lower bounding $\Ebb_{\overline{l}}[\Var_l[\Ebb_{1,\dots,l-1}[\mathcal{L}(\thv)]]]$} (i.e., taking into account the rest of parameters) which results in averaged double derivative terms  $\Ebb_{\thv}\left[\left.\left(\frac{d^2\mathcal{L}(\thv)}{d\theta_l^2}\right)\right|_{\theta_l=\phi_l}\right]$. 
    \item \textbf{Further lower bounding the averaged double derivative term $\Ebb_{\thv}\left[\left.\left(\frac{d^2\mathcal{L}(\thv)}{d\theta_l^2}\right)\right|_{\theta_l=\phi_l}\right]$} by splitting the term into the constant and the perturbation terms. This leads to a generic lower bound in Eq.~\eqref{eq:var-lower-bound-general} presented in the theorem.
    \item \textbf{Proving a theoretical guarantee on the sub-region} using the generic variance bound as well as the assumption on a non-vanishing second derivative.
\end{enumerate}

\subsection{Proofs of analytical results}\label{appendix:proof-analytical-results}

This sub-section contains the full detailed proofs of our theoretical statements including
\begin{itemize}
    \item Theorem~\ref{th:var_formal} in Appendix~\ref{appendix:warm_up_proof_th1}
    \item Theorem~\ref{th:var_formal_cor} in Appendix~\ref{appendix:lower_bound_correlated}
    \item Corollary~\ref{cor:var_minimum_app} in Appendix~\ref{subsec:var_minimum_app}.
\end{itemize}

\subsubsection{Proof of Theorem~\ref{th:var_formal}: the  uncorrelated/spatial correlated part of Theorem~\ref{th:var}}\label{appendix:warm_up_proof_th1}

In this section, we provide the exact formula for the variance lower bound and for a generic quantum circuit with uncorrelated/spatial correlated parameters and the scaling of the region of attraction around different points in the parameter hyperspace. 

%\subsection{Exact formula for the lower bound}
\begin{proof}
We recall that the loss function in Eq.~\eqref{eq:loss} is of the form 
    \begin{equation}
        \mathcal{L}(\thv)= \Tr[\rho U^{\dagger}(\thv) O U(\thv)] \;.
    \end{equation}
   
In this theorem, we consider the uncorrelated/spatial correlated parameters where the parametrized quantum circuit in Eq.~\eqref{eq:circuit} has $\nparams$ distinct parameters and $\nHam=\nparams$ generators with a one-to-one trivial map $\SC$ between generators and the variational parameters ($m$ = $M$ and $\SC(l) = l$)
    \begin{equation}
        U(\thv) = \prod_{l=1}^\nparams U_l(\theta_l)=\prod_{l=1}^\nparams V_l e^{-i\theta_l H_l}\;.
    \end{equation}
%\stcom{Maybe a few words to say how $m=M$ case can be used to describe spatial correlations ?}

Crucially, we emphasize that our convention for labeling the gate sequence in the proof may differ from the traditional convention. That is, we use the convention that the circuit gate $U_{\nparams}$ is the first gate acting on the initial state whereas $U_1$ is the gate closest to the observable. Indeed, we adopt this convention because in the following proof we will mainly focus on studying the back propagated observable in the Heisenberg picture $U(\vtheta)^\dagger O U(\vtheta) = U_\nparams^\dagger \dots U_1^\dagger O U_1 \dots U_\nparams$. Hence, it is more natural to start indexing the unitaries/generators with reference to the observable. %\stcom{In addition, ...}  

\medskip

To study the landscape of the loss $\mathcal{L}(\boldsymbol{\theta})$ around a point of interest $\vec{\phi}$, it is convenient to re-express the parameterized gates as the perturbation around the point $\vphi$ i.e., $\theta_l = \phi_l + \delta \theta_l$ for all $l$, and absorb $e^{-i \phi_l H_l}$ within a non parameterized unitary $V_l$. Note that for simplicity in notation we will also redefine the perturbation variable $\delta \theta_i \rightarrow \theta_i$. In particular, we have
\begin{align}
    U(\thv) & = \prod_{l=1}^\nparams V_l e^{- i (\phi_l + \delta \theta_l) H_l } \\
    & = \prod_{l=1}^\nparams \left( V_l e^{- i  \phi_i H_l}\right) e^{- i \delta \theta_i H_l} \\
    & = \prod_{l=1}^\nparams \tilde{ V}_l(\phi_l) e^{-i \theta_l H_l} \;.
\end{align}
In addition, without loss of generality, the loss variance can also be re-expressed in terms of the perturbation around $\boldsymbol{\phi}= \vec{0}$. 
\begin{equation}
   \Var_{\thv \sim \uni(\vec{\phi},r)}\left[\mathcal{L}(\thv)\right] =  \Var_{\thv \sim \uni(\vec{0},r)}\left[\mathcal{L}(\thv + \vec{\phi})\right] \;.
\end{equation}
For convenience, we will simply write $\LC(\thv) := \LC(\thv + \vec{\phi})$ from here. But we stress again that $\thv$ is now a perturbation around $\vphi$ and the information about the point $\vphi$ is implicitly encoded in $\LC(\thv)$ with $\{ \tilde{ V}_l(\phi_l) \}_{l}$. In addition, we will replace $\tilde{ V}_l(\phi_l)$ with simply $V_l$ and again we emphasize that with this simplified notation $V_l$ depends on $\phi_l, H_l$ and the original fixed gate $V_l$.

\medskip

\bigskip

We have now set the scene for the proof. Before moving forwards, the readers are highly encouraged to go through Appendix~\ref{appx:main-thm-overview} for the five key proof steps. Then, brace yourself and let us now delve into the proof.

\bigskip

\paragraph*{\underline{1.~Expressing the loss variance as a sum of single parameter variance contributions.}} Since the parameters $\{\theta_{l}\}_{l=1}^\nparams$ are independently sampled, one can directly apply the result from Proposition~\ref{prop:var_decomp} and express the loss function variance as a sum of single parameter variance contributions with an average over the remaining parameters.
\begin{equation}\label{eq:VarM}
\Var_{\thv \sim \uni(\vec{0},r)}\left[\mathcal{L}(\boldsymbol{\theta})\right]  =\sum_{l=1}^{\nparams} \Ebb_{\overline{l}}[\Var_l[\Ebb_{1,\dots,l-1}[\mathcal{L}(\thv)]]]
\end{equation}
where we introduce the shorthand $\overline{l}:= l+1,\dots,\nparams$. Moreover, we simply use the shorthands $\Var_{l} [\cdot] := \Var_{\theta_l \sim \uni(0,r)} [\cdot]$ and $ \Ebb_{l}[\cdot] = \Ebb_{\theta_l \sim \uni(0,r)}[\cdot]$, and we will keep using these notations in the remainder of the proof.

Next, we rewrite each term in the sum using the following compact notations
\begin{align}
    \rho_{\overline{l}} &:=  U_{l+1}(\theta_{l+1}) \dots U_{\nparams}(\theta_{\nparams} )  \rho_0 U^{\dagger}_{\nparams}(\theta_{\nparams}) \dots U_{l+1}^{\dagger}(\theta_{l+1}) \label{eq:rho-l-bar}\\
    \UC_{l,\th_l}(A)&:=e^{iH_l\th_l}Ae^{-iH_l\th_l}\\
    \UC_{l,\phi_l}^{(k)}(A)&:= \left.\frac{d^k \UC_{l,\th_l}(A)}{d\th_l^k}\right|_{\th_l=\phi_l}   
\end{align}
where $\rho_{\overline{l}}$ denotes the evolved state from the gate $\nparams$ to the gate $l+1$ and implicitly depends on $\{ \theta_{l+1}, \theta_{l+2}, \dots, \theta_\nparams\}$, and $\UC_{l,\th_l}(\cdot)$ denotes the unitary superoperator at layer $l$ with the parameter $\theta_l$.
We also introduce the recursive expectation of the back propagated observable up to layer $\mindex$, i.e. $A_{0,\mindex}$: 
\begin{equation}\label{eq:Am_def_recursive}
\begin{aligned}
    A_{0,\mindex}&:=  \Ebb_{\mindex}[U_{\mindex}^{\dagger}(\theta_{\mindex}) A_{0,\mindex-1}U_{\mindex}(\theta_{\mindex})]= \Ebb_{\mindex}[e^{i\theta_{\mindex} H_{\mindex}}  V_{\mindex}^{\dagger} A_{0,\mindex-1} V_{\mindex} e^{-i\theta_{\mindex} H_{\mindex}}] \quad , \forall \mindex \geq 1\\
    A_{0,0} &:= O
    \end{aligned}
\end{equation}
and rewrite the each term in the sum $\Ebb_{\overline{l}} \Var_l[\Ebb_{1,\dots,l-1}[\mathcal{L}(\thv)]]$ in Eq.~\eqref{eq:VarM} %for a fixed layer $l$ 
as follows
\begin{equation}\label{eq:var_lEl}
    \Ebb_{\overline{l}}\left[\Var_l[\Ebb_{1,\dots,l-1}[\mathcal{L}(\thv)]]\right] = \Ebb_{\overline{l}}\left[\Var_l\left[\Tr[\rho_{\overline{l}}\UC_{l,\th_l}\left(V_l^\dagger A_{0,l-1}V_l\right)]\right]\right] \;.
\end{equation}

In what follows, we focus on lower bounding each individual term $\Ebb_{\overline{l}}[\Var_l[\Ebb_{1,\dots,l-1}[\mathcal{L}(\theta)]]],  \forall l \in \{1,\dots,\nparams\}$ in Eq.~\eqref{eq:VarM}. We first lower bound the variance $\Var_l[\Ebb_{1,\dots,l-1}[\mathcal{L}(\theta)]]$ in the next step and then the whole term $\Ebb_{\overline{l}}[\Var_l[\Ebb_{1,\dots,l-1}[\mathcal{L}(\theta)]]]$ in the step after.
%Hence, we fix a parameter index $l$ and  proceed to lower bound the corresponding term in the variance decomposition.  

\bigskip

\paragraph*{\underline{2. Lower bounding each individual single parameter variance $\Var_l[\Ebb_{1,\dots,l-1}[\mathcal{L}(\thv)]]$.}}
From Eq.~\eqref{eq:var_lEl}, we have each individual variance term expressed as
\begin{align}
    \Var_l[\Ebb_{1,\dots,l-1}[\mathcal{L}(\thv)]] = \Var_l\left[\Tr[\rho_{\overline{l}}\UC_{l,\th_l}\left(V_l^\dagger A_{0,l-1}V_l\right)]\right] \;.
\end{align}
This variance is only with respect to a single variable $\theta_l$ and hence we can invoke Corollary~\ref{cor:var_LB_1param_1layer} with $\th \rightarrow \th_l, \rho \rightarrow \rho_{\overline{l}}$ and $ O  \rightarrow A_{0,l-1}$. 

That is, if $r$ obeys
\begin{equation}\label{eq:cond_r_remainder}
    r \leq \frac{3}{4 \omega^{(\rm max)}_l} \;,
\end{equation}
with $\omega_{l}^{(\rm max)}=\lambda_{\max}(H_l) - \lambda_{\min}(H_l)$ as defined in Eq.~\eqref{eq:omega-max-corr-def}, then we have the lower bound as
\begin{align}
    \Var_l[\Ebb_{1,\dots,l-1}[\mathcal{L}(\theta)]]  
    \label{eq:proof-theorem-uncorr-varl-double-bracket}
    & \geq \frac{r^4}{45}\Tr[\rho_{\overline{l}}\UC_{l,0}^{(2)}(V_l^\dagger A_{0,l-1} V_l)]^2 -\frac{2 \left[\omega_{l}^{\rm (max) }\right]^2  \norm{[H_l,[H_l,(V_l^\dagger A_{0,l-1} V_l)]]}^2}{135} r^6 \\
    & \geq \frac{r^4}{45}\Tr[\rho_{\overline{l}}\UC_{l,0}^{(2)}(V_l^\dagger A_{0,l-1} V_l)]^2 -\frac{32\left[ \omega_{l}^{\rm (max) }\right]^{6}  \norm{O}^2}{135} r^6
\end{align}
where in the first inequality we apply Corollary~\ref{cor:var_LB_1param_1layer} to get the lower bound~\footnote{When expanding around zero, we can further tighten this bound by using Corollary~\ref{cor:var_LB_1param_1layer_local}.} and in the second inequality, we applied Lemma~\ref{lemma:bound-nested-commutator} and upper bound $\norm{A_{0,l-1}}$ by using its recursive definition in Eq.~\eqref{eq:Am_def_recursive}. Indeed, for $\forall \mindex \in \{1,\dots,\nparams\}$ we have
\begin{align}
    \norm{A_{0,\mindex}} &:= \norm{\Ebb_{\mindex}[U_{\mindex}^{\dagger} A_{0,\mindex-1}U_{\mindex}]}\\
    & \leq \Ebb_{\mindex}\left[\norm{U_{\mindex}^{\dagger} A_{0,\mindex-1}U_{\mindex}}\right]\\
    & = \Ebb_{\mindex}\left[\norm{ A_{0,\mindex-1}}\right]\\
    & = \norm{A_{0,\mindex-1}}\\
    & \leq \norm{A_{0,0}} \\
    & = \norm{O}\;,
\end{align}
where we used Jensen's inequality in the first inequality, the second equality is due to unitary invariance of the norm and the last inequality is obtained by induction.

Notice that in the case where $l=1$, we can tighten this bound by noticing that the second term on the right-hand side in Eq.~\eqref{eq:proof-theorem-uncorr-varl-double-bracket} can be expressed as the infinity norm of the second derivative with respect to $\theta_1$ which in turn can be written in the form of $\omega_{1}^{\rm (eff)}(\phi_1)$ in Eq.~\eqref{eq:omega-eff-1-def}. This leads to the expression
\begin{equation}
  \Var_1[\mathcal{L}(\theta)]  \geq \frac{r^4}{45}\Tr[\rho_{\overline{1}}\UC_{1,0}^{(2)}(V_1^\dagger O V_1)]^2 -\frac{2 \left[\omega^{(\rm max)}_{1}\right]^2 \left[\omega^{\rm (eff)}_{1}(\phi_1)\right]^4 }{135} r^6\;.
\end{equation}

Hence, the lower bound can be expressed more compactly as
\begin{align}
    \Var_1[\mathcal{L}(\theta)]  \geq \frac{r^4}{45}\Tr[\rho_{\overline{l}}\UC_{l,0}^{(2)}(V_l^\dagger A_{0,l-1} V_l)]^2 -\frac{\beta_l}{45} r^6 \;, \label{eq:varl_LB}
\end{align}
with the notation $\beta_l$ defined in Eq.~\eqref{eq:beta_l}, that is $\beta_l = \frac{32( \omega_{l}^{\rm (max) })^{6}  \norm{O}^2}{3}$ for $l>1$, and $\beta_l = \frac{2 \left[\omega^{(\rm max)}_{1}\right]^2  (\omega^{\rm (eff)}_{1}(\phi_1))^4}{3}$ for $l=1$.

\bigskip

\paragraph*{\underline{3.~Lower bounding $\Ebb_{\overline{l}}[\Var_l[\Ebb_{1,\dots,l-1}[\LC(\thv)]]]$.}} 
Now, with $\Var_l[\Ebb_{1,\dots,l-1}[\LC(\thv)]]$ expressed in Eq.~\eqref{eq:varl_LB}, 
let us include the expectation over parameters $\{\theta_{l+1},\cdots, \theta_m\}$ (indicated with the index $\overline{l}$) %in Eq.~\eqref{eq:varl_LB} 
to get a lower bound for $\Ebb_{\overline{l}}[\Var_l[\Ebb_{1,\dots,l-1}[\LC(\thv)]]]$ which leads to
\begin{align}
    \Ebb_{\overline{l}}[\Var_l[\Ebb_{1,\dots,l-1}[\LC(\thv)]]] &\geq \Ebb_{\overline{l}}\left[ \frac{r^4}{45}\Tr[\rho_{\overline{l}}\UC_{l,0}^{(2)}(V_l^\dagger A_{0,l-1} V_l)]^2 -\frac{\beta_l}{45} r^6 \right]
    \\
    & = \frac{r^4}{45}\left(\Ebb_{\overline{l}}\left[\Tr[\rho_{\overline{l}}\UC_{l,0}^{(2)}(V_l^\dagger A_{0,l-1} V_l)]^2\right] -\beta_l r^2\right) \\
    & \geq \frac{r^4}{45}\left(\left[\Ebb_{\overline{l}}\Tr[\rho_{\overline{l}}\UC_{l,0}^{(2)}(V_l^\dagger A_{0,l-1} V_l)]\right]^2 -\beta_l r^2\right) \label{eq:lowerbound_ElVarl}
\end{align}
where the second inequality is due to Jensen's inequality.

Now, we express the lower bound in a form that the average over other parameters (except the parameter $\theta_l$) is completely presented as only one operator. To do so, we introduce an index $\mindex$ such that $l\leq \mindex \leq \nparams$ and introduce an \textit{average operator} as
\begin{equation}\label{eq:Aml_def}
    A_{l, \mindex}:=\Ebb_{l+1,\dots,\mindex}\left[U_{\mindex}^\dagger \dots U_{l+1}^\dagger [H_l,[H_l,V_l^\dagger A_{0,l-1} V_l]]U_{l+1} \dots U_{\mindex}\right]    \;.
\end{equation}
We emphasize again the operator $A_{l, \nparams}$ contains averages over all parameters except $\th_l$; explicitly expressed there for $\theta_{l+1},\dots, \theta_{m}$ and implicitly in $A_{0,l-1}$ for  $\theta_1, \dots, \theta_{l-1}$.
  
Hence, the lower bound in Eq.~\eqref{eq:lowerbound_ElVarl} can be expressed in terms of $A_{l, \nparams}$ as
\begin{equation}\label{eq:end_res_stepb}
    \Ebb_{\overline{l}}[\Var_l[\Ebb_{1,\dots,l-1}[\LC(\thv)]]] \geq \frac{r^4}{45} \left( \Tr[\rho A_{l,\nparams}]^2 - \beta_l r^2 \right) \;,
\end{equation}
where we recall that $\UC_{l,0}^{(2)}(A) = - [H_l,[H_l,A]]$ as well as the definition of $\rho_{\overline{l}}$ in Eq.~\eqref{eq:rho-l-bar}.

It is worth noting that the average operator $A_{l, m}$ is closely related to the average second derivative of the loss. In particular, we have
\begin{align}
\Tr\left[ \rho  A_{l, m}\right] =\Ebb_{\thv}\left[\left.\left(\frac{d^2\mathcal{L}(\thv)}{d\theta_l^2}\right)\right|_{\theta_l=0 }\right] \;,
\end{align}
which is the next quantity of our interest.

\bigskip

\paragraph*{\underline{4.~Lower bounding the average double derivative term $\Tr[\rho \; A_{l,m}]^2$.}}
We decompose the average operator $A_{l,\mu}$ into a constant component $D_{l,\mu}$ which solely depends on the fixed point on the landscape $\vphi$ and a perturbation dependent component $T_{l,\mu}$. In particular, for the index $\mu$ such that $l \leq \mu \leq \nparams$, we denote
\begin{align}
     D_{l,\mindex} &:=V_{\mindex}^\dagger \dots V_{l+1}^\dagger [H_l,[H_l,D_{0,l}]]V_{l+1} \dots V_{\mindex} \label{eq:Dml_def} \;, \\
      T_{l,\mindex} &:=A_{l,\mindex}-D_{l,\mindex} \label{eq:Tml_def} \;,
\end{align}
where $D_{0,l}:= V_{l}^\dagger \dots V_{1}^\dagger O V_{1} \dots V_{l}$. In particular, for $l=0$ we have $D_{0,0}=O$. 

We then lower bound the term in Eq.~\eqref{eq:end_res_stepb} as follows.
\begin{align}
    \Tr[\rho \; A_{l,m}]^2  &=\left(\Tr[\rho D_{l,\nparams}]+\Tr[\rho T_{l,\nparams}]\right)^2 \\
    &\geq  \left(\left|\Tr[\rho D_{l,\nparams}]\right|-\left|\Tr[\rho T_{l,\nparams}]\right|\right)^2  \;, \label{eq:LB_tilde}
\end{align}
where the inequality is due to the property $|a+b|\geq ||a|-|b||$ for any reals $a$ and $b$.

One can notice that $c_l(\vec{\phi}) :=\left|\Tr[\rho D_{l,\nparams}]\right|=\left|\left.\left(\frac{\partial^2\mathcal{L}(\thv)}{\partial\th_l^2}\right)\right|_{\thv=\vec{\phi}}\right|$ is nothing but the curvature of the loss function with respect to the parameter $\th_l$ evaluated at $\vphi$ (without the average) while $\left|\Tr[\rho T_{l,\nparams}]\right|$ contains higher order contributions in $r$, which justify why it is called the perturbation term.

To further lower bound Eq.~\eqref{eq:LB_tilde}, the idea is to restrict $r$ such that the perturbation is sufficiently small. In particular, we consider a condition for $r$ such that
\begin{equation}\label{eq:cond1}
c_l(\vec{\phi})  \geq \left|\Tr[\rho T_{l,\nparams}]\right| \;.
\end{equation}

In order to find the restriction on $r$ for which Eq.~\eqref{eq:cond1} is satisfied, it is sufficient to fulfill the two following conditions: (i).~find some function $\alpha_l(\vphi)$ such that
$ \alpha_l(\vec{\phi}) r^2 \geq \left|\Tr[\rho T_{l,\nparams}]\right|$ and (ii).~choose $r$ such that 
\begin{equation}\label{eq:cond_r_inside}
    c_l(\vec{\phi})  \geq \alpha_l(\vec{\phi})  r^2\;,
\end{equation}
which ensures that $c_l(\vec{\phi})  \geq \alpha_l(\vec{\phi})  r^2 \geq \left|\Tr[\rho T_{l,\nparams}]\right|$ and hence Eq.~\eqref{eq:cond1} is hold. Note that as a reminder there is another condition on $r$ stated earlier Eq.~\eqref{eq:cond_r_remainder} which has to be simultaneously fulfilled. 

If we were to assume that all these conditions are met, by plugging in $ \alpha_l(\vec{\phi}) r^2 \geq \left|\Tr[\rho T_{l,\nparams}]\right|$ to Eq.~\eqref{eq:LB_tilde} we obtain
\begin{equation}\label{eq:sstepc_result}
    \Tr[\rho A_{l,m}]^2 \geq (c_l(\vec{\phi}) - \alpha_l(\vec{\phi})r^2)^2 \;.
\end{equation}
By substituting Eq.~\eqref{eq:sstepc_result} into Eq.~\eqref{eq:end_res_stepb} together with Eq.~\eqref{eq:VarM}, the variance lower bound can be expressed as 
\begin{align}\label{eq:var-alpha-function}
    \Var[\LC(\thv)]  \geq  \frac{r^4}{45}\left(\left(c_l(\vec{\phi}) - \alpha_l(\vec{\phi})  r^2 \right)^2-\beta_l r^2\right).
\end{align}
As one can see the form in Eq.~\eqref{eq:var-alpha-function} now looks somewhat similar to the variance lower bound stated in the theorem in Eq.~\eqref{eq:var-lower-bound-general}. It now hinges on (i).~identifying $\alpha_l(\vphi)$ and (ii).~finding the restriction on $r$ which will satisfy both conditions in Eq.~\eqref{eq:cond_r_remainder} and Eq.~\eqref{eq:cond_r_inside} at the same time.

\medskip

\textit{\underline{4.1~Identifying $\alpha_l(\vphi)$.}} 

We focus on determining $\alpha_l(\vphi)$ which can be analyzed by upper bounding the term $\left|\Tr[\rho T_{l,\nparams}]\right|$. First, from H\"{o}lder inequality, we have
\begin{align}
    \left|\Tr[\rho T_{l,\nparams}]\right| &\leq \lVert \rho \rVert_1\norm{T_{l,\nparams}} \\
    &=\norm{T_{l,\nparams}}\;,
\end{align}
where we used $\lVert \rho \rVert_1=1$ due to the positivity and normalization of a density operator.

Our next step is to upper bound $\norm{T_{l,\nparams}}$ which can be derived using a recursive relation for each parameter $\mindex\geq l+1$ as follows
\begin{align}
    \norm{T_{l,\mindex}} &= \norm{A_{l,\mindex}-D_{l,\mindex}} \\
    &\leq \norm{A_{l,\mindex}-\Ebb_{\th_{\mindex}}\UC_{\theta_{\mindex}}\left(D_{l,\mindex}\right)} + \norm{\Ebb_{\th_{\mindex}}\UC_{\theta_{\mindex}}\left(D_{l,\mindex}\right)-D_{l,\mindex}} \\
    &\leq \norm{A_{l,\mindex}-\Ebb_{\th_{\mindex}}\UC_{\theta_{\mindex}}\left(D_{l,\mindex}\right)}+ \frac{r^2}{6}\norm{[H_{\mindex},[H_{\mindex},D_{l,\mindex}]]}  \\ &=\norm{\Ebb_{\th_{\mindex}}\UC_{\theta_{\mindex}}\left(V_\mindex^{\dagger}A_{l,\mindex-1}V_\mindex -D_{l,\mindex}\right)} + \frac{r^2}{6}\norm{[H_{\mindex},[H_{\mindex},D_{l,\mindex}]]}  \\
    &\leq \norm{V_\mindex^{\dagger} A_{l,\mindex-1}V_\mindex -D_{l,\mindex}} + \frac{r^2}{6}\norm{[H_{\mindex},[H_{\mindex},D_{l,\mindex}]]}  \\
    &=\norm{A_{l,\mindex-1} -D_{l,\mindex-1}} + \frac{r^2}{6}\norm{[H_{\mindex},[H_{\mindex},D_{l,\mindex}]]}  \\
    &=\norm{T_{l,\mindex-1}} + \frac{r^2}{6}\norm{[H_{\mindex},[H_{\mindex},D_{l,\mindex}]]}  \label{eq:trash_bound_1} \;,
\end{align}
where the triangle inequality is used in the first inequality, the second inequality is obtained by applying Proposition~\ref{prop:expectation_taylor_product}, the third inequality is due to the Jensen inequality together with unitarily invariance of the norm, and the third equality is due to unitarily invariance of the norm.

Next, we can proceed to upper bound $\norm{T_{l,\nparams}}$ by writing the term as the telescopic sum
\begin{align}
    \norm{T_{l,\nparams}}  & = \norm{T_{l,l}} + \sum_{\mindex=l+1}^\nparams \left( \norm{T_{l,\mindex}} - \norm{T_{l,\mindex-1}}\right) \\
    & \leq \norm{T_{l,l}} + \frac{r^2}{6} \sum_{\mindex=l+1}^\nparams \norm{[H_{\mindex},[H_{\mindex},D_{l,\mindex}]]}  \;, \label{eq:trash_UB_part0}
\end{align}

Then, the term $\norm{T_{l,l}}$ can be expressed as 
\begin{align}
     \norm{T_{l,l}} & = \norm{A_{l,l} - D_{l,l}} \\
     & = \norm{ [H_l, [H_l, V_l^\dagger A_{0,l-1}V_l]] - [H_l, [H_l, V_l^\dagger D_{0,l-1} V_l]]} \\
     & = \norm{ [H_l, [H_l, V_l^\dagger T_{0,l-1}V_l]] } \\
     & \leq 4 \left[\omega^{(\rm max)}_l\right]^2 \norm{T_{0,l-1}} \label{eq:trash_UB_part1} \;,
\end{align}
where the inequality is by using Lemma~\ref{lemma:bound-nested-commutator}.

Similarly, we can recursively upper bound $\norm{T_{0,l-1}}$ in the same procedure as above and remark that  $T_{0,0} = 0$ by definition leading to 
\begin{align}
    \norm{T_{0,l-1}} \leq \frac{r^2}{6} \sum_{\mindex=1}^{l-1} \norm{[H_{\mindex},[H_{\mindex},D_{0,\mindex}]]} \label{eq:trash_UB_part2}
\end{align}

Finally, by grouping Eq.~\eqref{eq:trash_UB_part0}, Eq.~\eqref{eq:trash_UB_part1} and Eq.~\eqref{eq:trash_UB_part2}, we obtain the following upper bound on $ \norm{T_{l,\nparams}}$ and can identify $\alpha_l(\vphi)$ %(\ref{eq:trash_UB_part0},\ref{eq:trash_UB_part1},\ref{eq:trash_UB_part2}), we finally obtain the following upper bound on $ \norm{T_{l,\nparams}}$:
\begin{align}\label{eq:UB_trash}
    \norm{T_{l,\nparams}} &\leq \frac{r^2}{6}\left(4 \left[\omega_{l}^{\rm (max) }\right]^2\sum_{\mindex=1}^{l-1}\norm{[H_{\mindex},[H_{\mindex},D_{0,\mindex}]]}  + \sum_{\mindex=l+1}^{\nparams}\norm{[H_{\mindex},[H_{\mindex},D_{l,\mindex}]]}  \right) \\
    &:= \frac{r^2}{6}\left(4 \left[\omega_{l}^{\rm (max) }\right]^2\sum_{\mindex=1}^{l-1} (\omega_{\mu}^{\rm (eff) }(\vec{\phi}))^2  + \sum_{\mindex=l+1}^{\nparams}\left[\widetilde{\omega}_{l, \mu }^{(\rm eff)}(\vec{\phi})\right]^2  \right)\\
    &:= \alpha_l(\vec{\phi}) r^2 
\end{align}
where we have
\begin{align}\label{eq:alpha-function-proof}
    \alpha_l(\vec{\phi}) &:= \frac{1}{6} \left(4 \left[\omega_{l}^{\rm (max) } \right]^2\sum_{\mindex=1}^{l-1} \left[\omega_{\mu}^{\rm (eff) }(\vec{\phi}) \right]^2 + \sum_{\mindex=l+1}^{\nparams}\left[\widetilde{\omega}_{l, \mu }^{(\rm eff)}(\vec{\phi})\right]^2  \right) \;,
\end{align}
with
\begin{align}
    \omega_{\mu}^{\rm (eff) }(\vec{\phi}) &:= \sqrt{\norm{[H_{\mindex},[H_{\mindex},D_{0,\mindex}]]}}\\
    \widetilde{\omega}_{l, \mu }^{(\rm eff)}(\vec{\phi}) &:= \sqrt{\norm{[H_{\mindex},[H_{\mindex},D_{l,\mindex}]]}}\;.
\end{align}

Here, we note that the effective frequencies $\omega_{\mu}^{\rm (eff) }(\vec{\phi})$ and $\widetilde{\omega}_{l, \mu }^{(\rm eff)}(\vec{\phi})$ are nothing but norms of partial derivatives of the backpropagated observable evaluated at $\vec{\phi}$. Specifically, we have 
\begin{align}
    \omega_{\mu}^{\rm (eff) }(\vec{\phi}) &= \sqrt{\norm{ \left.\frac{\partial^2 [U(\thv)^\dagger O U(\thv)]}{\partial \th_\mu^2}\right|_{\thv= \vec{\phi}}}}\;,\\
    \widetilde{\omega}_{l, \mu }^{(\rm eff)}(\vec{\phi}) &= \sqrt{\norm{ \left.\frac{\partial^4 [U(\thv)^\dagger O U(\thv)]}{\partial \th_\mu^2 \partial \th_l^2}\right|_{\thv= \vec{\phi}}}}\;.
\end{align}

Upon substituting $\alpha_l(\vphi)$ in Eq.~\eqref{eq:alpha-function-proof} into Eq.~\eqref{eq:var-alpha-function}, we achieve the variance lower bound in Eq.~\eqref{eq:var-lower-bound-general} presented in the theorem.

\medskip

\textit{\underline{4.2~.Determining the restriction on $r$.}} 
Crucially for the variance lower bound to hold, the width $2r$ of the region has to simultaneously satisfy Eq.~\eqref{eq:cond_r_remainder} and Eq.~\eqref{eq:cond_r_inside}. Additionally, $r$ has to be sufficiently small such that the variance lower bound in Eq.~\eqref{eq:var-alpha-function} is non-negative, that is 
\begin{align} \label{eq:cond_r_non_zero_var}
\left(c_l(\vec{\phi}) - \alpha_l(\vec{\phi})  r^2 \right)^2-\beta_l r^2 > 0 \;.
\end{align}

We now focus on determining the sufficient condition on $r$ that can achieve three constrains Eq.~\eqref{eq:cond_r_remainder}, Eq.~\eqref{eq:cond_r_inside} and Eq.~\eqref{eq:cond_r_non_zero_var} all at the same time. Taken an inspiration from the form in Eq.~\eqref{eq:cond_r_non_zero_var},
we start by considering
\begin{align}
  \left(c_l(\vec{\phi}) - \alpha_l(\vec{\phi})  r^2 \right)^2- \beta_l r^2  & = c_l(\vec{\phi}) ^2+(\alpha_l(\vec{\phi})  r^2)^2-2c_l\alpha_l r^2-\beta_l r^2 \\
  &\geq c_l(\vec{\phi})^2-(2c_l(\vec{\phi}) \alpha_l(\vec{\phi}) +\beta_l)r^2 \;, \label{eq:inspire-r-cond}
\end{align}
where we lower bound by removing the positive term $(\alpha_l(\vec{\phi}) 
 r^2)^2$. 

One can notice that Eq.~\eqref{eq:cond_r_non_zero_var} is satisfied if we enforce the right-hand of Eq.~\eqref{eq:inspire-r-cond} to follow
\begin{align}\label{eq:proof-thm3-near-the-end}
    c_l(\vec{\phi})^2-(2c_l(\vec{\phi}) \alpha_l(\vec{\phi}) +\beta_l)r^2 \geq c_l(\vphi)^2 (1- \Delta) \;,
\end{align}
for some $\Delta$ such that $0<\Delta <1$. Upon rearranging, we have the restriction on $r$ as follows
\begin{align}
   r^2 &  \leq  \frac{\Delta c_l(\vec{\phi})^2}{2 c_l(\vec{\phi}) \alpha_l(\vec{\phi}) +  \beta_l}\label{eq:cond3final}\;.
\end{align}

Next, we show that by choosing $r$ which satisfies Eq.~\eqref{eq:cond3final}, then the condition in Eq.~\eqref{eq:cond_r_inside} is also satisfied since
\begin{equation}
   r \leq \frac{\Delta c_l(\vec{\phi})^2}{2 c_l(\vec{\phi}) \alpha_l(\vec{\phi}) +  \beta_l} \leq \frac{\Delta}{2}\frac{c_l(\vec{\phi})}{\alpha_l(\vec{\phi})} \leq \frac{c_l(\vec{\phi})}{\alpha_l(\vec{\phi})}\;,
\end{equation}
where we first use $\beta_l\geq 0$, and then $\Delta<1$. 

Lastly, one can also verify that Eq.~\eqref{eq:cond_r_remainder} is satisfied under Eq.~\eqref{eq:cond3final} for some choice of the constant $\Delta$. In particular, we have
\begin{equation}\label{eq:r_final_to_corr}
    \frac{\Delta c_l(\vec{\phi})^2}{2 c_l(\vec{\phi}) \alpha_l(\vec{\phi}) +  \beta_l}\leq \frac{\Delta c_l(\vec{\phi})^2}{  \beta_l} \leq \frac{3\Delta}{2\left[\omega_{l}^{\rm (max) }\right]^2}\;,
\end{equation}
where the first inequality is due to $c_l(\vec{\phi}) \alpha_l(\vec{\phi})\geq 0$. In the second inequality, we recall $\beta_l=\frac{32\left[ \omega_{l}^{\rm (max) }\right]^{6}  \norm{O}^2}{3}$ and using $c_l(\vec{\phi})\leq \left[\omega^{\rm (eff)}_{l}(\vec{\phi})\right]^2\leq 4\left[\omega_{l}^{\rm (max) }\right]^2\norm{O}$ (from  Lemma~\ref{lemma:bouded_unitary_product}).
 Therefore, by setting $\Delta=3/8$, we recover the condition in Eq.~\eqref{eq:cond_r_remainder}.

Consequently, for any layer $l$ such that the perturbation $r$ obeys 
 \begin{equation}\label{eq:gorge_l}
   \begin{aligned}
     r^2 &\leq \frac{3c_l(\vec{\phi})^2}{8(2c_l(\vec{\phi})\alpha_l(\vec{\phi}) + \beta_l)} \\
     &= \frac{9c_l(\vec{\phi})^2}{8\left(c_l(\vec{\phi})\left(4 \left[\omega_{l}^{\rm (max) }\right]^2 \sum_{\mindex=1}^{l-1} \left[\omega_{\mu}^{\rm (eff) }(\vec{\phi})\right]^2  + \sum_{\mindex=l+1}^{\nparams}\left[\widetilde{\omega}_{l, \mu }^{(\rm eff)}(\vec{\phi})\right]^2  \right) + 3\beta_l \right)} \\
     & :=r^2_{{\rm patch},l}(\vec{\phi})
     \end{aligned}
 \end{equation}

 we have the following lower bound
 \begin{align}
       \Ebb_{\overline{l}}[\Var_l[\Ebb_{1,\dots,l-1}[\mathcal{L}(\thv)]]] &\geq \frac{r^4}{45}\left(\left(c_l(\vec{\phi})- \alpha_l(\vec{\phi}) r^2 \right)^2- \beta_l r^2\right)\\
    & \geq \frac{1}{72}    c_l(\vec{\phi})^2  r^4 \;,\label{eq:var_LB_l}
 \end{align}
where to reach the second inequality we use Eq.~\eqref{eq:proof-thm3-near-the-end} with $\Delta = 3/8$.

 We recall here that the above lower bound is also a lower bound of the loss function variance corresponding to a single term contribution in the variance decomposition in Eq.~\eqref{eq:VarM}. Thus, obtaining a lower bound including the contribution from any subset of parameters $\Lambda \subset \{1,\dots,\nparams\}$ amounts to restricting more the region with guarantees.
Indeed, if the perturbation $r$ satisfies
\begin{equation}
    r^2 \leq \min_{l \in \Lambda} r^2_{{\rm patch},l}(\vec{\phi}) := r^2_{{\rm patch}}(\vec{\phi}) 
\end{equation}
 then the loss function variance is lower bounded as
 
 \begin{align}
    \Var_{\vtheta \sim \uni(\vec{0},r)}[L(\thv)] &\geq \sum_{l \in \Lambda} \frac{r^4}{45}\left(\left(c_l(\vec{\phi})- \alpha_l(\vec{\phi}) r^2 \right)^2- \beta_l r^2\right)\\
    & \geq \frac{1}{72}   \left(\sum_{l \in \Lambda} c_l(\vec{\phi})^2 \right) r^4 \;.
 \end{align}

 Moreover, for $r = r_{\rm patch}(\vphi)$, we obtain
 \begin{equation}
     \Var_{\vtheta \sim \uni(\vec{0},r_{\rm patch})}[\LC(\thv)] \geq \frac{1}{72}   \left(\sum_{l \in \Lambda} c_l(\vec{\phi})^2 \right) r^4_{\rm patch}(\vphi)
 \end{equation}

 \bigskip

\paragraph*{\underline{5.~Proving a theoretical guarantee on the region of attraction.}}
    Going back to the loss function variance lower bound derived from the single contribution of a term $\Ebb_{\overline{l}}[\Var_l[\Ebb_{1,\dots,l-1}[\mathcal{L}(\thv ]]]$, we showed that for a perturbation $r$ obeying the condition 
    \begin{equation}\label{eq:informal_cond}
        r^2 \leq r^2_{{\rm patch},l}(\vec{\phi})
    \end{equation}
    where $r_{{\rm patch},l}(\vec{\phi})$ is defined in Eq.~\eqref{eq:gorge_l}, the variance is lower bounded as (See Eq.~\eqref{eq:var_LB_l}).
    \begin{equation}\label{eq:informal_LB}
         \Var_{\vtheta \sim \uni(\vec{0},r)}[\LC(\thv)] \geq \frac{1}{72} c_l(\vec{\phi})^2 r^4 
    \end{equation}

    We can further obtain a more generic  condition than the one in Eq.~\eqref{eq:informal_cond} by upper bounding the denominator of $r_{\rm patch,l}^2$ under the assumption that $c_l(\vec{\phi})$ is at most polynomially vanishing (i.e. $c_l(\vec{\phi}) \in \Omega\left(\frac{1}{{\rm poly}(n)}\right)$).

    First we recall that $r_{\rm patch,l}^2(\vphi) = \frac{3c_l(\vec{\phi})^2}{8(2c_l(\vec{\phi})\alpha_l(\vec{\phi}) + \beta_l)}$. 
    Hence,  Eq.~\eqref{eq:informal_cond} still holds if we replace $\alpha_l(\vec{\phi})$ by a more generic upper bound.
    Specifically,  we can use the property $\norm{[A,[A,B]]} \leq 4 \norm{A}^2 \norm{B}$ and obtain
\begin{align}
    \alpha_l(\vphi) &\leq \frac{2}{3} \left(  \left[\omega_{l}^{\rm (max) }\right]^2 \sum_{\mindex=1}^{l-1} \left[\omega_{\mu}^{\rm (eff) }(\vec{\phi})\right]^2  +  \sum_{\mindex=l+1}^{\nparams} \left[\omega^{(\rm max)}_{\mu}\right]^2 \norm{D_{l,\mindex}} \right) \label{eq:proof-roa-1}\\
    &= \frac{2}{3} \left(  \left[\omega_{l}^{\rm (max) }\right]^2 \sum_{\mindex=1}^{l-1} \left[\omega_{\mu}^{\rm (eff) }(\vec{\phi})\right]^2  +  \left[\omega^{\rm (eff)}_{l}(\vec{\phi})\right]^2 
 \sum_{\mindex=l+1}^{\nparams} \left[\omega^{(\rm max)}_{\mu}\right]^2 \right) \\
    & \leq \frac{8 }{3} \norm{O}\left(  \left[\omega_{l}^{\rm (max) }\right]^2 \sum_{\mindex=1}^{l-1} \left[\omega^{(\rm max)}_{\mu}\right]^2  + \left[\omega_{l}^{\rm (max) }\right]^2 \sum_{\mindex=l+1}^{\nparams} \left[\omega^{(\rm max)}_{\mu}\right]^2  \right)\\
    &= \frac{8 }{3} \norm{O} \left[\omega_{l}^{\rm (max) }\right]^2 \left(\sum_{\substack{\mu=1\\ \mu \neq l}}^\nparams  \left[\omega^{(\rm max)}_{\mu}\right]^2 \right)\\
    &\in \mathcal{O}\left(  \nparams \cdot \poly(n)\right)
\end{align}
where we mainly used the inequalities $\left[\widetilde{\omega}_{l, \mu }^{(\rm eff)}(\vphi)\right]^2 \leq 4 \left[\omega^{(\rm max)}_{\mu} \omega^{\rm (eff)}_{l}(\vphi)\right]^2$ and $\left[\omega_{\mu}^{\rm (eff) }(\vphi)\right]^2 \leq 4 \left[\omega^{(\rm max)}_{\mu}\right]^2 \norm{O}$ and the assumption that $\norm{O},\norm{H_l} \in \poly(n)$.

Moreover, the curvature $c_l(\vec{\phi})$ can be upper bounded as 
\begin{equation} 
    \begin{aligned}
         c_l(\vec{\phi}) &\leq 4 \norm{O} \left[\omega_{l}^{\rm (max) }\right]^2\\
         &\in \mathcal{O}(\poly(n))
    \end{aligned}
\end{equation}

Therefore, given that the curvature $c_l(\vec{\phi})$ satisfies
\begin{equation}\label{eq:proof_th1_nonvanish}
    c_l(\vec{\phi}) \in \Omega\left(\frac{1}{\poly(n)}\right)
\end{equation}
for a perturbation $r$ scaling as
\begin{equation}\label{eq:proof_th1_cond}
    r^2 \in \Omega\left(\frac{1}{\poly(n) \cdot \nparams}\right)
\end{equation}

the variance is lower bounded as 
\begin{equation}\label{eq:proof_th1_LB}
    \begin{aligned}
               \Var_{\vtheta \sim \uni(\vec{0},r)}[\LC(\thv)] &\in \Omega\left( c_l(\vec{\phi})^2 \frac{1}{\poly(n) \cdot m^2} \right)\\
               &\in \Omega\left(  \frac{1}{\poly(n) \cdot m^2} \right)
    \end{aligned}
\end{equation}
which concludes the proof.

\bigskip

\bigskip
      
\end{proof}

\subsubsection{Proof of Theorem~\ref{th:var_formal_cor}: The arbitrary correlated part of Theorem~\ref{th:var}}\label{appendix:lower_bound_correlated}

In this section, we provide a proof for Theorem~\ref{th:var_formal_cor} which is a generalization of Theorem \ref{th:var_formal}. In particualr, the theorem generalizes the uncorrelated/spatial correlated parameter setting with arbitrary generators to further account for temporal correlations (See Fig \ref{fig:corr_schematic}). However, this generalization comes at the cost of having more restricted regions with guaranteed substantial variance.

\begin{proof}
The main proof strategy here is identical to the proof strategy of Theorem~\ref{th:var_formal} (see Appendix~\ref{appx:main-thm-overview}) and indeed many step-by-step calculations are very similar. Nevertheless, the extra difficulty of course arises from temporally separate generators sharing the same parameter. Namely, this setting makes the expression of the loss partial derivatives more intricate, requiring a trickier tracking of the set of generators sharing the same parameter. 

We start by setting up some notations to match the general forms in the Lemmas and Corollaries from Appendix \ref{appendix:proof-analytical-results}. Then, we reiterate the same steps in the proof of Theorem \ref{th:var_formal} according to the newly introduced notation. For each proof step, we give a brief comment on how the generalized proof differs from the more restricted setting of Theorem \ref{th:var_formal}. Though not necessary, the readers are encouraged to go through the proof of Theorem~\ref{th:var_formal} beforehand. Otherwise, for those adventurous ones feel free to just delve in. 

\medskip

\paragraph*{\underline{0.~Setting some initial notations.}} We consider a loss function $\LC(\thv)$ of the form in Eq.~\eqref{eq:loss} and a parametrized circuit of the form in Eq.~\eqref{eq:circuit} with a set of $M$ generators $\{ H_l\}_{l=1}^\nHam$ and a set of $m$ independent parameters $\thv = \{\theta_k \}_{k=1}^\nparams$ such that $\nHam \geq \nparams$
\begin{align}\label{eq:circuit_th4}
    U(\vec{\th}) & = \prod_{l=1}^M V_l \, U_l\bigl(\theta_{\SC(l)}\bigr)\,,\\
    & = \prod_{\layerindex=1}^\nHam V_\layerindex e^{-i \th_{\mathcal{S}(\layerindex)} H_\layerindex} \;,
\end{align}
where a surjective map $\mathcal{S}:\{1\dots,\nHam\} \rightarrow \{1\dots,\nparams\}$ assigns each generator index $\layerindex$ (and hence generator $H_l$) to its associated parameter index $\paramindex = \SC(\layerindex)$, such that at the $\layerindex^{\rm th}$ gate we have the parametrized gate $e^{- i \th_\paramindex H_\layerindex}$. This allows different generators to share the same parameter, i.e. allows for \emph{correlated parameters} as sketched in Fig.~\ref{fig:corr_schematic}. 

Moreover, we denote the inverse map $\SC^{-1}$ such that $\mathcal{S}^{-1}(\paramindex)$ is a vector of size $K_p$ whose components are the gate indices associated with the parameter $\th_\paramindex$. For convenience, we arrange the components of the vector $\mathcal{S}^{-1}(\paramindex)$ in increasing order such that
    \begin{align}\label{eq:S_vec}
        [\mathcal{S}^{-1}(\paramindex)]_1 < [\mathcal{S}^{-1}(\paramindex)]_2 < \dots <  [\mathcal{S}^{-1}(\paramindex)]_{K_p}
    \end{align}
Here, we emphasize that we adopted the perhaps unusual convention of indexing the circuit gates with respect to the observable $O$ such that $U_1$ is the first gate acting on $O$. Hence, the layer index $[\mathcal{S}^{-1}(\paramindex)]_1$ corresponds to the first gate acting on the observable $O$ which has the parameter $\th_p$. 

\medskip

Similarly to the proof of Theorem~\ref{th:var_formal}, to study the patch of the loss landscape $\mathcal{L}(\boldsymbol{\theta})$ around a point $\vec{\phi}$, one can re-write the parametrized gates in form of the perturbation around the fixed point $\theta_l = \phi_l + \delta \theta_l$ and absorb $e^{-i \phi_{\mathcal{S}(\layerindex)} H_\layerindex}$ into the non parameterized unitaries $
V_\layerindex$. Note that for the ease of notations in the following long proof, we will also redefine the perturbation variable as $\delta \th_{S(l)} \rightarrow \th_{S(l)}$ in the last line below (and for the rest of the proof). 
\begin{align}
U_{\layerindex}(\theta_{S(l)} = \phi_{\mathcal{S}(\layerindex)} + \delta \th_{\mathcal{S}(\layerindex)}) & = V_l e^{-i (\phi_{\mathcal{S}(\layerindex)} + \delta \th_{\mathcal{S}(\layerindex)})H_l}\\ 
&= \left[V_l e^{-i \phi_{\mathcal{S}(\layerindex)} H_\layerindex}\right] e^{-i \delta \th_{\mathcal{S}(\layerindex)} H_\layerindex}\\
&= \widetilde{V}_\layerindex (\phi_{\mathcal{S}(\layerindex)}) e^{-i  \th_{\mathcal{S}(\layerindex)} H_\layerindex} \;,
\end{align}
where we used the shorthand $ \widetilde{V}_\layerindex (\phi_{\mathcal{S}(\layerindex)})  := V_l e^{-i \phi_{\mathcal{S}(\layerindex)} H_\layerindex}$. 

In addition, without loss of generality, the loss variance can also be re-expressed in terms of the perturbation around $\boldsymbol{\phi}= \vec{0}$. 
\begin{equation}
   \Var_{\thv \sim \uni(\vec{\phi},r)}\left[\mathcal{L}(\thv)\right] =  \Var_{\thv \sim \uni(\vec{0},r)}\left[\mathcal{L}(\thv + \vec{\phi})\right] \;.
\end{equation}
For convenience, we will simply write $\LC(\thv) := \LC(\thv + \vec{\phi})$ from here. But we stress again that $\thv$ is now a perturbation around $\vphi$ and the information about the point $\vphi$ is implicitly encoded in $\LC(\thv)$ with $\{ \tilde{ V}_l(\phi_l) \}_{l}$.

\medskip
    
For convenience, we also express the loss function in terms of a composition of the unitary channels 
\begin{align}
    \UC_{\th_{\mathcal{S}(l)},l}(\cdot)&:= e^{i \th_{\mathcal{S}(\layerindex)} H_\layerindex} (\cdot) e^{-i \th_{\mathcal{S}(\layerindex)} H_\layerindex} \\
    \mathcal{V}_l(\cdot) &:= \left[\widetilde{V}_\layerindex (\phi_{\mathcal{S}(\layerindex)})\right]^\dagger (\cdot) \left[\widetilde{V}_\layerindex (\phi_{\mathcal{S}(\layerindex)})\right] \;,
\end{align}
where we note that $\mathcal{V}_l(\cdot)$ implicitly depends on the component $\phi_l$ of the fixed point. Then, the loss can be re-expressed in terms of the channels as

\begin{align}
    \mathcal{L}(\vec{\th}) &= \Tr\left[\rho \; U^\dagger(\vec{\th}) O U(\vec{\th})\right]\\
    &= \Tr\left[\rho \; \left(\UC_{\th_{\mathcal{S}(\nHam),\nHam}} \circ \mathcal{V}_{\nHam}\right) \circ\dots  \circ\left(\UC_{\th_{\mathcal{S}(1),1}} \circ \mathcal{V}_{1}\right)[O] \right]\\
    &=  \Tr\left[\rho \; \channel_{\thv}(O) \right] \label{eq:channel_def}
\end{align}
In the second inequality, we express the back-propagated observable $U^\dagger(\vec{\th}) O U(\vec{\th})$ as a composition of unitary channels acting on the observable $O$ and in 
 the last equality, we simply introduce the shorthand
 \begin{align}\label{eq:channel-whole-evolution}
     \channel_{\thv}[O]:= \bigcirc_{l=M}^1 (\UC_{\th_{\mathcal{S}(l),l}} \circ \mathcal{V}_l)[O] =\left(\UC_{\th_{\mathcal{S}(\nHam),\nHam}} \circ \mathcal{V}_{\nHam}\right) \circ\dots  \circ\left(\UC_{\th_{\mathcal{S}(1),1}} \circ \mathcal{V}_{1}\right)[O] \;.
 \end{align}

Now, that we have recalled the general settings of the Theorem, we will start diving into the core steps of the proof. Concretely, we will be going through the same steps from the proof of Theorem \ref{th:var_formal} but with trickier notations. 

\bigskip
    
\paragraph*{\underline{1.~Expressing the loss variance as a sum of single parameter variance contributions.}}
    
    Since the parameters $\{\theta_{k}\}_{k=1}^\nparams$ are  independently sampled, one can directly use Proposition~\ref{prop:var_decomp} in order to express the variance of the loss function $\mathcal{L}[\boldsymbol{\theta}]$  as 
    \begin{equation}\label{eq:Var_decomp_cor}
    \Var_{\thv \sim \uni(\vec{0},r)}\left[\mathcal{L}(\thv )\right]  =\sum_{k=1}^\nparams \Ebb_{\pi(\nparams),\dots,\pi(k+1)}[\Var_{\pi(k)}[\Ebb_{\pi(k-1),\dots,\pi(1)}[\mathcal{L}(\vec{\th})]]]
    \end{equation}
    for any permutation $\pi: \{1,\dots,\nparams\} \rightarrow \{1,\dots,\nparams\}$.

    Let us now fix some parameter index $\paramindex_0 \in \{1,\dots,\nparams\}$ and choose the permutation such that $\pi(1) = p_0$. Then, the first term in the sum in Eq.~\eqref{eq:Var_decomp_cor} becomes $\Ebb_{\overline{p_0}}[\Var_{p_0}[\mathcal{L}(\thv)]]$ with $\overline{p_0}:= \{1,\dots,p_0-1,p_0+1,\dots,\nparams\}$. Note that this \textit{bar} convention is different from the one used in the proof of Theorem~\ref{th:var_formal}. %The permutation on other parameter indices are chosen such that 
    
    We note here that all the terms in the variance decomposition in Eq.~\eqref{eq:Var_decomp_cor} are positive. Hence, we can lower bound the loss function variance with the single parameter $\th_{p_0}$ variance contribution as
    \begin{equation}\label{eq:proof_th4_step1}
    \begin{aligned}
         \Var_{\thv \sim \uni(\vec{0},r)}[\mathcal{L}(\thv)] &\geq \Ebb_{\overline{p_0}}[\Var_{p_0}[\mathcal{L}(\thv)]]\;.
    \end{aligned}
    \end{equation}
    Here, note that we chose to lower bound the loss variance with the single parameter variance contribution and 
 permute an associated parameter $\th_{p_0}$ to the first parameter to be taken variance (i.e. no expectations are taken before applying the variance). This step is different from earlier in the proof of Theorem~\ref{th:var_formal} where all the single parameter variance terms are kept and is the consequence of the temporal correlation which makes much more technical challenging to keep the other terms here.

    In the following step, we focus on further lower bounding the right hand side in Eq.~\eqref{eq:proof_th4_step1}.

    \bigskip

\paragraph*{\underline{2.~ Lower bounding the single parameter variance $\Var_{p_0}[\mathcal{L}(\thv)]$.}} The variance term on the right-hand side is computed only over the parameter $\theta_{p_0}$. The strategy here is to express the channel $\EC_{\thv}(O)$ and consequentially $\LC(\thv)$ in the form such that we can justify the application of Corollary~\ref{cor:var_LB_1param_layers} to lower bound $\Var_{p_0}[\Tr[\rho \mathcal{E}_{\thv}(O)]]$. That is, we argue that the channel $\EC_{\thv}$ in Eq.~\eqref{eq:channel-whole-evolution} can be written in the form of %$\SC^{-1}(p_0)$
\begin{align}\label{eq:channel_regrouped}
     \channel_{\thv}[O] = \widetilde{\VC}_{K_{p_0} +1} \circ \left[\bigcirc_{l=K_{p_0}}^1 (\UC_{\theta_p, [S^{-1}(p)]_{l}} \circ \widetilde{\mathcal{V}}_l)\right][O] \;,
\end{align}
where the set of newly introduced shorthands i.e., $ \{ \widetilde{\VC}_l \}_{l=1}^{K_{p_0}+1}$ implicitly depends on the other parameters $\thv_{\overline{p_0}}$ and does not contain dependence on $\theta_{p_0}$. While the exact forms of these tilde unitaries are presented below, the main message here is that the form of $\channel_{\thv}[O]$ in Eq.~\eqref{eq:channel_regrouped} justifies the application of Corollary~\ref{cor:var_LB_1param_layers}. 

\medskip

To see how the whole parametrized channel $\EC_{\thv}$ can be re-written in Eq.~\eqref{eq:channel_regrouped}, let us consider a toy parametrized circuit with $5$ gate generators and $3$ independent parameters of the form%and let further say we are interested in $\theta_p = \theta_1$
\begin{align}%\label{eq:channel_toy}
     \mathcal{E}_{\thv}^{({\rm Toy})}
     &= \left(\UC_{\th_{3},5} \circ \mathcal{V}_{5}\right) \circ \left(\UC_{\th_{2},4} \circ \mathcal{V}_{4}\right) \circ \left(\UC_{\th_{3},3} \circ \mathcal{V}_{3}\right) \circ \left(\UC_{\th_{2},2} \circ \mathcal{V}_{2}\right) \circ \left(\UC_{\th_{1},1} \circ \mathcal{V}_{1}\right) \;.
\end{align}
As a reminder, for a parametrized channel of the form $\UC_{\theta_i,l}$, the first subscript refers the parameter (index) and the second subscript refers to the gate index. For this toy example, we have the following map $\SC$ which maps the gate index to a parameter index as
\begin{align}
    \SC(1) & =  1 \;, \\
    \SC(2) & = \SC(4) = 2 \;, \\
    \SC(3) & = \SC(5) = 3  \;, 
    %\SC(4) & = 4 \;.
\end{align}
as well as, the inverse map $\SC^{-1}$ which maps a parameter index to a vector whose components are gate indices
\begin{align}
    \SC^{-1}(1) & = (1) \;,\\
    \SC^{-1}(2) & = (2, 4) \;, \label{eq:proof-inverse-map-toy} \\
    \SC^{-1}(3) & = (3,5) \;.
\end{align}

Let us further say that we are interested in $\theta_{p_0} = \theta_2$. Then by grouping the channels that do not depend on $\theta_2$ together, and using the inverse map in Eq.~\eqref{eq:proof-inverse-map-toy} (i.e., $[\SC^{-1}(2)]_1 = 2, [\SC^{-1}(2)]_2 = 4$ and $K_2 = 2$), one can see that the toy channel can be re-expressed as
\begin{align}
    \mathcal{E}_{\thv}^{({\rm Toy})}
     &= \left(\UC_{\th_{3},5} \circ \mathcal{V}_{5}\right) \circ \UC_{\th_{2},4} \circ \left( \mathcal{V}_{4}\circ \UC_{\th_{3},3} \circ \mathcal{V}_{3}\right) \circ \UC_{\th_{2},2} \circ \left(\mathcal{V}_{2}\circ \UC_{\th_{1},1} \circ \mathcal{V}_{1}\right) \\
     & = \widetilde{\VC}_3 \circ \UC_{\theta_2, [\SC^{-1}(2)]_2} \circ \widetilde{\VC}_{2} \circ \UC_{\th_{2},[\SC^{-1}(2)]_1} \circ \widetilde{\VC}_{1} \;,
\end{align}
which now matches the form in Eq.~\eqref{eq:channel_regrouped} for this toy example. 
It is then straight forward to generalize this to an arbitrary parametrized circuit and obtain the form in Eq.~\eqref{eq:channel_regrouped}. Crucially, while an exact form of an individual $\widetilde{\VC}_l$ could be messy and is not of importance to our proof, all terms $\{ \widetilde{\VC}_l \}_{l=1}^{K_{p_0} + 1}$ together depends on all the rest of parameters except $\theta_{p_0}$. 

%\medskip

Before moving forwards, we introduce some compact notations for parameters and second derivative of channels. First, $\thv_{\overline{p_0}}$ is denoted as a collection of parameters \textit{except} $\theta_{p_0}$. Next, since we will need to specify which parameters are still free parameters or already evaluated at zero for the second derivative (with respect to $\theta_{p_0}$) in the upcoming proof steps, it is much more convenient to simply omit those parameters that are already evaluated in the subscripts. For example, if $\theta_{p_0}$ is already evaluated at zero, the second derivative channel can be expressed as
\begin{align}
    \channel^{(2)}_{\thv_{\overline{p_0}}} := \left. \left(\frac{\partial^2 \channel_{\thv}}{\partial \theta_{p_0}^2}\right)\right|_{\th_{p_0} = 0} \;.
\end{align}
More generally, if only $\{\th_1, \th_2, ..., \th_k \}$ with some $k < m$ are left as free parameters, the second derivative with respect to $\th_{p_0}$ can be written as
\begin{align}
    \channel^{(2)}_{\th_1, \th_2,\dots, \th_k} := \left. \left(\frac{\partial^2 \channel_{\thv}}{\partial \theta_{p_0}^2}\right)\right|_{\th_{k+1}= \th_{k+2} = \dots=\th_\nparams = 0} \;.
\end{align}
We remark that the parameter $\theta_{p_0}$ which the second derivative is with respect to is implicitly hidden in this compact notation.

%\medskip

We now apply Corollary~\ref{cor:var_LB_1param_layers} to lower bound the loss variance $\LC(\thv)$ with respect to a single parameter $\theta_{p_0}$ in Eq.~\eqref{eq:proof_th4_step1}. In particular, with  $K \rightarrow K_{p_0}, \mathcal{V}_k \rightarrow \widetilde{\mathcal{V}}_{l}$ and $H_k \rightarrow H_{[\mathcal{S}^{-1}(\paramindex_0)]_k}$, we obtain
\begin{align}
\Var_{\paramindex_0}\left[\Tr\left[ \rho \;\mathcal{E}_{\thv}(O)\right]\right] &\geq \frac{r^4}{45} \Tr\left[\rho \; \mathcal{E}_{\vec{\th}_{\overline{{\paramindex_0}}}}^{(2)}(O) \right]^2 - \frac{32 \left(\sum_{i \in \mathcal{S}^{-1}(\paramindex_0)} \omega^{(\rm max)}(H_i)\right)^6 \norm{O}^2}{135} r^6\;\\
&=\frac{r^4}{45} \left(\Tr\left[\rho \; \mathcal{E}_{\vec{\th}_{\overline{{\paramindex_0}}}}^{(2)}(O) \right]^2 - \frac{32 \left( \omega^{(\rm max)}_{p_0}\right)^6 \norm{O}^2}{3} r^2\right)\;, \label{eq:final_LB_step1_th4}
\end{align}
provided that the perturbation $r$ satisfies
\begin{equation}\label{eq:cond1_NEW0}
    r^2 \leq \frac{9}{16 \;\left[\omega^{(\rm max)}_{p_0}\right]^2}\;,
\end{equation}
where we denote $\omega^{(\rm max)}_{p_0}= \sum_{i \in \mathcal{S}^{-1}(\paramindex_0)} \omega^{(\rm max)}(H_i)$.

\bigskip

\paragraph*{\underline{3.~Lower bounding $\Ebb_{\overline{{\paramindex_0}}}\left[\Var_{p_0}[\mathcal{L}(\thv)]\right]$.}} 
By applying the expectation over the remaining parameters $\th_{\overline{p_0}}$ in Eq.~\eqref{eq:final_LB_step1_th4} and further denoting $\beta_{p_0} =\frac{32}{3} \left[\omega^{(\rm max)}_{p_0}\right]^6 \norm{O}^2$, we get 
\begin{align}
    \Var\left[\mathcal{L}(\boldsymbol{\theta})\right]  &\geq\Ebb_{\overline{{\paramindex_0}}}[\Var_{\paramindex_0}[\mathcal{L}(\theta)]]\\
    &\geq \frac{r^4}{45}\left( \Ebb_{\overline{{\paramindex_0}}}\left[\Tr[\rho \; \mathcal{E}_{\vec{\th}_{\overline{{\paramindex_0}}}}^{(2)}(O) ]^2\right]  - \beta_{p_0} r^2\right)\\
    &\geq \frac{r^4}{45} \left(\left(\Ebb_{\overline{{\paramindex_0}}}\left[\Tr\left[\rho \; \mathcal{E}_{\vec{\th}_{\overline{{\paramindex_0}}}}^{(2)}(O) \right]\right]\right)^2  - \beta_{p_0} r^2\right)\\
     &= \frac{r^4}{45} \left(\Tr[\rho \; \Ebb_{\overline{{\paramindex_0}}}\left[\mathcal{E}_{\vec{\th}_{\overline{{\paramindex_0}}}}^{(2)}(O) \right]]^2  - \beta_{p_0} r^2\right),\label{eq:var_lb_remainder}
\end{align}
where in the third inequality we used Jensen's inequality and in the last equality we used the linearity of the trace. 

\bigskip

\paragraph*{\underline{4. Lower bound $\Tr[\rho \; \Ebb_{\overline{{\paramindex_0}}}\left[\mathcal{E}_{\vec{\th}_{\overline{{\paramindex_0}}}}^{(2)}(O) \right]]^2$. }} 
Similarly to the fourth key proof step of Theorem \ref{th:var_formal}, the core idea to lower bound the term $\Tr[\rho \; \Ebb_{\overline{{\paramindex_0}}}\left[\mathcal{E}_{\vec{\th}_{\overline{{\paramindex_0}}}}^{(2)}(O) \right]]^2$ is to decompose  
$\Ebb_{\overline{{\paramindex_0}}}\left[\mathcal{E}_{\vec{\th}_{\overline{{\paramindex_0}}}}^{(2)}(O) \right]$ into a constant and a perturbation dependent component.
Specifically, we recall here that $\mathcal{E}_{\vec{\th}_{\overline{{\paramindex_0}}}}^{(2)}(O)$ still depends on the parameters $\vec{\th}_{\overline{{\paramindex_0}}}$ and we can decompose its expectation with respect to $\vec{\th}_{\overline{{\paramindex_0}}}$  into  a constant term $\channel^{(2)}(O)$ (i.e., the second derivative of $\channel_{\thv}(O)$ with respect to $\th_{p_0}$ evaluated at $\thv = \vec{0}$) and a perturbation dependent term $\mathcal{T}_{p_0}(O)$ of the form %leading to
\begin{align}
 \mathcal{T}_{p_0}(O) = \Ebb_{\vec{\th}_{\overline{{\paramindex_0}}}} \left[\mathcal{E}_{\vec{\th}_{\overline{{\paramindex_0}}}}^{(2)}(O) \right] - \mathcal{E}^{(2)}(O) \label{eq:decomp} \;.
\end{align}

By plugging this decomposition in Eq.~\eqref{eq:decomp} in the variance lower bound in Eq.~\eqref{eq:var_lb_remainder}, we obtain
\begin{align}
     \Var\left[\mathcal{L}[\boldsymbol{\theta}]\right]  &\geq \frac{r^4}{45} \left(\Tr[\rho \; \left(\mathcal{E}^{(2)}(O) + \mathcal{T}_{p_0}(O)\right)]^2  - \beta_{p_0} r^2 \right)\\
     & \geq \frac{r^4}{45} \left(\left(\left|\Tr[\rho \;\mathcal{E}^{(2)}(O)]\right|  -\left|\Tr[\rho \; \mathcal{T}_{p_0}(O)]\right|\right)^2 - \beta_{p_0} r^2 \right) \;, \label{eq:LB_sep}
\end{align}
where the last inequality is obtained using the reverse triangle inequality.

\medskip

To further lower bound Eq.~\eqref{eq:LB_sep}, our strategy is to choose $r$ such that the perturbation is sufficiently small. In particular, we are interested in a condition for $r$ such that
\begin{align}\label{eq:cond1_NEW}
    \left|\Tr[\rho \;\mathcal{E}^{(2)}(O)]\right|  \geq \left|\Tr[\rho \; \mathcal{T}_{p_0}(O)]\right| \;.
\end{align}

In order to find the restriction on the perturbation $r$ for which  Eq.~\eqref{eq:cond1_NEW} is satisfied, it is sufficient to fulfill the following two conditions: (i).~identify the parameter-independent function $ \gamma_{\paramindex_0}$ such that $\left|\Tr[\rho \; \mathcal{T}_{p_0}(O)]\right| \leq \gamma_{\paramindex_0} r^2$ and (ii).~choose $r$ such that
\begin{equation}\label{eq:cond_r_inside_general}
    \left|\Tr[\rho \;\mathcal{E}^{(2)}(O)]\right|  \geq \gamma_{p_0} r^2 \;,
\end{equation}
which leads to $\left|\Tr[\rho \;\mathcal{E}^{(2)}(O)]\right|  \geq \gamma_{p_0} r^2 \geq \left|\Tr[\rho \; \mathcal{T}_{p_0}(O)]\right|$ that in turn ensures the validation of Eq.~\eqref{eq:cond1_NEW}. As a friendly reminder, there is the other condition on $r$ also required to be satisfied as stated earlier in Eq.~\eqref{eq:cond1_NEW0}.

Similarly to the proof of Theorem~\ref{th:var_formal}, given that all conditions are satisfied, by using $\gamma_{p_0} r^2 \geq \left|\Tr[\rho \; \mathcal{T}_{p_0}(O)]\right|$ in Eq.~\eqref{eq:LB_sep} we obtain the following loss variance lower bound 
\begin{align}\label{eq:var-gamma-function}
    \Var_{\thv}\left[\mathcal{L}[\boldsymbol{\theta}]\right]   \geq \frac{r^4}{45} \left(\left(c_{p_0}(\vphi)  -\gamma_{p_{0}}r^2\right)^2 - \beta_{p_0} r^2 \right) \;,
\end{align}
where we notice that $c_{p_0}(\vphi) =  |\Tr[\rho \;\mathcal{E}^{(2)}(O)]| = \left| \left. \left(\frac{\partial \LC(\thv)}{\partial\th_{p_0}} \right) \right|_{\thv = \vphi}\right|$. We remark again that our compact notation for $\channel^{(2)}$ implicitly contains the dependence on $\vphi$ and $\theta_{p_0}$ (which is the parameter that got derivative with respect to). 

Note that the lower bound in Eq.~\eqref{eq:var-gamma-function} is already lookalike to the bound stated in the theorem in Eq.~\eqref{eq:thm-correlated-variance}. We now proceed by (i).~identifying $\gamma_{p_0}$ and (ii).~determining the restriction on $r$.

\medskip

\noindent\underline{\textit{4.1 Identifying } $\gamma_{\paramindex_0}$.}
The task in this setting is more intricate than one in Theorem \ref{th:var_formal}. This technical difficulty mainly arises due to non-consecutive gates being able to share the same parameter. Nonetheless, we will show that we do obtain similar results in both settings. 

First, from H\"{o}lder inequality, we have
\begin{align}
    \left|\Tr[\rho \; \mathcal{T}_{p_0}(O)]\right| &\leq \lVert \rho \rVert_1\norm{\mathcal{T}_{p_0}(O)} \\
    &=\norm{\mathcal{T}_{p_0}(O)}
    %&\leq \norm{\mathcal{T}_p} \norm{O}
\end{align}
where we used $\lVert \rho \rVert_1=1$ due to the positivity and normalization of a density operator.
Now, we focus on upper bounding $\norm{\mathcal{T}_{p_0}(O)} \leq \gamma_{p_0}$ such that $\gamma_{p_0}$ have no dependence on the parameters $\thv_{\overline{p_0}}$. Indeed, the end result will be the following
\begin{equation}\label{eq:gamma_p0}
    \gamma_{p_0} = \frac{8}{3}r^2\norm{O}\left(\omega^{(\rm max)}_{p_0}\right)^2\sum_{j=1}^{m-1}\left(\omega^{(\rm max)}_{p_{j}}\right)^2\;.
\end{equation}

To prove Eq.~\eqref{eq:gamma_p0}, let us first clarify about something that we have been being elusive on purpose, which is the permutation $\pi$. This permutation indeed gives us another map which tells us which parameters are being averaged first in Eq.~\eqref{eq:proof_th4_step1}. In other words, since we consider the single parameter $\theta_{p_0}$ for the variance, we have freedom due to the choice of the permutation map 
 $\pi$ over the order of other parameters that are averaged. In particular, $\pi: \{1, 2, \dots, \nparams\} \rightarrow \{p_0, p_1,\dots, p_{m-1} \}$ where $\pi$ respects the permutation. Crucially, the order of $\{\th_{p_0}, \th_{p_1}, \dots, \th_{p_{m-1}}\}$ does not at all have to follow the gate sequence. From here, to highlight this freedom, we will work with the permutation parameter indices.

Next, we consider the following cute quantity
\begin{equation}\label{eq:def-tj}
    t_j := \norm{\Ebb_{\overline{p_0,p_1,...,p_j}} \left[\mathcal{E}_{\overline{p_0}}^{(2)}(O) \right] - \mathcal{E}_{p_0,p_1,...,p_j}^{(2)}(O)}\;,
\end{equation}
for $j\in\{0,1,...,m-1\}$. Here it could be a good stopping point to remind some notations. First, a bar index means excluding that parameter i.e., $\overline{p_0,p_1,...,p_j}$ means the set of \textit{other} parameters \textit{excluding} $\{\th_{p_0}, \th_{p_1}, .., \th_{p_j}\}$. For the average $\Ebb$, the subscript indices refer to those parameters that are averaged over. Hence, $\Ebb_{\overline{p_0,p_1,...,p_j}} [\cdot]$ means we average over other parameters that are \textit{NOT} $\{\th_{p_0}, \th_{p_1}, .., \th_{p_j}\}$. On the other hand, the second derivative of the channel contains free parameters indicated as subscript and parameters that are already evaluated which are omitted. In particular, $\mathcal{E}_{p_0,p_1,...,p_j}^{(2)}$ has free parameters $\{\th_{p_0}, \th_{p_1}, .., \th_{p_j}\}$ and other parameters already evaluated (at the fixed point). From the definition in Eq.~\eqref{eq:def-tj}, we can see that $t_0=\norm{\mathcal{T}_{p_0}(O)}$ and $t_{m-1}=0$. 

Let us derive a relation between $t_j$ and $t_{j+1}$ to upper bound $t_0$. First, we will add and remove a term in the norm and use triangle inequality to split it into two terms as follows
\begin{align}
    t_j&=\norm{\Ebb_{\overline{p_0,p_1,...,p_j}} \left[\mathcal{E}_{\overline{p_0}}^{(2)}(O) \right] - \mathcal{E}_{p_0,p_1,...,p_j}^{(2)}(O)} \\
    &=\norm{\Ebb_{\overline{p_0,p_1,...,p_j}} \left[\mathcal{E}_{\overline{p_0}}^{(2)}(O) \right]-\Ebb_{p_{j+1}}\left[\mathcal{E}_{p_0,p_1,...,p_j,p_{j+1}}^{(2)}(O)\right] + \Ebb_{p_{j+1}}\left[\mathcal{E}_{p_0,p_1,...,p_j,p_{j+1}}^{(2)}(O)\right] - \mathcal{E}_{p_0,p_1,...,p_j}^{(2)}(O)}\\
    &\leq \norm{\Ebb_{\overline{p_0,p_1,...,p_j}} \left[\mathcal{E}_{\overline{p_0}}^{(2)}(O) \right]-\Ebb_{p_{j+1}}\left[\mathcal{E}_{p_0,p_1,...,p_j,p_{j+1}}^{(2)}(O)\right]} + \norm{\Ebb_{p_{j+1}}\left[\mathcal{E}_{p_0,p_1,...,p_j,p_{j+1}}^{(2)}(O)\right] - \mathcal{E}_{p_0,p_1,...,p_j}^{(2)}(O)}\;. \label{eq:proof-trash-induction-step-two-term-thm4}
\end{align}
The first term can be upper bounded by $\Ebb_{p_{j+1}}[t_{j+1}]$ as follows
\begin{align}
\norm{\Ebb_{\overline{p_0,\dots,p_j}} \left[\mathcal{E}_{\overline{p_0}}^{(2)}(O) \right]-\Ebb_{p_{j+1}}\left[\mathcal{E}_{p_0,\dots,p_j,p_{j+1}}^{(2)}(O)\right]} &=\norm{\Ebb_{p_{j+1}}\left[\Ebb_{\overline{p_0,\dots,p_j,p_{j+1}}} \left[\mathcal{E}_{\overline{p_0}}^{(2)} (O)\right]-\mathcal{E}_{p_0,\dots,p_j,p_{j+1}}^{(2)}(O)\right]} \\
&\leq \Ebb_{p_{j+1}}\left[\norm{\Ebb_{\overline{p_0,p_1,...,p_j,p_{j+1}}} \left[\mathcal{E}_{\overline{p_0}}^{(2)}(O) \right]-\mathcal{E}_{p_0,p_1,...,p_j,p_{j+1}}^{(2)}(O)}\right]\\
&=\Ebb_{p_{j+1}}\left[t_{j+1}\right]\;,
\end{align}
where the inequality is due to the convexity of the norm together with Jensen's inequality. For the readers who are exhausted, please hang in there. One of the authors often found watching some cute cat videos on the internet a big help for refreshing.

Now, the second term in Eq.~\eqref{eq:proof-trash-induction-step-two-term-thm4} can be upper bounded using Proposition~\ref{prop:expectation_taylor_product} with respect to the variable $\theta_{j+1}$  identifying $\Lambda_{\theta_{j+1}}(O) \rightarrow \mathcal{E}_{p_0,p_1,...,p_j,p_{j+1}}^{(2)}(O)$ and $\Lambda_{0}(O) \rightarrow \mathcal{E}_{p_0,p_1,...,p_j}^{(2)}(O)$. 
Moreover, the upper bound on the second derivative of $\mathcal{E}_{p_0,p_1,...,p_j,p_{j+1}}^{(2)}(O)$ with respect to $\theta_{p_{j+1}}$ (i.e. the constant $\gamma$ in the statement of Proposition~\ref{prop:expectation_taylor_product}) can be directly obtained using Lemma~\ref{lemma:double_deriv_bound} such that 
\begin{align}
   \norm{\left.\left(\frac{\partial^2}{\partial\th_{p_{j+1}}^2} [\channel^{(2)}_{p_1,\dots,p_j,p_{j+1}}(O)]\right)\right|_{\th_{p_{j+1}}= \nu} } &= \norm{\left.\left(\frac{\partial^4 \mathcal{E}_{p_1,\dots,p_{j+1} }(O)}{\partial \th_p^2 \; \partial  \th_{p_{j+1}}^2}\right)\right|_{\th_{p_0}=0,\th_{p_{j+1}}=\nu}} \\
   &\leq \left(\sum_{\layerindex \in \mathcal{S}^{-1}(p_0)} 2\omega^{(\rm max)}(H_l)\right)^2 \left(\sum_{\layerindex \in \mathcal{S}^{-1}(p_{j+1})} 2\omega^{(\rm max)}(H_l)\right)^2 \norm{O}\\
   &=16\left(\omega^{(\rm max)}_{p_0}\right)^2\left(\omega^{(\rm max)}_{p_{j+1}}\right)^2\norm{O}
\end{align}
Combining it with  Proposition~\ref{prop:expectation_taylor_product} ensures that
\begin{align}
    \norm{\Ebb_{p_{j+1}}\left[\mathcal{E}_{p_0,p_1,...,p_j,p_{j+1}}^{(2)}(O)\right] - \mathcal{E}_{p_0,p_1,...,p_j}^{(2)}(O)} &\leq \frac{8}{3}r^2\left(\omega^{(\rm max)}_{p_0}\right)^2\left(\omega^{(\rm max)}_{p_{j+1}}\right)^2\norm{O} \\ 
    &:=\alpha_{j+1}\;.
\end{align}
Therefore, we have shown 
\begin{equation}
    t_j\leq \Ebb_{p_{j+1}}[t_{j+1}]+\alpha_{j+1}\;.
\end{equation}
Given that each $\alpha_j$ does not depend on the variational parameters, we have, by induction, that
\begin{align}
    \norm{\mathcal{T}_{p_0}(O)} &=t_0\\
    &\leq \Ebb_{p_{1}}[t_{1}] + \alpha_{1}\\
    & \leq \Ebb_{p_{1},p_2}[t_2] + \alpha_1 + \alpha_2\\
    &\;\; \vdots\\
    &\leq \Ebb_{p_1,\dots,p_{m-1}}[t_{m-1}]+ \sum_{j=1}^{m-1} \alpha_j \\
    &=\sum_{j=1}^{m-1} \alpha_j\;,
\end{align}
where we recall that $t_{m-1}=0$ in the last equality. So, we finally end up with 
\begin{align}
    \norm{\mathcal{T}_{p_0}(O)}& \leq \frac{8}{3}r^2\norm{O}\left(\omega^{(\rm max)}_{p_0}\right)^2\sum_{j=1}^{m-1}\left(\omega^{(\rm max)}_{p_{j}}\right)^2 \\
    & := \gamma_{p_0} r^2\;,
\end{align}
which proves the claim in Eq.~\eqref{eq:gamma_p0}. Upon specifying $\gamma_{p_0}$ in Eq.~\eqref{eq:gamma_p0} to the lower bound in Eq.~\eqref{eq:var-gamma-function}, we achieve the loss variance bound promised in Eq.~\eqref{eq:thm-correlated-variance} in the main theorem. 

\medskip

\noindent\underline{\textit{4.2 Determining the restriction on } $r$.} 
The obtained lower bound does not work for any arbitrary $r$ but rather subjected to some conditions. Here we focus on precisely identifying this restriction. In particular, for the variance bound to hold, it is sufficient to simultaneously satisfy Eq.~\eqref{eq:cond1_NEW0} and Eq.~\eqref{eq:cond_r_inside_general} together with one additional condition directly obtained form the lower bound. The latest one is basically the condition that the lower bound in Eq.~\eqref{eq:var-gamma-function} is non-zero i.e.,
\begin{align}
    \left(c_{p_0}(\vphi)  -\gamma_{p_{0}}r^2\right)^2 - \beta_{p_0} r^2  \geq 0 \;.
\end{align}

Crucially, note that this setting to determine the restriction on $r$ is identical to the one in the proof of Theorem~\ref{th:var_formal} (with $\alpha_{l}(\vphi)$ replaced by $\gamma_{p_0}$ and with $\beta_l$ replaced by $\beta_{p_0}$). By following the same steps from Eq.~\eqref{eq:cond_r_non_zero_var} to Eq.~\eqref{eq:r_final_to_corr}, the following unified condition on $r$ is obtained

\begin{align}
    r^2 & \leq \frac{3c^2_{p_0}(\vec{\phi})}{8(2c_{p_0}(\vec{\phi})\gamma_{p_0} + \beta_{p_0})} \; \\
    & = \frac{9c^2_{p_0}(\vec{\phi})}{8\left(16 c_{p_0}(\vec{\phi}) \norm{O} \left[\omega^{(\rm max)}_{p_0}\right]^2\sum_{j=1}^{m-1}\left[\omega^{(\rm max)}_{p_{j}}\right]^2 + 32 \norm{O}^2 \left[\omega^{(\rm max)}_{p_0}\right]^6 \right) }\\
    & := r^2_{\rm patch}(\vphi) \label{eq:cond_merged_cor}\;.
\end{align}

Hence, given that $r \leq r_{\rm patch}(\vphi)$, we have the following variance lower bound
\begin{align}
    \Var_{\thv}\left[\mathcal{L}[\boldsymbol{\theta}]\right]   & \geq \frac{r^4}{45} \left(\left(c_{p_0}(\vphi)  -\gamma_{p_{0}}r^2\right)^2 - \beta_{p_0} r^2 \right) \\
    & \geq \frac{1}{72}c_{p_0}(\vphi)r^4 \;,
\end{align}
where the same approach as in Eq.~\eqref{eq:var_LB_l} is used to reach the final line.

Moreover, by using $r = r_{\rm patch}$, we have
\begin{equation}
    \Var_{\thv}[\LC(\thv)]  \geq \frac{1}{72}    c_{p_0}(\vec{\phi})^2   r_{\rm patch}^4(\vphi) \;.
\end{equation}

\bigskip

\paragraph*{\underline{5.~Theoretical guarantee on the substantial gradient region.}} By assuming that $\norm{O}, \norm{H_k} \in \mathcal{O}(\poly(n))$ and $c_{p_0}(\vec{\phi}) \in \Omega\left(\frac{1}{\poly(n)}\right)$, we can prove the guarantee on the region of attraction by following the similar proof steps from Eq.~\eqref{eq:proof-roa-1} to Eq.~\eqref{eq:proof_th1_nonvanish} in the proof of Theorem~\ref{th:var_formal}, which lead to the polynomial variance lower bound as 
\begin{align}
     \Var_{\theta \sim \uni(\vec{\phi}, r)}[\LC(\thv)] \in \Omega\left(\frac{1}{\poly(n) m^2}\right) \;.
\end{align}
if $r$ is chosen such that 
\begin{equation}
    r \in \Theta\left(\frac{1}{ \poly(n) \cdot \sqrt{m }}\right) \;.
\end{equation}

\end{proof}

\subsubsection{Proof of Corollary~\ref{cor:var_minimum_app}: Theoretical guarantees of the region of attraction with the polynomial width around the global minimum }\label{subsec:var_minimum_app}

\begin{proof}
The proof of this Corollary mainly applies  Theorem~\ref{th:var_formal_cor} to a sufficiently good minimum $\thv^*$ based on two core assumptions concerning the observable, and the interaction between the circuit and the initial state:
\begin{itemize}
    \item Assumption~1: The gap $\gap$ between the ground state and first excited state energies of $O$ vanishes at worst polynomially with the number of qubits. The ground state is also assumed to be non-degenerate. In particular, with the eigen decomposition of $O = \sum_{i=1}^{2^n} \lambda_i| \lambda_i \rangle\langle\lambda_i|$ where $\lambda_i$ is an eigen energy with an associated eigenstate $|\lambda_i\rangle$ labeled in increasing order i.e.,  $\lambda_1 < \lambda_2 \leq \dots \leq \lambda_{2^n}$, the gap satisfies
    \begin{align}
        \gap = \lambda_2 - \lambda_1 \in \Omega\left( \frac{1}{\poly(n)}\right) \;.
    \end{align}
    \item Assumption~2: The interaction between the gate closest to the initial state and the initial state itself has to be non-trivial. More specifically, given the initial state $\rho$ and the gate generator (closest to the initial state) $H_M$, we require
    \begin{align}
        \Var_\rho(H_M) = \Tr[\rho H_M^2] - \Tr[\rho H_M]^2\in\Omega(1/{\rm poly}(n)) \;.
    \end{align}
    Additionally, the parameter associated with the gate generator $H_M$ is required to be non-temporally correlated with other parameters. Remark that all other parameters can be arbitrarily correlated among themselves. 
\end{itemize}

Before proceeding, we want to emphasize the following: since it is equivalent to put the circuit $U(\thv^*)$ forward propagated on the initial state $\rho$ or back-propagated on the observable $O$
\begin{align}
    \Tr[U(\thv^*)\rho U^\dagger(\thv^*) O] = \Tr[\rho^* O] = \Tr[\rho O^*] \;,
\end{align}
with the evolved state $\rho^* = U(\thv^*)\rho U^\dagger(\thv^*)$ and the effective back-propagated observable $O^* =  U^\dagger (\thv^*) O U(\thv^*)$, the overlap of the initial state $\rho$ with the ground state of $O^*$ is the same as the overlap of the evolved state $\rho^*$ with the ground state of $O$. In this proof, we choose to work with the basis of $O^*$, that is, in the back-propagated observable.

Now, we aim to apply Theorem~\ref{th:var_formal_cor} at the fixed point $\thv^*$ under the assumptions stated above. We choose the parameter of interest to be the one closest to the initial state i.e., $\theta_{p_0} = \theta_m$ corresponding to the gate generator $H_M$. For ease of notation, we drop the subscript $\nHam$ in $H_\nHam$.

We then proceed to analyze the second derivative of the loss $c_\nparams\left(\thv^*\right)$ which is the core ingredient to compute the variance lower bound. 

\begin{align}
    c_\nparams\left(\thv^*\right) &= \left|\Tr[\rho [H,[H, U^\dagger\left(\thv^*\right)OU\left(\thv^*\right)]]]\right|\\
     &= \left|\Tr[\rho [H,[H, O^*]]]\right|\\
    &= \left|\Tr[\rho (H^2 O^* + O^* H^2 -2 H O^* H)]\right|\\
    \label{eq:eq_bound_derivative_min0}
    &= \left|\Tr[O^*\rho H^2  + \rho O^* H^2 -2 \rho H O^* H]\right| \;,
\end{align}
where in the second equality we have simply used the notation mentioned above $O^* = U^\dagger\left(\thv^*\right)OU\left(\thv^*\right)$. In the third equality, we have explicitly expanded the commutator, and finally, in the last one we have used the cyclic properties of the trace to rearrange the terms in a more comfortable manner. 

Let us denote the eigen decomposition of the back-propagated observable $O^* = \sum_{i=1}^{2^n} \lambda_i\ketbra{\lambda_i^*}$ with $\ket{\lambda_i^*} = U^\dagger (\thv^*) \ket{\lambda_i}$ and note that $O$ and $O^*$ share the same eigen values. We then write the initial state as $\rho = \ketbra{\psi}$, where 
\begin{align}
    \ket{\psi} = k_1 \ket{\lambda_1^*} + \epsilon \ket{\lambda_1^{* \perp}} \;,
\end{align}
with  $k_1,\epsilon\in\mathbb{C}$, $|k_1| = \sqrt{1-|\epsilon|^2}$, and $\braket{\lambda_1^*}{\lambda_1^{* \perp}} = 0$. We can now leverage that the minimum that we are in is ``good enough'' such that $|\epsilon|\leq|k_1|$. Later on, we will provide a more specific condition (upperbound) on $|\epsilon|$. 
    
By using that $O^*\ket{\lambda_1^*}  = \lambda_1 \ket{\lambda_1^*}$ we can simplify Eq.~\eqref{eq:eq_bound_derivative_min0} and find 
\begin{align}
    c_\nparams(\thv^*) &= \left|\Tr[O^*\rho H^2  + \rho O^* H^2 -2 \rho H O^* H]\right|\\
    \label{eq:eq_bound_derivative_min1}
    &= 2 |k_1|^2 \left| \tau_0+ \tilde{\tau}_\epsilon \right| \\
    &\geq 2|k_1|^2 \left( |\tau_0| - |\tilde{\tau}_\epsilon|\right) \label{eq:proof-coro-min-main-inequa} \;,
\end{align}
where we have 
\begin{align}
     \tau_0 =& \lambda_1 \expval{H^2}{\lambda_1^*}  - \sum_{i=1}^{2^n} \lambda_i \expval{H\ketbra{\lambda_i^*}H}{\lambda_1^*}  \label{eq:bound_other_terms_c_intovariance} \;, \\
 \tilde{\tau}_\epsilon = & \frac{1}{2|k_1|^2}\left( \Tr\left[ O^*\sigma H^2 + \sigma O^* H^2 -2\sigma HO^*H \right] + |\epsilon|^2  \Tr\left[ O^*\ketbra{\lambda_1^{* \perp}} H^2 \right.\right. \\ 
 & \left.\left. \;\;\;\;\;\;\;\;\;\;\;\;\;\;\;+ \ketbra{\lambda_1^{* \perp}} O^* H^2 -2\ketbra{\lambda_1^{* \perp}} HO^*H \right]  \right) \;, \label{eq:tau_1}
\end{align}
with $\sigma = \epsilon k_1^*\ketbra{\lambda_1^{* \perp}}{{\lambda_1^*}} + \epsilon^*k_1\ketbra{{\lambda_1^*}}{\lambda_1^{* \perp}}$. The inequality in Eq.~\eqref{eq:proof-coro-min-main-inequa} comes from the reverse triangle inequality and is valid when
\begin{align} \label{eq:cond-tau}
    |\tau_0| \geq |\tilde{\tau}_\epsilon| \;.
\end{align}
Before moving forward, we make a few remarks. $\tau_0$ mainly represents the contribution from $\rho$ that aligns with the ground state. On the other hand, $\tilde{\tau}_\epsilon$ represents the contribution from the perpendicular component and disappears if $\epsilon = 0$. Hence, in the scenario where $|\epsilon| \ll 1$, it can be interpreted as a perturbation which in turn gives some intuition behind what condition would make Eq.~\eqref{eq:cond-tau} hold. Lastly, we notice that for $\sigma$ in $\tilde{\tau}_\epsilon$, it has only two non-zero eigenvalues which are $\pm|\epsilon k_1|$. %(non-degenerate). 

Our strategy is to individually bound $|\tau_0|$ and $|\tilde{\tau}_\epsilon|$, and identify the scenario when Eq.~\eqref{eq:cond-tau} holds. 

\bigskip 

\paragraph*{\underline{1.~Bounding the perturbation $\tilde{\tau}_\epsilon$.}} 
 
First, we use the triangle inequality (i.e. $\|A + B\|\leq \|A\| + \|B\|$) to bound the first trace term in Eq.~\eqref{eq:tau_1} %expression
\begin{align}
    \Tr\left[ O^*\sigma H^2 + \sigma O^* H^2 -2\sigma HO^*H \right]\leq& \|O^*\sigma H^2 + \sigma O^* H^2 -2\sigma HO^*H\|_1\\
    \leq & \|O^*\sigma H^2\|_1 + \| \sigma O^* H^2 \|_1 + \|2\sigma HO^*H\|_1 \\
    \leq& 4\|O\|_{\infty}\|H\|_{\infty}^2\|\sigma\|_1\\ 
=&8|\epsilon k_1|\|O\|_{\infty}\|H\|_{\infty}^2, \label{eq:tau_bound_trace1}
\end{align}
where we simplified the norm of a sum of matrices, with the sum of the different norms. The third inequality is by using Holder's inequality to bound the $1$-norm of products of matrices by a product of $1$-norm and infinity-norms together with $\|\sigma\|_1 = 2|\epsilon k_1|$.

Similarly, for the second trace in Eq.~\eqref{eq:tau_1} we can bound it using triangular and Holder's inequality to obtain\begin{align}\label{eq:tau_bound_trace2}
    \Tr\left[ O^*\ketbra{\lambda_1^{* \perp}} H^2 + \ketbra{\lambda_1^{* \perp}} O^* H^2 -2\ketbra{\lambda_1^{* \perp}} HO^*H \right]\leq 4\|O\|_{\infty}\|H\|_{\infty}^2
\end{align}
Together, we can upper bound $\tilde{\tau}_\epsilon$ in Eq.~\eqref{eq:tau_1} by using the results obtained in Eq.~\eqref{eq:tau_bound_trace1} and Eq.~\eqref{eq:tau_bound_trace2} to obtain
\begin{equation}\label{eq:final_bound_tau}
    |\tilde{\tau}_\epsilon|\leq \frac{2|\epsilon|}{|k_1|^2}(2|k_1| + |\epsilon|)\|O\|_{\infty}\|H\|_{\infty}^2.
\end{equation}

\bigskip

\paragraph*{\underline{2.~Bounding the term $\tau_0$}}
Now we focus on the term $\tau_0$ in Eq.~\eqref{eq:bound_other_terms_c_intovariance}.

To proceed, we re-write the term slightly
\begin{align}
  |\tau_0| & = \left|\lambda_1 \expval{H^2}{\lambda_1^*} - \lambda_1\expval{H}{\lambda_1^*}^2  - \sum_{i=2}^{2^n} \lambda_i \expval{H\ketbra{\lambda_i^*}H}{\lambda_1^*}\right|  \\
   & = \left|\lambda_1 \Var_{\lambda_1^*}(H)  - \sum_{i=2}^{2^n} \lambda_i \expval{H\ketbra{\lambda_i^*}H}{\lambda_1^*} \right| \;, %\label{eq:bound_other_terms_c_intovariance}
\end{align}
where ${\rm Var}_{\lambda_1^*}(H)=\expval{H^2}{\lambda_1^*}-\expval{H}{\lambda_1^*}^2$ is the variance of an generator $H$ with respect to $\ket{\lambda_1^*}$.

Next, we focus on the remaining terms in the sum. We can readily prove that the aforementioned sum is larger than $\lambda_1\Var_{\lambda_1^*}(H)$. Indeed, by leveraging that $\lambda_2\leq \lambda_i\, \forall \, i \neq1$
\begin{align}\sum_{i=2}^{2^n}\lambda_i\expval{H\ketbra{\lambda_i^*}H}{\lambda_1^*} &\geq \lambda_2  \sum_{i=2}^{2^n}\expval{H\ketbra{\lambda_i^*}H}{\lambda_1^*}\\
     & = \lambda_2 \expval{H\left(\1-\ketbra{\lambda_1^*}\right)H}{\lambda_1^*}\\ 
     & = \lambda_2 \Var_{\lambda_1^*}(H) \label{eq:eq_bound_derivative_min30}\\
     & >\lambda_1 \Var_{\lambda_1^*}(H)\;,\label{eq:eq_bound_derivative_min3}
\end{align}
where in the first equality we used that $\sum_{i=1}^{2^n} \ketbra{\lambda_i^*} = \1$, and in the last equality we used the fact that 
\begin{align}
    \expval{H\left(\1-\ketbra{\lambda_1^*}{\lambda_1^*}\right)H}{\lambda_1^*} =\expval{H^2}{\lambda_1^*} - \expval{H}{\lambda_1^*}^2 = \Var_{\lambda_1^*}(H)\;.
\end{align}

Next, we can lower bound $|\tau_0|$ as follows
\begin{align}
     |\tau_0| & =\left|\lambda_1 \Var_{\lambda_1^*}(H)  - \sum_{i=2}^{2^n} \lambda_i \expval{H\ketbra{\lambda_i^*}H}{\lambda_1^*} \right|\\
     & \geq   \sum_{i=2}^{2^n}\lambda_i\expval{H\ketbra{\lambda_i^*}{\lambda_i^*}H}{\lambda_1^*} -\lambda_1 {\rm Var}_{\lambda_1^*}(H) \\
    & \geq (\lambda_2-\lambda_1)\Var_{\lambda_1^*}(H)\\
    &=  \gap \Var_{\lambda_1^*}(H) \;. \label{eq:sec_deriv_minimum_LB}
\end{align}
where the first inequality and the second inequality comes from Eq.~\eqref{eq:eq_bound_derivative_min3} and Eq.~\eqref{eq:eq_bound_derivative_min30}, respectively.

\bigskip

\paragraph*{\underline{3.~Identify the condition when $|\tau_0| \geq |\tilde{\tau}_\epsilon|$.}} 

Given the scenario where  $|\tau_0| \geq |\tilde{\tau}_\epsilon|$ in Eq.~\eqref{eq:cond-tau} holds,
we lower bound $c_\nparams(\thv^*)$ in Eq.~\eqref{eq:proof-coro-min-main-inequa} leading to \begin{equation}\label{eq:lowerbound_c_prevar}
    c_\nparams(\thv^*)\geq 2 |k_1|^2 \gap \Var_{\lambda_1^*}(H) - 4|\epsilon|(2|k_1| + |\epsilon|)\norm{O}\norm{H}^2\;.
\end{equation}

Now, to be in this scenario $|\tau_0| \geq |\tilde{\tau}_\epsilon|$, it is sufficient to have
\begin{equation}\label{eq:condition_lb_aroundmin}
    \gap \Var_{\lambda_1^*}(H)>  \frac{2|\epsilon|}{|k_1|^2}(2|k_1| + |\epsilon|)\norm{O}\norm{H}^2\;.
\end{equation}

This condition in Eq.~\eqref{eq:condition_lb_aroundmin} is not trivial to obtain. Indeed, computing $\Var_{\lambda_1^*}(H)$ exactly is difficult to perform in general. However, computing $\Var_{\rho}(H)$ can be done. That is because $\rho$ is the initial state, and $H$ is the first generator, so in most cases this quantity can be classically computed. Therefore, to finish this, we try to find a lower bond of $\Var_{\lambda_1^*}(H)$ with respect to $\Var_{\rho}(H)$. 
Recalling $\ket{\psi} = k_1 \ket{\lambda_1^*} + \epsilon \ket{\lambda_1^{*\perp}}$ and $\sigma = \epsilon k_1^*\ketbra{\lambda_1^{* \perp}}{{\lambda_1^*}} + \epsilon^*k_1\ketbra{{\lambda_1^*}}{\lambda_1^{* \perp}}$, we can write $\rho=|k_1|^2\ketbra{\lambda_1^*}{\lambda_1^*}+|\epsilon|^2\ketbra{\lambda_1^{* \perp}}{\lambda_1^{* \perp}}+\sigma$. Thus we expand $\Var_{\rho}(H)$ as follows
    \begin{align}
     \Var_{\rho}(H) =&   |k_1|^2\expval{H^2}{\lambda_1^*}+|\epsilon|^2\expval{H^2}{\lambda_1^{* \perp}}+\Tr[\sigma H^2]\nonumber\\ 
     &-\left(|k_1|^2\expval{H}{\lambda_1^*}+|\epsilon|^2\expval{H}{\lambda_1^{* \perp}}+\Tr[\sigma H]\right)^2 \\
     =& |k_1|^2\Var_{\lambda_1^*}(H) + |\epsilon|^2\Var_{\lambda_1^{* \perp}}(H) + \Tr[\sigma H^2] - \Tr[\sigma H]^2\nonumber\\
     &+|\epsilon|^2|k_1|^2\left(\expval{H}{\lambda_1^{* \perp}}-\expval{H}{\lambda_1^*}\right)^2-2\Tr[\sigma H]\left(|\epsilon|^2\expval{H}{\lambda_1^{* \perp}}+|k_1|^2\expval{H}{\lambda_1^*}\right)\;,
    \end{align}
    where the last equality is obtained using $|k_1|^2+|\epsilon|^2=1$ (that is, to rewrite $|k_1|^4=|k_1|^2-|k_1|^2|\epsilon|^2$). Therefore, we have
    \begin{align}
        |k_1|^2\Var_{\lambda_1^*}(H) &= \Var_{\rho}(H)-\zeta\;,
    \end{align}
    where 
    \begin{align}
        \zeta:=& |\epsilon|^2\Var_{\lambda_1^{* \perp}}(H) + \Tr[\sigma H^2] - \Tr[\sigma H]^2+|\epsilon|^2|k_1|^2\left(\expval{H}{\lambda_1^{* \perp}}-\expval{H}{\lambda_1^*}\right)^2 \nonumber \\
        & -2\Tr[\sigma H]\left(|\epsilon|^2\expval{H}{\lambda_1^{* \perp}}+|k_1|^2\expval{H}{\lambda_1^*}\right)\;.
    \end{align}
    We can lower bound this quantity by using Holder's inequality, as we did in Eq.~\eqref{eq:tau_bound_trace1}. Let us upper bound each relevant term in the definition of $\zeta$ separately as follows. First, we have
    \begin{align}
        \Var_{\lambda_1^{* \perp}}(H)&\leq \expval{H^2}{\lambda_1^{* \perp}}\leq \norm{H}^2\;.
    \end{align}
    Then, as $\sigma$ is Hermitian, we have
    \begin{equation}
         \Tr[\sigma H^2] - \Tr[\sigma H]^2 \leq \Tr[\sigma H^2] \leq \|\sigma\|_1\norm{H}^2=2|\epsilon k_1|\norm{H}^2\;. 
    \end{equation}
Similarly, we have
\begin{align}
    \Tr[\sigma H]\left(|\epsilon|^2\expval{H}{\lambda_1^{* \perp}}+|k_1|^2\expval{H}{\lambda_1^*}\right)&\leq \|\sigma\|_1\|H\| \left| |\epsilon|^2\expval{H}{\lambda_1^{* \perp}}+|k_1|^2\expval{H}{\lambda_1^*}\right|\\
    &\leq \|\sigma\|_1\norm{H}^2 \left( |\epsilon|^2+|k_1|^2\right)\\
    &=2|\epsilon k_1|\norm{H}^2 \left( |\epsilon|^2+|k_1|^2\right)=2|\epsilon k_1|\norm{H}^2\;,
\end{align}
where we used $|\epsilon|^2+|k_1|^2=1$ in the last equality.
Finally, for the remaining term we can use 
\begin{equation}
    \left(\expval{H}{\lambda_1^{* \perp}}-\expval{H}{\lambda_1^*}\right)^2\leq 4\norm{H}^2\;.
\end{equation}
Therefore, we can upper bound $\zeta$ using triangle inequality and previous bounds for each relevant terms to get
\begin{align}
  |\zeta| &\leq |\epsilon|^2\norm{H}^2+2|\epsilon k_1|\norm{H}^2+ 4|\epsilon|^2|k_1|^2\norm{H}^2+4|\epsilon k_1|\norm{H}^2 \\
  &=|\epsilon|\left(6| k_1|+|\epsilon|(1+4|k_1|^2)\right)\norm{H}^2\\
  &\leq 11|\epsilon|\norm{H}^2\;,
\end{align}
where we used $|\epsilon|\leq|k_1|\leq 1$ in the last inequality. Therefore, we finally have the bound on $\Var_{\lambda_1^*}(H)$ in terms of $\Var_{\rho}(H)$ to be
\begin{equation}\label{eq:bound-var-lamb-var-rho}
    |k_1|^2\Var_{\lambda_1^*}(H) \geq \Var_{\rho}(H)-11|\epsilon| \norm{H}^2\;.
\end{equation}

By inserting Eq.~\eqref{eq:bound-var-lamb-var-rho} into Eq.~\eqref{eq:condition_lb_aroundmin} and Eq.~\eqref{eq:lowerbound_c_prevar}, we can ensure that
\begin{equation}
\label{eq:lowerbound_c_aftervar}
c_\nparams(\thv^*)\geq 2  \gap \left(\Var_{\rho}(H)-11|\epsilon| \norm{H}^2\right) - 4|\epsilon|(2|k_1| + |\epsilon|)\norm{O}\norm{H}^2
\end{equation}
if the following condition (referring to Eq.~\eqref{eq:condition_lb_aroundmin}) is satisfied 
 \begin{equation} \label{eq:condition_lb_aroundmin_aftervar}
        \gap (\Var_{\rho}(H)-11|\epsilon| \norm{H}^2)>  2|\epsilon|(2|k_1| + |\epsilon|)\norm{O}\norm{H}^2\;.
    \end{equation}

    For now, we focus on simplifying the condition in Eq.~\eqref{eq:condition_lb_aroundmin_aftervar}. Particularly, we start by using, again, that $|\epsilon|\leq|k_1|\leq 1$, and thus $(2|k_1| + |\epsilon|) < 3$. Indeed, with this we can rewrite the condition as
    \begin{equation}\label{eq:condition_lb_aroundmin_aftervar2}
       \gap \Var_{\rho}(H) > |\epsilon|(6\norm{O} + 11\gap)\norm{H}^2 \;.
     \end{equation}  
Upon rearranging terms, we have the following condition
\begin{equation}\label{eq:perturbation-cond-coro1}
   |\epsilon|< \frac{\gap \Var_{\rho}(H)}{(6\norm{O} + 11\gap)\norm{H}^2 } \;.
\end{equation}

    Similarly, Eq.~\eqref{eq:lowerbound_c_aftervar} can be simplified in the same manner to obtain the following bound on the second derivative.
\begin{equation}\label{eq:lowerbound_c_aftervar_2}
        c_\nparams(\thv^*)\geq 2  \gap \Var_{\rho}(H) - 2|\epsilon|\|H\|_{\infty}^2 (11\gap + 6\|O\|_{\infty}) 
    \end{equation}

\bigskip

\paragraph*{\underline{4.~Theoretical guarantee for the region of attraction.}} 
With the bound on the second derivative in Eq.~\eqref{eq:lowerbound_c_aftervar_2} subjected to the small $|\epsilon|$ infidelity in Eq.~\eqref{eq:perturbation-cond-coro1}, we are ready to invoke Theorem~\ref{th:var_formal_cor} around $\thv^*$. We also choose the parameter to be taken the variance in the theorem to be the one closest to the initial state i.e., $\theta_p = \theta_m$. From the theorem, for any perturbation $r$ follows
\begin{align}
     r^2 \leq r^2_{\rm patch} := \frac{9c^2_{m}(\thv^*)}{8\left(16 c_{m}(\thv^*) \norm{O} \left[\omega^{(\rm max)}_{m}\right]^2\sum_{j=1}^{m-1}\left[\omega^{(\rm max)}_{j}\right]^2 + 32 \norm{O}^2 \left[\omega^{(\rm max)}_{m}\right]^6 \right) } \;,
\end{align}
we have that the variance lower bound scales as
\begin{align}
   \Var_{\sim \uni(\thv^*, r)}[\LC(\thv)] \geq  \frac{1}{72}    c_m(\thv^*)^2   r^4 \;.
\end{align}

With the bound on the second derivative in Eq.~\eqref{eq:lowerbound_c_aftervar_2}, we now identify the new patch's size $r_{\rm patch}^*$ such that
\begin{align}
    r_{\rm patch}^* < r_{\rm patch} \;,
\end{align}
and hence the variance lower bound holds. To do so, we keep performing a series of lower bounds on the $ r_{\rm patch}$ as follows
\begin{align}\label{eq:lowerbound_on_r_cor1}
    r^2_{\rm patch} & = \frac{9c^2_{m}(\vec{\phi})}{8\left(16 c_{m}(\vec{\phi}) \norm{O} \left[\omega^{(\rm max)}_{m}\right]^2\sum_{j=1}^{m-1}\left[\omega^{(\rm max)}_{j}\right]^2 + 32 \norm{O}^2 \left[\omega^{(\rm max)}_{m}\right]^6 \right) } \\
    & \geq \frac{9 \left[  \gap \Var_{\rho}(H) - |\epsilon|\|H\|_{\infty}^2 (11\gap + 6\|O\|_{\infty}) \right]^2}{2\left(16 c_{m}(\vec{\phi}) \norm{O} \left[\omega^{(\rm max)}_{m}\right]^2\sum_{j=1}^{m-1}\left[\omega^{(\rm max)}_{j}\right]^2 + 32 \norm{O}^2 \left[\omega^{(\rm max)}_{m}\right]^6 \right) } \\
    & \geq  \frac{9 [ \gap \Var_{\rho}(H) - |\epsilon|\|H\|_{\infty}^2 (11\gap + 6\|O\|_{\infty}) ]^2 }{128 \norm{O}^2 \left[ \omega^{\rm (max)}_{m} \right]^4 \sum_{j=1}^{\nparams-1} \left[ \omega^{\rm (max)}_{j} \right]^2 } \\
    & := r_{\rm patch}^{*2} \;,
\end{align}
where the first inequality is from the lower bound in Eq.~\eqref{eq:lowerbound_c_aftervar_2}, and the second inequality is by upper bounding the second derivative of the loss $c_{m}(\thv^*)$ with Lemma~\ref{lemma:bounded_unitary} and grouping the terms together. 

Within the new patch with $r^*_{\rm patch}$, the lower bound of the loss variance is
\begin{align}
    \Var_{\thv \sim \uni(\thv^*,r_{\rm patch}^*)}[\mathcal{L}(\thv)] \geq \frac{[ \gap \Var_{\rho}(H) - |\epsilon|\|H\|_{\infty}^2 (11\gap + 6\|O\|_{\infty}) ]^2}{18}  (r_{\rm patch}^*)^4 \;,
\end{align}
subjected to the small $|\epsilon|$ condition in Eq.~\eqref{eq:perturbation-cond-coro1}. 

Lastly, we can also obtain a  generic scaling  of the patch with guaranteed polynomially vanishing variance around the minimum $\thv^*$ by directly re-invoking the condition in Eq.~\eqref{eq:proof_th1_cond} and the associated variance lower bound in Eq.~\eqref{eq:proof_th1_LB} and substituting the condition on the second derivative in Eq.~\eqref{eq:proof_th1_nonvanish} by 
\begin{equation}
    \gap \Var_{\rho}(H_M) \in \Omega\left( \frac{1}{{\rm poly}(n)}\right) \;,
\end{equation}
which concludes the proof.

\end{proof}

\section{Applications to different specific circuit architectures}\label{appendix:architectures}

In this section, we provide more technical details and further discussion, in addition to Sec.~\ref{sec:architectures_main} in the main text, on the applications of the main theorems to specific circuit architectures. These include analytical expressions for scalings of the patch sizes and loss variances together with their derivations and further discussions on numerical results.

\bigskip

\paragraph*{\underline{The recipe for computing theoretical guarantees on substantial gradient regions for specific circuit architectures.}} Here, we provide a step-by-step guide on how to analyze the variance lower bound and the patch's size given a specific circuit.
\begin{enumerate}
    \item Given a circuit description in the form of Eq.~\eqref{eq:circuit} including sets of generators $\{H_i \}_{i=1}^\nHam$, sets of parameters $\{ \theta_l\}_{l=1}^\nparams$ and a point of interest $\vphi$, choose which theorem to use based on how parameters are correlated. In particular, we apply:
    \begin{itemize}
        \item  Theorem~\ref{th:var_formal} if the parameters are uncorrelated or spatial correlated.
        \item Theorem~\ref{th:var_formal_cor} if the parameters are arbitrarily correlated.
    \end{itemize} 
    \item For each parameter $\theta_l$, compute a second derivative of the loss with respect to the point $\vphi$ 
    \begin{align}
        c_l(\vphi) = \left|\left.\left(\frac{\partial^2\LC(\thv)}{\partial\th_l^2}\right)\right|_{\thv=\vec{\phi}}\right| \;.
    \end{align}
    \item Identify relevant parameters:
    \begin{itemize}
        \item For Theorem~\ref{th:var_formal}, choose the largest subset of parameter indices $\Lambda \subset \{1,...,m\}$ such that all the second derivatives in the subset do not exponentially vanish i.e., $\min_{l \in \Lambda} c_l(\vphi) \notin \OC(b^{-n})$ for any $b>1$. 
        \item For Theorem~\ref{th:var_formal_cor}, choose a parameter corresponding to the largest second derivative i.e., $p = \argmax_{l \in \{1,...,m\}} c_l(\vphi)$. %to maximize the patch's size.
    \end{itemize}
    \item Compute relevant frequencies:
    \begin{itemize}
        \item For Theorem~\ref{th:var_formal}, compute the maximal frequencies
    \begin{align}\label{eq:max_freqs_E}
        \omega_{\mu}^{(\rm max)}&:=\omega^{(\max)}(H_{\mu})=\lambda_{\max}(H_\mu) - \lambda_{\min}(H_\mu) \;,
    \end{align}
        where $\lambda_{\max}(H)$ and $\lambda_{\min}(H)$ are largest and smallest eigenvalues of $H$, as well as compute the effective frequencies including
        \begin{align}
         \omega_{\mu}^{\rm (eff) }(\vec{\phi})&= \sqrt{\norm{\left.\frac{\partial^2 [U(\thv)^\dagger O U(\thv)]}{\partial \th_\mu^2}\right|_{\thv= \vec{\phi}}}} \;,\label{eq:eff_freqs_E}\\ 
    \widetilde{\omega}_{l, \mu }^{(\rm eff)}(\vec{\phi}) &= \sqrt{\norm{ \left.\frac{\partial^4 [U(\thv)^\dagger O U(\thv)]}{\partial \th_\mu^2 \partial \th_l^2}\right|_{\thv= \vec{\phi}}}} \;,\label{eq:eff_freqs_tilde_E}
    \end{align}
    \item For Theorem~\ref{th:var_formal_cor}, compute the maximal frequencies $\omega^{(\rm max)}_q$ for all parameters
    \begin{align}
        \omega^{\rm (max)}_{q}  & =\sum_{l \in \mathcal{S}^{-1}(q)} \omega^{(\rm max)}(H_l) \;,
    \end{align}
    where $\mathcal{S}:\{1\dots,\nHam\} \rightarrow \{1\dots,\nparams\}$  maps every layer index $\layerindex$ and hence generator $H_l$ to its corresponding parameter index $\paramindex_0$, such that at layer $\layerindex$ and $\SC^{-1}$ is an inverse map of $\SC$.
    \end{itemize}
    \item Compute the patch's size $r_{\rm patch}$
    \begin{itemize}
        \item For Theorem~\ref{th:var_formal}, we have 
    \begin{align}\label{eq:rpatch_th3_recipe}
        r_{{\rm patch}}^{2} = \min_{l \in \Lambda} \frac{9c_l(\vec{\phi})^2}{8 c_l(\vec{\phi}) \mathcal{A}_{l}(\vphi) + 24\beta_l(\vphi)} 
    \end{align}
    where $\mathcal{A}_{l}(\vphi)$ and $\beta_l(\vphi)$ depends on which parameter we consider i.e.,
   
    \begin{align}
    \mathcal{A}_{l}(\vphi) &= 4 \left[ \omega^{(\rm max)}_l\right]^2 \sum_{\mu=1}^{l-1} \left[ \omega^{(\rm eff)}_\mu(\vphi) \right]^2 + \sum_{\mu=l+1}^{m} \left[ \widetilde{\omega}^{(\rm eff)}_{l,\mu}(\vphi) \right]^2 \label{eq:recipe_A}\\
            \beta_l(\vphi) &=\begin{cases}
       \frac{32( \omega_{l}^{\rm (max) })^{6}  \norm{O}^2}{3}\; {\rm if }\; l>1 \;,\\
       \frac{2 \left[\omega^{(\rm max)}_{1}\right]^2  \left[\omega^{\rm (eff)}_{1}(\phi_1)\right]^4}{3} \; {\rm if }\; l=1\;.
    \end{cases}  \label{eq:beta_recipe}
        \end{align}

        \item For Theorem~\ref{th:var_formal_cor}, we have
        
        \begin{align}\label{eq:rpatch_th4_recipe}
            r_{\rm patch}^2 = \frac{3c^2_{p}(\vec{\phi})}{8(2c_{p}(\vec{\phi})\gamma_{p} + \widetilde{\beta}_{p})} \;,
        \end{align}
        where $\widetilde{\beta}_{p}$ and $\gamma_{p}$ are given by
        \begin{align}
           \gamma_{p}&= \frac{8}{3}\norm{O}\left(\omega^{(\rm max)}_{p}\right)^2\sum_{\substack{q=1 \\ q\neq p}}^{m}\left(\omega^{(\rm max)}_{q}\right)^2 \label{eq:gamma_recipe} \\ 
           \widetilde{\beta}_p&=\frac{32}{3} \left[\omega^{(\rm max)}_{p}\right]^6 \norm{O}^2\;. \label{eq:beta_tilde_recipe}
        \end{align}
    \end{itemize}
    \item Compute the loss variance lower bound:
    \begin{itemize}
        \item For Theorem~\ref{th:var_formal}, the lower bound can be obtained as 
        \begin{align}\label{eq:Var_recipe_th3}
                \Var_{\thv \sim \uni(\vec{\phi}, r)}[\LC(\thv)]  \geq  \frac{1}{72}    \left( \sum_{l \in \Lambda} c_l(\vec{\phi})^2 \right)   r_{\rm patch}^4 
        \end{align}
        \item For Theorem~\ref{th:var_formal_cor}, the lower bound can be obtained as
        \begin{align}\label{eq:Var_recipe_th4}
        \Var_{\thv \sim \uni(\vec{\phi}, r)}[\LC(\thv)]  \geq \frac{1}{72}    c_p(\vec{\phi})^2   r_{\rm patch}^4
        \end{align}
    \end{itemize}
\end{enumerate}

\bigskip

\paragraph*{\underline{Outline for the rest of the section.}}

In the following subsections, we use this general recipe for different circuit architectures. For each circuit family, we present theoretical scaling of a patch's size, discussion about the implications to prior literature as well as a comparison with the numerical simulation. Note that in the numerics, we investigate the scaling of the patch's size corresponding to the maximum variance $r_{\rm max}$ and while the overall trends between analytics and numerics align, the mismatch in the exact scalings comes with no surprise. Lastly, all the proofs of theoretical results are presented at the end of the section.  

The rest of the section is structured as follows
\begin{itemize}
    \item Results and discussion for specific architectures.
    \begin{itemize}
        \item In Appendix~\ref{app:toy_worst_case}, we have Tensor product ansatz.
        \item In Appendix~\ref{subsec:HEA}, we have Hardware Efficient Ansatz. 
        \item In Appendix~\ref{subsec:HVA}, we have Hamiltonian Variational Ansatz.
        \item In Appendix~\ref{subsec:UCC}, we have Unitary Coupled Cluster Ansatz.
    \end{itemize}
    \item The proofs of theoretical results.
    \begin{itemize}
        \item In Appendix~\ref{subsubsec:proof_toy}, we have Tensor product ansatz.
        \item In Appendix~\ref{subsubsec:HEA_proof}, we have Hardware Efficient Ansatz.
        \item In Appendix~\ref{subsubsec:HVA_proof}, we have Hamiltonian Variational Ansatz.
        \item In Appendix~\ref{subsubsec:UCC_proof}, we have Unitary Coupled Cluster Ansatz.
    \end{itemize}
\end{itemize}

\subsection{Tensor product ansatz}\label{app:toy_worst_case}
\subsubsection{Circuit description}\label{subsubsec:circuit_toy}

We consider two parameterized circuits \( U(\thv_{\rm uncorrelated}) \) and \( U(\thv_{\rm correlated}) \), both having the same circuit architecture which consists of a single layer of \( R_X(\th) \), \( R_Z(\th) \) and \( R_X(\th) \) rotations applied to each qubit and differing only in their parameter correlation structure. In particular, we have
\begin{align}%\label{eq:circuits_toy_app}
    U(\vtheta_{\rm uncorrelated})&=  \left(\bigotimes_{i=1}^n e^{-i\theta_{i,x_1}\sigma_x^{(i)}}\right) \left(\bigotimes_{i=1}^n e^{-i\theta_{i,z}\sigma_z^{(i)}}\right)
\left(\bigotimes_{i=1}^n e^{-i\theta_{i,x_2}\sigma_x^{(i)}}\right)\label{eq:toy_uncorr_circuit} \;, \\
    U(\vtheta_{\rm correlated})
    &= \left(\bigotimes_{i=1}^n e^{-i\theta \sigma_x^{(i)}}\right) \left(\bigotimes_{i=1}^n e^{-i\theta \sigma_z^{(i)}}\right)
\left(\bigotimes_{i=1}^n e^{-i\theta \sigma_x^{(i)}}\right)\\
    &=  e^{-i\theta \sum_i \sigma_x^{(i)}} e^{-i\theta \sum_i \sigma_z^{(i)}}
     e^{-i\theta \sum_i \sigma_x^{(i)}} \\
     & :=  
     e^{-i\theta H_2} e^{-i\theta H_1}
     e^{-i\theta H_2}\, ,\label{eq:toy_cor_circuit}
\end{align}
where $\sigma_x^{(i)}$ and $\sigma_z^{(i)}$ are single-qubit Pauli-X and Pauli-Z on the $i^{\rm th}$ qubit respectively, and we introduce the shorthands $H_1 = \sum_i \sigma_z^{(i)}$ and $H_2 = \sum_i \sigma_x^{(i)}$ in the last equality.

We note that both circuits share the same set of \( \nHam = 3n \) generators, \( \{\sigma_{x}^{(i)}, \sigma_z^{(i)},\sigma_{x}^{(i)}\}_{i=1}^n \), but differ in the size of their parameter vectors (due to the correlation). The uncorrelated parameter vector \( \vtheta_{\rm uncorrelated} \) with components \( \{\th_{i,k}\}_{i=1,k \in \{z,x_1,x_2\}}^{n} \) has a total of \( \nparams_{\rm uncorrelated} = 3n \) independent components such that each generator has a unique parameter. In contrast, the fully correlated parameter vector \( \vtheta_{\rm correlated} = \th \) has \( \nparams_{\rm correlated} = 1 \), with a single parameter \( \th \) shared across all generators.

We also associate the circuits $U(\thv_{\rm uncorrelated})$ and  $U(\thv_{\rm correlated})$ to the loss functions $\mathcal{L}(\thv_{\rm uncorrelated})$ and $\mathcal{L}(\thv_{\rm correlated})$ of the form in Eq.~\eqref{eq:loss}   with the following initial state and observable,
\begin{align}
    \rho &= \frac{\left(|+\rangle^{\otimes n} + |-\rangle^{\otimes n}\right) \left(\langle-|^{\otimes n} + \langle+|^{\otimes n}\right)}{2}\label{eq:state_init_toy} \;, \\
    O &= \sigma_z^{\otimes n}\;. \label{eq:obs_toy}
\end{align}

\subsubsection{Analytical and numerical results}
We study the size of the region  with guaranteed gradients centered around zero for each of the loss functions.  
By adapting Theorem \ref{th:var_formal} (for the uncorrelated case) and Theorem \ref{th:var_formal_cor} (for the correlated case) to the loss functions $\mathcal{L}(\thv_{\rm uncorrelated})$ and $\mathcal{L}(\thv_{\rm correlated})$ around $\vphi = \vec{0}$, %described in the previous section \ref{subsubsec:circuit_toy}, 
we prove the following scaling of the patch size centered around zero with guaranteed substantial loss variance with the proof details presented in Appendix~\ref{subsubsec:proof_toy}. Furthermore, we retrieve the numerical scalings from Fig.~\ref{fig:scalings}. We summarize the numerical and analytical scaling results for comparison in Table~\ref{table:tensor_product_ansatz_results}. We see a similar scaling behavior in $r_{\rm patch},\, r_{\max}$ when comparing the correlated and uncorrelated cases. In particular, we see that both $r_{\rm patch}$ and $ r_{\max}$ are larger when the parameters are uncorrelated.

\begin{table}[ht!]
\begin{tabular}{ll|l|l|}
\cline{3-4}
                                                  &                    & Uncorrelated & Correlated \\ \hline
\multicolumn{1}{|l|}{\multirow{2}{*}{Analytical}} & $r_{\rm patch}$    & $\Theta( M^{-0.5})$            & $\Theta( M^{-1})$         \\ \cline{2-4} 
\multicolumn{1}{|l|}{}                            & $\Var_{\rm patch}$ & $\Omega(M^{-1})$               & $\Omega(1)$             \\ \hline
\multicolumn{1}{|l|}{\multirow{2}{*}{Numerical}}  & $r_{\rm max}$      & $\Theta( M^{-0.53})$           & $\Theta( M^{-1.05})$          \\ \cline{2-4} 
\multicolumn{1}{|l|}{}                            & $\Var_{\rm max}$   & $\Theta( M^{0.22})$            & $\Theta( M^{-0.03})$         \\ \hline
\end{tabular}
\caption{\textbf{Tensor product ansatz:} Summary of the scalings of the analytical $r_{\rm patch},\, \Var_{\rm patch}$, defined as the region in which we can guarantee non-exponential vanishing variance, and the numerical results $r_{\rm max},\, \Var_{\rm max}$, defined in Eq.~\eqref{eq:varmax} and Eq.~\eqref{eq:rmax} as the point $r$ such that the variance is maximized, and the value of such variance.}
\label{table:tensor_product_ansatz_results}
\end{table}

Finally, we highlight that this does not contradict previous results that relate correlating parameters to potentially better gradients~\cite{holmes2021connecting, volkoff2021large}. Certainly, as can be seen in Fig.~\ref{fig:corr_vs_uncorr}, the fact that the maximal value of the variance appears first for the correlated case, does not mean that the uncorrelated has a bigger variance. The figure shows how the variance of uncorrelated parameters decays faster than the variance of the correlated case when the hypercube is increased beyond $r_{\rm max}$.

\subsection{Hardware Efficient Ansatz}\label{subsec:HEA}
\subsubsection{Circuit description}\label{subsubsec:HEA_circuit}
We consider the Hardware Efficient Ansatz which has $L$ layers; each consists of $R_y$ and $R_z$ rotations on all qubits followed by a 1D ladder of CZ operations with periodic boundary conditions (i.e. CZ$(n,n+1):=$ CZ$(n,1)$).
\begin{equation}\label{eq:HEA_circuit}
\begin{aligned}
    U(\thv)&= \prod_{l=1}^L \left(\prod_{i=1}^n e^{-i\th_{l,y_i} \sigma_y^{(i)}} e^{-i\th_{l,z_i} \sigma_z^{(i)}} \right)V_l\;,
    \end{aligned}
\end{equation}
where the non parametrized gate is a layer of controlled Z gates $V_l = \prod_i {\rm CZ}(i,i+1)$. Here all the parameters in the circuit are independent. Indeed, we have that the total number of layers is  $\nHam = 2\cdot L \cdot n$ and the number of independent parameters is $\nparams = \nHam$. 
We also consider two loss functions  both of the form in Eq.~\eqref{eq:loss} with the initial state as $\rho = |0\rangle\langle0|^{\otimes n}$ and the parametrized circuit described in Eq.~\eqref{eq:HEA_circuit}.
We denote these loss functions by $\mathcal{L}_G(\thv)$ and $\mathcal{L}_L(\thv)$, which correspond to the use of a global and local observables respectively. More precisely, we consider the observables and the initial state.
\begin{align}
    O_G &:=  \sigma_z^{\otimes n}\label{eq:obs_G} \\
    O_L &:= \sigma_z \otimes  \sigma_z \otimes \1^{\otimes(n-2)} \label{eq:obs_L}
\end{align}

\subsubsection{Analytical and numerical results}\label{subsubsec:HEA_results}
With this specific architecture, we retrieve similar scalings of the patch's size $r_{\rm patch}$ with guaranteed gradients around zero as in Ref.~\cite{wang2023trainability}. In addition, while our result here is based on uniform distribution, extending it to other distributions, such as Gaussian distribution in Ref.~\cite{shi2024avoiding}, is straightforward.

On another note, we show the effect of the observable locality on the scaling of $r_{\rm patch}$ through the lens of effective frequencies (see Eq.~\eqref{eq:eff_freqs_E}). In particular, the scaling of the patch's size with the local observable is better than the other one with the global observable. This can be understood from the light cone argument (as will be detailed in Sec.~\ref{subsubsec:HEA_proof}), which results in the local observable having more zero effective frequencies.  

Table~\ref{table:hea_results} presents the analytical and numerical scaling results around $\vphi=\vec{0}$ of the architecture for $L = n$. The proof of analytical results as well as general expressions of the scalings for an arbitrary circuit depth are presented in Appendix~\ref{subsubsec:HEA_proof}. The expected trend from the effect of the observable locality is observed not only in the analytical result but also in the numerical result where $r_{\rm max}$ is considered.

\begin{table}[ht!]
\begin{tabular}{ll|l|l|}
\cline{3-4}
                                                  &                    & Global & Local \\ \hline
\multicolumn{1}{|l|}{\multirow{2}{*}{Analytical}} & $r_{\rm patch}$    & $\Theta( M^{-0.5})$            & $\Theta( M^{-0.25})$         \\ \cline{2-4} 
\multicolumn{1}{|l|}{}                            & $\Var_{\rm patch}$ & $\Omega(M^{-1})$               & $\Omega(M^{-0.5})$             \\ \hline
\multicolumn{1}{|l|}{\multirow{2}{*}{Numerical}}  & $r_{\rm max}$      & $\Theta( M^{-0.68})$           & $\Theta( M^{-0.34})$          \\ \cline{2-4} 
\multicolumn{1}{|l|}{}                            & $\Var_{\rm max}$   & $\Theta( M^{-0.53})$            & $\Theta( M^{-0.01})$         \\ \hline
\end{tabular}
\caption{\textbf{Harware Efficient Ansatz (HEA):} Summary of the scaling of the analytical $r_{\rm patch},\, \Var_{\rm patch}$, defined as the region in which we can guarantee non-exponential vanishing variance, and the numerical results $r_{\rm max},\, \Var_{\rm max}$, defined in Eqs.~(\ref{eq:varmax},\ref{eq:rmax}) as the point $r$ such that the variance is maximized, and the value of such variance.}
\label{table:hea_results}
\end{table}

\subsection{Hamiltonian Variational Ansatz}\label{subsec:HVA}

\subsubsection{Circuit description}\label{subsubsec:circuit_HVA}

We consider two variants of the Hamiltonian Variational Ansatz where the first one implements an $L$-step trotter evolution of a Hamiltonian $H= \sum_{k=1}^K H_k$ and the second one is a relaxation of the former where the parameters between each trotter step are varied independently.
Precisely, the parametrized circuits are of the form
\begin{align}
    U(\thv_{\rm Trotter}) &= \prod_{l=1}^L \prod_{k=1}^K e^{-i \th_{k}H_k} \label{eq:HVA_trotter}\\
    U(\thv_{\rm Relaxed}) &= \prod_{l=1}^L \prod_{k=1}^K e^{-i \th_{k,l}H_k}\label{eq:HVA_relaxed}
\end{align}

Here the parameter vector $\thv_{\rm Trotter}=\{\th_k\}_{k=1}^K$ is of size $\nparams_{\rm Trotter} = K$  while the parameter vector $\thv_{\rm Relaxed}=\{\th_{k,l}\}_{k,l=1}^{K,L}$is of size $\nparams_{\rm Relaxed} = KL$, where we recall that the components of each parameter vector are independent.

We further assume that $H= \sum_{k=1}^K H_k$ is \textit{geometrically} $\kappa$-local for a constant $\kappa$, i.e. each term $H_k$ acts on at most $\kappa \in \Theta(1)$ \textit{geometrically nearby} qubits in a given constant-dimensional lattice structure. 
Precisely, we have
\begin{align}
H&= \sum_{k=1}^K H_k \text{ with } K \in \Theta(1)\label{eq:loc_1}\\
   H_k &= \sum_{j=1}^{N_k} h_j^{(k)} \; ,\forall 1 \leq k\leq K \;, \label{eq:loc_2}
\end{align}
 where $h_j^{(k)}$ are commuting $\kappa$-local Pauli strings and  $N_k \in \mathcal{O}({\rm poly}(n))$ is the number of terms in the Pauli decomposition of each Hamiltonian $H_k$.

Under this general setting, we also do consider the specific example of the Heisenberg Hamiltonian with periodic boundary conditions, i.e.
\begin{align}\label{eq:heisenberg}
   \widetilde{H} & = \sum_{i=1}^n \left( \sigma_z^{(i)}\otimes\sigma_z^{(i+1)} + \sigma_y^{(i)}\otimes\sigma_y^{(i+1)} +  \sigma_x^{(i)}\otimes\sigma_x^{(i+1)} \right)  \\
   & = \widetilde{H}_1 + \widetilde{H}_2 + \widetilde{H}_3 \;,
\end{align}
where we denote the interaction terms in the Hamiltonian $\widetilde{H}$ by
\begin{align}\label{eq:heisenberg_terms}
    \widetilde{H}_1 = \sum_{i=1}^n \sigma_z^{(i)}\otimes\sigma_z^{(i+1)}, \quad
    \widetilde{H}_2 = \sum_{i=1}^n \sigma_y^{(i)}\otimes\sigma_y^{(i+1)}, \quad
    \widetilde{H}_3 = \sum_{i=1}^n \sigma_x^{(i)}\otimes\sigma_x^{(i+1)}\;.
\end{align}

Here, we clearly see that $\widetilde{H}$ satisfies the locality assumptions in  Eq.~\eqref{eq:loc_1} and Eq.~\eqref{eq:loc_2} with $K=3$ and each $H_k$ being $2$-local with $ N_k =n \;, \forall k\in \{1,2,3\}$.

\subsubsection{Analytical and numerical results}\label{subsusbsec:HVA_results}

\paragraph*{\underline{Generic Relaxed HVA with some geometrically local Hamiltonian $H$.}}\label{para:gen_hva_setting}

We begin by characterizing the region centered around zero with guaranteed substantial gradients for the loss function $\mathcal{L}(\thv_{\rm Relaxed})$ of the form in Eq.~\eqref{eq:loss}, with a circuit of the form in Eq.~\eqref{eq:HVA_relaxed}, i.e.
\begin{equation}\label{eq:hva_res_1}
    \mathcal{L}(\thv_{\rm Relaxed}) = \Tr[\rho U^\dagger(\thv_{\rm Relaxed}) O U(\thv_{\rm Relaxed})]\,.
\end{equation}
 Furthermore, we assume that the Hamiltonian \(H\) used to construct the circuit \(U(\thv_{\rm Relaxed})\) satisfies the conditions in Eq. \eqref{eq:loc_1} and Eq. \eqref{eq:loc_2}, meaning that \(H\) is a sum of \(\kappa\)-local Hamiltonian terms \(H_k\), where each term is a sum of commuting \(\kappa\)-local Pauli strings. We also assume that the observable $O$ is $\kappa_{O}$-local with  $\kappa_O \in  \mathcal{O}(1)$, i.e. it is composed of Pauli strings acting non trivially on at most $\kappa_O$ neighboring qubits, i.e. 
 \begin{equation}\label{eq:local_obs_hva}
     O = \sum_{i=1}^{N_O} P_i \;.
 \end{equation}
 
 Furthermore, we assume that the curvature (i.e. the second derivative of the loss function) with respect to the parameter $\th_{(1,1)}$ evaluated at zero and denoted by $c_{(1,1)}(\vec{0})$, scales as , i.e. %\stcom{how strong is this assumption ? }
 \begin{align}\label{eq:hva_cond_curv}
     c_{(1,1)}(\vec{0}) \in \Theta\left( N_O  \right)\;,
 \end{align}
 where we recall that the parameter $\th_{(1,1)}$ associated to the Hamiltonian term $H_1$ in the first trotter layer is the closest to the observable.

 Under this setting, we show that the size of the patch  with substantial loss variance scales as 
 \begin{align}\label{eq:hva_res_gen_patch}
    r_{\rm patch} \in \Theta\left( \frac{1}{\sqrt{M}}\right) \;,
\end{align}
where we recall that for this circuit $M =  K L$. The variance within this patch is lower bounded as\begin{equation}\label{eq:hva_res_gen_var}
     \Var_{\thv \sim \uni(\vec{0},r_{\rm patch})}[\mathcal{L}(\thv)] \in \Omega\left(\frac{N_O^2}{M^2}\right)\;.
 \end{equation}
 The proof of these scalings as well as other HVA results presented below is detailed in Appendix~\ref{subsubsec:HVA_proof}.

These results corroborate the findings in Ref.~\cite{park2023hamiltonian}. Here the authors demonstrated that if the loss function's  gradients scale as a constant in the system size at $\vphi = 0$, then the loss gradients will not vanish exponentially in a region $r \sim \frac{1}{KL \norm{O}n}$ with $K$ the number of Hamiltonian terms (see Eq.~\eqref{eq:loc_1}) and $L$ is the number of trotter layers. Moreover, note that our results extend beyond their assumption of non-exponentially vanishing initial gradients, as we can guarantee the presence of regions with substantial loss variance that are centered around a point with zero gradients.

\medskip

\paragraph*{\underline{Generic Trotter HVA with some geometrically local Hamiltonian $H$.}}\label{para:gen_hva_setting_trotter}

Next, we focus on characterizing the region centered around zero with guaranteed substantial gradients for the loss function $\mathcal{L}(\thv_{\rm Trotter})$ of the form in Eq.~\eqref{eq:loss}, i.e.
\begin{equation}\label{eq:hva_res_2}
    \mathcal{L}(\thv_{\rm Trotter}) = \Tr[\rho U^\dagger(\thv_{\rm Trotter}) O U(\thv_{\rm Trotter})]
\end{equation}
 using the Relaxed version of the HVA ansatz 
 in Eq.~\eqref{eq:HVA_trotter} for some local Hamiltonian $H$ verifying the assumptions in Eq.~\eqref{eq:loc_1} and Eq.~\eqref{eq:loc_2}. We also assume that the observable $O$ is local and can be decomposed into a sum of local Pauli terms as in Eq.~\eqref{eq:local_obs_hva}.

 Moreover, we assume that the loss second derivative w.r.t to the parameter $\th_1$ evaluated at zero and denoted by $c_1(\vec{0})$, scales as , i.e.
 \begin{align}\label{eq:hva_cond_curv_trotter}
     c_1(\vec{0}) \in \Theta\left( N_O \cdot L^2 \right)\;.
 \end{align}
 where we recall that $L$ is the number of Trotter layers and $N_O$ is the number of terms in the Pauli decomposition of the observable $O$. We also introduce the parameter $N = \max_k N_k$, where we recall that $N_k$ is the number of terms in the Pauli decomposition of each Hamiltonian $H_k$ (see Eq.~\eqref{eq:loc_2}).

 Under this setting, we show that the size of the patch  with substantial loss variance scales as 
 \begin{align}\label{eq:hva_res_gen_patch_trotter}
    r_{\rm patch} \in \Theta\left( \frac{1}{LN^3}\right) \underset{N , L\sim n}{=} \Theta\left( \frac{1}{M^4}\right)  \;,
\end{align}
where we recall that in this setting $M = KL \in \Theta(L)$ since we assume that $K \in \Theta(1)$.  The variance within this patch is lower bounded as
 \begin{equation}\label{eq:hva_res_gen_var_trotter}
     \Var_{\thv \sim \uni(\vec{0},r_{\rm patch})}[\mathcal{L}(\thv)] \in \Omega\left(\frac{N_O^2}{N^{12}}\right)  \underset{N , N_O, L\sim n}{=} \Omega\left( \frac{1}{M^{10}}\right)\;.
 \end{equation}

\paragraph*{\underline{Heisenberg model.}}

Next, we analytically and numerically study the HVA ansatz for the Heisenberg Hamiltonian $\widetilde{H}$ in Eq.~\eqref{eq:heisenberg}  using both the Trotter and the Relaxed variants, or in other words, with and without time correlations. Precisely, we define the loss functions $\widetilde{\mathcal{L}}(\thv_{\rm Relaxed})$ and $\widetilde{\mathcal{L}}(\thv_{\rm Trotter})$ as,
\begin{align}
    \widetilde{\mathcal{L}}(\thv_{\rm Relaxed}) &= \Tr[|\psi\rangle\langle\psi| \widetilde{U}^\dagger(\thv_{\rm Relaxed}) \widetilde{H}\widetilde{U}(\thv_{\rm Relaxed}) ] \label{eq:loss_relax_heisenberg}\\
    \widetilde{\mathcal{L}}(\thv_{\rm Trotter}) &= \Tr[|\psi\rangle\langle\psi| \widetilde{U}^\dagger(\thv_{\rm Trotter}) \widetilde{H}\widetilde{U}(\thv_{\rm Trotter}) ] \label{eq:loss_trotter_heisenberg}
\end{align}
where $\widetilde{U}(\thv_{\rm Relaxed})$ and $\widetilde{U}(\thv_{\rm Trotter})$ are of the form in Eq.~\eqref{eq:HVA_relaxed} and Eq.~\eqref{eq:HVA_trotter} associated with the Heisenberg Hamiltonian $\widetilde{H}$ in Eq.~\eqref{eq:heisenberg}. We also consider $|\psi\rangle$ to be the N\'eel state defined as $|\psi\rangle = \frac{|01\rangle^{\otimes n/2} + |10\rangle^{\otimes n/2}}{\sqrt{2}}$. Finally, we have $L = 8n$. Crucially, we verify that the assumption in Eq~\eqref{eq:hva_cond_curv_trotter} is held for this specific setting.

We summarize the analytical results proven in Section~\ref{subsubsec:HVA_proof} and the numerical results presented in Fig.~\ref{fig:scalings}, all condensed in Table~\ref{table:hva_results}. Note that we do the analysis in a hypercube centered at 0, i.e. $\vol(\vec{0},r)$. Here we see clearly, how both the numerical analysis and the analytic results show that correlating parameters reduces both $r_{\rm patch}$ and $r_{\rm max}$. We see that the scaling of numerical results of the relaxed version align well with those from the analytical ones. However, there is a big discrepancy in for the Trotter case.

\begin{table}[ht!]
\begin{tabular}{ll|l|l|}
\cline{3-4}
                                                  &                    & Relaxed & Trotter \\ \hline
\multicolumn{1}{|l|}{\multirow{2}{*}{Analytical}} & $r_{\rm patch}$    & $\Theta( M^{-0.5})$            & $\Theta( M^{-4})$         \\ \cline{2-4} 
\multicolumn{1}{|l|}{}                            & $\Var_{\rm patch}$ & $\Omega(M^{-1})$               & $\Omega(M^{-10})$             \\ \hline
\multicolumn{1}{|l|}{\multirow{2}{*}{Numerical}}  & $r_{\rm max}$      & $\Theta( M^{-0.48})$           & $\Theta( M^{-0.88})$          \\ \cline{2-4} 
\multicolumn{1}{|l|}{}                            & $\Var_{\rm max}$   & $\Theta( M^{1.86})$            & $\Theta( M^{1.74})$         \\ \hline
\end{tabular}
\caption{\textbf{Hamiltonian Variational Ansatz (HVA):} Summary of the scalings of the analytical $r_{\rm patch},\, \Var_{\rm patch}$, defined as the region in which we can guarantee non-exponential vanishing variance, and the numerical results $r_{\rm max},\, \Var_{\rm max}$, defined in Eqs.~(\ref{eq:varmax},\ref{eq:rmax}) as the point $r$ such that the variance is maximized, and the value of such variance.}
\label{table:hva_results}
\end{table}

\subsection{Unitary Coupled Cluster Ansatz}\label{subsec:UCC}

\subsubsection{Circuit description}\label{subsubsec:UCC_circuit}
In this section, we study two variants of the Unitary Coupled Cluster (UCC) ansatz which systematically performs single and double electron excitations known as Unitary Coupled Cluster Single and Double excitations (UCCSD),  together with a Hartree Fock state as an initial state. The first variant consists in using the $L$-steps Trotter approximation of the UCCSD ansatz \cite{mao2023barren} described by the parametrized unitary evolution 
\begin{equation}\label{eq:UCCSD_Trotter}
    U(\thv_{\rm Trotter}) = \prod_{l=1}^L \prod_{k=1}^K e^{\th_k (\hat{\tau}_k - \hat{\tau}_k^\dagger)} \; .
\end{equation}
The second variant additionally decouples the variational parameters such that 
\begin{equation}\label{eq:UCCSD_relax}
    U(\thv_{\rm Relaxed}) = \prod_{l=1}^L \prod_{k=1}^K e^{\th_{k,l} (\hat{\tau}_k - \hat{\tau}_k^\dagger)} \; .
\end{equation}
Here $\hat{\tau}_k \in \{\hat{a}_p^\dagger\hat{a}_q,\hat{a}_p^\dagger \hat{a}_q^\dagger \hat{a}_r \hat{a}_s\}$ are single and double excitation operators i.e., $\hat{a}_p^\dagger\hat{a}_q$ excites a single electron from the orbital $q$ to $p$, and $\hat{a}_p^\dagger \hat{a}_q^\dagger \hat{a}_r \hat{a}_s$ excites double electrons from the orbitals $r$ and $s$ to the orbitals $p$ and $q$, and $K$ is the total number of single and double excitation operators. These operators can be mapped using the Jordan-Wigner transformation to a ``qubit'' version  $ \{Q_p^\dagger Q_q,Q_p^\dagger Q_q^\dagger Q_r Q_s\}$ with  $Q_p =\frac{ \sigma_x^{(p)} + i\sigma_y^{(p)}}{2}$. Similarly to the analysis performed in Ref.~\cite{mao2023barren}, we consider all single excitations $Q_p^\dagger Q_q \; \forall p>q$ and all double excitations $Q_p^\dagger Q_q^\dagger Q_r Q_s \; \forall p>q>r>s$. 
Thus the total number of generators for the trotter and relaxed settings is $\nHam = KL $ with $K \sim n^4$ whereas the total number of distinct parameters is $\nparams_{\rm Trotter}= K$ in the trotter setting and $\nparams_{\rm Relaxed}= KL$ in its relaxed version.

For convenience, we further express these anti-Hermitian single and double excitations in the Pauli basis as follows
\begin{align}
   Q_p^\dagger Q_q := -i H_{pq} &= \frac{(\sigma_x^{(p)}-i\sigma_y^{(p)})(\sigma_x^{(q)}+i\sigma_y^{(q)})}{4} - h.c. \\
    &= -i \frac{\sigma_y^{(p)}\sigma_x^{(q)} - \sigma_x^{(p)}\sigma_y^{(q)}}{2} \label{eq;H_pq} \;, \\
    Q_p^\dagger Q_q^\dagger Q_r Q_s :=-i H_{pqrs} =& \frac{(\sigma_x^{(p)}-i\sigma_y^{(p)}) (\sigma_x^{(q)}-i\sigma_y^{(q)})(\sigma_x^{(r)}+i\sigma_y^{(r)}) (\sigma_x^{(s)}+i\sigma_y^{(s)})}{16} - h.c. \\
    =& -i \frac{\sigma_y^{(p)}\sigma_x^{(q)}\sigma_x^{(r)}\sigma_x^{(s)} + \sigma_x^{(p)}\sigma_y^{(q)}\sigma_x^{(r)}\sigma_x^{(s)} + \sigma_y^{(p)}\sigma_y^{(q)}\sigma_x^{(r)}\sigma_y^{(s)} + \sigma_y^{(p)}\sigma_y^{(q)}\sigma_y^{(r)}\sigma_x^{(s)}}{8}\\
    & +i \frac{\sigma_x^{(p)}\sigma_x^{(q)}\sigma_y^{(r)}\sigma_x^{(s)} + \sigma_x^{(p)}\sigma_x^{(q)}\sigma_x^{(r)}\sigma_y^{(s)} + \sigma_x^{(p)}\sigma_y^{(q)}\sigma_y^{(r)}\sigma_y^{(s)} + \sigma_y^{(p)}\sigma_x^{(q)}\sigma_y^{(r)}\sigma_y^{(s)}}{8}\;. \label{eq:H_pqrs}
\end{align}
Here we notice that the generators of the form $H_{pq} \;, p>q$ are 2-local  and those of the form $H_{pqrs} \;, p>q>r>s$ are 4-local, consisting each of a sum of a constant number of commuting Pauli terms.

Furthermore, we do precise the order of application of the parametrized unitaries  in the Relaxed and Trotter circuits (See Eq.~\eqref{eq:UCCSD_relax} and Eq.~\eqref{eq:UCCSD_Trotter}) by using the newly introduced Hamiltonians $H_{pq}$ and $H_{pqrs}$ in Eq.~\eqref{eq;H_pq} and Eq.~\eqref{eq:H_pqrs}. Similarly to the circuits adapted in Ref \cite{mao2023barren}, we first apply all the single excitations and then all the double excitations in each layer as follows 
\begin{align}
    U(\thv_{\rm Trotter}) &= \prod_{l=1}^L \left(\prod_{p>q} e^{-i\th_{(p,q)} H_{pq}} \prod_{p>q>r>s} e^{-i\th_{(p,q,r,s)} H_{pqrs}} \right) \;, \label{eq:trotter_ucc_ham}\\
      U(\thv_{\rm Relaxed}) &=\prod_{l=1}^L \left(\prod_{p>q} e^{-i\th_{(p,q),l} H_{pq}} \prod_{p>q>r>s} e^{-i\th_{(p,q,r,s),l} H_{pqrs}} \right) \label{eq:relax_ucc_ham}
\end{align}

For ease of notation, we introduce the sets $\Gamma_1$ and $\Gamma_2$ defined as
\begin{align}
    \Gamma_1 &:= \{(p,q)\;, p>q\} \label{eq:gamma_1}\\
    \Gamma_2 &:= \{(p,q,r,s)\;, p>q>r>s\} \;,\label{eq:gamma_2} 
\end{align}

where we use the subscript $\mu_1 \in \Gamma_1 $  to index the set of generators corresponding to first excitations acting non trivially on qubits $\mu_1 = p,q$ and the subscript  $\mu_2 \in \Gamma_2 $ to index the set of generators corresponding to double excitations acting non trivially on qubits $\mu_2 = p,q,r,s$.  For convenience, we also consider that the elements of the sets $\Gamma_1$ and $\Gamma_2$ are mapped to consecutive integers according to the position of their  associated gate in the circuit in Eq.~\eqref{eq:relax_ucc_ham}. Thus, the elements of $\Gamma_1$ goes from $1$ to $|\Gamma_1|$ and the elements of $\Gamma_2$ goes from $|\Gamma_1|+1$ to $K=|\Gamma_1|+|\Gamma_2|$. Note $|\Gamma_1| = \frac{1}{2} n (n+1)$, $|\Gamma_2| = \frac{1}{24} n (n+1) \left(n^2+9 n+26\right)$, thus $ K = \frac{1}{24} n (n+1) (n (n+9)+38)$.

\subsubsection{Analytical and numerical results}\label{subsubsec:UCC_results}

\paragraph*{\underline{UCCSD ansatz with an arbitrary  observable.}}\label{para:UCCSD_arbitrary_setting}
We consider two loss functions $\mathcal{L}(\thv_{\rm Trotter})$ and $\mathcal{L}(\thv_{\rm Relaxed})$ of the form in Eq.~\eqref{eq:loss} using the Trotter and the Relaxed version of the UCCSD ansatz introduced in Eq.~\eqref{eq:trotter_ucc_ham} and Eq.~\eqref{eq:relax_ucc_ham} respectively. We also consider an \textit{arbitrary} observable $O$ and an initial state $\rho$ such that the loss's  second derivative at zero w.r.t a subset of parameters $\Lambda \subset \Gamma_1 \cup \Gamma_2$ (See Eq.~\eqref{eq:gamma_1} and Eq.~\eqref{eq:gamma_2}) scale as a constant in the number of qubits. Formally, we have
\begin{align}
    &\left|\left.\left(\frac{\partial^2 \mathcal{L}(\thv_{\rm Relaxed})}{\partial \th_{\mu,l}^2}\right)\right|_{\thv_{\rm Relaxed} = \vec{0}}\right| \in \Theta(1) \;,\forall \mu \in \Lambda \text{ and } l \in \{1,\dots,L\}\;,\label{eq:cond_curv_relax_ucc}\\
    &\left|\left.\left(\frac{\partial^2 \mathcal{L}(\thv_{\rm Trotter})}{\partial \th_{\bar{\mu}}^2}\right)\right|_{\thv_{\rm Trotter} = \vec{0}}\right| \in \Theta(L^2) \;,\text{ for some } \bar{\mu} \in \Lambda \;.\label{eq:cond_curv_trotter_ucc}
\end{align}
Here we note that in the Relaxed setting,  the loss second derivative evaluated at zero with respect to some parameter $\th_{\mu,l}$   associated with the generator $H_\mu$ will be independent on $l \in \{1,\dots,L\}$, i.e. the second derivative with respect to the parameter $\theta_{\mu,l}$ will be the same as the derivative with respect to $\theta_{\mu,l'}$. This is because there are no parametrized layers and we are computing the derivative around $\thv = \vec{0}$. Therefore, all the gates evaluated at $\vec{0}$ are just identity. 

Under this setting and these conditions, we show that the size of the region with guaranteed substantial loss variance around zero, denoted by  $r_{\rm patch}^{\rm Relaxed}$ and $r_{\rm patch}^{\rm Trotter}$ in the relaxed and trotter version , scales as 
\begin{align}
    r_{\rm patch}^{\rm Relaxed} &\in \Theta\left( \frac{1}{\sqrt{KL} \norm{O}} \right) \label{eq:ucc_relax_patch} \;,\\
    r_{\rm patch}^{\rm Trotter} &\in \Theta\left(\frac{1}{L \sqrt{\norm{O}(K+ \norm{O})}}\right) \label{eq:ucc_trotter_patch} \;,
\end{align}
and for that region, the variance is lower bounded as
\begin{align}
   \Var_{\thv_{\rm Relaxed} \sim \uni(\vec{0},r_{\rm patch}^{\rm Relaxed})}[\mathcal{L}(\thv_{\rm Relaxed})] & \in \Omega\left(\frac{L |\Lambda|}{M^2 \norm{O}^4}\right)\;,   \label{eq:ucc_relax_var} \\
    \Var_{\thv_{\rm Trotter} \sim \uni(\vec{0},r_{\rm patch}^{Trotter})}[\mathcal{L}(\thv_{\rm Trotter})] &\in  \Omega\left(\frac{1}{\norm{O}^2 (K+ \norm{O})^2}\right) \;.\label{eq:ucc_trotter_var}
\end{align}
The proof is detailed in Appendix~\ref{subsubsec:UCC_proof}.

Crucially, we note that this theoretical guarantee of the substantial gradient patch is complementary to Ref.~\cite{mao2023barren} which shows the exponential concentration of the loss with respect to the number of electrons with random initialization over the entire landscape.

\bigskip

\paragraph*{\underline{Toy example.}} We further specify the scalings of the patch's size in Eq.~\eqref{eq:ucc_relax_patch} and Eq.~\eqref{eq:ucc_trotter_patch}, and of the loss variance in Eq.~\eqref{eq:ucc_relax_var} and Eq.~\eqref{eq:ucc_trotter_var}, by considering the setting where an observable is of the form  $O= \sum_{i=1}^{n}\sigma_z^{(i)}\otimes\sigma_z^{(i+1)}$ (with periodic boundary conditions) with an intial state $\ket{\psi} = \ket{1}^{\otimes \frac{n}{2}} \otimes \ket{0}^{\otimes \frac{n}{2}}$.%\stcom{the initial state $\rho = ...$}. 

Crucially, we verify the assumptions in Eq.~\eqref{eq:cond_curv_relax_ucc} and Eq.~\eqref{eq:cond_curv_trotter_ucc} under this specific setting which leads to 
the worst case the width of the patch around identity with guaranteed polynomial gradients scales as
\begin{align}
    r_{\rm patch}^{\rm Relaxed}&\in\Theta\left(\frac{1}{\sqrt{KL + n^2}}\right)\label{eq:r_uncor_ucc}\\
    r_{\rm patch}^{\rm Trotter}&\in\Theta\left(\frac{1}{L\sqrt{n(K+n)}}\right)\label{eq:r_corr_ucc}
\end{align}
and for that region, the variance is lower bounded as
\begin{align}\label{eq:var_ucc}
    \Var_{\thv_{\rm Relaxed} \sim \uni(\vec{0},r_{\rm patch}^{\rm Relaxed})}[\mathcal{L}(\thv_{\rm Relaxed})]  &\in \Omega\left(\frac{KL}{(KL+ n^2)^2}\right)\\
   \Var_{\thv_{\rm Trotter} \sim \uni(\vec{0},r_{\rm patch}^{Trotter})}[\mathcal{L}(\thv_{\rm Trotter})]&\in \Omega\left( \frac{1}{n^2(K+n)^2}\right) \;. 
\end{align}

Note that this results illustrate the effects of correlating parameters in time that we have highlighted in the previous sections. Indeed, here we see that the regions decrease when correlating parameters. 
Furthermore in our numerics, we use $L=2$ and $K\sim n^4$. With this, Table~\ref{table:ucc_results} presents specific analytical scalings together with numerical scalings obtained from Fig.~\ref{fig:scalings}.

\begin{table}[ht!]
\begin{tabular}{ll|l|l|}
\cline{3-4}
                                                  &                    & Relaxed & Trotter \\ \hline
\multicolumn{1}{|l|}{\multirow{2}{*}{Analytical}} & $r_{\rm patch}$    & $\Theta( M^{-0.5})$            & $\Theta( M^{-0.625})$         \\ \cline{2-4} 
\multicolumn{1}{|l|}{}                            & $\Var_{\rm patch}$ & $\Omega(M^{-1})$               & $\Omega(M^{-2.5})$             \\ \hline
\multicolumn{1}{|l|}{\multirow{2}{*}{Numerical}}  & $r_{\rm max}$      & $\Theta( M^{-0.55})$           & $\Theta( M^{-0.61})$          \\ \cline{2-4} 
\multicolumn{1}{|l|}{}                            & $\Var_{\rm max}$   & $\Theta( M^{-0.77})$            & $\Theta( M^{-0.93})$         \\ \hline
\end{tabular}
\caption{\textbf{Unitary Coupled Cluster (UCC):} Summary of the scalings of the analytical $r_{\rm patch},\, \Var_{\rm patch}$, defined as the region in which we can guarantee non-exponential vanishing variance, and the numerical results $r_{\rm max},\, \Var_{\rm max}$, defined in Eqs.~(\ref{eq:varmax},\ref{eq:rmax}) as the point $r$ such that the variance is maximized, and the value of such variance.}
\label{table:ucc_results}
\end{table}

Clearly, we see that the scaling of the analytical patch size aligns with the size of the maximal region with gradients, both showcasing that the Trotter version yields a smaller region with substantial gradients. This also aligns with the intuition given by the effective/maximal frequencies. Indeed, the Trotter (time-correlated) case requires an analysis based on maximal frequencies, and the relaxed version allows for effective frequencies. Moreover, we also see in that the $\Var_{\max}$ is also smaller in the Trotter case compared to the Relaxed case. This aligns with the lower-bounds obtained for $r_{\rm patch}$. However, one discrepancy which is on the variance scaling in the Trotter case could be explained as follows. We recall that in Theorem~\ref{th:var_formal_cor}, we only kept the variance contribution of a single parameter as opposed to Theorem~\ref{th:var_formal} which includes all parameters contribution. So, we could have suspected that an additional factor $M$ multiplying the variance would be closer to the true scaling of the variance.

\subsection{Proofs of analytical results for specific circuit families}
\subsubsection{Tensor product ansatz: Proof of analytical results}\label{subsubsec:proof_toy}

Following the general recipe detailed in the introduction of Appendix \ref{appendix:architectures},  we begin by identifying the type of correlations in the circuits described in Eq.~\eqref{eq:toy_uncorr_circuit} and Eq.~\eqref{eq:toy_cor_circuit} in order to characterize the patch with guaranteed gradients, i.e. $r_{\rm patch}$, centered around $\vphi = \vec{0}$. 
Precisely, we apply Theorem~\ref{th:var_formal} to analyze $\mathcal{L}(\thv_{\rm uncorrelated})$ and Theorem \ref{th:var_formal_cor} for $\mathcal{L}(\thv_{\rm correlated})$ exhibiting both spatial (between qubits) and temporal ($R_X$ and $R_Z$ rotations) correlations. To ease the notation, we will use $\thv_{\rm uncorr}, \thv_{\rm corr}$ instead of $\thv_{\rm uncorrelated}, \thv_{\rm correlated}$.

\paragraph*{\underline{Characterization of the region with guaranteed substantial gradients for $\mathcal{L}(\thv_{\rm uncorr})$.}}

Let us start by preparing the key ingredients necessary to adapt Theorem \ref{th:var_formal} to the loss function  $\mathcal{L}(\thv_{\rm uncorr})$ associated to the circuit in Eq.~\eqref{eq:toy_uncorr_circuit}. Namely, we compute the loss function second derivatives, as well as the associated maximal and effective frequencies at $\vphi=\vec{0}$.  

Since we are looking at the region with gradients around zero, the loss second derivatives with respect to a parameter $\th_{i,k} $ for $ k \in \{z,x_1,x_2\}$ denoted by $c_{i,k}(\vec{0})$ simplifies $\forall i,j \in \{1,\dots,n\}, \forall k \in \{z,x_1,x_2\}$ to

\begin{align}
    c_{i,k}(\vec{0}) &= \left|  \left.\left(\frac{\partial^2 \mathcal{L}(\thv_{\rm uncorr})}{\partial\th_{i,k}^2} \right)\right|_{\thv= \vec{0}} \right|\\
    &= \left| \Tr\left[\rho \left.\left(\frac{\partial^2 [e^{i\th_{i,k}\sigma_k^{(i)}}Oe^{-i\th_{i,k}\sigma_k^{(i)}}]}{\partial\th_{i,k}^2}\right)\right|_{\th_{i,k}=0}\right]  \right| \\
    &= \left| \Tr[\rho [\sigma_k^{(i)},[\sigma_k^{(i)}, O]]] \right|\;.
\end{align}
where all the parameters except $\th_{i,k}$ are set to zero (resulting in the associated unitaries being identities) in the second equality and  we use Eq.~\eqref{eq:p-order-grad-A} from Lemma \ref{lemma:bounded_unitary} in the last equality.

Hence, one can easily see that the loss function second derivatives with respect to parameters associated with a $\sigma_z$ generator are simply vanishing ,
\begin{align}
    c_{i,z}(\vec{0}) = \left|\Tr[\rho [\sigma_z^{(i)},[\sigma_z^{(i)},\sigma_z^{\otimes n}]]]\right| =0 \;.
\end{align}

On the other hand, the loss second derivatives with respect to the remaining parameters of the form $\th_{i,x}$ associated with generators $\sigma_x^{(i)}$ scale as constants in the number of qubits. Indeed, we show that
\begin{align}\label{eq:toy_curv}
    c_{i,x_l}(\vec{0}) &=  \left|\Tr[\rho [\sigma_{x}^{(i)},[\sigma_{x}^{(i)},\sigma_z^{\otimes n}]]]\right| = 4\;, l \in \{1,2\}\;,
\end{align}
where we used the fact that the double commutator gives $[\sigma_{x}^{(i)},[\sigma_{x}^{(i)},Z^{\otimes n}]] = 4 \sigma_z^{\otimes n}$ and that the initial state $\rho$ defined in Eq.~\eqref{eq:state_init_toy} is an eigenvector of $\sigma_z^{\otimes n}$ associated with eigenvalue $1$.

Consequently, we choose to apply Theorem \ref{th:var_formal} with the contribution of the parameters in the subset $\Lambda=\{(i,x_1),(i,x_2)\}_{i=1}^n$ as explained in the third step of  the general recipe. 

Next step form the recipe is to compute relevant frequencies. Here the effective frequencies denoted by $\omega^{(\rm eff)}_{(i,k)}(\vec{0})$ and $\widetilde{\omega}^{(\rm eff)}_{(i,k),(j,k')}(\vec{0})$ as defined in Eq.~\eqref{eq:eff_freqs_E} simplifies to 
\begin{align}
    \left[\omega^{(\rm eff)}_{(i,k)}(\vec{0})\right]^2&= \norm{\left.\left(\frac{\partial^2 [U(\thv_{\rm uncorr})^\dagger O U(\thv_{\rm uncorr})]}{\partial \th_{i,k}^2}\right)\right|_{\thv_{\rm uncorr= \vec{0}}} }\\
    &= \norm{\left.\left(\frac{\partial^2 [e^{i\th_{i,k}\sigma_k^{(i)}}Oe^{-i\th_{i,k}\sigma_k^{(i)}}]}{\partial \th_{i,k}^2}\right)\right|_{\th_{i,k}=0} }\\
    &= \norm{ [\sigma_k^{(i)},[\sigma_k^{(i)}, O]]} \label{eq:toy_eff_freqs} \;,\\
    \left[\widetilde{\omega}^{(\rm eff)}_{(i,k),(j,k')}(\vec{0})\right]^2 &= \norm{\left.\left(\frac{\partial^4 [U(\thv_{\rm uncorr})^\dagger O U(\thv_{\rm uncorr})]}{\partial \th_{i,k}^2 \partial \th_{j,k'}^2}\right)\right|_{\thv_{\rm uncorr= \vec{0}}} }\\
    &= \norm{\left.\left(\frac{\partial^2 [e^{i\th_{j,k'}\sigma_{k'}^{(j)}}e^{i\th_{i,k}\sigma_k^{(i)}}Oe^{-i\th_{i,k}\sigma_k^{(i)}}e^{-i\th_{j,k'}\sigma_{k'}^{(j)}}]}{\partial \th_{i,k}^2 \partial \th_{j,k'}^2}\right)\right|_{\th_{i,k},\th_{j,k'}=0} }\\
    &= \norm{ [\sigma_{k'}^{(j)},[\sigma_{k'}^{(j)},[\sigma_k^{(i)},[\sigma_k^{(i)}, O]]]]} \; , \label{eq:toy_eff_tilde_freqs}
\end{align}
where we used again Eq.~\eqref{eq:p-order-grad-A} from Lemma \ref{lemma:bounded_unitary} applied consecutively for the second derivatives with respect to $\th_{i,k}$ and $\th_{j,k'}$ in the last equality.

We then show that the effective frequencies are either zero or constant in the system size by simply computing the commutators between $\sigma_z$ and $\sigma_x$ as detailed in Eq.~\eqref{eq:toy_eff_freqs} and Eq.~\eqref{eq:toy_eff_tilde_freqs}. 
\begin{align}
    &\omega^{(\rm eff)}_{(i,z)}(\vec{0}) =0\\
    &\omega^{(\rm eff)}_{(i,x_l)}(\vec{0})= 2 \;, l \in \{1,2\}\label{eq:non_null_eff_toy}\\ 
    &\widetilde{\omega}^{(\rm eff)}_{(j,k'),(i,k)}(\vec{0})= \begin{cases}
        4 &\;, \forall  k,k' \in \{x_1,x_2\} \label{eq:non_null_eff_tilde_toy}\\
        0 &\;, \text{otherwise}.
    \end{cases}
\end{align}
Besides, the maximal frequencies, as defined in Eq.~\eqref{eq:max_freqs_E} associated with a $\sigma_z$ or $\sigma_x$ generators are simply   
\begin{align}
    \omega^{(\rm max)}_{(i,z)} &= \lambda_{\max}(\sigma_z^{(i)})-\lambda_{\min}(\sigma_z^{(i)})= 2 \;, \\
    \omega^{(\rm max)}_{(i,x_l)} &=\lambda_{\max}(\sigma_{x}^{(i)})-\lambda_{\min}(\sigma_{x}^{(i)})= 2\; , l \in \{1,2\}. \label{eq:max_toy}
\end{align}

Finally that we have gone through the first 4 steps of the general recipe, we proceed to substituting the second derivative, the maximal and effective frequencies in Eq.~\eqref{eq:rpatch_th3_recipe} and Eq.~\eqref{eq:Var_recipe_th3} to get the scaling of $r_{\rm patch}^{\rm uncorr}$ and the associated variance lower bound respectively. 
Specifically, we map the index $l$  in  Eq.~\eqref{eq:rpatch_th3_recipe} to the index $(i,x_l)$. Hence, the summation over the index $\mu$ in the term $\mathcal{A}_{(i,x_l)}(\vec{\phi})$ defined in Eq.~\eqref{eq:recipe_A} becomes a double summation over the $n$ qubits and the 2 layers indexed by $\nu \in \{1,2\}$ having $\sigma_{x_\nu}$ generators thus, enabling us to properly index the non zero effective frequencies computed in Eq.~\eqref{eq:non_null_eff_toy} and Eq.~\eqref{eq:non_null_eff_tilde_toy}. Precisely, we obtain 
\begin{align}
   \mathcal{A}_{(i,x_l)}(\vec{0}) &=  4 \left[\omega_{(i,x_l)}^{\rm (max)}\right]^2 \left(\sum_{\nu=1}^{l-1} \sum_{j=1}^{n} \left[\omega_{(j,x_\nu)}^{\rm (eff) }(\vec{0})\right]^2 +  \sum_{j=1}^{i-1} \left[\omega_{(j,x_l)}^{\rm (eff) }(\vec{0})\right]^2 \right) +  \sum_{j=i+1}^{n} \left[ \widetilde{\omega}^{(\rm eff)}_{(i,x_l),(j,x_l)}(\vec{0}) \right]^2 +  \sum_{\nu=l+1}^2 \sum_{j=1}^{n} \left[ \widetilde{\omega}^{(\rm eff)}_{(i,x_l),(j,x_\nu)}(\vec{0}) \right]^2\\
   &= \begin{cases}
         4 \left[\omega_{(i,x_1)}^{\rm (max)}\right]^2  \sum_{j=1}^{i-1} \left[\omega_{(j,x_1)}^{\rm (eff) }(\vec{0})\right]^2 + \sum_{j=i+1}^{n} \left[ \widetilde{\omega}^{(\rm eff)}_{(i,x_1),(j,x_1)}(\vec{0}) \right]^2+  \sum_{j=1}^{n} \left[ \widetilde{\omega}^{(\rm eff)}_{(i,x_1),(j,x_2)}(\vec{0}) \right]^2&\;, l=1 \\
         4 \left[\omega_{(i,x_2)}^{\rm (max)}\right]^2  \left(\sum_{j =1}^n \left[\omega_{(j,x_1)}^{\rm (eff) }(\vec{0})\right]^2 + \sum_{j =1}^{i-1} \left[\omega_{(j,x_2)}^{\rm (eff) }(\vec{0})\right]^2\right) + \sum_{j=i+1}^{n} \left[ \widetilde{\omega}^{(\rm eff)}_{(i,x_2),(j,x_2)}(\vec{0}) \right]^2  &\;, l=2 
   \end{cases}\\
   &= \begin{cases}
       16(2n+3i-4) &\;, l=1 \\
      16 (5n+3i-4) &\;, l=2 \\
   \end{cases} \label{eq:final_A_toy}
\end{align}
where we emphasize that the $i$ appearing at the end result refers to the index $i$, not to a complex valued result.

Consequently, by re-invoking Eq.~\eqref{eq:rpatch_th3_recipe} with the appropriate parameter indexing, we obtain the following scaling of $r_{\rm patch}^{\rm uncorr}$, i.e. the region with guaranteed substantial gradients centered around zero for the loss function $\mathcal{L}(\thv_{\rm uncorr})$
\begin{align}\label{eq:r_inter_toy}
     (r_{\rm patch}^{\rm uncorr})^2 = \min_{\substack{i\in \{1,\dots,n\}\\ l \in \{1,2\}}} \frac{9c_{i,x_l}(\vec{0})^2}{8 c_{i,x_l}(\vec{\phi}) \mathcal{A}_{(i,x_l)} + 24\beta_{(i,x_l)}}
\end{align}

We can easily check from Eq.~\eqref{eq:final_A_toy} that the index pair $(i,x_l)$ that will maximize $\mathcal{A}_{(i,x_l)}$ is $i = n,\, l = 2$. We can also recall from Eq.~\eqref{eq:toy_curv} that $c_{i,x_l}$ is independent on this indexes. Finally, we can readily check that $\beta_{(i,x_l)}$ is also maximized for $i = n,\, l=2$. Let us recall Eq.~\eqref{eq:beta_recipe}. We have shown that in this example that $\omega^{(\max)}_{(i,x_l)} = \omega^{(\rm eff)}_{(i,x_l)}$ (see Eqs.~(\ref{eq:max_toy},\ref{eq:non_null_eff_toy})), and because $\|\sigma_z^{\otimes n}\|_{\infty} = 1$, then we have the following 
    \begin{align}\label{eq:beta_toy}
            \beta_{(i,x_l)}(\vec{0}) =\begin{cases}
        \frac{2 \left[\omega^{(\rm max)}_{(1,x_1)}\right]^2  \left[\omega^{\rm (eff)}_{(1,x_1)}(0)\right]^4}{3} \;  {\rm if }\; i=1,\,l=1 \;,\\
       {\tiny \,}\\
       \frac{32\left[ \omega_{(i,x_l)}^{\rm (max) }\right]^{6}  \norm{O}^2}{3}\; {\rm otherwise}\;.
    \end{cases}  
    \Rightarrow \;\;\;\;\;
     \beta_{(i,x_l)}(\vec{0}) =\begin{cases}
     \frac{2^7}{3} \;  {\rm if }\; i=1,\,l=1\;,\\
       \;\\
       \frac{2^{11}}{3}\;  {\rm otherwise} \;.
      \end{cases}  
\end{align}

Thus, it is clear that the indices that maximize this quantity are $i = n,\, l = 2$. Hence, by plugging Eq.~\eqref{eq:final_A_toy} and the larger version of Eq.~\eqref{eq:beta_toy} in Eq.~\eqref{eq:r_inter_toy}, we finally get
\begin{align}
    (r_{\rm patch}^{\rm uncorr})^2 &= \frac{9}{2^7(2n+7)} \in \Theta\left(\frac{1}{n}\right) =  \Theta\left(\frac{1}{M}\right) \;.
\end{align}
where we recall that $M = 3n$

 Moreover, we have from Eq.~\eqref{eq:Var_recipe_th3} that the variance of the loss function at $r_{\rm patch}^{\rm uncorr}$ centered around zero is lower bounded as 
\begin{align}
    \Var_{\thv \sim \uni\left(\vec{0},r_{\rm patch}^{\rm uncorr}\right)}[\mathcal{L}(\vtheta_{\rm uncorr})] &\in \Omega\left( \sum_{l=1}^2 \sum_{i=1}^n c_{i,x_l}(\vec{0})^2(r_{\rm patch}^{\rm uncorr})^4\right)= \Omega\left( \frac{1}{n} \right) =  \Omega\left( \frac{1}{M} \right) 
\end{align}
where we recall that $M=3n$.

\paragraph*{\underline{Characterization of the region with guaranteed substantial gradients for $\mathcal{L}(\thv_{\rm corr})$.}}

We repeat the same steps for the fully correlated loss function $\mathcal{L}(\thv_{\rm corr})$ associated with the circuit in Eq.~\eqref{eq:toy_cor_circuit}. As we mentioned in the beginning of the proof, we will be using Theorem \ref{th:var_formal_cor} to characterize the region with substantial gradients centered around $\vphi = \vec{0}$ since the circuit exhibits temporal correlation (two non commuting gates on the same qubit sharing the same parameter). We also recall that in this setting we have $M =3n$ generators but $m_{\rm corr}=1$ parameter shared across all of the circuit generators.
Hence, the loss function depends on a single parameter $\th$ and adapting Theorem \ref{th:var_formal} boils down to computing $c(\vec{0})$, i.e. the second derivative of the loss $\mathcal{L}(\th)$ and the maximal frequency $\omega^{(\rm max)}$ associated to that single parameter.

We start by computing compute $c(\vec{0})$. To do so we use Eq.~\eqref{eq:pth_deriv_corr} from Lemma \ref{lemma:bouded_unitary_product} and the definitions $H_1 = \sum_{i=1}^{n}\sigma_z^{(i)},\, H_2 = \sum_{i=1}^{n}\sigma_x^{(i)}$ to obtain
\begin{align}
    c(\vec{0}) &= \left| \Tr\left[\rho\left.\left(\frac{d^2 [U^\dagger(\th) O U(\th)]}{d\th^2} \right)\right|_{\th =0}\right]\right|\\
    &= \left|\Tr\left[\rho \left( [H_1,[H_1,O]] + 2 [H_2,[H_2,O]] + 2 [H_1,[H_2,O]] + 2 [H_2,[H_1,O]] + 2 [H_2,[H_2,O]]\right) \right]\right|\\
    &= 4 \left|\Tr\left[\rho \ [H_2,[H_2,O]] \right]\right| \;.\label{eq:c0_toy}
\end{align}
where we used that the observable is $O = \sigma_z^{\otimes n}$, in order to obtain $[H_1,O] = 0$. Furthermore, we used that $\Tr[\rho[H_1,[H_2,O]]] = 0 $. This can be shown as follows: first we use the well known identity $[A,[B,C]] + [C,[A,B]] + [B,[C,A]] = 0$ as follows
\begin{align}
    [H_1,[H_2,O]]] + [O,[H_1,H_2]]]+ [H_2,[O,H_1]]] = 0 \Rightarrow  [H_1,[H_2,O]]] = - [O,[H_1,H_2]]]
\end{align}
where we used that $[H_1,O] = 0$. With this, we can use the cyclicity of the trace, particularly the fact that $\Tr[A,[B,C]] = \Tr[[A,B]C]$, to show
\begin{align}\label{eq:toy_trace_equal_zero}
    \Tr[\rho[H_1,[H_2,O]]] =& - \Tr[\rho[O,[H_1,H_2]]]\\
    =& - \Tr[[\rho,O]\cdot[H_1,H_2]]] = 0
\end{align}
where we used that $[\rho,O]=0$, as $\rho$ is an eigenstate of $O$. Indeed, recall from Eq.~\eqref{eq:state_init_toy} that $\rho = \ketbra{\psi},\, \ket{\psi} = \frac{1}{\sqrt{2}}(\ket{+}^{\otimes n}+ \ket{-}^{\otimes n})$. And because $\sigma_z\ket{+} = \ket{-},\, \sigma_z\ket{-} = \ket{+}$, it is trivial that $\sigma_z^{\otimes n }\ket{\psi} = \ket{\psi}$.

\medskip

After all these simplifications, we can recover Eq~\eqref{eq:c0_toy} and are now ready to evaluate the remaining term. To evaluate this term, we will show $[H_2,[H_2,\sigma_z^{\otimes n}]] = 4H_2^2 \sigma_z^{\otimes n}$, where we will write $O=\sigma_z^{\otimes n}$ explicitly. To prove this we start with the trivial observation that $\sigma_x^{i}\sigma_z^{\otimes n} = -\sigma_z^{\otimes n}\sigma_x^{i}$ as the Pauli matrices anti-commute. With this observation we can show that 
\begin{align}
    [H_2,\sigma_z^{\otimes n}] =& \sum_{i=1}^n[\sigma_x^{(i)},\sigma_z^{\otimes n}]\\
    =&\sum_{i=1}^n(\sigma_x^{(i)}\sigma_z^{\otimes n} - \sigma_z^{\otimes n}\sigma_x^{(i)})\\
    =&2\sum_{i=1}^n \sigma_x^{(i)}\sigma_z^{\otimes n} = 2 H_2\sigma_z^{\otimes n}
\end{align}
where in the first equality we expanded $H_2 = \sum_{i=1}^n \sigma_x^{(i)}$ and used the linearity of the commutator with the sum, to take the $\sum_{i=1}^n$ outside of the commutator. In the second equality we expanded the commutator, and in the third inequality we used that $\sigma_x^{(i)},\sigma_z^{\otimes n}$ anti-commute as explained before. With this, we can trivially show that $[H_2,[H_2,\sigma_z^{\otimes n}]] = 4H_2^2\sigma_z^{\otimes n}$.

With this result we can recover Eq.~\eqref{eq:c0_toy} to simplify the final term. Indeed, now we can write
\begin{align}
     c(\vec{0}) =& 16 \left|\Tr\left[\rho \ H_2^2 \sigma_z^{\otimes n} \right]\right| \\
     =& 16 \left|\Tr\left[\rho \ H_2^2  \right]\right|\\
     \label{eq:second_deriv_toy_cor}
     =& 16 n^2
\end{align}
where in the second inequality we used that the trace is cyclic and that $\sigma_z^{\otimes n}\rho = \rho$ as shown right after Eq.~\eqref{eq:toy_trace_equal_zero}. In the last equality we used that $\Tr[H_2^2\rho]=n^2$. Indeed, this can be shown by recovering the initial state $\rho$ in Eq.~\eqref{eq:state_init_toy}, then 
\begin{align}
    H_2^2\frac{1}{\sqrt{2}}\left(\ket{+}^{\otimes n} + \ket{-}^{\otimes n} \right) = 
 H_2 \, n\frac{1}{\sqrt{2}}\left(\ket{+}^{\otimes n} - \ket{-}^{\otimes n} \right) = n^2\frac{1}{\sqrt{2}}\left(\ket{+}^{\otimes n} + \ket{-}^{\otimes n} \right)
\end{align}
where we used that $\sigma_x\ket{\pm} = \pm\ket{\pm}$.

Furthermore, the largest and lowest eigenvalues of both $H_1$ and $H_2$ are respectively $n$  and $-n$ (sum of commuting Pauli matrices), so we have
\begin{align}\label{eq:omega_max_toy_cor}
\omega^{\rm (max) } &=\omega^{(\rm max)}(H_1)+2\omega^{(\rm max)}(H_2)= 6n\;. 
\end{align}
where we used that we have 2 sets of $R_X$ and one set of $R_Z$.

Thus, by plugging in Eq.~\eqref{eq:second_deriv_toy_cor} and Eq.~\eqref{eq:omega_max_toy_cor} in Eq.~\eqref{eq:rpatch_th4_recipe} (and using Eqs.~(\ref{eq:gamma_recipe},\ref{eq:beta_tilde_recipe})), we get that within the patch defined by $r_{\rm patch}^{\rm corr}$ and centered at zero, 
\begin{align}
    (r_{\rm patch}^{\rm corr})^2 &= \frac{9 c_0^2(\vec{0})}{256 \norm{O}^2 \left[\omega^{\rm (max) }\right]^6 } = \frac{1 }{5184 n^2}\in\Theta\left(\frac{1}{M^2}\right) \;,
\end{align}
where we recall that $M = 3n$. 

The variance of the loss function is lower bounded as 
\begin{align}
    \Var_{\thv \sim \uni\left(\vec{0},r_{\rm patch}^{\rm corr}\right)}[\mathcal{L}(\vtheta_{\rm corr})] &\in \Omega\left(\frac{c(\vec{0})^2}{M^4}\right) \in \Omega(1) \;.
\end{align}

\subsubsection{Hardware Efficient Ansatz: Proof of analytical results}\label{subsubsec:HEA_proof}

We consider the setting where all the parameters in the circuit defined in Eq.~\eqref{eq:HEA_circuit} are independent. Hence, as mentioned in the first step of the general recipe (see the beginning of this Appendix \ref{appendix:architectures}), we apply Theorem~\ref{th:var_formal} to characterize the region with gradients centered around $\vphi=\vec{0}$ for each of the loss functions $\mathcal{L}_G(\thv)= \Tr[| 0 \rangle\langle 0 |^{\otimes n} U^\dagger(\thv) O_G U(\thv)]$ and $\mathcal{L}_L(\thv) = \Tr[| 0 \rangle\langle 0 |^{\otimes n} U^\dagger(\thv) O_L U(\thv)]$ associated to the observables $O_G$ and $O_L$ introduced in Eq.~\eqref{eq:obs_G} and Eq.~\eqref{eq:obs_L} respectively.

We now proceed to the steps $2,3 \text{ and } 4$ of the general recipe mainly consisting of computing the second derivatives of the loss functions w.r.t all the parameters, as well as the maximal and effective frequencies. Specifically, we develop these steps for a loss function of the form $\mathcal{L}(\thv)= \Tr[|0\rangle\langle0|^{\otimes n} U^\dagger(\thv)OU(\thv)]$ using the HEA ansatz $U(\thv)$ described in Eq.~\eqref{eq:HEA_circuit} for any observable $O$. Then, we tailor the final expressions to the loss functions of interest $\mathcal{L}_G(\thv)$ and $\mathcal{L}_L(\thv)$ described above.

Since we are looking at the region with gradients around zero, the computations of the second derivatives and the effective frequencies boils down to evaluating simple doubly nested commutators. Specifically, we can express $c_{l,k_i}(\vec{0})$, i.e. the second derivative of the loss function $\mathcal{L}(\thv)$ w.r.t the parameter $\th_{l,k_i}$ evaluated at zero for $ k \in\{y,z\}$ as 
\begin{align}
    c_{l,k_i}(\vec{0}) &= \left| \Tr\left[|0\rangle\langle0|^{\otimes n} \left.\left(\frac{\partial^2 [U^\dagger(\thv)OU(\thv)]}{\partial\th_{l,k_i}^2}\right)\right|_{\thv=\vec{0}}\right] \right|\\
    &= \left| \Tr\left[|0\rangle\langle0|^{\otimes n} V_L^\dagger \dots V^\dagger_{l+1} \left.\left(\frac{\partial^2 [e^{i \th_{l,k_i}\sigma_{k_i}^{(i)}}( V^\dagger_l \dots V^\dagger_1O V_1 \dots V_l )e^{-i \th_{l,k_i} \sigma_{k_i}^{(i)}}]}{\partial\th_{l,k_i}^2}\right)\right|_{\th_{l,k_i}=0} V_{l+1} \dots V_L\right] \right|\\
    &= \left| \Tr\left[|0\rangle\langle0|^{\otimes n} V_L^\dagger \dots V^\dagger_{l+1} [\sigma_{k_i}^{(i)},[\sigma_{k_i}^{(i)},V^\dagger_l \dots V^\dagger_1O V_1 \dots V_l]]  V_{l+1} \dots V_L\right] \right|\\
    &= \left| \Tr\left[|0\rangle\langle0|^{\otimes n}  [\sigma_{k_i}^{(i)},[\sigma_{k_i}^{(i)},V^\dagger_l \dots V^\dagger_1O V_1 \dots V_l]] \right] \right| \label{eq:curv_HEA}
\end{align}
where we used Eq.~\eqref{eq:p-order-grad-A} from Lemma \ref{lemma:bounded_unitary} in the third equality and in the last equality we used the fact that $V_i = CZ(i,i+1)$ acts trivially on the initial state $|0\rangle\langle0|^{\otimes n}$.

Similarly, we can use the same arguments to express the effective frequencies $\omega^{(\rm eff)}_{(l,k_i)}(\vec{0})$ associated to the parameter $\th_{l,k_i}$ and defined in Eq.~\eqref{eq:eff_freqs_E}  as 
\begin{align}
    \left[\omega^{(\rm eff)}_{(l,k_i)}(\vec{0})\right]^2 &= \norm{\left.\frac{\partial^2 [U(\thv)^\dagger O U(\thv)]}{\partial \th_{l,k_i}^2}\right|_{\thv = \vec{0}}} \\
    &= \norm{ V_L^\dagger \dots V^\dagger_{l+1} [\sigma_{k_i}^{(i)},[\sigma_{k_i}^{(i)},V^\dagger_l \dots V^\dagger_1O V_1 \dots V_l]]  V_{l+1} \dots V_L}\\
    &= \norm{[\sigma_{k_i}^{(i)},[\sigma_{k_i}^{(i)},V^\dagger_l \dots V^\dagger_1O V_1 \dots V_l]]} \label{eq:eff_HEA}
\end{align}
where we used that the infinity norm is invariant under unitary operations, i.e. $\|UAU^\dagger\|_{\infty} = \|A\|_{\infty}$.

Moreover, the effective frequencies $\widetilde{\omega}^{(\rm eff)}_{(l,k_i),(l',k'_j)}(\vec{0})$ associated to the parameters $\th_{(l,k_i)}$ and $\th_{(l',k'_j)}$ such that $\th_{(l,k_i)}$ is closer to the observable than $\th_{(l',k'_j)}$, as  introduced in Eq.~\eqref{eq:eff_freqs_tilde_E}, can be expressed as 
\begin{align}
    \left[\widetilde{\omega}^{(\rm eff)}_{(l,k_i),(l',k'_j)}(\vec{0})\right]^2 &= \norm{ \left.\left(\frac{\partial^4 [U(\thv)^\dagger O U(\thv)]}{ \partial \th_{(l',k'_j)}^2 \partial \th_{(l,k_i)}^2}\right)\right|_{\thv= \vec{0}}}\\
    &=  \norm{ \left.\frac{\partial^4 [e^{i \th_{l',k'_j} \sigma^{(j)}_{k'_j}} V_{l'}^\dagger \dots V^\dagger_{l+1} e^{i \th_{l,k_i} \sigma^{(i)}_{k_i}} V_l^\dagger \dots V_1^\dagger O V_1 \dots V_l e^{-i \th_{l,k_i} \sigma^{(i)}_{k_i}} V_{l+1} \dots V_{l'} e^{-i \th_{l',k'_j} \sigma^{(j)}_{k'_j}}]}{ \partial \th_{(l',k'_j)}^2 \partial \th_{(l,k_i)}^2}\right|_{\th_{l,k_i},\th_{l',k'_k}=0 }}\\
    &= \norm{[\sigma^{(j)}_{k'_j},[\sigma^{(j)}_{k'_j},V_{l'}^\dagger \dots V^\dagger_{l+1}[\sigma^{(i)}_{k_i},[\sigma^{(i)}_{k_i}, V_l^\dagger \dots V_1^\dagger O V_1 \dots V_l]]V_{l+1} \dots V_{l'}]]} \label{eq:eff_tilde_HEA}
\end{align}
where we used Eq.~\eqref{eq:p-order-grad-A} from Lemma \ref{lemma:bounded_unitary} twice in the last equality.

Finally, since all of our generators are Pauli strings,  the maximal frequencies defined in Eq.~\eqref{eq:max_freqs_E} are simply all ones.
\begin{align}\label{eq:max_HEA}
    \omega^{(\rm max)}_{(l,k_i)} = 2 \quad,\; \forall 1\leq l \leq L, 1\leq i \leq n \text{ and } k \in \{y,z\} \;.
\end{align}

Let us now further develop Eq.~\eqref{eq:curv_HEA}, Eq.~\eqref{eq:eff_HEA} and Eq.~\eqref{eq:eff_tilde_HEA} for the observables $O_G$ and $O_L$ in order to finally obtain the expression of $r_{\rm patch}^{\rm G}$ and $r_{\rm patch}^{\rm L}$, i.e. the regions with guaranteed substantial gradients around zero for the loss functions $\mathcal{L}_G(\thv)$   and $\mathcal{L}_L(\thv)$ respectively according to Eq.~\eqref{eq:rpatch_th3_recipe}. From Eq.~\eqref{eq:Var_recipe_th3}, we also get the associated variance lower bounds.

\paragraph*{\underline{Characterization of the region with gradients $r_{\rm patch}^{G}$ centered around zero for the loss function $\mathcal{L}_G(\thv)$.}}
     Here, we consider the global observable $O_G = Z^{\otimes n}$. Thus, any $V_i:= CZ(i,i+1)$ acts trivially on the observable $O_G$.
     Hence, the second derivative $c_{l,k_i}(\vec{0})$ in Eq.~\eqref{eq:curv_HEA} and the effective frequencies in Eq.~\eqref{eq:eff_HEA} and Eq.~\eqref{eq:eff_tilde_HEA} simplifies to
     \begin{align}
         c_{l,k_i}(\vec{0})
         &= \left| \Tr\left[|0\rangle\langle0|^{\otimes n}  [\sigma_{k_i}^{(i)},[\sigma_{k_i}^{(i)}, \sigma_z^{\otimes n}]] \right] \right| \label{eq:curv_global}\\
         \omega^{(\rm eff)}_{(l,k_i)}(\vec{0}) &= \sqrt{\norm{[\sigma_{k_i}^{(i)},[\sigma_{k_i}^{(i)},\sigma_z^{\otimes n}]]}} \label{eq:eff_global}\\
         \widetilde{\omega}^{(\rm eff)}_{(l,k_i),(l',k'_j)}(\vec{0}) &=\sqrt{\norm{[\sigma^{(j)}_{k'_j},[\sigma^{(j)}_{k'_j},V_{l'}^\dagger \dots V^\dagger_{l+1}[\sigma^{(i)}_{k_i},[\sigma^{(i)}_{k_i}, \sigma_z^{\otimes n}]]V_{l+1} \dots V_{l'}]]}}\label{eq:eff_tilde_global}
     \end{align}
     where we used the fact that $\sigma_z^{\otimes n}$ commutes with $V_l := \prod_i CZ(i,i+1) \;, \forall 1\leq i \leq n$.

     Clearly, one can see that $ c_{l,z_i}(\vec{0}) = 0$ given that the commutator between the generators and the observable is zero. The curvature with respect to the $R_Y$ (deppending on $\th_{l,y_i}$) gates is non zero, but are constants in the system size. We can trivially show this by using that the Pauli matrices anticommute, i.e. $\sigma_y^{(i)}\sigma_z^{(i)} =-\sigma_z^{(i)}\sigma_y^{(i)}$, and with this we find
     \begin{align}\label{eq:computation_c0y_hea}
         c_{l,y_i}(\vec{0}) &= \left| \Tr\left[|0\rangle\langle0|^{\otimes n}  [\sigma_{y}^{(i)},[\sigma_{y}^{(i)}, \sigma_z^{\otimes n}]] \right] \right|\\
         &= 2  \left| \Tr\left[|0\rangle\langle0|^{\otimes n}   \sigma_y^{(i)}[\sigma_y^{(i)},\sigma_z^{\otimes n}] \right] \right|\\
         &= 4 \;.
     \end{align}
     where in the second equality we used the anti-commutation rules for the first time and took all the constants outside of the absolute value. In the last inequality we used the same approach to simplify the second commutator and $\sigma_y^2 = \1$ to further simplify the result.  Finally, we use that $\sigma_z\ket{0}=\ket{0}$.

For the effective frequencies in Eq.~\eqref{eq:eff_global}  and Eq.~\eqref{eq:eff_tilde_global}, we also obtain a similar scaling. We start by computing Eq.~\eqref{eq:eff_global}
\begin{align}
    \left[\omega^{(\rm eff)}_{(l,k_i)}(\vec{0})\right]^2 &= \norm{[\sigma_{k_i}^{(i)},[\sigma_{k_i}^{(i)},\sigma_z^{\otimes n}]]} = 
    \begin{cases}
        4 \, {\rm if} \, k=y\\
        0 \, {\rm if} \, k=z\\
    \end{cases}
\end{align}
where we used the anti-commutation rules of the Pauli matrices, similar to Eq.~\eqref{eq:computation_c0y_hea}.

Similarly, we compute $\widetilde{\omega}^{(\rm eff)}_{(l,k_i),(l',k'_i)}(\vec{0})$. We can start by noting that $k_i = y_i$ as otherwise the commutator will be zero. 
\begin{align}
    \left[\widetilde{\omega}^{(\rm eff)}_{(l,y_i),(l',k'_j)}(\vec{0})\right]^2 &= \| [\sigma^{(j)}_{k'_j},[\sigma^{(j)}_{k'_j},V_{l'}^\dagger \dots V^\dagger_{l+1}[\sigma^{(i)}_{y_i},[\sigma^{(i)}_{y_i}, \sigma_z^{\otimes n}]]V_{l+1} \dots V_{l'}]] \|_{\infty}\\
    &= 4 \| [\sigma^{(j)}_{k'_j},[\sigma^{(j)}_{k'_j}, \sigma_z^{\otimes n}]] \|_{\infty}
\end{align}
where in the first equality we imposed the condition on $k_i$ such that the commutator is not zero. In the second equality we used the anti-commutation rules to simplify one commutator. We can finally simplify this further by noting that, again, we need $k'_j=y_j$ in order for the commutator not to vanish and thus we can obtain
\begin{align}
    \left[\widetilde{\omega}^{(\rm eff)}_{(l,y_i),(l',y_j)}(\vec{0})\right]^2 =4 \| [\sigma^{(j)}_{y'_j},[\sigma^{(j)}_{k'_j}, \sigma_z^{\otimes n}]] \|_{\infty} = 16
\end{align}

All in all, can be summarized as follows
\begin{align}
    \omega^{(\rm eff)}_{(l,k_i)} &= \begin{cases}
        0 &\; \text{if } k=z\;,\\
        2  &\; \text{if } k=y\;.
    \end{cases}\label{eq:eff_freqs_values_global1} \\
    \widetilde{\omega}^{(\rm eff)}_{(l,k_i),(l',k'_j)} 
    &= \begin{cases}\label{eq:eff_freqs_values_global2}
        4 &\;   \text{if } k=k'= y\;,\\
        0  &\; \text{otherwise} \;.
    \end{cases}
\end{align}
Consequently, we apply Theorem \ref{th:var_formal} with the subset of indices $\Lambda = \{(l,y_i)\}_{l=1,i=1}^{L,n}$ as detailed in the third step of the general recipe.

Now, we can re-invoke the expression of $r_{\rm patch}^{\rm G}$ from Eq.~\eqref{eq:rpatch_th3_recipe} with the adequate parameter indices.
\begin{align}\label{eq:global_patch_HEA}
     (r_{\rm patch}^{\rm G})^2 = \min_{\substack{l\in \{1,\dots,L\}\\i\in \{1,\dots,n\}}} \frac{9c_{l,y_i}(\vec{0})^2}{8 c_{l,y_i}(\vec{0}) \mathcal{A}_{(l,y_i)}(\vec{0}) + 24\beta_{(l,y_i)}}
\end{align}

We have concluded in Eq~\eqref{eq:computation_c0y_hea} that $c_{l,y_i}(\vec{0})$ is independent of the indices. We then recover Eqs.~(\ref{eq:recipe_A},\ref{eq:beta_recipe}) to see which one will minimize Eq.~\eqref{eq:global_patch_HEA}. We start with the term $\mathcal{A}_{(l,y_i)}(\vec{0})$ defined in Eq.~\eqref{eq:recipe_A} with the appropriate parameter indices of the form $(l,y_i)$ as follows. To compute this quantity we need to recall the values of the different frequencies in Eqs.~(\ref{eq:max_HEA},\ref{eq:eff_freqs_values_global1},\ref{eq:eff_freqs_values_global2})
\begin{align}\label{eq:denom_HEA}
    \AC_{(l,y_i)}(\vec{0}) &= 4 \left[\omega_{(l,y_i)}^{\rm (max)}\right]^2 \left(\sum_{\nu=1}^{l-1} \sum_{j=1}^{n} \left[\omega_{(\nu,y_j)}^{\rm (eff) }(\vec{0})\right]^2 +  \sum_{j=1}^{i-1} \left[\omega_{(l,y_j)}^{\rm (eff) }(\vec{0})\right]^2 \right) +  \sum_{j=i+1}^{n} \left[ \widetilde{\omega}^{(\rm eff)}_{(l,y_i),(l,y_j)}(\vec{0}) \right]^2 +  \sum_{\nu=l+1}^L \sum_{j=1}^{n} \left[ \widetilde{\omega}^{(\rm eff)}_{(l,y_i),(\nu,y_j)}(\vec{0}) \right]^2\\
    &= 16\left[n(L+3l-3) +3i -4\right]\;, 
\end{align}
where we only keep the contribution of the non zero effective frequencies in the first equality and use the exact values of the effective and maximal frequencies in the final equality. We see that this is maximized for $i=n, l=L$, which gives
\begin{equation}\label{eq:denom_HEA_final}
     \AC_{(L,y_n)}(\vec{0}) = 64\left[nL -1\right] \, .
\end{equation}

Similarly to in the case of the Tensor product ansatz, specifically in Eq.~\eqref{eq:beta_toy}, we see that the term in $\beta_{(l,y_i)}$ that will be largest is when $l>1, i>1$. We use Eq.~\eqref{eq:max_HEA} i.e. the following
\begin{align}\label{eq:beta_HEA}
    \beta_{(l,y_i)}  =\frac{32\left[  \omega^{(\rm max)}_{(l,y_i)}\right]^{6}  \norm{\sigma_z^{\otimes n}}^2}{3} = \frac{2^{11}}{3}
\end{align}

Consequently, we can take $l = L, i = n$ as the indices that minimize Eq.~\eqref{eq:global_patch_HEA}. Then by substituting Eq.~\eqref{eq:denom_HEA_final} and Eq.~\eqref{eq:beta_HEA} in Eq.~\eqref{eq:global_patch_HEA}, we obtain
\begin{align}
     (r_{\rm patch}^{\rm G})^2 &=  \frac{9}{128[nL+7]} \in \Theta\left(\frac{1}{M}\right)
\end{align}
where we recall that $M =  2 L n$

From Eq.~\eqref{eq:Var_recipe_th3}, we can also get the associated variance lower bound (for $\Lambda=\{(l,y_i)\}_{l=1,i=1}^{L,n}$),
\begin{align}
    \Var_{\thv \sim \uni(\vec{0},r_{\rm patch}^{\rm G})}[\mathcal{L}(\thv)] &\geq \frac{1}{72} \sum_{l=1}^L \sum_{i=1}^n c^2_{l,y_i}(\vec{0}) (r_{\rm patch}^{\rm G})^4\\
    &\in \Omega\left(\frac{1}{M}\right)\;.
\end{align}

\paragraph*{\underline{Characterization of the region with gradients $r_{\rm patch}^{L}$ centered around zero for the loss function $\mathcal{L}_L(\thv)$.}}
Here we consider the local observable $O_L:= \sigma_z \otimes  \sigma_z \otimes \1^{\otimes(n-2)}$. Similarly, to the setting with the global observable $O_G$, we have that any unitary $V_i:= CZ(i,i+1)$ acts trivially on the observable $O_L$.
Hence, the expressions of the second derivatives and the effective frequencies in Eq.~\eqref{eq:curv_global}, Eq.~\eqref{eq:eff_global} and Eq.~\eqref{eq:eff_tilde_global} carry on to the local observable case. Precisely, we have
 \begin{align}
         c_{l,k_i}(\vec{0})
         &= \left| \Tr\left[|0\rangle\langle0|^{\otimes n}  [\sigma_{k_i}^{(i)},[\sigma_{k_i}^{(i)}, \sigma_z \otimes  \sigma_z \otimes \1^{\otimes(n-2)}]] \right] \right| \label{eq:c_hea_local}\\
         \omega^{(\rm eff)}_{(l,k_i)}(\vec{0}) &= \sqrt{\norm{[\sigma_{k_i}^{(i)},[\sigma_{k_i}^{(i)},\sigma_z \otimes  \sigma_z \otimes \1^{\otimes(n-2)}]]}} \label{eq:w_hea_local}\\
         \widetilde{\omega}^{(\rm eff)}_{(l,k_i),(l',k'_i)}(\vec{0}) &=\sqrt{\norm{[\sigma^{(j)}_{k'_j},[\sigma^{(j)}_{k'_j},V_{l'}^\dagger \dots V^\dagger_{l+1}[\sigma^{(i)}_{k_i},[\sigma^{(i)}_{k_i}, \sigma_z \otimes  \sigma_z \otimes \1^{\otimes(n-2)}]]V_{l+1} \dots V_{l'}]]}}\label{eq:tildew_hea_local}
     \end{align}

Here, we can see that the only non zero second derivatives and effective frequencies are the ones corresponding to $k=y$ as in the case of the global observable (recall Eqs~(\ref{eq:computation_c0y_hea},\ref{eq:eff_freqs_values_global1},\ref{eq:eff_freqs_values_global2})) and $i \in \{1,2\}$ and for these parameter indices, we get $\forall 1 \leq l \leq L$ and  $i,j \in \{1,2\}$, as otherwise the commutators will be between the generators and identity (and thus automatically zero)
\begin{align}
         c_{l,y_i}(\vec{0})
         &= 4 \\
         \omega^{(\rm eff)}_{(l,y_i)}(\vec{0}) &= 2 \\
         \widetilde{\omega}^{(\rm eff)}_{(l,y_i),(l',y_j)}(\vec{0}) &=4 \;.
\end{align}
indeed, even in the case of $\widetilde{\omega}^{(\rm eff)}_{(l,y_i),(l',y_j)}$, both $i,j$ need to be either 2 or 1.

Consequently, we apply Theorem \ref{th:var_formal} with the subset of indices $\Lambda=\{(l,y_i)\}_{l=1,i\in\{1,2\}}^L$ and obtain from Eq.~\eqref{eq:rpatch_th3_recipe} that $r_{\rm patch}^{\rm L}$ is defined as
\begin{align}\label{eq:r_local_hva}
    (r_{\rm patch}^{\rm L})^2 = \min_{\substack{l \in \{1,\dots,L\}\\ i \in \{1,2\}}} \frac{9c_{l,y_i}(\vec{0})^2}{8 c_{l,y_i}(\vec{0}) \mathcal{A}_{(l,y_i)}(\vec{0}) + 24\beta_{(l,y_i)}}
\end{align}

Similarly to what we did before, we analyze which coefficients will minimize the previous equation. We start by computing $\mathcal{A}_{(l,y_i)}(\vec{0})$ defined in Eq.~\eqref{eq:recipe_A} by only keeping the contribution of the non zero frequencies $\forall i,j\in \{1,2\}$ and $l \in \{1,\dots,L\}$ (defined in Eqs.(\ref{eq:c_hea_local}\ref{eq:w_hea_local},\ref{eq:tildew_hea_local})) , i.e.
\begin{align}\label{eq:a_local_HEA}
    \AC_{(l,y_i)}(\vec{0}) &= 4 \left[\omega_{(l,y_i)}^{\rm (max)}\right]^2 \left(\sum_{\nu=1}^{l-1} \sum_{j =1}^{2} \left[\omega_{(\nu,y_j)}^{\rm (eff) }(\vec{0})\right]^2 +  \sum_{j=1}^{i-1} \left[\omega_{(l,y_j)}^{\rm (eff) }(\vec{0})\right]^2 \right) +  \sum_{j=i+1}^{2} \left[ \widetilde{\omega}^{(\rm eff)}_{(l,y_i),(l,y_j)}(\vec{0}) \right]^2 +  \sum_{\nu=l+1}^L \sum_{j=1}^{2} \left[ \widetilde{\omega}^{(\rm eff)}_{(l,y_i),(\nu,y_j)}(\vec{0}) \right]^2\\
    &= 16(2L + 6l-10 +3i)\;, 
\end{align}
where in the first equality we imposed the condition $\forall i,j\in \{1,2\}$. This is maximized when $l = L$ and $i =2 $, where we find
\begin{align}
     \AC_{(L,y_2)}(\vec{0}) = 64(2L - 1)
\end{align}

Here,  we consider the following definition of $\beta_{(l,y_i)}$ $,\forall i\in \{1,2\}$ and $l \in \{1,\dots,L\}$, as we did for the global case in Eq.~\eqref{eq:beta_HEA}
\begin{align}\label{eq:beta_local_HEA}
    \beta_{(l,y_i)}  =\frac{32\left[  \omega^{(\rm max)}_{(l,y_i)}\right]^{6}  \norm{ \sigma_z \otimes  \sigma_z \otimes \1^{\otimes(n-2)}}^2}{3} = \frac{2^{11}}{3}
\end{align} 

Hence, we finally obtain the following scaling of $r_{\rm patch}^{\rm L}$, by recovering Eqs.~(\ref{eq:beta_local_HEA},\ref{eq:a_local_HEA}) and pluging them in Eq.~\eqref{eq:r_local_hva} i.e.
\begin{align}
    (r_{\rm patch}^{\rm L})^2 =  \frac{9}{128(2L+7)} \in \Theta\left(\frac{1}{L}\right) =\Theta\left(\frac{1}{\sqrt{M}}\right)
\end{align}
and the associated variance lower bound defined in Eq.~\eqref{eq:Var_recipe_th3} is given by
\begin{align}
    \Var_{\thv \sim \uni(\vec{0},r_{\rm patch}^{\rm L})}[\mathcal{L}(\thv)] &\geq \frac{1}{72} \left(\sum_{l=1}^L \sum_{i=1}^2 c^2_{l,y_i}(\vec{0})\right) (r_{\rm patch}^{\rm G})^4\\
    &\in \Omega\left(\frac{1}{L}\right) = \Omega\left(\frac{1}{\sqrt{M}}\right)\;.
\end{align}
where we recall that $M = 2nL$, and that $n \in \Theta(L)$.

\subsubsection{Hamiltonian Variational Ansatz: Proof of analytical results}\label{subsubsec:HVA_proof}

\paragraph*{\underline{Generic Relaxed HVA for a geometrically local Hamiltonian:}}\label{para:generic_relax_hva}

Here we consider the generic relaxed HVA loss function $\mathcal{L}(\thv_{\rm Relaxed})$ for a geometrically local Hamiltonian described in Eq.~\eqref{eq:hva_res_1} and apply the general recipe detailed in the introduction of Appendix \ref{appendix:architectures} to characterize the size of the patch centered around zero with guaranteed substantial loss variance, which we denote by $r_{\rm patch}^{\rm Relaxed}$, as well as the associated variance lower bound. 
Since all the $\nparams_{\rm Relaxed}= KL$ parameters in the Relaxed circuit $U(\thv_{\rm Relaxed})$ are independent (See Eq.~\eqref{eq:HVA_relaxed}), we henceforth apply Theorem~\ref{th:var_formal} to derive the final results given in Eq.~\eqref{eq:hva_res_gen_patch} and Eq.~\eqref{eq:hva_res_gen_var}.

For the second and third steps of the general recipe, we can already determine the scaling of the loss's second derivative and identify the subset of parameters contributing to the variance lower bound, based on the assumption in Eq. \eqref{eq:hva_cond_curv}. Specifically, we choose to focus on the parameter $\th_{(1,1)}$ and we will assume that this parameter fulfills the aforementioned condition in Eq.~\eqref{eq:hva_cond_curv}. With this, we now shift our focus to compute the relevant frequencies associated with the parameter \( \th_{(1,1)} \), as outlined in the fourth step of the general recipe. In particular, we will fisrt demonstrate that the effective frequencies of the form  \( \widetilde{\omega}^{(\rm eff)}_{(1,1),(l,k)} \) for all 
%\( (l,k) > (1,1) \),
$k>1,\,l>1$ ,as defined in  Eq. \eqref{eq:eff_freqs_tilde_E}, can be upper-bounded by  \( N_O \), the number of Pauli terms in the observable decomposition in Eq.~\eqref{eq:local_obs_hva}, up to a multiplicative constant, i.e.
\begin{align}\label{eq:eff_freqs_hva_bound}
    \widetilde{\omega}^{(\rm eff)}_{(1,1),(l,k)} \leq c N_O \;, \forall (l,k) > (1,1) \;,
\end{align}
 where $c$ is a constant that depends on characteristics of the Hamiltonian terms $H_k$ and $H_1$.

Let us now proceed to prove it and identify the constant $c$ in Eq.~\eqref{eq:eff_freqs_hva_bound}.

First, since we are studying the region with gradients around zero,
the effective frequencies simplifies to
\begin{align}
    \left[\widetilde{\omega}^{\rm (eff)}_{(1,1),(l,k)}(\vec{0})\right]^2 &= \norm{ \left.\frac{\partial^4 [U^\dagger(\thv_{\rm Relaxed})OU(\thv_{\rm Relaxed})]}{\partial\th_{(l,k)}^2 \partial\th_{(1,1)}^2}\right|_{\thv_{\rm Relaxed} = \vec{0}} }\\
    &= \norm{ \left.\frac{\partial^4 [e^{i \th_{(l,k)} H_k} e^{i \th_{(1,1)} H_1}Oe^{-i \th_{(1,1)} H_1} e^{-i \th_{(l,k)} H_k}]}{\partial\th_{(l,k)}^2 \partial\th_{(1,1)}^2}\right|_{\th_{(1,1)},\th_{(l,k)}=0} }\\
    &= \norm{ [H_k,[H_k,[H_1,[H_1,O]]]]}  \label{eq:eff_freq_simple_hva}
\end{align}
where we invoked Eq.~\eqref{eq:p-order-grad-A} from Lemma \ref{lemma:bounded_unitary} twice in the final equality, i.e. first for the second derivative with respect to $\th_{(1,1)}$ and second for the second derivative with respect to $\th_{(l,k)}$.
Hence, our task boils down to bounding the infinity norm of nested commutators.

To move forward, let us introduce the set $\mathcal{C}_k(P)$ 
corresponding to the Hamiltonian term $H_k$ and enumerating its Pauli components defined in Eq.~\eqref{eq:loc_1} which anti commute with a given local Pauli string $P$, i.e.
\begin{equation}
    \mathcal{C}_k(P) = \{1 \leq j \leq N_k \;, \{h_j^{(k)},P\}=0\}\;.
\end{equation}
Clearly, the size of this set can be upper bounded by a constant $s_k$  for any local Pauli string acting non trivially on neighboring qubits  due to the locality of the $h_j^{(k)}$ components, i.e. $|\mathcal{C}_k(P)| \leq s_k$.

Back to the effective frequencies in Eq.~\eqref{eq:eff_freq_simple_hva}, we invoke twice Eq.~\eqref{eq:comm_local_lemma} from Lemma \ref{lemma:bound_nested_com_local} with $p=2$  first for $H_1$ then $H_k$ to further simplify it. Specifically, we can first reduce the nested commutator $[H_1,[H_1,O]]$ using Eq.~\eqref{eq:comm_local_lemma} as follows:
\begin{align}
    [H_1,[H_1,O]] &= \sum_{i=1}^{N_O} [H_1,[H_1,P_i]]\\
    &= \sum_{i=1}^{N_O} [\widetilde{H}_1(P),[\widetilde{H}_1(P),P_i]] \label{eq:htilde_com1}
\end{align}
where we substitute the observable $O$ with its Pauli decomposition given in Eq.~\eqref{eq:local_obs_hva} in the first equality and applied Eq.~\eqref{eq:comm_local_lemma} in the second equality by introducing $\widetilde{H}_1(P) = \sum_{j \in \mathcal{C}_k(P_i)} h_j^{(1)}$.
Now, we can further apply the double commutator with the Hamiltonian term $H_k$ on top of Eq.~\eqref{eq:htilde_com1} and obtain
\begin{align}
    [H_k,[H_k,[H_1,[H_1,O]]]] &= \sum_{i=1}^{N_O} 
    [H_k,[H_k,[\widetilde{H}_1(P),[\widetilde{H}_1(P),P_i]]]]\\
    &= \sum_{i=1}^{N_O} \sum_{j,j' \in \mathcal{C}_1(P_i)} 
    [H_k,[H_k,[h_{j'}^{(1)},[h_{j}^{(1)},P_i]]]]\\
     &= \sum_{i=1}^{N_O} \sum_{j,j' \in \mathcal{C}_1(P_i)}\left[H_k,\left[H_k, P_i h_j^{(1)}h_{j'}^{(1)}+ h_j^{(1)}h_{j'}^{(1)} P_i -  h_j^{(1)}P_ih_{j'}^{(1)}- h_{j'}^{(1)}P_ih_{j}^{(1)}\right]\right] \label{eq:double_com_h1_reduce}
\end{align}

By applying the infinity norm in Eq.~\eqref{eq:double_com_h1_reduce}, we retrieve the expression of the effective frequencies from Eq.~\eqref{eq:eff_freq_simple_hva} and we can further upper bound it by re-invoking Eq.~\eqref{eq:comm_reduce_norm_lemma} from Lemma \ref{lemma:bound_nested_com_local}. Precisely, we have
\begin{align}
    \left[\widetilde{\omega}^{\rm (eff)}_{(1,1),(l,k)}(\vec{0})\right]^2 =& \norm{[H_k,[H_k,[H_1,[H_1,O]]]]}\\
    =&  \norm{\sum_{i=1}^{N_O} \sum_{j,j' \in \mathcal{C}_1(P_i)}\left[H_k,\left[H_k, P_i h_j^{(1)}h_{j'}^{(1)}+ h_j^{(1)}h_{j'}^{(1)} P_i -  h_j^{(1)}P_ih_{j'}^{(1)}- h_{j'}^{(1)}P_ih_{j}^{(1)}\right]\right]} \\
     \leq &\sum_{i=1}^{N_O}   \sum_{j,j' \in \mathcal{C}_1(P_i)} \bigg( 2\norm{[H_k,[H_k,  P_i h_j^{(1)}h_{j'}^{(1)}]]} + \norm{[H_k,[H_k,  h_j^{(1)}h_{j'}^{(1)}P_i]]} \nonumber\\
     &+ \norm{[H_k,[H_k, h_j^{(1)}P_ih_{j'}^{(1)}]]} + \norm{[H_k,[H_k, h_{j'}^{(1)}P_ih_{j}^{(1)}]]} \bigg) \label{eq:com_aaaa}
\end{align}
where we used the triangular inequality to get Eq.~\eqref{eq:com_aaaa}.

Here, the Pauli strings $P_ih_j^{(1)}h_{j'}^{(1)}, \, h_j^{(1)}h_{j'}^{(1)}P_i,\, h_{j'}^{(1)}P_ih_{j}^{(1)}$ and $h_j^{(1)}P_ih_{j'}^{(1)}$ in Eq.~\eqref{eq:com_aaaa} are at most acting non trivially on $\kappa_O + 2 \kappa$  qubits where we recall that $P_i$ is $\kappa_O$-local and $h_j^{(1)}$ is $\kappa$-local. Hence, we invoke  Lemma~\ref{lemma:bound_nested_com_local}, precisely the result in Eq.~\eqref{eq:lemma_bound_com_local} for $p=2$ to upper bound the infinity norm terms in Eq.~\eqref{eq:com_aaaa} as follows
\begin{align}
    \norm{[H_k,[H_k,  P_ih_j^{(1)}h_{j'}^{(1)}]]} & \leq 4 s_k^2( P_ih_j^{(1)}h_{j'}^{(1)})\\
    \norm{[H_k,[H_k,  h_j^{(1)}h_{j'}^{(1)}P_i]]} & \leq 4 s_k^2( h_j^{(1)}h_{j'}^{(1)}P_i)\\
    \norm{[H_k,[H_k, h_j^{(1)}P_ih_{j'}^{(1)}]]} &\leq 4 s_k^2( h_j^{(1)}P_ih_{j'}^{(1)})\\
    \norm{[H_k,[H_k, h_{j'}^{(1)}P_ih_{j}^{(1)}]]} &\leq 4 s_k^2( h_{j'}^{(1)}P_ih_{j}^{(1)})
\end{align}
where we recall that $s_k(X)$, introduced in Lemma~\ref{lemma:bound_nested_com_local} is the number of Pauli terms in $H_k$ that do Not commute with $X$. Furthermore, $s_k(P_i h_j^{(1)}h_{j'}^{(1)}),\,s_k( h_j^{(1)}h_{j'}^{(1)}P_i)$, $s_k( h_j^{(1)}P_ih_{j'}^{(1)})$ and $s_k( h_{j'}^{(1)}P_ih_{j}^{(1)})$ are all constant in the system size according to Lemma \ref{lemma:bound_nested_com_local}. Hence, we can simply upper bound the terms above simultaneously by maximizing over $j,j'  \in \mathcal{C}_1(P_i)$ and $P_i \;, \forall 1 \leq i \leq N_O$ such that $s_k$ defined as,
\begin{align}
    s_k = \max_{ 1 \leq i \leq N_O} \max_{j,j'  \in \mathcal{C}_1(P_i)} \max\left[s_k(P_i h_j^{(1)}h_{j'}^{(1)}),s_k( h_j^{(1)}h_{j'}^{(1)}P_i),s_k( h_j^{(1)}P_ih_{j'}^{(1)}),s_k( h_{j'}^{(1)}P_ih_{j}^{(1)})\right]\;,
\end{align}
is still a constant in the system size.

Therefore, the upper bound on the effective frequencies in Eq.~\eqref{eq:com_aaaa} becomes
\begin{align}
    \left[\widetilde{\omega}^{\rm (eff)}_{(1,1),(l,k)}(\vec{0})\right]^2 &\leq  \sum_{i=1}^{N_O}   \sum_{j,j' \in \mathcal{C}_1(P_i)} 16 s_k^2\\
    &\leq 16 s_k^2 \sum_{i=1}^{N_O} s_1^2(P_i) \\
    &\leq 16 s_k^2 s_1(O)^2 N_O \label{eq:final_eff_freq_hva_bound}
\end{align}
where we introduced $s_1(O) = \max_{ 1 \leq i \leq N_O} s_1(P_i)$ which is constant in the system size since all $s_1(P_i)$ are constants.

Hence, we proved the result in Eq.~\eqref{eq:eff_freqs_hva_bound} with the constant $c$ defined as $c= 16 s_k^2 s_1^2(O)$.

Now, that we have a final upper bound on the effective frequencies associated with the parameter $\th_{(1,1)}$  in Eq.~\eqref{eq:final_eff_freq_hva_bound}, we can invoke the patch size expression $r_{\rm patch}^{\rm Relaxed}$ with guaranteed gradients around zero from Eq.~\eqref{eq:rpatch_th3_recipe} (with $\Lambda=\{(1,1)\}$),

 \begin{align}\label{eq:patch_hva}
     (r_{\rm patch}^{\rm Relaxed})^2 =  \frac{9c^2_{1,1}(\vec{0})}{8 c_{1,1}(\vec{0}) \mathcal{A}_{(1,1)}(\vec{0}) + 24 \beta_{(1,1)}(\vec{0})}
 \end{align}
where $\mathcal{A}_{(1,1)}(\vec{0})$ is defined in Eq.~\eqref{eq:recipe_A}  as
\begin{align}
     \mathcal{A}_{(1,1)}(\vec{0}) &= \sum_{k=2}^{K}  (\widetilde{\omega}^{\rm (eff)}_{(1,1),(1,k)}(\vec{0}))^2+ \sum_{l=2}^L \sum_{k=1}^{K}  (\widetilde{\omega}^{\rm (eff)}_{(1,1),(l,k)}(\vec{0}))^2 \;.
\end{align}

Using the upper bound on the effective frequencies in Eq.~\eqref{eq:final_eff_freq_hva_bound}, we can obtain the following upper bound on $\mathcal{A}_{(1,1)}(\vec{0})$. Indeed, because we are using the parameter index $(1,1)$, the first sum in Eq.~\eqref{eq:recipe_A} is zero, and thus we are left with the following
\begin{align}\label{eq:final_A_local}
    \mathcal{A}_{(1,1)}(\vec{0}) &= \sum_{k=2}^{K} (\widetilde{\omega}^{\rm (eff)}_{(1,1),(1,k)}(\vec{0}))^2+ \sum_{l=2}^L \sum_{k=1}^{K} (\widetilde{\omega}^{\rm (eff)}_{(1,1),(l,k)}(\vec{0}))^2 \\
    &\leq 16 s_1^2(O) N_O\left(-s_1^2+L\sum_{k=1}^Ks_k^2\right)\in \mathcal{O}(M \cdot N_O)
\end{align}
where we recall that $s_1(O)$ and $s_k$ are constants and  $M = KL$ is the number of distinct parameters in the parametrized circuit.

Moreover, the expression of $\beta_{(1,1)}(\vec{0})$ defined in Eq.~\eqref{eq:beta_recipe} can be made tighter (smaller) under the generators and observable locality by  slightly modifying  the first step in the proof of Theorem~\ref{th:var_formal}. Specifically, instead of using Corollary~\ref{cor:var_LB_1param_1layer} to obtain Eq.~\eqref{eq:varl_LB}, we rather use Corollary~\ref{cor:var_LB_1param_1layer_local} to get 
 \begin{equation}\label{eq:final_beta_local}
     \beta_{(1,1)}(\vec{0}) = \frac{2^5 N_O^2 s_1(O)^6 }{3}\in\Theta(N_O^2)\;.
 \end{equation}
 Note that this is possible due to the locality assumptions on both the observable and the generators of the circuit. This is indeed a very strong assumption in general, but that fits perfectly into this case. That is why we present this bound separately.

Note, that similarly to what we did in the proof of Corollary~\ref{cor:var_minimum}, particularly in Eq.~\eqref{eq:lowerbound_on_r_cor1}, we can still obtain the results for a lower-bound on $r_{\rm patch}$. Hence, by combining the new expression of $\beta_{(1,1)}(\vec{0})$ in Eq.~\eqref{eq:final_beta_local}  and the upper bound on $\mathcal{A}_{(1,1)}(\vec{0})$ in Eq.~\eqref{eq:final_A_local},  the final scaling of $r_{\rm patch}^{\rm Relaxed}$ is 
 \begin{align}
     (r_{\rm patch}^{\rm Relaxed})^2 \in \Theta\left( \frac{c^2_{1,1}(\vec{0})}{ c_{1,1}(\vec{0}) M \cdot N_O + N_O^2} \right)
 \end{align}

Finally, by plugging in the assumption in Eq.~\eqref{eq:hva_cond_curv} on the loss second derivative $c_{1,1}(\vec{0})$, we get
 \begin{align}\label{eq:patch_hva_relaxed_final}
     (r_{\rm patch}^{\rm Relaxed})^2 \in \Theta\left( \frac{1}{ M} \right)
 \end{align}
 
 Moreover, according to Eq.~\eqref{eq:Var_recipe_th3}, the corresponding variance lower bound  scales as 
 \begin{equation}\label{eq:var_hva_relaxed_final}
     \Var_{\thv \sim \uni(\vec{0},r_{\rm patch}^{\rm Relaxed})}[\mathcal{L}(\thv)] \in \Omega\left(c^2_{1,1}(\vec{0}) (r_{\rm patch}^{\rm Relaxed})^4\right) = \Omega\left(\frac{N_O^2}{M^2}\right)
 \end{equation}

 \medskip

\paragraph*{\underline{Bounds for Relaxed HVA with the Heisenberg model :}}

The region centered around zero, where the loss variance for the loss function \(\widetilde{\mathcal{L}}(\thv_{\rm relaxed})\) associated with the relaxed version of the HVA for the Heisenberg Hamiltonian \(\widetilde{H}\) (defined in Eq. \eqref{eq:loss_relax_heisenberg}) is at most polynomially vanishing, can be characterized directly from the generic Relaxed HVA guarantees for a local Hamiltonian developed in the previous section \ref{para:generic_relax_hva}.

Hence, it suffices to verify that the Heisenberg Hamiltonian $\widetilde{H}$ in Eq.~\eqref{eq:heisenberg} fulfills the locality assumptions considered in section \ref{para:gen_hva_setting} for a generic HVA loss with a local Hamiltonian  and that the loss function  \(\widetilde{\mathcal{L}}(\thv_{\rm relaxed})\) second derivatives w.r.t $\th_{(1,1)}$ obeys the condition in Eq.~\eqref{eq:hva_cond_curv}.

First, as outlined in Section \ref{subsubsec:circuit_HVA},  the Heisenberg Hamiltonian $\widetilde{H}= \widetilde{H}_1 + \widetilde{H}_2 + \widetilde{H}_3 $ where the terms $\widetilde{H}_k \;, k \in \{1,2,3\}$ defined in Eq.~\eqref{eq:heisenberg_terms} are each 2-geometrically local. Moreover, we have that $N_{\widetilde{H}}$, i.e. the number of Pauli terms in $\widetilde{H}$ is $N_{\widetilde{H}} = 3n$ where $n$ is the number of qubits.

Now, we show that the loss function second derivative w.r.t $\th_{(1,1)}$ evaluated at zero, denoted by $c_{(1,1)}(\vec{0})$  scales as $N_{\widetilde{H}} = 3n$, satisfying the condition  in Eq.~\eqref{eq:hva_cond_curv}.

First, the expression of $c_{(1,1)}(\vec{0})$ can be written as 
\begin{align}
    c_{1,1}(\vec{0}) &= \left| \Tr\left[|\psi\rangle\langle\psi| \left.\frac{\partial^2 [\widetilde{U}^\dagger(\thv_{\rm Relaxed}) \widetilde{H}\widetilde{U}(\thv_{\rm Relaxed})]}{\partial\th_{(1,1)}^2}\right|_{\thv_{\rm Relaxed}= \vec{0}}\right] \right|\\
    &= \left| \Tr\left[|\psi\rangle\langle\psi| \left.\frac{\partial^2 [e^{i \th_{(1,1)} \widetilde{H}_1} \widetilde{H} e^{-i \th_{(1,1)} \widetilde{H}_1}]}{\partial\th_{(1,1)}^2}\right|_{\th_{(1,1)}= 0}\right] \right|\\
    &= \left| \Tr\left[|\psi\rangle\langle\psi| [\widetilde{H}_1,[\widetilde{H}_1,\widetilde{H}]]\right] \right| \label{eq:com_heisenberg}
\end{align}
where we invoked Eq.~\eqref{eq:p-order-grad-A} from Lemma \ref{lemma:bounded_unitary} for $p=2$ in the last equality.

To further develop the expression of $c_{(1,1)}(\vec{0})$ in Eq.~\eqref{eq:com_heisenberg}, we focus in a first step on computing the nested commutator $[\widetilde{H}_1,[\widetilde{H}_1,\widetilde{H}]]$ where we recall for completeness the definition of the Hamiltonian terms  $\widetilde{H}_k \;, k \in \{1,2,3\}$ from Eq.~\eqref{eq:heisenberg_terms}.
\begin{align}
    & \widetilde{H} = \widetilde{H}_1 + \widetilde{H}_2 + \widetilde{H}_3 \;, \\
     &\widetilde{H}_3 = \sum_{i=1}^n \sigma_z^{(i)}\otimes\sigma_z^{(i+1)}, \quad
    \widetilde{H}_2 = \sum_{i=1}^n \sigma_y^{(i)}\otimes\sigma_y^{(i+1)}, \quad
    \widetilde{H}_1 = \sum_{i=1}^n \sigma_x^{(i)}\otimes\sigma_x^{(i+1)}\;.
\end{align}

Hence, we start by computing the nested commutator $[\widetilde{H}_1,\widetilde{H}_k]$. Trivially $[\widetilde{H}_1,\widetilde{H}_1] = 0$, and thus we only need to focus on the other two. We will do this computation as follows: we assign $k$ to be a variable that can take values either $k = \{2,3\}$ and to ease the notation we will use $\sigma_2 = \sigma_y,\,\sigma_3 = \sigma_z$. With this we can compute the commutators in a very compact form
\begin{align}\label{eq:idontlikethis_commutators}
    [\widetilde{H}_1,\widetilde{H}_k] &= \sum_{i=1}^n \sum_{j=1}^n [\sigma_x^{(i)}\otimes\sigma_x^{(i+1)},\sigma_k^{(j)}\otimes\sigma_k^{(j+1)}] \\
    &= \sum_{j=1}^{n} \left(\sigma_x^{(j-1)}\otimes[\sigma_x^{(j)},\sigma_k^{(j)}]\otimes\sigma_k^{(j+1)} + \sigma_k^{(j-1)}\otimes[\sigma_k^{(j)},\sigma_x^{(j)}]\otimes\sigma_x^{(j+1)}\right)\\
    &= 2i\epsilon_{x,k,l}\sum_{j=1}^{n} \left(\sigma_x^{(j-1)}\otimes\sigma^{(j)}_l\otimes\sigma_k^{(j+1)} - \sigma_k^{(j-1)}\otimes\sigma^{(j)}_l\otimes\sigma_x^{(j+1)}\right)\label{eq:commutatin_paulis_forH}
\end{align}
where in the first equality we used that if $[\sigma_l^{(i)}\otimes\sigma_l^{(i+1)}, \sigma_k^{(i)}\otimes\sigma_k^{(i+1)}] = 0\, \forall \, l,k$ and that if the Pauli matrices act of different qubits the commutator is also zero. Furthermore, we are abusing notation and defining $\sigma^{(0)} = \sigma^{(n)}, \sigma^{(n+1)}=\sigma^{(1)}$. In third equality we used the well known equality for the commutator of Pauli matrices $[\sigma_a,\sigma_b] = 2i\epsilon_{a,b,c}\sigma_c$, where $\epsilon_{a,b,c}$ is the Levi-Civita symbol. Therefore, we can use this to find
\begin{align}
     [\widetilde{H}_1,\widetilde{H}_2] = &2i\sum_{j=1}^{n} \left(\sigma_x^{(j-1)}\otimes\sigma^{(j)}_z\otimes\sigma_y^{(j+1)} - \sigma_y^{(j-1)}\otimes\sigma^{(j)}_z\otimes\sigma_x^{(j+1)}\right)\\
      [\widetilde{H}_1,\widetilde{H}_3] = & 2i\sum_{j=1}^{n} \left(\sigma_z^{(j-1)}\otimes\sigma^{(j)}_y\otimes\sigma_x^{(j+1)} - \sigma_x^{(j-1)}\otimes\sigma^{(j)}_y\otimes\sigma_z^{(j+1)}\right)
\end{align}

If we recover Eq.~\eqref{eq:idontlikethis_commutators}, we can now compute $ [\widetilde{H}_1,[\widetilde{H}_1,\widetilde{H}_k]]$. We proceed as in the previous case
\begin{align}
    [\widetilde{H}_1,[\widetilde{H}_1,\widetilde{H}_k]] = & 2i\epsilon_{z,k,l}\sum_{i=1}^n \sum_{j=1}^{n}[ \sigma_x^{(i)}\otimes\sigma_x^{(i+1)}, \sigma_x^{(j-1)}\otimes\sigma^{(j)}_l\otimes\sigma_k^{(j+1)} - \sigma_k^{(j-1)}\otimes\sigma^{(j)}_l\otimes\sigma_x^{(j+1)} ]\\
    =& 2i\epsilon_{z,k,l}\sum_{j=1}^n \bigg( \left(\sigma_x^{(j-1)}\right)^2\otimes[\sigma_x^{(j)},\sigma_l^{(j)}]\otimes\sigma_k^{(j+1)} - \sigma_k^{(j)}\otimes [\sigma_l^{(j)},\sigma_x^{(j)}]\otimes \left(\sigma_x^{(j+1)}\right)^2 \\
    &+\sigma_x^{(j-1)}\otimes\sigma^{(j)}_l\otimes[\sigma_x^{(j+1)},\sigma_k^{(j+1)}]\otimes \sigma_x^{(j+2)} - \sigma_x^{(j-2)}\otimes[\sigma_x^{(j-1)},\sigma_k^{(j-1)}]\otimes\sigma^{(j)}_l\otimes\sigma_x^{(j+1)}\bigg)\\
    =&-8\epsilon_{z,k,l}\sum_{j=1}^n \bigg( \epsilon_{z,l,k}\sigma_k^{(j)}\otimes\sigma_k^{(j+1)} +\epsilon_{z,k,l}\sigma_x^{(j-1)}\otimes\sigma^{(j)}_l\otimes\sigma_l^{(j+1)}\otimes \sigma_x^{(j+2)} \bigg)\\
    =& 8\sum_{j=1}^n \bigg( \sigma_k^{(j)}\otimes\sigma_k^{(j+1)} -\sigma_x^{(j-1)}\otimes\sigma^{(j)}_l\otimes\sigma_l^{(j+1)}\otimes \sigma_x^{(j+2)} \bigg)
\end{align}
where in the second equality we used that $[\sigma_x\otimes\sigma_x, \sigma_a,\sigma_b] = 0$ if $a,b\neq x$. In the second to last equality we used that $\sigma_x^2 = \1$ and applied the same commutation rules for the Pauli matrices as in Eq.~\eqref{eq:commutatin_paulis_forH}. We also used the trivial identity $[A,B] = -[B,A]$. Finally, in the last equality we use $\epsilon_{a,b,c} = -\epsilon_{a,c,b}$ and $\epsilon_{a,b,c}^2 = 1$ as long as the indices $\{a,b,c\}$ are different (zero otherwise). Therefore, from here onward we assume that $k\neq l\neq z$. With this we can finally compute $ [\widetilde{H}_1,[\widetilde{H}_1, H]]$. Indeed, with the previous equation we find
\begin{align}
     [\widetilde{H}_1,[\widetilde{H}_1,\widetilde{H}_2]] =& 8\sum_{j=1}^n \bigg( \sigma_y^{(j)}\otimes\sigma_y^{(j+1)} -\sigma_x^{(j-1)}\otimes\sigma^{(j)}_z\otimes\sigma_z^{(j+1)}\otimes \sigma_x^{(j+2)} \bigg)\\
     [\widetilde{H}_1,[\widetilde{H}_1,\widetilde{H}_3]] =& 8\sum_{j=1}^n \bigg( \sigma_z^{(j)}\otimes\sigma_z^{(j+1)} -\sigma_x^{(j-1)}\otimes\sigma^{(j)}_y\otimes\sigma_y^{(j+1)}\otimes \sigma_x^{(j+2)} \bigg)
\end{align}
and by adding these two we find that $ [\widetilde{H}_1,[\widetilde{H}_1, H]]$ is
\begin{align}\label{eq:final_com_ricard}
     [\widetilde{H}_1,[\widetilde{H}_1,H]]  =& 8\sum_{j=1}^n \left[\sigma_z^{(j)}\otimes\sigma_z^{(j+1)}  + \sigma^{(j)}_y\otimes\sigma_y^{(j+1)} -\sigma_x^{(j-1)}\otimes\left( \sigma_z^{(j)}\otimes\sigma_z^{(j+1)}  + \sigma^{(j)}_y\otimes\sigma_y^{(j+1)} \right)\otimes \sigma_x^{(j+2)} \right]
\end{align}

Now that we have computed the nested commutator $[\widetilde{H}_1,[\widetilde{H}_1,\widetilde{H}]]$, we can recover Eq.~\eqref{eq:com_heisenberg}, and recalling that the state $|\psi\rangle = \frac{|01\rangle^{\otimes n/2} + |10\rangle^{\otimes n/2}}{\sqrt{2}}$ we can easily compute $c_{(1,1)}(\vec{0})$.
\begin{align}
    c_{(1,1)}(\vec{0}) =& \frac{1}{2}\Bigg| \left(\bra{01}^{\otimes n/2}+\bra{10}^{\otimes n/2}\right) 8\sum_{j=1}^n \bigg[\sigma_z^{(j)}\otimes\sigma_z^{(j+1)}  + \sigma^{(j)}_y\otimes\sigma_y^{(j+1)} \\
    &-\sigma_x^{(j-1)}\otimes\left( \sigma_z^{(j)}\otimes\sigma_z^{(j+1)}  + \sigma^{(j)}_y\otimes\sigma_y^{(j+1)} \bigg)\otimes \sigma_x^{(j+2)} \right]\left(\ket{01}^{\otimes n/2}+\ket{10}^{\otimes n/2}\right) \Bigg| \\
    =& 8 \frac{1}{2} \Bigg| \left(\bra{01}^{\otimes n/2}+\bra{10}^{\otimes n/2}\right)\sum_{j=1}^n \sigma_z^{(j)}\otimes\sigma_z^{(j+1)}  \left(\ket{01}^{\otimes n/2}+\ket{10}^{\otimes n/2}\right) \Bigg|\\
    =&  8 \Bigg|-\sum_{j=1}^n 1 \Bigg| = 8n\label{eq:deriv_ricard}
\end{align}
where in the second equality we have assumed the number of qubits to be larger than 4, and thus the expected value of all the Pauli matrices that appear and are not $\sigma_z$ are zero. In the second to last equality we used that $\sigma_z\otimes\sigma_z\ket{01} = -\ket{01}$. Hence we see that $c_1(\vec{0}) = 8n$, and thus fulfills the condition  in Eq.~\eqref{eq:hva_cond_curv}.

After verifying that all the conditions adopted in the previous section for a geometrically local Hamiltonian hold for the Heisenberg loss function $\widetilde{\mathcal{L}}(\thv_{\rm Relaxed})$ , the characterization of its region with guaranteed substantial loss variance  follows directly from the generic one obtained in Eq.~\eqref{eq:patch_hva_relaxed_final}.
Similarly, the variance lower bound within this region follows from Eq.~\eqref{eq:var_hva_relaxed_final}.

\paragraph*{\underline{Generic Trotter HVA for a local Hamiltonian:}}\label{para:generic_trotter_hva}

Here we consider the generic Trotter HVA loss function $\mathcal{L}(\thv_{\rm Trotter})$ for a geometrically local Hamiltonian described in the previous section \ref{subsusbsec:HVA_results} in Eq.~\eqref{eq:hva_res_2} and apply the general recipe detailed in the introduction of Appendix \ref{appendix:architectures} to characterize the size of the patch centered around zero with guaranteed substantial loss variance, which we denote by $r_{\rm patch}^{\rm Trotter}$, as well as the associated variance lower bound. 
In this setting, we have only $m_{\rm Trotter}= K$ independent parameters repeated over the $L$ Trotter layers (See Eq.~\eqref{eq:HVA_trotter}).
Hence, we invoke Theorem \ref{th:var_formal_cor} to derive the final results given in Eq.~\eqref{eq:hva_res_gen_patch_trotter} and Eq.~\eqref{eq:hva_res_gen_var_trotter}.

Having identified the theorem we will adapt, we now proceed to the second and third steps of the general approach, choosing to apply Theorem \ref{th:var_formal_cor} with the parameter contribution \(\th_1\), for which the scaling of the loss function's second derivative with respect to it is already given in Eq. \eqref{eq:hva_cond_curv_trotter}.

Now, we move to the fourth step consisting in computing the maximal frequencies $\omega^{(\rm max)}_k$ associated with the parameters $\th_k\;, k \in \{1,\dots,K\}$ as defined in Eq.~\eqref{eq:max_freqs_E}, i.e.
\begin{equation}\label{eq:omega_max_trotter_1}
    \omega^{(\rm max)}_k = \sum_{l=1}^L \omega^{(\rm max)}(H_k) = L \cdot \omega^{(\rm max)}(H_k)\;.
\end{equation}
In addition, since we have that each $H_k$ can be decomposed as a sum of commuting Pauli strings, as outlined in Section \ref{subsubsec:circuit_HVA}, the maximal frequency associated to a Hamiltonian term $H_k$ is nothing but $\omega^{(\rm max)}(H_k) = 2N_k$, where $N_k$ is the number of terms in the Pauli decomposition of $H_k$ (See Eq.~\eqref{eq:loc_2}).
Therefore, Eq.~\eqref{eq:omega_max_trotter_1} becomes
\begin{equation}\label{eq:omega_max_trotter}
    \omega^{(\rm max)}_k = 2LN_k\;.
\end{equation}

Consequently, we can now invoke Eq.~\eqref{eq:rpatch_th4_recipe} to obtain the scaling of $r_{\rm patch}^{\rm Trotter}$, i.e. the size of the region with guaranteed substantial gradients centered around zero for the loss function $\mathcal{L}(\thv_{\rm Trotter})$ defined in Eq.~\eqref{eq:hva_res_2}. Precisely, by plugging the expression of the effective frequencies in Eq.~\eqref{eq:omega_max_trotter} and the condition on the second derivative in Eq.~\eqref{eq:hva_cond_curv_trotter} in Eq.~\eqref{eq:rpatch_th4_recipe}  , we get

 \begin{align}\label{eq:rpatch_trotter_hva}
            (r_{\rm patch}^{\rm Trotter})^2 = \frac{3c^2_{1}(\vec{0})}{8(2c_{1}(\vec{0})\gamma_{1} + \widetilde{\beta}_{1})} \;,
        \end{align}
        where $\widetilde{\beta}_{1}$ and $\gamma_{1}$ are respectively defined in Eq.~\eqref{eq:gamma_recipe} and Eq.~\eqref{eq:beta_tilde_recipe}.
        By plugging the expression of the effective frequencies in Eq.~\eqref{eq:omega_max_trotter} and the condition on the second derivative in Eq.~\eqref{eq:hva_cond_curv_trotter} , we get
        \begin{align}
           \gamma_{1}&= \frac{8}{3}\norm{O}\left(\omega^{(\rm max)}_{1}\right)^2\sum_{k=2}^{K}\left(\omega^{(\rm max)}_{k}\right)^2 \\ 
           &= \frac{2^7}{3}\norm{O} L^4 N_1^2 \left(\sum_{k=2}^K N_k^2\right) \\
           \widetilde{\beta}_1&=\frac{32}{3} \left[\omega^{(\rm max)}_{1}\right]^6 \norm{O}^2\\
           &= \frac{2^{11}}{3} L^6 N_1^6 \norm{O}^2 \;. 
        \end{align}
Hence, Eq.~\eqref{eq:rpatch_trotter_hva} simplifies to
\begin{align}
    (r_{\rm patch}^{\rm Trotter})^2 &= \frac{9  c_1^2(\vec{0})}{2^{11} c_1(\vec{0}) \norm{O} L^4 N_1^2 \left(\sum_{k=2}^K N_k^2\right) + 2^{14} L^6 N_1^6 \norm{O}^2 }\\
    &= \frac{ 9 b^2 L^4 N_O^2}{2^{11} b L^6 N_O \norm{O}   N_1^2 \left(\sum_{k=2}^K N_k^2\right) + 2^{14} L^6 N_1^6 \norm{O}^2 }\\
    & \geq \frac{ 9 b^2 L^4 N_O^2}{2^{11} b L^6 N_O^2  N_1^2\left(\sum_{k=2}^K N_k^2\right) + 2^{14} L^6 N_1^6 N_O^2 }\\ \label{eq:lower-bound-rpatch-HVA-trotter-general}
    & = \frac{ 9 b^2 }{2^{11} b L^2   N_1^2 \left(\sum_{k=2}^K N_k^2\right) + 2^{14} L^2 N_1^6  }\;,
\end{align}
where in the second equality we plugged in the condition in Eq.~\eqref{eq:hva_cond_curv_trotter} by making the substitution $c_1(\vec{0}) = b L^2 N_O$ for some constant $b$. In the first inequality, we used the fact that $\norm{O} \leq N_O$, which can be obtained by simply using the triangular inequality in Eq.~\eqref{eq:local_obs_hva}.

As mentioned for the Relaxed case, we can still obtain the same results for a lower-bound on $r_{\rm patch}$ (as we did in the proof of Corollary~\ref{cor:var_minimum}, particularly in Eq.~\eqref{eq:lowerbound_on_r_cor1}). 
Therefore, we redefine the lower-bound in Eq.~\eqref{eq:lower-bound-rpatch-HVA-trotter-general} as $r_{\rm patch}^{\rm Trotter}$ and give its final scaling by assuming that $N_k \leq N \;, \forall 1 \leq k\leq K$ and that $K \in \Theta(1)$. 
\begin{align}
        \left(r_{\rm patch}^{\rm Trotter}\right)^2& \in \mathcal{O}\left(\frac{1}{L^2 N_1^2  \left(N_1^4 + \sum_{k=2}^K N_k^2\right)}\right)\\
    &= \mathcal{O}\left(\frac{1}{L^2N^6 }\right)\\
     & \underset{L,N \sim n}{=} \mathcal{O}\left(\frac{1}{M^8}\right)\label{eq:final_trotter_hva_patch}
\end{align}
Consequently, the associated variance lower bound for $r_{\rm patch}^{\rm Trotter} \in \Theta\left(\frac{1}{LN^3}\right)$ can be derived from Eq.~\eqref{eq:Var_recipe_th4} as follows
\begin{equation}\label{eq:final_trotter_hva_var}
    \Var_{\thv_{\rm Trotter} \sim \uni(\vec{0},r_{\rm patch}^{\rm Trotter})}[\mathcal{L}(\thv_{\rm Trotter})] \in \Omega\left(\frac{L^4N_O^2}{L^4 N^{12}}\right) \underset{L,N,N_O \sim n}{=} \Omega\left(\frac{1}{M^{10}}\right)
\end{equation}

\paragraph*{\underline{Bounds for Trotterized HVA with the Heisenberg model :}}

The region centered around zero, where the loss variance for the loss function \(\widetilde{\mathcal{L}}(\thv_{\rm Trotter})\) associated with the Tottter version of the HVA for the Heisenberg Hamiltonian \(\widetilde{H}\) (defined in Eq. \eqref{eq:loss_trotter_heisenberg}) is at most polynomially vanishing, can be characterized directly from the generic Trotter HVA guarantees for a local Hamiltonian developed in the previous section \ref{para:generic_trotter_hva}.

Hence, it suffices to verify that  the loss function  \(\widetilde{\mathcal{L}}(\thv_{\rm relaxed})\) second derivatives w.r.t $\th_{1}$ obeys the condition in Eq.~\eqref{eq:hva_cond_curv_trotter}.
Let us begin by evaluating the loss curvature at zero w.r.t $\th_{1}$ using Lemma \ref{lemma:bouded_unitary_product}. To match the channel form used in Lemma \ref{lemma:bouded_unitary_product}, we express the backpropagated observable $\widetilde{U}^\dagger(\thv_{\rm Trotter})\widetilde{H}\widetilde{U}^\dagger(\thv_{\rm Trotter})$  as a composition of unitary channels applied to $\widetilde{H}$ as follows
\begin{align}
    \channel_{\th_1,0_2,0_3}(\widetilde{H}) &:= \left(\prod_{l=1}^L e^{-i \th_1 \widetilde{H}_1}\right)^\dagger \widetilde{H} \left(\prod_{l=1}^L e^{-i \th_1 \widetilde{H}_1}\right)\\
    &= \UC_{\th_1,L} \circ \dots \circ  \UC_{\th_1,1}(\widetilde{H})
\end{align}
where we recall that $\UC_{\th_1,l}(\cdot) := e^{i \th_1 \widetilde{H}_1} (\cdot)  e^{-i \th_1 \widetilde{H}_1}$.

Hence, the curvature w.r.t $\th_1$ can be written as
\begin{align}\label{eq:curv_heisenberg_cor_exp}
    c_1(\vec{0}) = \left|\Tr[|\psi\rangle\langle\psi| \left.\frac{\partial^2}{\partial\th_1^2} [\channel_{\th_1,0_2,0_3}(\widetilde{H})]\right|_{\th_1=0}]\right|\;.
\end{align}

Now, we apply Lemma \ref{lemma:bouded_unitary_product} to compute the second order derivative of $\channel_{\th_1,0_2,0_3}(\widetilde{H})$. Precisely, we have
\begin{align}
    \left.\frac{\partial^2}{\partial\th_1^2} [\channel_{\th_1,0_2,0_3}(\widetilde{H})]\right|_{\th_1=0} &= \sum_{\substack{\vec{a}=(a_1,\dots,a_L) \label{eq:trotter_curv_HVA1}\\ a_1+ \dots+a_L = 2}} \binom{2}{\vec{a}} \UC_{0_1,L}^{(a_L)} \circ \dots \circ \UC_{0_1,1}^{(a_1)}(\widetilde{H})\\
    &= \sum_{l=1}^L \UC_{0_1,L}^{(0)} \circ \dots \circ \UC_{0_1,l}^{(2)} \circ \dots \circ \UC_{0_1,1}^{(0)}(\widetilde{H}) + \sum_{l=1}^L 2 \sum_{t > l} \UC_{0_1,L}^{(0)} \circ \dots \circ \UC_{0_1,t}^{(1)} \circ \dots \circ \UC_{0_1,l}^{(1)} \circ \dots \circ \UC_{0_1}^{(0)}(\widetilde{H})\\
    &= \sum_{l=1}^L \UC_{0_1,l}^{(2)}(\widetilde{H}) +\sum_{l=1}^L 2 \sum_{t > l} \UC_{0_1,t}^{(1)}  \circ \UC_{0_1,l}^{(1)} (\widetilde{H})\\
    &= - \sum_{l=1}^L [\widetilde{H}_1,[\widetilde{H}_1,\widetilde{H}]] -\sum_{l=1}^L 2 \sum_{t > l} [\widetilde{H}_1,[\widetilde{H}_1,\widetilde{H}]]\\
    &= - L^2  [\widetilde{H}_1,[\widetilde{H}_1,\widetilde{H}]] \label{eq:trotter_curv_HVA2}
\end{align}
where in the first equality we invoked Eq.~\eqref{eq:pth_deriv_corr} from Lemma \ref{lemma:bouded_unitary_product} and in the second equality we expanded the sum over all the possible realizations of the vector $\boldsymbol{a}=(a_1,\dots,a_L)$ such that $a_1 + \dots + a_L =2$. In the third equality, we use the identity $\UC_{0_l,l}^{(0)} = Id$ and the fourth equality we invoke Eq.~\eqref{eq:p-order-grad-A} from Lemma \ref{lemma:bounded_unitary}. 

Therefore, by plugging in the result from Eq.~\eqref{eq:trotter_curv_HVA2} in Eq.~\eqref{eq:curv_heisenberg_cor_exp} and using the result already derived in Eq.~\eqref{eq:deriv_ricard},  the second derivative expression simplifies to  
\begin{align}
    c_1(\vec{0}) &=  L^2 \left|\Tr[|\psi\rangle\langle\psi| [\widetilde{H}_1,[\widetilde{H}_1,\widetilde{H}]] ]\right|\\
    &= 8L^2 n\;,
\end{align}
which satisfies the assumption on the second derivative in Eq.~\eqref{eq:hva_cond_curv_trotter} where we recall that $N_{\widetilde{H}} = 3n$.

Consequently, the scaling of the region with guaranteed gradients and the associated variance lower bound will be of the form in Eq.~\eqref{eq:final_trotter_hva_patch} and Eq.~\eqref{eq:final_trotter_hva_var} respectively. 

\subsubsection{Unitary Coupled Cluster Ansatz: Proof of analytical results}\label{subsubsec:UCC_proof}

\paragraph*{\underline{Relaxed UCCSD ansatz with an arbitrary local observable.}}\label{para:ucc_relax_proof}

In this section, we study the region with guaranteed substantial loss variance around zero for the loss function $\mathcal{L}(\thv_{\rm Relaxed})$ based on the relaxed version of the UCCSD anstaz in Eq.~\eqref{eq:relax_ucc_ham} as described in section \ref{para:UCCSD_arbitrary_setting}. 
We also recall that the main assumption covered in section \ref{para:UCCSD_arbitrary_setting} is the constant scaling of second derivatives in Eq.~\eqref{eq:cond_curv_relax_ucc}. 

Hence, we follow the guidelines of the general recipe outlined in the introduction of Appendix~\ref{appendix:architectures}, which leads to the final scalings in Eq.~\eqref{eq:ucc_relax_patch} and Eq.~\eqref{eq:ucc_relax_var}. We begin by recognizing that we will be adapting Theorem \ref{th:var_formal}, as the Relaxed UCCSD circuit in Eq.~\eqref{eq:relax_ucc_ham} involves only spatial correlations.

For the second and third steps of the general recipe,  the scaling of the loss's second derivative and hence the identification of the parameter subset contributing to the variance lower bound are already covered through the assumption in Eq. \eqref{eq:cond_curv_relax_ucc}. Specifically, we choose to focus on the parameters $\th_{\mu,l} \;, \forall \mu \in \Lambda\;, l \in \{1,\dots,L\}$, where we recall that $\Lambda$ is the set of parameters with constant second derivatives. With this, we now shift our focus to computing the relevant frequencies associated with the parameters of the form $\th_{\mu,l}$, as outlined in the fourth step of the general recipe. In particular, we will demonstrate that the effective frequencies, as defined in Eq.~\eqref{eq:eff_freqs_E} and  Eq. \eqref{eq:eff_freqs_tilde_E}, are upper bounded by the observable infinity norm up to a multiplicative constant.

First, we consider the effective frequencies of the form $\omega^{(\rm eff)}_{\mu}(\vec{0})\;, \mu \in \Lambda\subset \Gamma_1 \cup \Gamma_2$, where $\Gamma_{1,2}$ are defined in Eqs.~(\ref{eq:gamma_1},\ref{eq:gamma_2}), associated to a parameter $\th_{\mu_1,l}$ for any $l \in \{1,\dots,L\}$ as defined in Eq.~\eqref{eq:eff_freqs_E} and show that they can be written as infinity norms of nested commutators. Precisely, we have 
\begin{align}
    \left[\omega^{(\rm eff)}_{\mu}(\vec{0})\right]^2 &= \norm{\left.\frac{\partial^2 [U(\thv_{\rm Relaxed})^\dagger O U(\thv_{\rm Relaxed})]}{\partial \th_{\mu,l}^2}\right|_{\thv = \vec{0}}} \\
    &= \norm{\left.\frac{\partial^2 [e^{i \th_{\mu,l} H_{\mu}} Oe^{-i \th_{\mu,l} H_{\mu}}]}{\partial \th_{\mu,l}^2}\right|_{\th_{\mu,l} = 0}}\\
    &= \norm{[H_{\mu},[H_{\mu},O]]}\label{eq:eff_freq_Hpq}
\end{align}
where we invoke Eq.~\eqref{eq:p-order-grad-A} from Lemma \ref{lemma:bounded_unitary} in the last equality.

Similarly, the effective frequencies, as defined in Eq.~\eqref{eq:eff_freqs_tilde_E}, of the form $\widetilde{\omega}^{(\rm eff)}_{\mu,\mu'}(\vec{0})$ for $\mu,\mu' \in \Gamma_1 \cup \Gamma_2 $  associated to the parameters $\th_{\mu,l}$ and  $\th_{\mu,l'}$ for any $l,l' \in \{1,\dots,L\}$ such that the generator $H_{\mu}$ is closer to the observable than the generator $H_{\mu'}$   can be written as
\begin{align}
    \left[\widetilde{\omega}^{(\rm eff)}_{\mu,\mu'}(\vec{0})\right]^2 &= \norm{\left.\frac{\partial^2 [U(\thv_{\rm Relaxed})^\dagger O U(\thv_{\rm Relaxed})]}{\partial \th_{\mu',l'}^2 \partial \th_{\mu,l}^2}\right|_{\thv = \vec{0}}} \\
    &= \norm{\left.\frac{\partial^2 [e^{i \th_{\mu',l'} H_{\mu'}} e^{i \th_{\mu,l} H_{\mu}} Oe^{-i \th_{\mu,l} H_{\mu}} e^{-i \th_{\mu',l'} H_{\mu'}}]}{\partial \th_{\mu',l'}^2 \partial \th_{\mu,l}^2}\right|_{\th_{\mu',l'},\th_{\mu,l} = 0}}\\
    &= \norm{[H_{\mu'},[H_{\mu'},[H_{\mu},[H_{\mu},O]]]]}\label{eq:eff_freq_Hmu_Hmu'}
\end{align}
where we invoke Eq.~\eqref{eq:p-order-grad-A} from Lemma \ref{lemma:bounded_unitary} in the last equality twice, i.e. first for the partial second derivative w.r.t $\th_{\mu,l}$ and second for the partial second derivative w.r.t $\th_{\mu',l'}$. 

Here we recall, as outlined in section \ref{subsubsec:UCC_circuit}, that any Hamiltonian of the form  $H_{\mu_{k}} \;,\mu_k \in \Gamma_k \;, k\in \{1,2\}$  is the average of commuting Pauli strings  according to Eq.~\eqref{eq;H_pq} and Eq.~\eqref{eq:H_pqrs}.
This implies that 
\begin{equation}\label{eq:bounded_norm_gen_ucc}
    \norm{H_{\mu_k}} \leq 1 \;, \forall \mu_k \in \Gamma_k \;, k\in \{1,2\}\;.
\end{equation}
Hence,  the effective frequencies in Eq.~\eqref{eq:eff_freq_Hpq} and Eq.~\eqref{eq:eff_freq_Hmu_Hmu'} can be upper bounded by the observable infinity norm up to a multiplicative constant. Specifically, we have $\forall \mu \in \Gamma_1 \cup \Gamma_2 $
\begin{align}
    \left[\omega^{(\rm eff)}_{\mu}(\vec{0})\right]^2 &\leq 4 \norm{H_\mu}^2 \norm{O}\label{eq:eff_freq_inter} \\
     &\leq 4 \norm{O}  \;.\label{eq:eef_freqs_1_ucc}
\end{align}
Moreover, $\forall \mu,\mu' \in \Gamma_1 \cup \Gamma_2$ such that the generator $H_{\mu}$ is closer to the observable than the generator $H_{\mu'}$ according to Eq.~\eqref{eq:relax_ucc_ham}, we similarly get
\begin{align}
     \left[\widetilde{\omega}^{(\rm eff)}_{\mu,\mu'}(\vec{0})\right]^2 & \leq 16\norm{H_{\mu}}^2 \norm{H_{\mu'}}^2 \norm{O}\label{eq:eff_freq_tilde_inter} \\
     &\leq 16 \norm{O} \;, \label{eq:eef_freqs_2_ucc}
\end{align}
where we used the property $\norm{[A,[A,B]]} \leq 4 \norm{A}^2 \norm{B}$ in Eq.~\eqref{eq:eff_freq_inter} and Eq.~\eqref{eq:eff_freq_tilde_inter} and invoked in Eq.~\eqref{eq:eef_freqs_1_ucc} and Eq.~\eqref{eq:eef_freqs_2_ucc} the bounded norm of the generators in Eq.~\eqref{eq:bounded_norm_gen_ucc}.

Now, we focus on computing the maximal frequencies defined in Eq.~\eqref{eq:max_freqs_E}.
Indeed, in the subspace of the qubits $p$ and $q$, the generator $H_{pq}$ can be expressed as
\begin{align}
    H_{pq}=i|10\rangle\langle 01|-i|01\rangle\langle 10|\;,
\end{align}
which has eigenvalues $\pm 1$ (and $0$) and acts trivially on the remaining qubits. Therefore, we have 
\begin{equation}\label{eq:max_freq_1_ucc}
    \omega^{\rm (max)}(H_{\mu_1})=2\;, \mu_1 \in \Gamma_1.
\end{equation}
Similarly, in the subspace of the qubits $p$, $q$, $r$ and $s$, we have
\begin{equation}
    H_{pqrs}=i|1100\rangle\langle 0011|-i|0011\rangle\langle 1100|\;,
\end{equation}
which also has eigenvalues $\pm 1$ (and $0$), so the maximum frequency is also
\begin{equation}\label{eq:max_freq_2_ucc}
    \omega^{\rm (max)}(H_{\mu_2})=2\;, \mu_2 \in \Gamma_2.
\end{equation}

Now that we have computed the maximal frequencies of the circuit generators in Eq.~\eqref{eq:max_freq_1_ucc} and Eq.~\eqref{eq:max_freq_2_ucc} and provided an upper bound on the effective frequencies in Eq.~\eqref{eq:eef_freqs_1_ucc} and Eq.~\eqref{eq:eef_freqs_2_ucc}, we move forward to evaluating the quantities of interest that will determine $r_{\rm patch}^{\rm Relaxed}$, i.e. the size of the region around zero with guaranteed substantial loss variance $\mathcal{L}(\thv_{\rm Relaxed})$.
For completeness, we recall here the expression of $r_{\rm patch}^{\rm Relaxed}$ from Eq.~\eqref{eq:rpatch_th3_recipe} using the generator indices $\mu \in \Gamma_1 \cup \Gamma_2 $ .

\begin{align}\label{eq:patch_ucc_relax}
       (r_{\rm patch}^{\rm Relaxed})^2 = \min_{\substack{\mu \in \Lambda\\ l \in \{1,\dots,L\}}} \frac{9c_{\mu,l}(\vec{0})^2}{8 c_{\mu,l}(\vec{0}) \mathcal{A}_{\mu,l}(\vec{0}) + 24\beta_{\mu,l}(\vec{0})} 
    \end{align}
    where we recall that $\Lambda \subset \Gamma_1 \cup \Gamma_2$ is the subset of parameter indices which satisfies the assumption in Eq.~\eqref{eq:cond_curv_relax_ucc} and that the elements of $\Gamma_1$ are integers ranging from $1$ to $|\Gamma_1|$ and the elements of $\Gamma_2$ goes from $|\Gamma_1|+1$ to $K=|\Gamma_1|+|\Gamma_2|$ as mentioned in section \ref{subsubsec:UCC_circuit}.
    Using this indexing of generators and hence the associated parameters, the terms $\mathcal{A}_{\mu,l}(\vec{0})$ and $\beta_{\mu,l}(\vec{0})$ introduced in Eqs.~(\ref{eq:recipe_A},\ref{eq:beta_recipe}) can be written as follows $\forall 1 \leq l \leq L$ and $\forall 1 \leq \mu \leq K$,
   
    \begin{align}
    \mathcal{A}_{\mu,l}(\vec{0}) &= 4 \left[ \omega^{(\rm max)}_{\mu} \right]^2 \left( \sum_{\nu=1}^{l-1} \sum_{k=1}^{K} \left[ \omega^{(\rm eff)}_{k}(\vec{0}) \right]^2 + \sum_{k=1}^{\mu -1} \left[ \omega^{(\rm eff)}_{k}(\vec{0}) \right]^2 \right) \\
    &\quad + \sum_{k=\mu + 1}^{K} \left[ \widetilde{\omega}^{(\rm eff)}_{\mu,k}(\vphi) \right]^2  + \sum_{\nu=l+1}^L \sum_{k=1}^{K} \left[ \widetilde{\omega}^{(\rm eff)}_{\mu,k}(\vphi) \right]^2 \;, \label{eq:_ucc_gen}\\
     \beta_{\mu,l}(\vec{0}) &=\begin{cases}
       \frac{2 \left[\omega^{(\rm max)}_{1,1}\right]^2  \left[\omega^{\rm (eff)}_{1}(0)\right]^4}{3} \; {\rm if }\; \mu,l=1,1 \;,\\ \frac{32( \omega_{\mu}^{\rm (max) })^{6}  \norm{O}^2}{3}\; {\rm otherwise }\; \;.
    \end{cases}  
        \end{align}

Using the value of the maximal frequencies  in Eq.~\eqref{eq:max_freq_1_ucc} and Eq.~\eqref{eq:max_freq_2_ucc} and  the effective frequencies upper bounds in Eq.~\eqref{eq:eef_freqs_1_ucc} and Eq.~\eqref{eq:eef_freqs_2_ucc}, we can also upper bounded the terms $\mathcal{A}_{\mu,l}(\vec{0})$ and $\beta_{\mu,l}(\vec{0})$ $\forall 1 \leq l \leq L$ and $\forall 1 \leq \mu \leq K$ as 
\begin{align}
     \mathcal{A}_{\mu,l}(\vec{0}) & \leq 16 \norm{O} \left[ 4(l-1)K + 4\mu -4+K(L-l) + K - \mu\right]\\
     & = 16 \norm{O} \left[ KL + 3K(l-1) +3\mu - 4\right]\\
     & \leq 2^6 \norm{O} KL \label{eq:A_ucc_bound}\\
     \beta_{\mu,l}(\vec{0}) &\leq \frac{2^{11}}{3} \norm{O}^2
\end{align}
where in the first inequality we substitute the upper-bounds of the effective frequencies, and in the second one we group the terms and in Eq.~\eqref{eq:A_ucc_bound} we upper-bound all the remaining negative terms with zero to make the analysis clearer further down the line. The upper-bound on $\beta_{\mu,l}(\vec{0})$ is straightforward.

Moreover, the loss second derivative at zero w.r.t the parameter $\th_{\mu,l}$ denoted by $c_{\mu,l}(\vec{0})$ can be upper bounded as
\begin{align}
    c_{\mu,l}(\vec{0}) &= \left|\Tr\left[\rho \left.\frac{\partial^2 [U(\thv_{\rm Relaxed})^\dagger O U(\thv_{\rm Relaxed})]}{\partial \th_{\mu,l}^2}\right|_{\thv = \vec{0}} \right]\right|\\
    &\leq \|\rho\|_1 \norm{\left.\frac{\partial^2 [U(\thv_{\rm Relaxed})^\dagger O U(\thv_{\rm Relaxed})]}{\partial \th_{\mu,l}^2}\right|_{\thv = \vec{0}}}\\
    & \leq 4 \norm{O}
\end{align}
where we used the Hölder inequality in the first inequality and the result from Eq.~\eqref{eq:eff_freq_Hpq} in the last inequality.

Hence, we can redefine a smaller patch size than the one in Eq.~\eqref{eq:patch_ucc_relax} with the same guarantees of substantial loss variance.
\begin{align}\label{eq:final_patch_relax_gen_ucc}
     (r_{\rm patch}^{\rm Relaxed})^2 &= \min_{\substack{\mu \in \Lambda\\ l \in \{1,\dots,L\}}} \frac{9c_{\mu,l}(\vec{0})^2}{2^{11} \norm{O}^2  KL + 2^{14}\norm{O}^2} \\
     & \in \Theta\left( \frac{1}{M \norm{O}^2} \right)
\end{align}
where we used the assumption in Eq.~\eqref{eq:cond_curv_relax_ucc}.

Moreover, the variance lower bound within the region of size $r_{\rm patch}^{\rm Relaxed}$ is given according to Eq.~\eqref{eq:Var_recipe_th3} by
\begin{align}\label{eq:final_var_relax_gen_ucc}
    \Var_{\thv_{\rm Relaxed} \sim \uni(\vec{0},r_{\rm patch}^{\rm Relaxed})}[\mathcal{L}(\thv_{\rm Relaxed})] \in \Omega\left((r_{\rm patch}^{\rm Relaxed})^4\sum_{l=1}^L\sum_{\mu \in \Lambda} c^2_{\mu,l}(\vec{0})\right) = \Omega\left(\frac{L |\Lambda|}{M^2 \norm{O}^4}\right)\;.
\end{align}

\paragraph*{\underline{Trotterized UCCSD ansatz with an arbitrary  observable.}}\label{para:ucc_trotter_proof}
Here, we study the region with guaranteed substantial loss variance around zero for the loss function $\mathcal{L}(\thv_{\rm Trotter})$ based on the Trotter version of the UCCSD anstaz in Eq.~\eqref{eq:trotter_ucc_ham} as described in section \ref{para:UCCSD_arbitrary_setting}. Hence, we henceforth adapt Theorem \ref{th:var_formal_cor} to prove the guarantees on the region with gradients centered around zero in Eq.~\eqref{eq:ucc_trotter_patch} and Eq.~\eqref{eq:ucc_trotter_var}.

The analysis for the Trotterized UCCSD ansatz follows from the general recipe in the introduction of Appendix \ref{appendix:architectures} adapting Theorem \ref{th:var_formal_cor}. Specifically, the set  of parameters for which the loss second derivative scale as a constant and thus will be contributing to the variance lower bound is already given in Eq.~\eqref{eq:cond_curv_trotter_ucc} by the single parameter index $\bar{\mu} \in \Gamma_1 \cup \Gamma_2$. Moreover, the maximal frequencies computed in the relaxed setting in Eq.~\eqref{eq:max_freq_1_ucc} and Eq.~\eqref{eq:max_freq_2_ucc} will be key in computing the maximal frequencies associated with a parameter $\th_\mu$ associated with a generator $H_\mu$ for $\mu \in \Gamma_1 \cup \Gamma_2$ repeated across all of the $L$ trotter layers (See Eq.~\eqref{eq:gamma_1} and Eq.~\eqref{eq:gamma_2} for the definition of $\Gamma_1$ and $\Gamma_2$). 
Indeed, the maximal frequency $\omega^{(\rm max)}_\mu$ associated to the parameter $\th_\mu\;, \forall \mu \in \Gamma_1 \cup \Gamma_2$ can be written according to the definition in Eq.~\eqref{eq:max_freqs_E} as
\begin{align}
    \omega^{(\rm max)}_\mu = \sum_{l=1}^L \omega^{(\rm max)}(H_\mu) = 2L \;,
\end{align}
where we substituted  $\omega^{(\rm max)}(H_\mu)$ by its value already computed in Eq.~\eqref{eq:max_freq_1_ucc} and Eq.~\eqref{eq:max_freq_2_ucc}. 

Hence, we now invoke from Eq.~\eqref{eq:rpatch_th4_recipe} the expression of $r_{\rm patch}^{\rm Trotter}$, i.e. the size of the region around zero with guaranteed substantial loss variance $\mathcal{L}(\thv_{\rm Trotter})$.
\begin{align}\label{eq:rpatch_trotter_ucc1}
            (r_{\rm patch}^{\rm Trotter})^2 = \frac{3c^2_{\bar{\mu}}(\vec{\vec{0}})}{8(2c_{\bar{\mu}}(\vec{0})\gamma_{\bar{\mu}} + \widetilde{\beta}_{\bar{\mu}})} \;,
        \end{align}
        where $\widetilde{\beta}_{\bar{\mu}}$ and $\gamma_{\bar{\mu}}$ are given by
        \begin{align}
           \gamma_{\bar{\mu}}&= \frac{8}{3}\norm{O}\left[\omega^{(\rm max)}_{\bar{\mu}}\right]^2\sum_{\substack{k=1 \\ k\neq \bar{\mu}}}^{K}\left[\omega^{(\rm max)}_{k}\right]^2 =\frac{2^7}{3}\norm{O} L^4(K-1) \label{eq:eq1}\\
           \widetilde{\beta}_{\bar{\mu}}&=\frac{2^{5}}{3} \left[\omega^{(\rm max)}_{\bar{\mu}}\right]^6 \norm{O}^2 = \frac{2^{11}}{3} L^6 \norm{O}^2\;. \label{eq:eq2}
        \end{align}
        where we just substituted the values of the maximal frequencies computes above.

        Consequently, by combining the scaling of $c_{\bar{\mu}}(\vec{0})$ from Eq.~\eqref{eq:cond_curv_trotter_ucc} and the expressions of $ \gamma_{\bar{\mu}}$ and $\widetilde{\beta}_{\bar{\mu}}$ from Eq.~\eqref{eq:eq1} and Eq.~\eqref{eq:eq2}  in Eq.~\eqref{eq:rpatch_trotter_ucc1}, we finally obtain
        \begin{align}\label{eq:final_gen_ucc_trotter_patch}
            (r_{\rm patch}^{\rm Trotter})^2 \in \Theta\left(\frac{1}{L^2 \norm{O} (K+ \norm{O})}\right)\:.
        \end{align}
     Moreover, the variance lower bound within the region of size $r_{\rm patch}^{\rm Trotter}$ is given according to Eq.~\eqref{eq:Var_recipe_th4} by
\begin{align}\label{eq:final_gen_ucc_trotter_var}
    \Var_{\thv_{\rm Trotter} \sim \uni(\vec{0},r_{\rm patch}^{\rm Trotter})}[\mathcal{L}(\thv_{\rm Trotter})] \in \Omega\left((r_{\rm patch}^{\rm Trotter})^4 c^2_{\bar{\mu},l}(\vec{0})\right) = \Omega\left(\frac{1}{\norm{O}^2 (K+ \norm{O})^2}\right)\;.
\end{align}   

\paragraph*{\underline{Toy example proofs:}}

In this section, we consider the observable
\begin{equation}\label{eq:obs_ucc_num}
    O=\sum_{i=1}^n \sigma_z^{(i)} \otimes  \sigma_z^{(i+1)} \;,
\end{equation}
 such that $\sigma_z^{(n)} \otimes  \sigma_z^{(n+1)} := \sigma_z^{(n)} \otimes  \sigma_z^{(1)}$ and apply the results from sections \ref{para:ucc_relax_proof} and \ref{para:ucc_trotter_proof} for the specific observable $O$ in Eq.~\eqref{eq:obs_ucc_num} and the initial state $\ket{\psi}$ given by
 \begin{equation}\label{eq:state_fock}
     \ket{\psi} = \ket{1}^{\otimes \frac{n}{2}} \otimes \ket{0}^{\otimes \frac{n}{2}}
\;.
  \end{equation}
Specifically, we mainly focus on verifying that the conditions in Eq.~\eqref{eq:cond_curv_relax_ucc} is satisfied from some set $\Lambda \subset \Gamma_1 \cup \Gamma_2$ that the condition in Eq.~\eqref{eq:cond_curv_trotter_ucc} is fulfilled by some $\bar{\mu} \in \Gamma_1 \cup \Gamma_2$. Moreover, we show that in the relaxed setting, we can get a bigger patch with guarantees that the one obtained for a generic observable in Eq.~\eqref{eq:final_patch_relax_gen_ucc} mainly due to the observable $O$ defined in  Eq.~\eqref{eq:obs_ucc_num} being geometrically local.

\paragraph*{In the Relaxed setting,} we derive a tighter upper bound on the effective frequencies than the one obtained for an arbitrary observable in Eq.~\eqref{eq:eef_freqs_1_ucc}
 and Eq.~\eqref{eq:eef_freqs_2_ucc}. Precisely, using the explicit expression of the observable $O$ given in Eq.~\eqref{eq:obs_ucc_num}, we can explicitly compute the nested commutators for the generator $H_{pq};, p>q$ appearing in the effective frequencies expression in Eq.~\eqref{eq:eff_freq_Hpq}. 
 Let us first start by evaluating the commutator of the form $[H_{pq},O]$, i.e.
 \begin{align}
[H_{pq},O] &= \left[\frac{\sigma_x^{(q)}\sigma_y^{(p)}-\sigma_y^{(q)}\sigma_x^{(p)}}{2},\sigma_z^{(q-1)}\sigma_z^{(q)}\right] + \left[\frac{\sigma_x^{(q)}\sigma_y^{(p)}-\sigma_y^{(q)}\sigma_x^{(p)}}{2},\sigma_z^{(p)}\sigma_z^{(p+1)}\right]\\
&= \frac{1}{2} \left( \sigma_z^{(q-1)}[\sigma_x^{(q)},\sigma_z^{(q)}] \sigma_y^{(p)} - \sigma_z^{(q-1)}[\sigma_y^{(q)},\sigma_z^{(q)}] \sigma_x^{(p)} +  \sigma_x^{(q)}[\sigma_y^{(p)},\sigma_z^{(p)}] \sigma_z^{(p+1)} - \sigma_y^{(q)}[\sigma_x^{(p)},\sigma_z^{(p)}] \sigma_z^{(p+1)}\right)\\
&= -i \left( \sigma_z^{(q-1)}\sigma_y^{(q)}\sigma_y^{(p)} + \sigma_z^{(q-1)}\sigma_x^{(q)}\sigma_x^{(p)} - \sigma_x^{(q)}\sigma_x^{(p)} \sigma_z^{(p+1)} -\sigma_y^{(q)}\sigma_y^{(p)} \sigma_z^{(p+1)}\right)\;.
\end{align}

Now, we apply again the commutator w.r.t the same generator $H_{pq}$ and obtain,

\begin{align}
    [H_{pq},[H_{pq},O] ] =& [H_{pq}, \sigma_z^{(q-1)}\sigma_y^{(q)}\sigma_y^{(p)} + \sigma_z^{(q-1)}\sigma_x^{(q)}\sigma_x^{(p)} - \sigma_x^{(q)}\sigma_x^{(p)} \sigma_z^{(p+1)} -\sigma_y^{(q)}\sigma_y^{(p)} \sigma_z^{(p+1)}]\\
    =& -\frac{i}{2} (\sigma_z^{(q-1)}[\sigma_x^{(q)},\sigma_y^{(q)}](\sigma_y^{(p)})^2 - \sigma_z^{(q-1)}(\sigma_y^{(q)})^2[\sigma_x^{(p)},\sigma_y^{(p)}])\\
    &-\frac{i}{2} (\sigma_z^{(q-1)}(\sigma_x^{(q)})^2[\sigma_y^{(p)},\sigma_x^{(p)}] - \sigma_z^{(q-1)}[\sigma_y^{(q)},\sigma_x^{(q)}](\sigma_y^{(p)})^2)\\
    & +\frac{i}{2} ((\sigma_x^{(q)})^2[\sigma_y^{(p)},\sigma_x^{(p)}]\sigma_z^{(p+1)}-[\sigma_y^{(q)},\sigma_x^{(q)}](\sigma_x^{(p)})^2 \sigma_z^{(p+1)}) \\
    &+\frac{i}{2} ([\sigma_x^{(q)},\sigma_y^{(q)}]\1 \sigma_z^{(p+1)}-\1[\sigma_x^{(p)},\sigma_y^{(p)}] \sigma_z^{(p+1)})\\
     =& (\sigma_z^{(q-1)}\sigma_z^{(q)} - \sigma_z^{(q-1)} \sigma_z^{(p)}) - (\sigma_z^{(q-1)}\sigma_z^{(p)} - \sigma_z^{(q-1)}\sigma_z^{(q)})\\
     &+ (\sigma_z^{(p)}\sigma_z^{(p+1)}-\sigma_z^{(q)} \sigma_z^{(p+1)}) -(\sigma_z^{(q)} \sigma_z^{(p+1)}-\sigma_z^{(p)} \sigma_z^{(p+1)})\\
     =& 2\sigma_z^{(q-1)} (\sigma_z^{(q)} - \sigma_z^{(p)}) + 2 (\sigma_z^{(p)} - \sigma_z^{(q)}) \sigma_z^{(p+1)}\;.\label{eq:com_ucc}
\end{align}
where in the first equality, we explicitly substituted $H_{pq}$ given in Eq.~\eqref{eq;H_pq}. In the second inequality we did the same thing, and expanded the expression by pulling out of the commutators all those terms that did not affect them. In the second to last equation we compute the commutators and use the fact that $\sigma_i^2 = \1$ to simplify the expression. Finally we group together the remaining terms.

Hence, the effective frequencies of the form $\omega^{(\rm eff)}_{\mu_1}(\vec{0})\;, \mu_1 \rightarrow (p,q) \text{ with } p>q$ are nothing but the infinity norm of 8 commuting Pauli strings, i.e.
\begin{align}\label{eq:custom_eff1}
    [\omega^{(\rm eff)}_{\mu_1}(\vec{0})]^2 = 8\;.
\end{align}

Similarly, the effective frequencies of the form $\omega^{(\rm eff)}_{\mu,\mu'} \;, \mu,\mu' \in \Gamma_1 \cup \Gamma_2$ defined in Eq.~\eqref{eq:eff_freq_Hmu_Hmu'} can be upper bounded by a constant that is independent of the observable norm by using the same argument. Broadly, since the generators $H_{\mu}$ acts non trivially on at most 4 qubits and the observable $O$ contains exactly 2 Pauli terms acting non trivially on a specific qubit  (e.g. for qubit $i$ we have $\sigma_z^{(i-1)}\otimes\sigma_z^{(i)}$ and $\sigma_z^{(i)}\otimes\sigma_z^{(i+1)}$), each rotation might affect at most $2\times 4=8$ Pauli terms in $O$ (i.e. there are at most $8$ Pauli terms in $O$ that might not commute with a generator $H_\mu$). Let us define $\tilde{O}_\mu$ as the sum of Pauli terms in $O$ that does not commute with $H_\mu$ such that $[H_\mu,O]=[H_\mu,\tilde{O}_\mu]$. Therefore, we have $\norm{\tilde{O}_\mu}\leq 8$ and the effective frequencies of the form $\omega^{(\rm eff)}_{\mu,\mu'} \;, \mu,\mu' \in \Gamma_1 \cup \Gamma_2$ are upper bounded as follows

\begin{align}
    \left[\widetilde{\omega}_{\mu,\mu'}^{\rm (eff)}(\vec{0})\right]^2&=\norm{[H_{\mu'},[H_{\mu'},[H_\mu,[H_\mu,O]]]]}\\
    &=\norm{[H_{\mu'},[H_{\mu'},[H_\mu,[H_\mu,\tilde{O}_{\mu}]]]]}\\
    &\leq (2\norm{H_{\mu'}})^2(2\norm{H_\mu})^2\norm{\tilde{O}_\mu}\\
    &\leq 128\;. \label{eq:custom_eff2}
\end{align}
where the first equality is obtained from Eq.~\eqref{eq:eff_freq_Hmu_Hmu'} and the first inequality is obtained by applying the property $\norm{[A,[A,B]]} \leq 4 \norm{A}^2 
 \norm{B}$. In the last equality we simply substitute the upper-bounds on $\norm{\tilde{O}_\mu}$ explained right before the equation, and $\norm{H_{\mu'}}$ in Eq.~\eqref{eq:bounded_norm_gen_ucc}.

 Consequently, we can obtain in this setting a tighter upper bound on the term $\mathcal{A}_{\mu,l}(\vec{0})$ defined in Eq.~\eqref{eq:_ucc_gen} than the bound in Eq.~\eqref{eq:A_ucc_bound}. Precisely, by combining the new effective frequencies bounds in 
 Eq.~\eqref{eq:custom_eff1} and Eq.~\eqref{eq:custom_eff2}, we obtain the following scaling of  $\mathcal{A}_{\mu,l}(\vec{0}) \;, \forall \mu \in \Gamma_1 \cup \Gamma_2 ,  1 \leq l \leq L$.
 \begin{align}\label{eq:A_UB_custom}
     \mathcal{A}_{\mu,l}(\vec{0}) \in \mathcal{O}(KL) \;.
 \end{align}

Now, we focus on computing the loss function second derivatives with respect to all the parameters $\th_{\mu_1,l} \;, \mu_1 \in \Gamma_1 , 1 \leq l \leq L$.

Let us first show that the loss second derivative evaluated at zero denoted by $ c_{\mu_1,l}(\vec{0})$ boils down to computing the overlap between the initial state given in Eq.~\eqref{eq:state_fock} and the nested commutators $[H_{\mu_1},[H_{\mu_1},O]]$ in Eq.~\eqref{eq:com_ucc}.
\begin{align}
     c_{\mu_1,l}(\vec{0}) &= \left| \Tr\left[|\psi\rangle\langle\psi| \left.\frac{\partial^2 [U^\dagger(\thv_{\rm Relaxed}) O U(\thv_{\rm Relaxed})]}{\partial\th_{(\mu_1,l)}^2}\right|_{\thv_{\rm Relaxed}= \vec{0}}\right] \right|\\
    &= \left| \Tr\left[|\psi\rangle\langle\psi| \left.\frac{\partial^2 [e^{i \th_{(\mu_1,l)} H_{\mu_1}} O e^{-i \th_{(\mu_1,l)} H_{\mu_1}}]}{\partial\th_{(\mu_1,l)}^2}\right|_{\th_{\mu_1,l}= 0}\right] \right|\\
    &= \left| \Tr\left[|\psi\rangle\langle\psi| [H_{\mu_1},[H_{\mu_1},O]]\right] \right|\;. \label{eq:curv_ucc_relax_num} 
\end{align}

Here we recall that the initial state is given by $ \ket{\psi} = \ket{1}^{\otimes \frac{n}{2}} \otimes \ket{0}^{\otimes \frac{n}{2}}$. Thus, computing the overlap with the state in Eq.~\eqref{eq:curv_ucc_relax_num} will depend on the relative position of $(p,q) \rightarrow \mu_1$ compared to $n/2$. Precisely, the loss second derivative $ c_{\mu_1,l}(\vec{0})$ is non zero only if $q \leq n/2$ and $p>n/2$ and will be in this case constant in the system size, i.e. 
\begin{align} 
     c_{\mu_1,l}(\vec{0}) &= \left| \Tr\left[|\psi\rangle\langle\psi| [H_{\mu_1},[H_{\mu_1},O]]\right] \right|\\
     &=2 \left| \Tr\left[|\psi\rangle\langle\psi| \left(\sigma_z^{(q-1)} (\sigma_z^{(q)} - \sigma_z^{(p)}) +  (\sigma_z^{(p)} - \sigma_z^{(q)}) \sigma_z^{(p+1)}\right)\right] \right|\\
     &= 8 \label{eq:aaaaa}
\end{align}
Therefore, we choose the subset $\Lambda$ defined in Eq.~\eqref{eq:cond_curv_relax_ucc} to be $\Lambda=\{(p,q)\;,q \leq n/2 < p\}$ of size $|\Lambda|= n^2/4$.

Finally, given that we identified the subset $\Lambda$ of size $|\Lambda| \in \Theta( n^2)$ verifying the assumption in Eq.~\eqref{eq:cond_curv_relax_ucc} and provided a tighter upper bound 
on $\mathcal{A}_{\mu,l}(\vec{0})$ in Eq.~\eqref{eq:A_UB_custom}, we can plug in these results in Eq.~\eqref{eq:final_patch_relax_gen_ucc} and Eq.~\eqref{eq:final_var_relax_gen_ucc} and obtain that the patch around identity with guaranteed gradients for the UCCSD relaxed version , the observable $O$ in Eq.~\eqref{eq:obs_ucc_num} and the initial state in Eq.~\eqref{eq:state_fock} scales as
\begin{align}
    r_{\rm patch}^{\rm Relaxed}&\in\Theta\left(\frac{1}{\sqrt{KL + n^2}}\right)
\end{align}
and for that region, the variance is lower bounded as
\begin{align}\label{eq:var_ucc_proof}
    \Var_{\thv_{\rm Relaxed} \sim \uni(\vec{0},r_{\rm patch}^{\rm Relaxed})}[\mathcal{L}(\thv_{\rm Relaxed})]  &\in \Omega\left(\frac{Ln^2}{(KL+ n^2)^2}\right) \;, 
\end{align}
where we used $\norm{O}=n$.

\paragraph*{In the Trotter setting,}

we simply show that $\bar{\mu} = 1$ corresponding to the generator $H_{(2,1)}$ acting non trivially on the first and second qubits, as defined in Eq.~\eqref{eq:eff_freq_Hpq}, fulfills the assumption in Eq.~\eqref{eq:cond_curv_trotter_ucc} for the loss function $\mathcal{L}(\thv_{\rm Trotter})$ using the Trotter UCCSD circuit in Eq.~\eqref{eq:trotter_ucc_ham}, the observable $O$ in Eq.~\eqref{eq:obs_ucc_num} and the initial state in Eq.~\eqref{eq:state_fock}.

Indeed , for the loss second derivative  w.r.t $\th_1$ evaluated at zero, we follow the same analysis as done for the Trotter HVA (See Eq.~\eqref{eq:trotter_curv_HVA1} - Eq.~\eqref{eq:trotter_curv_HVA2} ) which leads to
\begin{align}
    c_1(\vec{0}) = L^2  \left|\Tr[\rho [H_1,[H_1,O]] ] \right|  \in \Theta(L^2 )\;.
\end{align}
where we used Eq.~\eqref{eq:aaaaa} showing that $\left|\Tr[\rho [H_1,[H_1,O]] ] \right| \in \Theta(1)$. Consequently, we retrieve the results in Eq.~\eqref{eq:final_gen_ucc_trotter_patch} and Eq.~\eqref{eq:final_gen_ucc_trotter_var}.

\section{Fourier expansion of the loss function}\label{App:Fourier_expansion}

As established in the main text, the characteristic (maximal and effective) frequencies of a parameterized quantum circuit can tightly constrain the size of the parameter region over which the variance of a loss function remains non-negligible. Specifically, we have proven that the patch size, $r_{\rm patch}$, is inversely proportional to sums of these characteristic frequencies. Here, we make this connection explicit by expanding the loss function in a Fourier basis and clarifying how these Fourier frequencies contribute to the effective and maximal frequencies.

Throughout this appendix, we consider a loss function $\mathcal{L}(\boldsymbol{\theta})$ of the form introduced in Eq.~\eqref{eq:loss}, with $\nHam$ generators $\{H_l\}_{l=1}^\nHam$ and $\nparams$ parameters $\{\theta_l\}_{l=1}^\nparams$. We describe the Fourier expansion in a discrete Fourier basis determined by the eigenvalue spectra of the $H_l$'s \cite{schuld2021effect}, first for the case $\nparams = \nHam$, where each generator $H_l$ is associated with a unique parameter $\theta_l$. We then explain how these results generalize when parameters are shared across multiple generators.

We begin by decomposing each generator \(H_l\) in its own eigenbasis:
\begin{equation}
    H_l = \sum_{\lambda_i^{(l)} \in \mathrm{Spec}(H_l) } \lambda_i^{(l)} P_i^{(l)} , \quad P_i^{(l)} = \sum_{j=1}^{dim(E_{\lambda_i^{(l)}}(H_l))} |\lambda_i^{(l)}, j\rangle \langle \lambda_i^{(l)}, j|.
\end{equation}

Here, $\mathrm{Spec}(H_l)$ is the spectrum (eigenvalues) of $H_l$, and $P_i^{(l)}$ is the projector onto the corresponding eigenspace $E_{\lambda_i^{(l)}}(H_l)$, accounting for degeneracy. By considering all pairwise differences of the eigenvalues of $H_l$, we denote the set of \emph{distinct} frequencies
\begin{equation}
    \Omega_l := \{\lambda_j^{(l)} - \lambda_i^{(l)} \;\big|\; \lambda_i^{(l)}, \lambda_j^{(l)} \in\mathrm{Spec}(H_l) \}.
\end{equation}
Note that many pairs $(j,i)$ can yield the same frequency $\omega_{l}$, and to keep track of frequency redundancy we label these pairs separately by the set 
\begin{equation}
 R(\omega_{l}) := \{(j,i)  \;\big|\;\lambda_j^{(l)} - \lambda_i^{(l)} = \omega_{l}\} , \quad \forall \omega_{l} \in \Omega_l. 
\end{equation}

A straightforward calculation shows that the unitary time evolution of an operator \(A\) under \(H_l\) in the Heisenberg picture admits a discrete Fourier expansion:
 \begin{align}
        e^{i\theta_l H_l} A e^{-i\theta_l H_l} &= \sum_{i,j} e^{-i\theta_l (\lambda_j^{(l)} - \lambda_i^{(l)})} P_i^{(l)} A P_j^{(l)}\\
        &= \sum_{\omega_{l} \in \,\Omega_l}   e^{-i\theta_l \omega_{l}} 
        \Bigl(\sum_{(j,i) \in R(\omega_{l})}P_i^{(l)} A P_j^{(l)}\Bigr)\\
        &= \sum_{\omega_{l} \in \,\Omega_l}  e^{-i\theta_l \omega_{l}} P_{\omega_{l}}\left(A\right),\label{eq:Fourier_sandwich}
\end{align}
where we have introduced the superoperator $P_{\omega_{l}}(\cdot) := \sum_{(j,i) \in R(\omega_{l})} P_i^{(l)} (\cdot) P_j^{(l)}$.

We now apply Eq.~\eqref{eq:Fourier_sandwich} repeatedly to expand the loss function $\mathcal{L}(\thv)=\Tr\!\bigl[\rho_0\,U^\dagger(\thv)\,O\,U(\thv)\bigr]$
in a multi-dimensional Fourier series. Recall that $U(\boldsymbol{\theta}) = \prod_{l=1}^{\nHam} \Bigl(V_l\,e^{-i\theta_lH_l}\Bigr)$, and that  $\rho_{\overline{l}} :=  U_{l+1}(\theta_{l+1}) \dots U_{M}(\theta_{M} )  \rho_0 U^{\dagger}_{M}(\theta_{M}) \dots U_{l+1}^{\dagger}(\theta_{l+1})$, the notation introduced in~\eqref{eq:rho-l-bar}. 
The loss function can then be written as
\begin{align}
        \mathcal{L}(\thv) &= \Tr[\rho_0 U^{\dagger}(\thv) O U(\thv)]\label{eq:Fourier_decomp_start}\\
        &= \Tr[\rho_{\overline{1}} e^{i\theta_1 H_1} V_1^{\dagger}O V_1 e^{-i\theta_1 H_1} ]\\
        &= \Tr[\rho_{\overline{1}} \sum_{\omega_{1} \in \,\Omega_1} e^{-i\theta_1 \omega_{1}}  P_{\omega_{1}}(V_1^\dagger O V_1)]\\
      &= \sum_{\omega_{1} \in \,\Omega_1} e^{-i\theta_1 \omega_{1}}\Tr[\rho_{\overline{1}}   P_{\omega_{1}}(V_1^\dagger O V_1)]\\
        &\; \; \vdots\\
        &=  \sum_{\omega_{1} \in \,\Omega_1} \dots  \sum_{\omega_{M} \in \,\Omega_\nHam} e^{-i (\th_1 \omega_{1}+\dots + \th_\nHam \omega_{M})}   \Tr[\rho_0 P_{\omega_{M}}(V_M^\dagger \dots P_{\omega_{1}}(V_1^\dagger O V_1)\dots V_M)]\\
          &=  \sum_{\omega_{1} \in \,\Omega_1} \dots  \sum_{\omega_{M} \in \,\Omega_\nHam}  e^{-i (\th_1 \omega_{1}+\dots + \th_\nHam \omega_{M})}   \Tr[\rho_0 O_{\boldsymbol{\omega}}]\label{eq:Fourier_decomp_end}\\
        &= \sum_{\boldsymbol{\omega} \in \Omega_1 \times \dots \times\Omega_M} a_{\boldsymbol{\omega}}  e^{-i \thv^T \boldsymbol{\omega} }, \label{eq:cost_expansion}
\end{align}
where we introduced  $a_{\boldsymbol{\omega}} := \Tr[\rho_0 O_{\boldsymbol{\omega}}]$ and  $O_{\boldsymbol{\omega}} := P_{\omega_{M}}(V_M^\dagger \dots P_{\omega_{1}}(V_1^\dagger O V_1)\dots V_M)$. This is the Fourier expansion of the loss function when each generator is associated with a unique parameter.

In the case of circuits with correlations ($\nparams < \nHam$), the above derivation still holds after a suitable regrouping. 
Specifically, recall that $\mathcal{S}:\{1,\dots,\nHam\}\to\{1,\dots,\nparams\}$ maps each generator index $l$ to the parameter index $\mathcal{S}(l)$. 
Then the dot product $\thv^T\boldsymbol{\omega}$ in Eq.~\eqref{eq:cost_expansion} can be grouped by distinct parameters as 
\begin{equation}
    \boldsymbol{\theta}^T\,\boldsymbol{\omega} 
    \;=\;
    \sum_{h=1}^{\nparams} 
    \theta_h\,
    \Bigl(\sum_{\,l \,\in\, \mathcal{S}^{-1}(h)} \omega_l\Bigr)
    \;=:\;
    \sum_{h=1}^{\nparams} 
    \theta_h\,
    \omega_h,
\end{equation}
where $\omega_h := \sum_{l \in \mathcal{S}^{-1}(h)}\omega_{l}$.
Using this grouping of the frequencies shared among the same parameter, one can define the spectrum associated to the parameter $\th_h$, denoted by $\Omega_h$, as
\begin{equation}
    \Omega_h :=  \{ \omega_h = \sum_{l \in \mathcal{S}^{-1}(h)} \omega_{l}  \;\Big|\; \omega_l \in \Omega_l \}.
\end{equation}

Hence, the \emph{maximal frequency} $\omega^{(\rm max)}_{h}$ in $\Omega_h$ is given by
\begin{equation}\label{eq:max_freq_fourier}
    \omega^{(\rm max)}_{h} = \sum_{l \in \mathcal{S}^{-1}(h)} \omega_{l}^{\rm (max) } \; ,
\end{equation}
where $\omega_l^{\mathrm{(max)}}$ is the maximal frequency in the spectrum of the individual generator $H_l$. This expression is Eq.~\eqref{eq:max_freqs} from the main text, which states that the maximal frequency associated to a given parameter $\theta_h$ is simply the sum of the maximal frequencies of the associated (correlated) generators.

\noindent\underline{\textit{Connection to effective frequencies}}:

We now illustrate how the \emph{effective frequencies} from Theorem~\ref{th:var_formal} relate to the Fourier expansion frequencies of the loss function under the setting that $\nparams = \nHam$. 
Specifically, we show that the effective frequencies can be seen as the weighted sum of the underlying Fourier frequencies.

First, we consider the back-propagated observable in the Heisenberg picture evaluated at $\thv + \vphi$, and expand it in the same discrete Fourier basis as before (see Eqns.~\eqref{eq:Fourier_decomp_start}-\eqref{eq:Fourier_decomp_end}):
\begin{align}
    U^\dagger(\thv + \vphi)OU(\thv + \vphi)&=  \left(\prod_{l=1}^M 
 V_l e^{-i (\th_l+ \phi_l) H_l}\right)^\dagger O \left(\prod_{l=1}^M 
 V_l e^{-i (\th_l+ \phi_l) H_l}\right)\\
 &= \left(\prod_{l=1}^M 
 (V_l e^{-i \phi_l H_l}) e^{-i \th_l H_l}\right)^\dagger O \left(\prod_{l=1}^M 
 (V_le^{-i \phi_l H_l}) e^{-i \th_l H_l}\right)\\
 &=   \sum_{\omega_{1} \in \,\Omega_1} \dots  \sum_{\omega_{M} \in \,\Omega_\nHam}  e^{-i (\th_1 \omega_{1}+\dots + \th_\nHam \omega_{\nHam})  } \left(P_{\omega_{M}}(e^{i \phi_M H_M}V_M^\dagger \dots P_{\omega_{1}}(e^{i \phi_1 H_1}V_1^\dagger O V_1 e^{-i \phi_1 H_1})\dots V_M e^{-i \phi_M H_M})\right)\\
 &= \sum_{\boldsymbol{\omega} \in \Omega_1\times \dots \times \Omega_\nHam}   e^{-i \thv^T \boldsymbol{\omega}  }  O_{\boldsymbol{\omega}}(\vphi),\label{eq: backprop_fourier}
\end{align}
where we have introduced  
\begin{equation}
O_{\boldsymbol{\omega}}(\vphi) := P_{\omega_M}( e^{i \phi_M H_M}V_M^\dagger \dots P_{\omega_1}(e^{i \phi_1 H_1}V_1^\dagger O V_1 e^{-i \phi_1 H_1})\dots V_M e^{-i \phi_M H_M}).    
\end{equation}

Recall that the \emph{effective} frequencies $\omega_{\mu}^{(\mathrm{eff})}(\vec{\phi})$ and $\widetilde{\omega}_{(\mu,k)}^{(\mathrm{eff})}(\vec{\phi})$ introduced in the main text appear when taking second or fourth partial derivatives (with respect to $\theta_\mu$ and $\theta_k$) of this back-propagated observable for some fixed $\vphi$. Specifically,
\begin{align}
    (\omega_{\mu}^{\rm (eff) }(\vec{\phi}) )^2
  & =\norm{ \left.\frac{\partial^2 [U(\thv)^\dagger O U(\thv)]}{\partial \th_\mu^2}\right|_{\thv= \vec{\phi}}} = \norm{ \left.\frac{\partial^2 [U(\thv+\vphi)^\dagger O U(\thv+\vphi)]}{\partial \th_\mu^2}\right|_{\thv= \vec{0}}},\\
  (\widetilde{\omega}^{\rm (eff)}_{(\mu,k)}(\vec{\phi}))^2 
  &=\norm{ \left.\frac{\partial^4 [U(\thv)^\dagger O U(\thv)]}{\partial \th_\mu^2 \partial \th_k^2}\right|_{\thv= \vec{\phi}}} = \norm{ \left.\frac{\partial^4 [U(\thv+ \vphi)^\dagger O U(\thv + \vphi)]}{\partial \th_\mu^2 \partial \th_k^2}\right|_{\thv= \vec{0}}}.
\end{align}

Substituting in the Fourier decomposition of the back-propagated observable in Eq.~\eqref{eq: backprop_fourier}, we obtain

\begin{align}
    (\omega_{\mu}^{\rm (eff) }(\vec{\phi}) )^2
  &= \norm{\sum_{\boldsymbol{\omega} \in \Omega_1 \times \dots \times \Omega_M} \omega_\mu^2 O_{\boldsymbol{\omega}}(\vphi)},\\
  (\widetilde{\omega}^{\rm (eff)}_{(\mu,k)}(\vec{\phi}))^2 
    &= \norm{\sum_{\boldsymbol{\omega} \in \Omega_1 \times \dots \times \Omega_M} \omega_\mu^2 \omega_k^2 O_{\boldsymbol{\omega}}(\vphi)}.
\end{align}

Here, we want to explicitly highlight how the effective frequencies are, the infinity norm of a sum of the Fourier frequencies, weighted by the different back-propagated observables. Indeed this is clearly what we see in these last equations presented. If all these observables commute, the result is trivial. However, in general these observables do not commute (i.e. $[O_{\vec{\omega_i}},O_{\vec{\omega_j}}]\neq 0$), and thus we see that this $O_{\vec{\omega}}$ act as some sort of weights, and thus the effective frequencies roughly become a weighted sum of the Fourier frequencies, dictated by how the different back-propagated observables interact.

\section{Upper bound on the loss variance} \label{app:upperbound}

\setcounter{proposition}{0}
\begin{proposition}
[Upper bound on the variance]\label{th:upperbound}
Consider a generic loss $\LC(\thv)$ of the form in Eq.~\eqref{eq:loss}. 
Suppose that when $\thv$ is uniformly sampled from the full parameter space $\vol(\vec{\phi},r_{\rm full})$, the average of $\LC(\thv)$ is zero, and its variance over this full landscape is exponentially vanishing in the system size $n$,
\begin{align}
    \Var_{\vtheta \sim \uni(\vphi, r_{\rm full})} \left[ \LC (\vtheta)\right] \in \OC\left( \frac{1}{b^n}\right)  \quad\text{for some } b>1.
\end{align}
Then, for any hypercube $\vol(\vec{\phi}, r)$ with 
\begin{align}
    r \;>\;\frac{r_{\rm full}}{b^{m/n}},
\end{align} 
the variance of $\LC(\thv)$ over this hypercube will also exponentially vanish in $n$, i.e., 
\begin{align}
    \Var_{\vtheta \sim \uni(\vec{\phi}, r)} \left[ \LC (\vtheta)\right]\in\order{\frac{1}{\beta^n}}    
    \quad\text{for some } \beta>1.
\end{align}
Consequently, if the number of parameters $m$ scales linearly with $n$ as $m = c n$, then the variance on any hypercube $\vol(\vec{\phi}, r)$ will vanish exponentially in $n$ provided that  \begin{align}
    r \;>\;\frac{r_{\rm full}}{b^{1/c}}.
\end{align}
\end{proposition}
\begin{proof}
    We prove the statement by relating the second moments of $\LC(\thv)$ over different hypercubes of width $2r\leq 2r_{\rm full}$.
    First, consider the second moment evaluated from a single parameter $\theta_j$, while the other parameters are held fixed.  Let $\theta_j$ be uniformly sampled from $\bigl[\phi_j-r_{\rm full}, \phi_j+r_{\rm full}\bigr]$. Then the second moment of $\LC$ with respect to $\theta_j$ is
\begin{align}
    \Ebb_{\theta_j\sim \uni(\phi_j,\,r_{\rm full})}\bigl[\LC^2(\thv)\bigr]
    \;=\; 
    \frac{1}{2r_{\rm full}}\int_{\phi_j-r_{\rm full}}^{\phi_j+r_{\rm full}} \LC^2(\thv)\,d\theta_j.
\end{align}

\noindent Now decompose this integral into two regions:
an inner interval $\bigl[\phi_j-r,\phi_j+r\bigr]$ (with $r \leq r_{\rm full}$),
and the two outer intervals of total length \(2r_{\rm full}-2r\). 
Since $\LC^2(\thv)\ge0$, one obtains
\begin{align}
    \frac{1}{2r_{\rm full}}
    \int_{\phi_j-r_{\rm full}}^{\phi_j+r_{\rm full}} 
    \LC^2(\thv)\, d\theta_j
    \;\geq\; 
    \frac{1}{2r_{\rm full}}
    \int_{\phi_j-r}^{\phi_j+r}\LC^2(\thv)\,d\theta_j.
\end{align}

 \noindent Multiplying both sides by $\tfrac{r_{\rm full}}{r}$ yields 
\begin{align}
    \frac{r_{\rm full}}{r}\,
    \Ebb_{\theta_j\sim \uni(\phi_j,\,r_{\rm full})}
    \bigl[\LC^2(\thv)\bigr]
    \;\;\geq\;\;
    \Ebb_{\theta_j\sim \uni(\phi_j,\,r)}\bigl[\LC^2(\thv)\bigr].
    \label{eq:bound_onevar}
\end{align}

    We now extend this single-parameter bound on the second moment to a multi-parameter setting. Let $\thv=(\theta_1,\dots,\theta_m)$ be uniformly drawn from a hypercube of width $2r \leq 2 r_{\rm full}$ centered around $\vphi$, i.e., $\thv\sim\uni(\vphi,r)$.  We can write
    \begin{equation}\label{eq:rewrite_expval}
        \Ebb_{\thv\sim\uni(\vphi,r)}[\LC^2(\thv)] = \Ebb_{\th_1\sim\uni(\phi_1,r)}[... [\Ebb_{\th_m\sim\uni(\phi_m,r)}[\LC^2(\thv)]]].
    \end{equation}

\noindent We can apply the bound~\eqref{eq:bound_onevar} on each $\theta_j$. 
For each parameter, the factor $\Ebb_{\theta_j\sim \uni(\phi_j,r_{\rm full})}[\LC^2(\thv)]$ will be upper bounded by $\tfrac{r_{\rm full}}{r} \Ebb_{\theta_j\sim \uni(\phi_j,r)}[\LC^2(\thv)]$.
It then immediately follows that
    \begin{equation}
        \left(\frac{r_{\rm full}}{r}\right)^\nparams\Ebb_{\thv\sim\uni(\vec{\phi},r_{\rm full})}[\LC^2(\thv)]\geq \Ebb_{\thv\sim\uni(\vec{\phi},r)}[\LC^2(\thv)].
    \end{equation}

 \noindent Consequently, if $\Ebb_{\thv\sim\uni(\vec{\phi},r_{\rm full})}[\LC^2(\thv)]\in\order{1/b^n}$ for some $b>1$, then 
    \begin{equation}
          \Var_{\thv\sim\uni(\vec{\phi},r)}[\LC^2(\thv)]\in\order{\left(\frac{r_{\rm full}}{r}\right)^m\frac{1}{b^n}}.
    \end{equation}
   
Choosing
\begin{align}
  r \;>\;\frac{r_{\rm full}}{b^{\,m/n}}
\end{align}
makes
$\bigl(\tfrac{r_{\rm full}}{r}\bigr)^{m} < \tfrac{1}{b^n}$,
so the second moment (and hence the variance) decays as $\mathcal{O}\!\bigl(\tfrac{1}{\beta^n}\bigr)$ for some $\beta>1$.

Finally, if $m = cn$ for some constant $c>0$, then the condition 
\begin{align}
  r \;>\;\frac{r_{\rm full}}{b^{1/c}}
\end{align}
ensures that 
$\bigl(\tfrac{r_{\rm full}}{r}\bigr)^{\,c\,n} \le \tfrac{1}{b^n}$.
This completes the proof.
\end{proof}

\section{Counter example: Identity initialization can fail to have a large variance} \label{appx:iden-initialization}

\begin{proposition}
Consider a state learning task with a target state $\rho_{\rm target} = |1\rangle\langle 1|^{\otimes n}$, an initial state  $\rho_0 = |0\rangle\langle 0|^{\otimes n}$ and a tensor product ansatz $U(\thv) = \bigotimes_{i=1}^n U_i(\theta_i)$ where each $U_i(\theta_i)$ is a single-qubit rotation around the y-axis of the $i^{\rm th}$ qubit. Suppose that the parameters $\thv$ are initialized near $\mathbf{0}$ (the identity initialization) within a hypercube of width $2r$, with $r<1$. Then the variance of the loss function $\LC(\thv)$ in that region vanishes exponentially in $n$, i.e.,

\begin{align}
    \Var_{\thv \sim \uni(\vec{0},r)}[\LC(\thv)] \in \OC\left( \frac{1}{b^n}\right) 
    \quad\text{for some } b>1.
\end{align}

\end{proposition}
\begin{proof}
Because $\rho_{0}$ and $\rho_{\mathrm{target}}$ are product states, 
the fidelity under the tensor-product ansatz factorizes.  Specifically, the loss function becomes
\begin{align}
    \LC(\thv) & = 1 - \Tr\left[ U(\thv) \rho_0 U^\dagger(\thv) \rho_{\rm target}\right] \\
    & = 1 - \prod_{i=1}^n \sin^2 (\theta_i) \;.
\end{align} 
Then, we can explicitly compute the variance to be
\begin{align}
    \Var_{\thv \sim \uni(\vec{0},r)}[\LC(\thv)] & = \Var_{\thv \sim \uni(\vec{0},r)}[1 - \LC(\thv)] \\
    & \leq \Ebb_{\thv\sim \uni(\vec{0},r)}[(1 - \LC(\thv))^2] \\ 
    & =\left[\frac{3}{8} - \frac{1}{2}\cdot \left(\frac{\sin(2r)}{2r}\right) + \frac{1}{8} \cdot\left(\frac{\sin(4r)}{4r} \right)\right]^n\\
    & \leq \left( \frac{r^4}{5}\right)^n \;,
\end{align}
where the second inequality is due to a direct Taylor expansion and holds for $r<1$. 
\end{proof}
\end{document}